\documentclass[a4paper,11pt,fleqn]{book}

\usepackage[T1]{fontenc}
\usepackage[utf8]{inputenc}
\usepackage{fourier} % Utopia font-typesetting including mathematical formula compatible with newer TeX-Distributions (>2010)
\usepackage{setspace} % increase interline spacing slightly

\usepackage{graphicx,xcolor}
\graphicspath{{images/}}

\usepackage{subfig}
\usepackage{booktabs}
\usepackage{lipsum}
\usepackage{microtype}
\usepackage{url}
\usepackage[final]{pdfpages}

\usepackage{fancyhdr}

\fancypagestyle{plain}{
	\fancyhf{}

	\fancyfoot[EL,OR]{\thepage}}
\fancypagestyle{addpagenumbersforpdfimports}{
	\fancyhead{}
	
	\fancyfoot{}
	\fancyfoot[RO,LE]{\thepage}
}

\usepackage{listings}
\usepackage{hyperref}
\hypersetup{pdfborder={0 0 0},
	colorlinks=true,
	linkcolor=black,
	citecolor=black,
	urlcolor=black}
\makeatletter
\def\cleardoublepage{\clearpage\if@twoside \ifodd\c@page\else
    \hbox{}
    \thispagestyle{empty}
    \newpage
    \if@twocolumn\hbox{}\newpage\fi\fi\fi}
\makeatother \clearpage{\pagestyle{plain}\cleardoublepage}

\usepackage{color}
\usepackage{tikz}
\usepackage[explicit]{titlesec}
\newcommand*\chapterlabel{}
\titleformat{\chapter}[display]  % type (section,chapter,etc...) to vary,  shape (eg display-type)
	{\normalfont\bfseries\Huge} % format of the chapter
	{\gdef\chapterlabel{\thechapter\ }}     % the label 
 	{0pt} % separation between label and chapter-title
 	  {\begin{tikzpicture}[remember picture,overlay]
    \node[yshift=-8cm] at (current page.north west)
      {\begin{tikzpicture}[remember picture, overlay]
        \draw[fill=black] (0,0) rectangle(35.5mm,15mm);
        \node[anchor=north east,yshift=-7.2cm,xshift=34mm,minimum height=30mm,inner sep=0mm] at (current page.north west)
        {\parbox[top][30mm][t]{15mm}{\raggedleft $\phantom{\textrm{l}}$\color{white}\chapterlabel}};  %the black l is just to get better base-line alingement
        \node[anchor=north west,yshift=-7.2cm,xshift=37mm,text width=\textwidth,minimum height=30mm,inner sep=0mm] at (current page.north west)
              {\parbox[top][30mm][t]{\textwidth}{\color{black}#1}};
       \end{tikzpicture}
      };
   \end{tikzpicture}
   \gdef\chapterlabel{}
  } % code before the title body

\titlespacing*{\chapter}{0pt}{50pt}{30pt}
\titlespacing*{\section}{0pt}{13.2pt}{*0}  % 13.2pt is line spacing for a text with 11pt font size
\titlespacing*{\subsection}{0pt}{13.2pt}{*0}
\titlespacing*{\subsubsection}{0pt}{13.2pt}{*0}

\newcounter{myparts}
\newcommand*\partlabel{}
\titleformat{\part}[display]  % type (section,chapter,etc...) to vary,  shape (eg display-type)
	{\normalfont\bfseries\Huge} % format of the part
	{\gdef\partlabel{\thepart\ }}     % the label 
 	{0pt} % separation between label and part-title
 	  {\setlength{\unitlength}{20mm}
	  \addtocounter{myparts}{1}
	  \begin{tikzpicture}[remember picture,overlay]
    \node[anchor=north west,xshift=-65mm,yshift=-6.9cm-\value{myparts}*20mm] at (current page.north east) % for unknown reasons: 3mm missing -> 65 instead of 62
      {\begin{tikzpicture}[remember picture, overlay]
        \draw[fill=black] (0,0) rectangle(62mm,20mm);   % -\value{myparts}\unitlength
        \node[anchor=north west,yshift=-6.1cm-\value{myparts}*20mm,xshift=-60.5mm,minimum height=30mm,inner sep=0mm] at (current page.north east)
        {\parbox[top][30mm][t]{55mm}{\raggedright \color{white}Part \partlabel $\phantom{\textrm{l}}$}};  %the phantom l is just to get better base-line alingement
        \node[anchor=north east,yshift=-6.1cm-\value{myparts}*20mm,xshift=-63.5mm,text width=\textwidth,minimum height=30mm,inner sep=0mm] at (current page.north east)
              {\parbox[top][30mm][t]{\textwidth}{\raggedleft \color{black}#1}};
       \end{tikzpicture}
      };
   \end{tikzpicture}
   \gdef\partlabel{}
  } % code before the title body

\usepackage{amsmath}
\makeatletter
\def\resetMathstrut@{%
  \setbox\z@\hbox{%
    \mathchardef\@tempa\mathcode`\(\relax
      \def\@tempb##1"##2##3{\the\textfont"##3\char"}%
      \expandafter\@tempb\meaning\@tempa \relax
  }%
  \ht\Mathstrutbox@1.2\ht\z@ \dp\Mathstrutbox@1.2\dp\z@
}
\makeatother

 	\usepackage{amsthm, amssymb}
\usepackage{algorithmic}
\usepackage{algorithm}
 \DeclareMathOperator{\ND}{\mathcal{N}}
 \newcommand{\PR}{\mathbb{P}}
 \newcommand{\PQ}{\mathbb{Q}}
 \newcommand{\PU}{\mathbb{U}}
 \newcommand{\PV}{\mathbb{V}}
 \newcommand{\PS}{\mathbb{S}}
\newcommand{\I}{\mathbb{1}}
\newcommand{\E}{\mathbb{E}}
\DeclareMathOperator{\Var}{Var}
\DeclareMathOperator{\Cov}{Cov}
\DeclareMathOperator{\sgn}{sgn}
\newcommand{\R}{\mathbb{R}}

\newcommand{\N}{\mathbb{N}}
\newcommand{\wt}[1]{\widetilde{{#1}}}
\newtheorem{theorem}{Theorem}
\newtheorem{remark}[theorem]{Remark}
\newtheorem{condition}[theorem]{Condition}

\newtheorem{definition}[theorem]{Definition}
\DeclareMathOperator*{\esssup}{ess\,sup}
\DeclareMathOperator{\Tr}{Tr}
\DeclareMathOperator{\IQR}{IQR}
\DeclareMathOperator{\ord}{ord}
\DeclareMathOperator{\eig}{eig}

\DeclareMathOperator{\diag}{diag}
\newtheorem{lemma}[theorem]{Lemma}
\DeclareMathOperator{\Pois}{Pois}
\newcommand{\wh}[1]{\widehat{{#1}}}
\newcommand{\mc}[1]{\mathcal{{#1}}}

\DeclareMathOperator{\id}{id}
\DeclareMathOperator{\ce}{ce}
\DeclareMathOperator{\ic}{ic}

\DeclareMathOperator{\var}{var}

\DeclareMathOperator{\est}{est}
\DeclareMathOperator{\msq}{msq}
\DeclareMathOperator{\mgf}{mgf}
\DeclareMathOperator{\dist}{d}

\DeclareMathOperator{\Sym}{Sym}

\title{Adaptive importance sampling via minimization of estimators of cross-entropy, mean square, and inefficiency constant}
\author{Tomasz Badowski  \footnote{Supported by the Berlin Mathematical School and Freie Universit\"at Berlin.} \\[11pt]
\emph{Freie Universit\"{a}t Berlin}}

\begin{document} 
\maketitle 

\chapter*{Abstract}
\addcontentsline{toc}{chapter}{Abstract}
\vspace{2em}
The inefficiency of using an unbiased estimator in a 
Monte Carlo procedure can be quantified using an 
inefficiency constant, equal to the product of the variance of the estimator 
and its mean computational cost. 
We develop methods 
for obtaining the parameters of the importance sampling (IS) change of measure via 
single- and multi-stage minimization 
of well-known estimators of cross-entropy and the mean square of the IS estimator, 
as well as of new estimators of such a mean square and inefficiency constant. 
We prove the convergence and asymptotic properties 
of the minimization results in our methods. 
We show that if a zero-variance IS parameter exists, 
then, under appropriate assumptions, 
minimization results of the new estimators converge to such a parameter 
at a faster rate than such results of the well-known estimators, and a 
positive definite asymptotic covariance matrix of the minimization results of the cross-entropy estimators is four times such a matrix 
for the well-known mean square estimators. 
We introduce criteria for comparing the asymptotic efficiency of stochastic optimization methods, 
applicable to the minimization methods of estimators considered in this work. 
In our numerical experiments for computing expectations of functionals of an Euler scheme, the
minimization of the new estimators led to the lowest inefficiency constants and variances of the IS estimators, 
followed by the minimization of the well-known mean square estimators, and the cross-entropy ones. 
%\end{abstract} 
%\newpage
\tableofcontents

\chapter{Introduction}\label{secInt} 
%\addcontentsline{toc}{chapter}{Introduction} 
In this work we consider the problem of estimating an expectation of the form $\E_{\PQ_1}(Z)$, 
where $\PQ_1$ is a probability and $Z$ is a $\PQ_1$-integrable random variable. 
Such expectations are of interest in a variety of fields. For instance, they arise as prices of derivatives in mathematical 
finance \cite{glasserman2004monte}, as committors in molecular dynamics \cite{HartmannEntr14,Prinz_2011, Allen_2009}, and as probabilities 
of buffer overflow in telecommunications, system failure in dependability modelling, 
or ruin in insurance risk modelling \cite{asmussen2007stochastic}. 
The Monte Carlo (MC) method relies on approximating such an expectation using an average of independent replicates of 
$Z$ under $\PQ_1$. The inefficiency of the MC method can be quantified using an inefficiency constant, also known as a work-normalized variance 
% defined as the product of the variance of $Z$ under $\PQ_1$ and the mean cost of computing replicates of $Z$ under $\PQ_1$\mu
\cite{Glynn_1992, effimprGlynn94, rubino2009rare, badowski2011, badowski2013}. 
We discuss such constants and their interpretations in more detail in Chapter \ref{secIneff}. 
Efficiency improvement techniques (EITs) (the term having been proposed in \cite{Glynn_1992}) try to improve the efficiency 
of the estimation of the expectation of interest over the crude MC as above, e.g. 
by using some MC method with a lower inefficiency constant. 
Popular statistical EITs include control variates, importance sampling (IS), antithetic variables, and 
stratified sampling; see e.g. \cite{asmussen2007stochastic, effimprGlynn94}. 
Control variates method relies on generating in an MC method replicates of a control variates estimator, 
equal to the sum of $Z$ and a $\PQ_1$-zero-mean random variable, called a control variate \cite{asmussen2007stochastic, Szechtman2001}. 
In importance sampling (IS), 
for a probability $\PQ_2$, called an IS distribution, and a random variable $L$ such that $\E_{\PQ_2}(ZL)=\E_{\PQ_1}(Z)$, called an IS density, 
one computes in an MC method replicates of the IS estimator $ZL$, under $\PQ_2$. 
IS has found numerous applications among others to the computation of the expectations mentioned above and is a useful tool for 
rare-event simulation \cite{Glynn_2012,asmussen2007stochastic, ZhangWang2014, bucklew2004, Jourdain2009, Lemaire2010}. 
Adaptive EITs use the information from the random drawings available to make the estimation method more efficient, e.g. by tuning some 
parameter of the method from some set $A \subset \R^l$. 
For instance, in adaptive control variates one typically tunes the parameter in some parametrization of the control variates, while 
in IS --- in some parametrizations $b\to \PQ(b)$ of the IS distributions and $b\to L(b)$ of the IS densities. 
Adaptive IS and control variates can have a two-stage form, in the first stage of which an adaptive parameter as above is obtained and in the second 
a separate IS or control variates MC procedure is performed using this parameter. 
Typically in the literature adaptive control variates and IS 
have attempted to find a parameter optimizing (i.e. minimizing or maximizing) some function $f:A\to\overline{\R}$.  
Frequently, such a function was the variance or equivalently the mean square of the adaptive estimator and it was minimized; 
see e.g. \cite{Rubinstein97optimizationof,Jourdain2009,arouna04,Lemaire2010,Lapeyre2011}
for adaptive IS and \cite{Szechtman2001,Nelson_1990,KimH07} for control variates. 
We say that two functions $f_i$, $i=1,2$, are positively (negatively) linearly equivalent, if 
$f_1= af_2+b$ for some linear proportionality constant $a\in (0,\infty)$ ($a \in(-\infty,0)$) and $b \in \R$. 
In a number of adaptive IS approaches it was proposed to maximize a certain function negatively linearly equivalent to
the cross-entropy distance (also known as Kullback-Leibrer divergence) of the zero variance IS distribution (if it exists) 
from the IS distribution considered \cite{Rubinstein_optim,Rubinstein_2004, pupa2011}. 

We define cross-entropy to be a certain function of the IS parameter, 
positively linearly equivalent to the cross-entropy distance of the zero-variance IS distribution from 
the IS distribution considered, 
even though this name is sometimes used in the literature as a synonym of the cross-entropy distance \cite{Rubinstein_optim, cover2006elements}.
In addition to minimizing the mean square and such a cross-entropy, 
in this work we also minimize inefficiency constant. 
To our knowledge, it is the first time when inefficiency constant is being minimized for adaptive MC. 
One reason why many previous works focused on the minimization of variance rather than inefficiency constant may be that 
for some problems considered in these works the mean computation cost was approximately constant in the function of the adaptive parameter 
and thus the inefficiency constant and variance were approximately proportional. 
For instance, this is typically the case in parametric adaptive control variates and in parametric IS for many 
problems of derivative pricing in computational finance \cite{glasserman2004monte, Jourdain2009, Lemaire2010}. However, in 
numerous current and potential applications of IS in which the computation of a replicate of the IS estimator involves simulating 
a stochastic process until a random time, the mean cost typically depends on the IS parameter 
and the minimization of the variance and the inefficiency constant is no longer equivalent. 
This is for instance typically the case when performing IS for 
pricing knock-out barrier options in computational finance \cite{glasserman2004monte, Jourdain2009}. 
Further examples are provided by the molecular dynamics applications in which one is interested in computing 
expectations of various functionals of discretizations of diffusions considered 
until their exit time of some set; see e.g. \cite{ZhangWang2014,DupuisSpil2012} and our numerical experiments. 
See also \cite{Glynn_2012} and references therein for some examples from queueing theory and dependability modelling. 

Two types of stochastic optimization methods have typically been used in the literature 
for optimizing some functions $f$ as above. Methods of the first type are 
stochastic approximation algorithms. These are multi-stage stochastic optimization methods using stochastic gradient descent, 
in which estimates of the values of gradients of such $f$ are computed in each stage. 
See e.g. \cite{KimH07} for an application of such methods to variance minimization 
in adaptive control variates and \cite{arouna04, Lemaire2010, Lapeyre2011} in adaptive IS. 
One problem with such methods is that their practical performance heavily depends on the choice of step sizes, 
and some heuristic tuning of them may be needed to achieve a reasonable performance \cite{KimH07}. 
Stochastic optimization methods of the second type rely, in their simplest form, 
on the optimization of $b \to \wh{f}(b,\omega)$ for an appropriate random function 
$\wh{f}:A\times \Omega \to \R$ (where $(\Omega,\mc{F},\mc{P})$ is the default probability space and $\omega \in \Omega$ 
is an elementary event). The function $\wh{f}$ can be thought of as an estimator or a stochastic counterpart of $f$, 
and thus the methods from this class have been called stochastic counterpart methods, 
alternative names including sample path and sample 
average approximation methods \cite{HomemdeMello02rareevent, KimH07, HandbookMonteCarlo, Shapiro2003}. 
See Chapter 6, Section 9, in \cite{Shapiro2003} for a historical review 
of such methods, related to M-estimation and in particular maximum likelihood estimation in statistics \cite{Vaart}. 
The most well-known example of an application of the 
stochastic counterpart method to efficiency improvement are linearly parametrized control variates 
\cite{asmussen2007stochastic,Szechtman2001,Nelson_1990}, in which to obtain the control variates parameter one 
minimizes the sample variance of the control variates estimator by solving a certain system of linear equations. 
See \cite{Rubinstein97optimizationof,Rubinstein_optim,Rubinstein_2004,Jourdain2009} for 
applications of the stochastic counterpart method to adaptive IS and \cite{KimH07} for an application to nonlinearly parametrized 
control variates. In some works on adaptive IS 
it was proposed to perform a multi-stage stochastic counterpart method (as opposed to the single-stage one as above), 
in which the optimization result from a given stage is used to construct 
the estimator optimized in the subsequent stage \cite{Rubinstein97optimizationof,Rubinstein_2004}. 
As discussed heuristically in Section 2 in \cite{Rubinstein97optimizationof}, such 
an approach may be better than the single-stage one 
because the asymptotic distribution of the optimization results of the estimators 
from its final stage may be 
less spread than when using some default estimators in the single-stage case. 

In this work we investigate single- and multi-stage 
stochastic counterpart methods minimizing some well-known estimators of
mean square \cite{Rubinstein97optimizationof, Jourdain2009} and 
cross-entropy \cite{Rubinstein_optim, Rubinstein_2004}, 
as well as newly proposed estimators of mean square and inefficiency constant. 
In our theoretical analysis we focus on the parametrizations of IS obtained
via exponential change of measure (ECM) and via linearly parametrized exponential tilting for Gaussian stopped sequences (LETGS). 
Using IS in some special cases of the ECM and LETGS settings 
has been demonstrated to lead to significant variance reductions e.g. 
in rare event simulation \cite{bucklew2004,asmussen2007stochastic} 
and when pricing options in computational finance \cite{Jourdain2009, Lemaire2010}. 
We provide sufficient and in some cases also necessary assumptions under which 
there exist unique minimum points of the cross-entropy 
and mean square as well as of their estimators in the ECM and LETGS settings 
and we give some sufficient conditions for these assumptions to hold in the Euler scheme case. 
It is well known that for some important parametrizations of IS 
the minimum points of the cross-entropy estimators can be found exactly, 
which makes these estimators more convenient to minimize than the well-known mean square estimators, 
for the minimization of which one typically uses some iterative methods. This is for instance the case 
in some special cases of the ECM setting, in IS for finite support distributions (see examples 3.5 and 3.6 in \cite{Rubinstein_2004}),
and when using the Girsanov transformation with a linear parametrization of IS drifts for diffusions \cite{ZhangWang2014}. 
We show that this is also the case in the LETGS setting. 

An important contribution of this work is the definition of versions of single- and multi-stage minimization methods of the above estimators
in the ECM and LETGS settings whose results enjoy appropriate strong convergence and asymptotic properties 
in the limit of the increasing 
budget of the single-stage minimization or the increasing number of stages of the multi-stage minimization. 
To ensure such properties of the multi-stage methods we use increasing numbers of simulations in 
the consecutive stages and projections of the minimization results onto some bounded sets. Furthermore, 
in the proofs we apply a new multi-stage strong law of large numbers. 
For the cross-entropy  estimators we consider their exact minimization utilising formulas for their minimum points,
and we prove the a.s. convergence of their minimization results to the unique minimum point of cross-entropy. 
We show that the well-known  mean square estimators in both settings and the new mean square estimators in the ECM setting are convex 
and we prove the a.s. convergence of the results of their minimization with gradient-based stopping criteria 
to the unique minimum point of mean square. 
For the new mean square estimators in the 
LETGS setting and the ones of the inefficiency constant in the ECM setting for a constant computation cost, 
we prove the a.s. convergence of their minimization results to the unique minimum point of the mean square 
when using the following two-phase minimization procedure. 
In its first phase some convex estimator of the mean square as above can be minimized, and then, using its minimization result as a starting point, 
one can carry out a constrained minimization of the considered estimator or an unconstrained minimization but of an appropriately modified such estimator. 
For the inefficiency constant estimators in the LETGS setting 
we propose a more complicated three-phase minimization procedure with gradient-based stopping criteria, 
the first phase of which can be as above. We prove the convergence of the minimization results in such a procedure 
to the set of the first-order critical points of the inefficiency constant which have not higher values of the inefficiency 
constant than in the minimum point of the variance, or even by at least some positive constant lower 
such values if the gradient of the inefficiency constant in the minimum point of the variance does not vanish.

Using the theory of the asymptotic behaviour of minimization results of random functions from \cite{Shapiro1993}, we develop 
such a theory for the minimization results of such functions when using gradient-based stopping criteria. 
We use it for proving the asymptotic properties of the single- and multi-stage minimization methods of the estimators as above. 
To our knowledge, previously in the literature only the
strong convergence and asymptotic properties of the single-stage minimization of the well-known mean square estimators were proved in
\cite{Jourdain2009}, but only in the limit of the increasing number of simulations, 
in the ECM setting for normal random vectors, under stronger integrability assumptions than in our work, and using 
exact minimization which cannot be implemented in practice as opposed to the minimization with gradient-based stopping 
criteria considered in this work. 

Another important contribution of this work is the definition of the
first- and second-order criteria for comparing 
the asymptotic efficiency of certain stochastic optimization methods for the minimization of a given function. 
A method more efficient in the first-order sense 
leads to lower values of the minimized function in the minimization results by at least a fixed positive constant 
with probability going to one as the budget of the method increases. The second-order asymptotic efficiency of the minimization  methods 
in which such values  converge in probability to the same constant 
can be quantified using some parameters, like the means, of some second-order asymptotic distributions of such values around such a constant. 
We apply such criteria to comparing the asymptotic efficiency 
of the single- and multi-stage minimization methods of the estimators discussed above. For these methods, the means of the distributions as above can be potentially 
estimated and adaptively minimized. 

We show that if $\PQ_1(Z\neq0)>0$ then there exists a unique IS distribution leading to the lowest variance of the IS estimator,
which we call the optimal-variance one.  %TODO write later what others did 
If additionally $Z\geq 0$, $\PQ_1$ a.s., then the optimal-variance IS distribution leads to a zero-variance IS estimator.  
IS parameters leading to such distributions are called optimal-variance or zero-variance ones respectively. 
We show that if there exists an optimal-variance IS parameter for the new mean square estimators or 
a zero-variance one for the 
inefficiency constant estimators, then under appropriate assumptions a.s. the minimization results of the exact 
single- and multi-stage minimization of such estimators are equal to such respective parameters
for a sufficiently large simulation budget used. Furthermore, for the single- or multi-stage minimization
of these estimators with gradient-based stopping criteria we can have a faster rate of convergence 
of the minimization results to such parameters than for the well-known estimators. 
We also show that if there exists a zero-variance IS parameter, 
then, under appropriate assumptions,  the asymptotic covariance matrix of the minimization results of 
the cross-entropy estimators is positive definite 
and is four times such a matrix for the well-known mean square estimators. 

We provide an analytical example in which  
all possible relations between the asymptotic variances (i.e. equalities and both strict inequalities) 
of the minimization results of different types of estimators converging to the same point are 
achieved for different parameters of the example, except that using the cross-entropy 
estimators always leads to not lower asymptotic variance than using the well-known mean square estimators. 

In our numerical experiments we consider an Euler scheme discretization of a diffusion in a potential. 
We address the problem of estimating the moment-generating function (MGF) of the exit time of such an Euler scheme of a domain, the probability to exit it 
by a fixed time, and the probabilities to leave it through given parts of the boundary, called committors.
Such quantities are of interest e.g. in molecular dynamics applications; see \cite{Dupuis2012,HartmannShuette2012, ZhangWang2014,
HartmannEntr14,Prinz_2011, Allen_2009}. 
We use IS in the LETGS setting, for which under the IS distribution we receive again 
an Euler scheme but this time with an additional drift depending on the IS parameter, 
called an IS drift. 
For the estimation of the above quantities we use a two-stage method as discussed above,
in the first stage of which to obtain the IS parameter we use 
simple multi-stage minimization of various estimators. % of estimators without various enhancements like 
%projections or function modifications needed to ensure appropriate convergence properties of the methods as discussed above. 
In our numerical experiments, 
the minimization of the new estimators of inefficiency constant and 
mean square led to the lowest variances and inefficiency constants of the IS estimators, followed by the minimization of the well-known mean square 
estimators, and of the cross-entropy ones. 
In one case, the minimization of the inefficiency constant estimators outperformed the minimization 
of the new mean square estimators by arriving at a
lower mean cost and a higher variance but so that their product, equal to the inefficiency constant, was lower. 
% When computing expectations of nonincreasing functions 
% of the exit time, like MGF and probability to leave the domain before a fixed time,
% the mean exit times under the IS parameters obtained from  were significantly smaller than for CMC. 
% An intuition behind this result is provided by a theorem, which we prove, that for $Z$ being nonincreasing 
% functions of the computation cost variable, under the zero-variance IS distribution we should have not higher mean cost than 
% for CMC, and under some additional assumptions this inequality is sharp. 
The variances and inefficiency constants of the adaptive IS estimators in our experiments strongly depended on the 
parametrization of the IS drifts used and could be reduced by adding appropriate positive constants to the variables $Z$ as above.
For a committor we also performed experiments comparing the spread of the IS drifts obtained from single-stage minimization, 
which yielded results qualitatively and quantitatively close to the 
case when a zero-variance IS parameter exists as discussed above. 
We provide some intuitions supporting the observed results. 

\chapter{\label{secIneff}Monte Carlo method and inefficiency constant} 
Let us further in this work denote $\N_p=\{p,p+1,\ldots\}$, $\N=\N_0$, $\N_+=\N_1$, and $\R_+=(0,\infty)$.
For a set $A \in \mc{B}(\R^l)$ for some $l \in \N_+$, or $A \in \mc{B}(\overline{\R})$ (where $\mc{B}(B)$ is the Borel $\sigma$-field on $B$),
the default measurable space which we shall consider on it is $(A,\mc{B}(A))$, further denoted simply as $\mc{S}(A)$. 
Consider a probability $\PQ_1$ on a measurable space $\mc{S}_1=(\Omega_1,\mc{F}_1)$ 
and let $Z$ be an $\R$-valued random variable on $\mc{S}_1$ (i.e. a measurable function from $\mc{S}_1$ to $\mc{S}(\R)$), 
such that $\E_{\PQ_1}(|Z|)<\infty$. We are interested in the estimation of $\alpha:=\E_{\PQ_1}(Z)$.
The above defined quantities shall be frequently used further in this work. 
%We say that such $Z$ is an unbiased estimator of $\alpha$ under $\PQ_1$. 
In the Monte Carlo (MC) method, for some $n \in \N_+$, one approximates $\alpha$ using an MC average $\wh{\alpha}_n:=\frac{1}{n}\sum_{i=1}^nZ_i$ 
of independent random variables $Z_i$, $i=1,\ldots n$, each having the same distribution as 
$Z$ under $\PQ_1$, shortly called independent replicates of $Z$ under $\PQ_1$. 
Variance of $\wh{\alpha}_n$ measures its mean squared error of approximation of $\alpha$, and for $\var: = \Var_{\PQ_1}(Z)$ we have 
$\Var(\wh{\alpha}_n)=\frac{\var}{n}$. 

When performing an MC procedure on a computer it is often the case that there exists a 
nonnegative random variable $\dot{C}$ 
on $\mc{S}_1$ such that for generated independent replicates $(Z_i,\dot{C}_i)$, $i=1,\ldots n$, of $(Z,\dot{C})$ under $\PQ_1$, 
$\dot{C}_i$ are typically approximately equal to some practical costs, 
like computation times, needed to generate $Z_i$. We call such $\dot{C}$ a practical cost variable (of an MC step). 
Often  we have 
$\dot{C}=p_{\dot{C}}C$ for some $p_{\dot{C}}\in \R_+$, which may be different for different computers and implementations (shortly, for different 
practical realizations) 
considered and a random variable $C$ on $\mc{S}_1$, called a theoretical cost (of an MC step), which is common for these practical realizations. 
In case when the practical costs of generating $Z_i$ are approximately constant, one can take  $C=1$. 
A random $C$ can be e.g. 
the internal duration time of a  stochastic process from which $Z$ is computed, like its hitting time of some set. 
For instance, when pricing knock-out barrier options in computational finance using the MC method \cite{glasserman2004monte, Jourdain2009} as
such $C$ one can typically take the minimum of the hitting time of the asset of the barrier and the expiry date of the option. 
We define a mean theoretical cost $c=\E_{\PQ_1}(C)$ and a theoretical inefficiency constant $\ic=c\var$ (whenever 
this product makes sense, i.e. when we do not multiply zero by infinity in it), 
and the practical ones $\dot{c}=\E_{\PQ_1}(\dot{C})=p_{\dot{C}}c$ and $\dot{\ic}=\dot{c}\var=p_{\dot{C}}\ic$.
For $\dot{c}$ and $\var$ finite, practical inefficiency constants are reasonable measures of the
inefficiency of MC procedures as above, i.e. higher such constants imply lower efficiency. 
The name inefficiency constant was coined in \cite{badowski2011, badowski2013}, 
while in some other works such a constant was called a work-normalized variance \cite{rubino2009rare}. 
However, the idea of using a reciprocal of a practical inefficiency constant to quantify the efficiency of MC methods was conceived much earlier,  
see \cite{Glynn_1992} for a historical review. 
Glynn and Whitt \cite{Glynn_1992} 
proposed more general criteria for quantifying the asymptotic efficiency of simulation estimators using asymptotic efficiency rates and values and 
the above practical inefficiency constant is equal to the reciprocal of their efficiency value in the 
special case of an MC method, in which the efficiency rate equals $\frac{1}{2}$.
See \cite{effimprGlynn94}, Section 10 of Chapter 3 in \cite{asmussen2007stochastic}, or Section 1.1.3
in \cite{glasserman2004monte} for accessible descriptions of their approach in the special case of MC methods. 
%TODO

Further on in this chapter we provide some interpretations of inefficiency constants, both from the literature and new ones, 
justifying their utility for quantifying the inefficiency of MC procedures. 
The theorems introduced in the process will be frequently used further on in this work. 
 We focus on theoretical inefficiency constants 
 (often dropping further the word theoretical), but analogous interpretations hold also for the practical ones. 

The following interpretation of inefficiency constants was given in Section 2.6 in \cite{badowski2013}. 
The ratio of positive finite inefficiency constants $\ic$ of different sequences of MC procedures as above (indexed 
by the numbers $n$ of replicates used in them)
is equal to the limit of ratios of their mean costs $n_{\epsilon}c$ corresponding 
to the minimum numbers of replicates $n_{\epsilon}=\lceil\frac{\var}{\epsilon}\rceil$ 
needed to reduce the variances $\frac{\var}{n_{\epsilon}}$ of the MC averages $\wh{\alpha}_n$ below a given threshold $\epsilon>0$ for $\epsilon \rightarrow 0$. 

Consider a function $f:\R_+^2\to\R_+$ such that for each $x,y \in \R_+$, $f(x,y)=f(y,x)$ and for each $a\in \R_+$, $af(x,y)=f(ax,ay)$. Let 
$g:\R_+^2\to[0,\infty)$ be such that $g(x,y)=\frac{|x-y|}{f(x,y)}$, so that $g(x,y)=g(y,x)$ and $g(ax,ay)=g(x,y)$, $a \in \R_+$. 
For instance, 
$f(x,y)$ can be equal to $\max(x,y)$, $\min(x,y)$, or $\frac{x+y}{2}$, in which case $g(x,y)$ can be interpreted 
as the relative difference of $x$ and $y$. For some $\delta>0$, we say that $x,y \in \R_+$ are 
$\delta$-approximately equal, which we denote as $x\approx_{\delta}y$, if $g(x,y)\leq \delta$. Note that 
$x\approx_{0}y$ implies that $x=y$. 
%for $\delta=0$, $\delta$-approximate equality is the ordinary equality. 
The below simple interpretations of inefficiency constants were given in sections 1.9 and 2.6 of \cite{badowski2013}
in the special case of $f=\min$ as above. 
For two MC procedures for estimating $\alpha$, one like above using $n$ replicates and an analogous primed one, assuming that $\ic, \ic' \in \R_+$, 
from an easy calculation we have
\begin{equation}
g(\frac{\frac{\var}{n}}{\frac{\var'}{n'}},\frac{\ic}{\ic'})= g(nc,n'c')
\end{equation}
and
\begin{equation}
g(\frac{nc}{n'c'},\frac{\ic}{\ic'})= g(\frac{\var}{n},\frac{\var'}{n'}). 
\end{equation}
In particular, the ratio of positive finite inefficiency constants 
of these procedures is $\delta$-approximately equal to the ratio of the
variances of their respective MC averages 
for $\delta$-approximately equal respective mean total costs 
and it is also $\delta$-approximately equal to the ratio of their average costs for $\delta$-approximately equal 
variances of their MC averages. 

% For some $\delta\geq 0$, we say that $x, y>0$ are $\delta$-approximately equal if 
% $\frac{|x-y|}{\min(x,y)}\leq \delta$. Note that for $\delta=0$, $\delta$-approximate equality is the ordinary equality. 
%Generalize 
%provides the following simple interpretations of inefficiency constants.  
%Indeed, if $c\approx_{\delta}c'$ 
%let us check the first claim, and let the quantities 
% For some nonnegative random variable $\wt{C}$ on $\mc{S}_1$, %(we will further be most interested in cases $\wt{C}=C$ and $\wt{C}=1$)
% let $\wt{c}=\E_{\PQ_1}(\wt{C})$, $\wt{\ic}=\var\wt{c}$, so that in case $\wt{C}=C$ we have 
% $\wt{c}=c$ and $\wt{\ic}=\ic$ and in case $\wt{C}=1$, $\wt{c}=1$ and $\wt{\ic}=\var$. 

Let $(Z_i,C_i),\ i \in \N_+$, be independent replicates of $(Z,C)$ under $\PQ_1$. 
Before providing further interpretations of inefficiency constants, let us recall some basic
facts about MC procedures as above. 
From the strong law of large numbers (SLLN), for $\wh{\alpha}_n=\frac{1}{n}\sum_{i=1}^nZ_i$ it holds a.s.
%\begin{equation}
$\lim_{n\to\infty}\wh{\alpha}_n =\alpha$,
%\end{equation}
and if $\var<\infty$, then 
from the central limit theorem (CLT), $\sqrt{n}(\wh{\alpha}_n-\alpha)\Rightarrow \ND(0,\var)$. Consider the following sample variance estimators
\begin{equation}\label{varN} 
\wh{\var}_n = \frac{n}{n-1}(\frac{1}{n}\sum_{i=1}^kZ_i^2- \wh{\alpha}_n^2),\quad n\in \N_2.
\end{equation}
If $\var <0$, then from the SLLN  a.s. $\lim_{n \to \infty} \wh{\var}_n=\var$ and if further $\var>0$, then from Slutsky's
lemma (see e.g. Lemma 2.8 in \cite{Vaart})
\begin{equation}\label{CLTv}
\sqrt{\frac{n}{\wh{\var}_n}}(\wh{\alpha}_n-\alpha) \Rightarrow \ND(0,1),
\end{equation}
which can be used to construct asymptotic confidence intervals for $\alpha$, as discussed e.g. in Chapter 3, 
Section 1 in \cite{asmussen2007stochastic}. 

For $n \in \N_+$, let $\wh{c}_n=\frac{1}{n}\sum_{i=1}^nC_i$ and % which for $\wt{C}=C$ is estimator of mean cost 
%and for $\wt{C}=1$, $\wh{c}_n=1$. 
for $n \in \N_2$, let 
$\wh{\ic}_n=\wh{c}_n\wh{\var}_n$. %$, which for $\wt{C}=C$ is estimator of inefficiency constant and for $C=1$, 
%$\wh{\ic}_n=\wh{\var}_n$.
Assuming that $c,\var<\infty$, from the SLLN, a.s. $\lim_{n\to \infty}\wh{c}_n= c$ and 
$\lim_{n\to \infty}\wh{\ic}_n= \ic$.
%For $\omega \in \Omega$, consider
Let $S_n = \sum_{i=1}^nC_i$, $n \in \N$ (in particular $S_0=0$),  
so that  $S_n$ is the cost of generating the first $n$ replicates of $Z$. %For $t \in [0,\infty)$, 
%consider $\N$-valued random variables
%\begin{remark}
For $t \in \R_+$, consider 
\begin{equation}\label{nt1}
N_t=\sup\{n \in \N: S_n\leq t\}, 
\end{equation}
or
\begin{equation}\label{nt2}
N_t=\inf\{n \in \N: S_n\geq t\}.
\end{equation}
The above defined $N_t$ are reasonable choices of the numbers of simulations to perform if we want to spend an approximate 
total budget $t$ (like e.g. some internal simulation time) on the whole MC procedure. Definition (\ref{nt1}) ensures that we do not exceed 
the budget $t$. Under definition (\ref{nt2}) we let ourselves
finish the last computation started before the budget $t$ is exceeded and thus we do not waste 
the computational effort already invested in it. Note that under (\ref{nt2}) we have $N_t>0$, $t \in \R_+$,
which does not need to be the case under (\ref{nt1}). 
If $C<\infty$, $\PQ_1$ a.s., then a.s. $C_i<\infty$, $i\in \N_+$, and thus under both definitions a.s.
\begin{equation}\label{ntoinfty}
\lim_{t\to \infty}N_t=\infty.
\end{equation}
For some subset $A$ of some set $D$ we denote $\I_A$ or $\I(A)$ the indicator function of $A$, i.e. a function equal to 
one on $A$ and to zero on $D\setminus A$.
For a real-valued random variable $Y$ we denote $Y_+=Y\I(Y>0)$ and $Y_-=-Y\I(Y<0)$. 
We have the following well-known slight generalization of the ordinary SLLN (see the corollary on page 292 in \cite{billingsley1979}).
\begin{theorem}\label{thSLLNG} 
 If an $\overline{\R}$-valued random variable $Y$ is such that  $\E(Y_{-})<\infty$, 
 then for $Y_1,Y_2, \ldots$, i.i.d. $\sim Y$, a.s. 
 \begin{equation} 
\frac{1}{n}\sum_{i=1}^nY_i \to \E(Y) \in \R\cup\{\infty\}. 
 \end{equation} 
\end{theorem}
Let $c>0$ (in particular we can have $c=\infty$). Then, from the above lemma a.s. $\lim_{n\to \infty}\frac{S_n}{n}= c$ and thus 
\begin{equation}\label{sntoinfty} 
\lim_{n\to \infty}S_n = \infty, 
\end{equation} 
so that under both definitions a.s. 
\begin{equation}\label{nt0} 
N_t < \infty,\quad t \geq 0. 
\end{equation} 
From renewal theory (see Theorem 5.5.2 in \cite{chung2001course}),  under definition (\ref{nt1}) we have  a.s. 
\begin{equation}\label{nttc} 
\lim_{t\to \infty}\frac{N_t}{t}=\frac{1}{c}. 
\end{equation} 
Since, marking $N_t$ given by (\ref{nt2}) with a prim, we have
$N_t\leq N_t'\leq N_t +1$, (\ref{nttc}) holds also when using definition (\ref{nt2}). 

Let us further consider general $\N\cup\{\infty\}$-valued random variables $N_t$, $t \in \R_+$. 
%Sometimes we will need to consider the case of many-dimensional random variables. 
Let $m \in \N_+$ and $Y$ be an $\R^m$-valued random vector such that $\E(Y_{i}^2)<\infty$, $i=1,\ldots,m$, with mean 
$\mu=\E(Y)$ and covariance matrix $W=\E((Y-\mu)(Y-\mu)^T)$. Let $X_1,X_2,\ldots$ be i.i.d. $\sim Y$. 
Let $\wh{\mu}_n=\frac{1}{n}\sum_{i=1}^nX_i$, $n \in \N_+$. 
For the string $\lambda$ substituted by each of the strings $\alpha$, $\var$, $\ic$, $c$, and $\mu$, 
for $p=2$ for $\lambda$ substituted by $\var$ or $\ic$ and for $p=1$ otherwise, 
consider an estimator $\wt{\lambda}_t$ of $\lambda$ corresponding to the total budget $t \in \R_+$ and with an initial value
$\lambda_0$, where $\lambda_0 \in \R^m$ for $\lambda$ substituted by $\mu$ and $\lambda_0\in \R$ otherwise, defined as follows 
\begin{equation}\label{lambdatdef} 
\wt{\lambda}_t= \wh{\lambda}_{N_t}\I(N_t\in \N_p) +\lambda_0\I(N_t\notin \N_p). 
\end{equation} 
%Note that under definition (\ref{nt2}) we have $\wt{\lambda}_t= \wh{\lambda}_{N_t}\I(\N_t<\infty)$, $t>0$. 
% \begin{condition}\label{condtau} 
% \end{condition} 
We shall need the following trivial remark.
\begin{remark}\label{remtaua}
For each $k \in \N_+,$ let $\tau_k$ be an a.s. $\N$-valued random variable (i.e. $\PR(\tau_k \in \N) =1$) and let 
a.s. $\lim_{k\to \infty}\tau_k=\infty$.  
Let further $a_k$, $k \in \N$, be random variables such that a.s. $\lim_{k\to \infty}a_k=a$. 
%(or for some set $B$ a.s. $a_k \in B$ for a sufficiently large $k$), 
Then, a.s. $\lim_{k\to \infty} \I(\tau_k \in \N)a_{\tau_k}=a$.  %(a.s. $a_{\tau_k} \in B$ for a sufficiently large $k$). 
\end{remark}
When we have a.s. (\ref{ntoinfty}), (\ref{nt0}), and for some $\lambda$ as above, a.s. $\lim_{n\to \infty}\wh{\lambda}_n=\lambda$, then  
%under both definitions above, 
from Remark \ref{remtaua}, a.s.
\begin{equation}\label{wtlambda}
\lim_{t\to\infty} \wt{\lambda}_t=\lambda.
\end{equation}

%Consider the following helper theorem, whose one dimensional version is often referred to as Anscombe's theorem in the literature. 
\begin{lemma}\label{lemAnsc}
Let $a_n$, $n\in \N_+$, be $\N_+$-valued random variables such that for some  $t_n\in \R_+$, $n \in \N_+$, such that $\lim_{n\to\infty}t_n= \infty$,
for some $b \in \R_+$, we have 
%\begin{equation}
$\frac{a_n}{t_n} \overset{p}{\to} b$.
%\end{equation}
Then, 
\begin{equation}
\frac{\sum_{i=1}^{a_n}(X_i-\mu)}{\sqrt{a_n}} \Rightarrow \ND(0,W). 
\end{equation}
\end{lemma}
\begin{proof} 
Using Cram\'{e}r-Wold device (see page 16 in \cite{Vaart}) it is sufficient to consider the case of $m=1$, which let us assume. 
For $W=0$ we have a.s. $X_i=\mu$, $i \in \N_+$, so that the thesis is obvious. 
The general case with $W>0$ can be easily inferred from the special case in which $\mu=0$ and $W=1$, 
which can be proved analogously as Theorem 7.3.2 in \cite{chung2001course}.  
%but using $a_n=N_{t_n}$ and in appropriate places $t_n$ in the place of $n$ in that proof. 
\end{proof} 

%For the general $N_t$, $t \in \R_+$, as above, 
Consider the following condition 
(which for $c \in \R_+$ follows e.g. from (\ref{nttc}) holding a.s.). %, implying (\ref{ntoinfty}) and (\ref{nt0}). 
\begin{condition}\label{condnt}
It holds $c \in \R_+$ and 
\begin{equation}
\frac{N_t}{t}\overset{p}{\to}\frac{1}{c}.
\end{equation}
\end{condition}

For $a, b \in \overline{\R}$, by 
$a\wedge b$ we denote their minimum and 
$a\vee b$ --- their maximum. 
\begin{theorem}\label{theConvIneff}
Under Condition \ref{condnt} we have
\begin{equation}\label{astoic1}
\sqrt{t}(\wt{\mu}_t -\mu)\Rightarrow \ND(0,cW).
\end{equation}
\end{theorem}
\begin{proof}
For each $t \in \R_+$, let $M_t=(\I(N_t\neq \infty)N_t) \vee 1$, 
which is an $\N_+$-valued random variable, equal to $N_t$ when $N_t \in \N_+$. 
From Condition \ref{condnt}, it holds 
\begin{equation}\label{pntinn}
\lim_{t\to\infty}\PR(M_t=N_t\in \N_+)= 1 
\end{equation}
and thus
\begin{equation}
\frac{M_t}{t}\overset{p}{\to}\frac{1}{c}. 
\end{equation}
Thus, from Lemma \ref{lemAnsc}
\begin{equation}\label{asympsig}
R_t:=\sqrt{M_t}(\wh{\mu}_{M_{t}} -\mu)\Rightarrow \ND(0,W).
\end{equation}
% Let $T_t\simeq R_t$ mean that $T_t-R_t\overset{p}{\to}0$ as $t\to \infty$, so that 
% from Slutsky's lemma, for some probability $\nu$, $T_t \Rightarrow \nu$ only if $R_t \Rightarrow \nu$. 
Let $\wt{R}_t=\I(N_t\in \N_+)R_t=\I(N_t\in \N_+)\sqrt{N_t}(\wt{\mu}_{t} -\mu)$.
From (\ref{pntinn}),
%\end{equation}
%\begin{equation}\label{asympeq}
%\begin{split}
$R_t-\wt{R}_t \overset{p}{\to}0$. 
%\end{split}
%\end{equation}
Therefore, from  (\ref{asympsig}) and Slutsky's lemma, $\wt{R}_t \Rightarrow \ND(0,W)$, and thus 
\begin{equation}\label{wtras}
\sqrt{c}\wt{R}_t \Rightarrow \ND(0,cW).  
\end{equation}
Let $G_t=\sqrt{t}(\wt{\mu}_t - \mu)$ and $\wt{G}_t=\I(N_t\in \N_+)G_t$. Then, 
$G_t-\wt{G}_t \overset{p}{\to}0$, so that to prove
(\ref{astoic1}) it is sufficient to prove that 
\begin{equation}\label{gtas}
\wt{G}_t\Rightarrow \ND(0,cW).  
\end{equation}
From (\ref{pntinn}), the continuous mapping theorem, and Slutsky's lemma, 
$S_t:=\I(N_t \in \N_+)\sqrt{\frac{t}{cN_t}}\overset{p}{\to}1$.
Thus, (\ref{gtas}) follows from (\ref{wtras}) and the fact that from Slutsky's lemma
\begin{equation}
\sqrt{c}\wt{R}_t-\wt{G}_t= \sqrt{c}\wt{R}_t(1-S_t)\overset{p}{\to}0.
\end{equation}
\end{proof}

In the below theorem and remark we extend 
the interpretations of inefficiency constants provided at the beginning of 
Section 10, Chapter 3 in \cite{asmussen2007stochastic} (see also \cite{effimprGlynn94} and Example 1 in \cite{Glynn_1992}). 

\begin{theorem}\label{thcostCLT}
If $\var<\infty$ and Condition \ref{condnt} holds, then
\begin{equation}\label{astoic}
\sqrt{t}(\wt{\alpha}_t -\alpha)\Rightarrow \ND(0,\ic). 
\end{equation}
If we further have $\var>0$ and a.s. (\ref{ntoinfty}) and (\ref{nt0}), then  
\begin{equation}\label{atvt}
\sqrt{\frac{t}{\wt{\ic}_t}}(\wt{\alpha}_t -\alpha)\Rightarrow \ND(0,1).
\end{equation}
\end{theorem}
\begin{proof}
Formula (\ref{astoic}) follows immediately from Theorem \ref{theConvIneff} and (\ref{atvt}) follows 
from (\ref{astoic}), (\ref{wtlambda}) holding a.s. for $\lambda =\ic$, and Slutsky's lemma.
\end{proof}

% For $X\sim \ND(0,1)$ and $p \in [0,\infty)$, let 
% %\begin{equation}\label{psipdef} 
%$\Psi(p)= \PR(|X|\leq p)$ and for $\beta \in (0,1)$, let 
%\end{equation} 
%Then, for $\gamma \in (0,1)$, $\Psi(z_{\beta})= 1-2\beta$. 

\begin{remark}\label{remConfIC}
Let $X \sim \ND(0,1)$ and let for $\beta \in (0,1)$, $z_{\beta}$ be the $\beta$-quantile of the normal distribution, i.e. $\PR(X \leq z_{\beta})=\beta$.
Let $\gamma \in (0,1)$ and $p_{\gamma}= z_{1-\frac{\gamma}{2}}$, so that  $\PR(|X|\leq p_{\gamma})=1-\gamma$.
Assuming (\ref{atvt}), for the random interval 
$I_{\gamma,t}=(\wt{\alpha}_t- p_{\gamma}\sqrt{\frac{\wt{\ic}_t}{t}},\wt{\alpha}_t +p_{\gamma}\sqrt{\frac{\wt{\ic}_t}{t}})$ we have 
\begin{equation}\label{aAsympInt}
\lim_{t \to \infty}\PR(\alpha \in I_{\gamma,t})=\PR(|X|\leq p_{\gamma})=1-\gamma,
\end{equation}
i.e. $I_{\gamma,t}$ is an asymptotic $1-\gamma$ confidence interval for $\alpha$. 
It follows that $\wt{\ic}_t$ and $\wt{\alpha}_t$ can play the same role when constructing the asymptotic confidence intervals for 
$\alpha$ for $t \to \infty$, as 
 $\wh{\var}_n$ and $\wh{\alpha}_n$ do for $n\to \infty$ as discussed below (\ref{CLTv}). 
For $C=1$, both approaches to constructing the asymptotic confidence intervals are equivalent. 
\end{remark} 

\chapter{\label{secImp}Importance sampling} 
\section{\label{secBackDens}Background on densities} 
Consider a measurable space $\mc{S}=(D,\mc{D})$, 
let $\mu_1$ and $\mu_2$ be measures on $\mc{S}$, and let $A \in \mc{D}$. We say 
that $\mu_1$ has a density $L$ (also a called Radon-Nikodym derivative) 
with respect to $\mu_2$ on $A$, which we denote as 
$L=(\frac{d\mu_1}{d\mu_2})_A$, if $L$ is a measurable function from $\mc{S}$ to $\mc{S}(\overline{\R})$ 
such that for each $B \in\mc{D}$, $\mu_1(A\cap B) =  \int\! L\I(A\cap B)\, \mathrm{d}\mu_2$.
If $L=(\frac{d\mu_1}{d\mu_2})_A$, then for each measurable function  $f$ from $\mc{S}$ to $\mc{S}(\overline{\R})$ such that 
that $\I_Af$ is nonnegative or $\mu_1$-integrable, it holds 
\begin{equation}\label{epr1ay} 
\int\! \I(A)f\, \mathrm{d}\mu_1=\int\! \I(A)fL\, \mathrm{d}\mu_2. 
\end{equation} 
Such an $L$ 
is uniquely defined $\mu_2$ a.e. on $A$, i.e. for some $L':\mc{S} \to \mc{S}(\overline{\R})$ 
we also have $L'=(\frac{d\mu_1}{d\mu_2})_A$ only if $L'=L$, $\mu_2$ a.e. on $A$ (i.e. if $\mu_2(\{L'=L\}\cap A)=\mu_2(A)$). 
Furthermore, such an $L$ is $\mu_2$ a.e. nonnegative on $A$. 
We say that $\mu_1$ is absolutely continuous with respect to $\mu_2$ on $A$, 
which we denote as $\mu_1 \ll_{A} \mu_2$, if for each $B \in \mc{D}$, 
from $\mu_2(A\cap B)=0$ it follows that $\mu_1(A\cap B)=0$. 
We say that $\mu_1$ and $\mu_2$ are mutually absolutely continuous on $A$ if 
$\mu_1 \ll_A \mu_2$ and $\mu_2 \ll_A \mu_1$, which we also denote as $\mu_1 \sim_A \mu_2$. 
If $L=(\frac{d\mu_1}{d\mu_2})_A$ exists, then it holds $\mu_1 \ll_{A} \mu_2$. 
We say that a measure $\mu$ on $\mc{S}$ is $\sigma$-finite on $A$ if $A$ 
is a countable union of sets from $\mc{D}$ with $\mu$-finite measure. Note that if $\mu$ is a probability distribution 
then it is $\sigma$-finite on $A$. 
From the Radon-Nikodym theorem, if $\mu_1$ and $\mu_2$ are $\sigma$-finite on $A$ and $\mu_1 \ll_A \mu_2$, 
then $L= (\frac{d\mu_1}{d\mu_2})_A$ exists. 
% \begin{remark}\label{remL}
% $L$ as above is  $\mu_2$ a.s. positive on $A$ since otherwise we would have $\mu_1(A\cap \{L\leq0\})=\E_{\mu_2}(L\I(A\cap\{L\leq 0\}))\leq 0$ 
% and \mu_2(L\leq0)>0$ 
% \end{remark} 
%Note that $L=(\frac{d\mu_1}{d\mu_2})_A$ is defined uniquely only $\PQ_1$ a.e. on $A$. 
%i.e. for each random variable $L_1$ such that $L_1=L$, $\PQ_1$ a.e. on $A$ we also have $L_1=(\frac{d\mu_1}{d\mu_2})_A$. 
%We have the following easy lemma. 
\begin{lemma}\label{lemsimA} 
Let 
%\begin{equation}
$L=(\frac{d\mu_1}{d\mu_2})_A$. 
%\end{equation}
Then, $\mu_1 \sim_A \mu_2$ only if $\mu_2(\{L=0\}\cap A)=0$, in which case
\begin{equation}\label{ln0l}
\frac{\I(L\neq0)}{L}=\left(\frac{d\mu_2}{d\mu_1}\right)_A. 
\end{equation}
\end{lemma}
 \begin{proof}
 If $\mu_2(\{L=0\}\cap A)=0$, then for $B \in \mc{D}$, from (\ref{epr1ay}),
 \begin{equation}
%\int\! \I(A)f\, \mathrm{d}\mu_1
 \int\! \I(A\cap B)\, \mathrm{d}\mu_2= \int\!L\frac{\I(L\neq0)}{L}\I(A\cap B)\, \mathrm{d}\mu_2= \int\!\frac{\I(L\neq0)}{L}\I(A\cap B)\, \mathrm{d}\mu_1
 \end{equation}
 so that we have (\ref{ln0l}) and  $\mu_1 \sim_A \mu_2$. 
 On the other hand, since $\mu_1(\{L=0\}\cap A)=\int\!  L\I(\{L=0\}\cap A)\, \mathrm{d}\mu_2=0$, 
 if $\mu_2(\{L=0\}\cap A)>0$ then we cannot have $\mu_2 \ll_A \mu_1$.  
 \end{proof}
For $A =D$ we omit $A$ in the above notations, e.g. we write $\mu_1 \ll \mu_2$, $\mu_1\sim\mu_2$, 
and $L=\frac{d\mu_1}{d\mu_2}$. % is the ordinary Radon-Nikodym derivative. 
We say that $q$ is a random condition on $\mc{S}$ if 
$\{x\in D: q(x)\} \in \mc{D}$. Often the event $\{x \in D:q(x)\}$ will be denoted simply as $\{q\}$ and we shall frequently write 
$q$ in the place of $\{q \}$ in various notations. 

\section{\label{secIS}IS and zero- and optimal-variance IS distributions} 
If for some probability $\PQ_2$ on $\mc{S}_1$, $\PQ_1 \ll_{Z\neq 0} \PQ_2$, then for $L=(\frac{d\PQ_1}{d\PQ_2})_{Z\neq 0}$ we have 
%\begin{equation}
$\alpha=\E_{\PQ_2}(ZL)$.
%\end{equation}
Importance sampling (IS) relies on estimating $\alpha$ by using in an MC method independent replicates of such an IS estimator
$ZL$ under $\PQ_2$.
%Such $L$ is called IS density, $\PQ_2$ --- the IS distribution, and $ZL$ --- the IS . 
The variance of the IS estimator fulfills
\begin{equation}\label{varzl}
\Var_{\PQ_2}(ZL)=\E_{\PQ_2}((ZL)^2)-\alpha^2= \E_{\PQ_1}(Z^2L)-\alpha^2.
\end{equation}
%The below condition states that $ZL$ is a zero-variance estimator of $\alpha$ under $\PQ_2$. 
\begin{condition}\label{condpqs}
It holds $\PQ_1 \ll_{Z\neq 0} \PQ_2$ and for some $L=(\frac{d\PQ_1}{d\PQ_2})_{Z\neq 0}$ we have  $\Var_{\PQ_2}(ZL)=0$
or equivalently $\PQ_2$ a.s. $ZL=\alpha$. 
\end{condition}

\begin{theorem}
Condition \ref{condpqs} holds only if it holds with 'for some' replaced by 'for each'. 
\end{theorem}
\begin{proof}
Let $L$ be as in Condition \ref{condpqs} and $L'=(\frac{d\PQ_1}{d\PQ_2})_{Z\neq 0}$. Then,
$L=L'$, $\PQ_2$ a.s. on $Z\neq 0$ and $0=ZL=ZL'$ on $Z=0$. Thus,  
from $\PQ_2$ a.s. $ZL=\alpha$ it also holds $\PQ_2$ a.s. $ZL'=\alpha$.  
\end{proof}

%  \begin{remark}
% Let Condition \ref{condpqs} hold and let $L'=(\frac{d\PQ_1}{d\PQ_2})_{Z\neq 0}$. Then,
% $L=L'$, $\PQ_2$ a.s. on $Z\neq 0$ and $0=ZL=ZL'$ on $Z=0$. Thus,  
% from $\PQ_2$ a.s. $ZL=\alpha$ it also holds $\PQ_2$ a.s. $ZL'=\alpha$.  
% \end{remark}
\begin{condition}\label{condn0}
It holds $\PQ_1(Z \neq 0)>0$. 
\end{condition}
\begin{condition}\label{condpqe}
It holds $\PQ_1(Z\neq 0)>0$ and either $\PQ_1$ a.s. $Z\geq 0$ or $\PQ_1$ a.s. $Z \leq 0$. 
%For $S=Z$ or $S=-Z$ we have  $\PQ_1(S>0) >0$ and $\PQ_1$ a.s. $S \geq 0$. 
% Equivalently,
% we have $\alpha \neq 0$ and $\PQ_1$ a.s. $Z\sgn(\alpha) \geq 0$.  
\end{condition}
%Note that Condition \ref{condpqe} implies Condition \ref{condn0}.
\begin{theorem}\label{thpqs}
If Condition \ref{condpqe} holds, then for a probability $\PQ^*$ given by
\begin{equation}\label{pqprzez}
\frac{d\PQ^*}{d\PQ_1}= \frac{Z}{\alpha},
\end{equation}
Condition \ref{condpqs} holds for $\PQ_2=\PQ^*$. Furthermore, $\PQ^*(Z\neq 0)=1$ and
$\PQ^* \sim_{Z\neq 0} \PQ_1$ with
\begin{equation}\label{lsdef}
L^*:=\I(Z\neq 0)\frac{\alpha}{Z}=(\frac{d\PQ_1}{d\PQ^*})_{Z\neq 0}.
\end{equation}
\end{theorem}
\begin{proof}
Condition \ref{condpqe} implies that $\PQ^*$ is well-defined. Furthermore,  
$\PQ^*(Z\neq 0)=\E_{\PQ_1}(\frac{Z}{\alpha})=1$ and 
from Lemma \ref{lemsimA} we have (\ref{lsdef}). 
In particular, $ZL^*=\alpha$, $\PQ^*$ a.s., that is Condition \ref{condpqs} holds for $\PQ_2=\PQ^*$. 
\end{proof}

\begin{lemma}\label{condpq2}
Assuming Condition \ref{condn0}, if there exists a probability $\PQ_2$ fulfilling Condition \ref{condpqs}, then
Condition \ref{condpqe} holds and such $\PQ_2$ is equal to the probability $\PQ^*$ as in Theorem \ref{thpqs}. 
\end{lemma}
\begin{proof}
Let conditions \ref{condpqs} and \ref{condn0} hold. Then, $\PQ_2$ a.s. 
\begin{equation}\label{azl}
\I(Z\neq 0)\frac{\alpha}{Z}= \I(Z\neq 0)L=(\frac{d\PQ_1}{d\PQ_2})_{Z\neq0}.
\end{equation}
Thus, from Condition \ref{condn0},
%\begin{equation}
$\E_{\PQ_2}\left(\I(Z\neq 0)\frac{\alpha}{Z}\right)= \PQ_1(Z \neq 0)>0$,
%\end{equation}
which implies that $\alpha \neq 0$. % and thus from (\ref{azl}) and Remark \ref{remL}, 
% \begin{equation}\label{zsgnpq2as}
% Z\sgn(\alpha)\geq 0,\quad  \PQ_2\text{ a.s.}
% \end{equation}
From (\ref{azl}) and Lemma \ref{lemsimA} we have $\PQ_1 \sim_{Z\neq 0} \PQ_2$ and 
%\begin{equation}
$\frac{Z}{\alpha} = (\frac{d\PQ_2}{d\PQ_1})_{Z\neq0}$.
%\end{equation}
Thus, $\PQ_2(Z \neq 0)=\E_{\PQ_1}(\I(Z \neq 0)\frac{Z}{\alpha})=1$ and
%\begin{equation}\label{dpq2pq1}
$\frac{Z}{\alpha} = \frac{d\PQ_2}{d\PQ_1}$. 
%\end{equation}
%Since $\frac{d\PQ_2}{d\PQ_1}$ is PQ_1 a.s. nonnegative 
In particular,  $\PQ_1$ a.s. $Z\sgn(\alpha)\geq 0$ and thus Condition \ref{condpqe} holds and $\PQ_2=\PQ^*$. 
% From discussion of uniqueness properties of $L$ in Section \ref{secBackDens} and from $\PQ^* \sim_{Z \neq 0} \PQ_1$,
% $L$ must be as in the thesis.
\end{proof}

\begin{theorem}\label{thvaruniq}
If Condition \ref{condpqe} holds, then the probability $\PQ^*$ as in Theorem \ref{thpqs} is the unique probability 
$\PQ_2$ for which Condition \ref{condpqs} holds. %(in which case it will be called the zero-variance IS distribution). 
%If Condition \ref{condn0} holds then 
%Condition \ref{condpqe} holds only if distribution 
%$\PQ_2$ for which \ref{condpqs} holds exists, in which case 
%\begin{equation}\label{pq2pqs}
%\PQ_2=\PQ^*. 
%\end{equation}
\end{theorem}
\begin{proof}
%The left implication of the only if statement with equality (\ref{pq2pqs}) are equivalent to Lemma \ref{condpq2}. 
Since Condition \ref{condpqe} implies Condition \ref{condn0} and Theorem \ref{thpqs}
implies the existence of $\PQ_2$ fulfilling Condition \ref{condpqs}, from Lemma \ref{condpq2}, $\PQ_2=\PQ^*$. 
\end{proof}
We shall call the probability $\PQ^*$ as in Theorem \ref{thpqs} the zero-variance IS distribution. 
Assuming that $L=(\frac{d\PQ_1}{d\PQ_2})_{Z\neq 0}$, from (\ref{varzl}),
\begin{equation}\label{varazl}
\Var_{\PQ_2}(|ZL|) =  \E_{\PQ_1}(Z^2L)- (\E_{\PQ_1}(|Z|))^2,
\end{equation}
and thus
\begin{equation}\label{varpqineq}
\begin{split}
\Var_{\PQ_2}(ZL)=\Var_{\PQ_2}(|ZL|)+(\E_{\PQ_1}(|Z|))^2-\alpha^2
\geq (\E_{\PQ_1}(|Z|))^2-\alpha^2,\\
\end{split}
\end{equation}
with equality holding only if Condition \ref{condpqs} holds for  $|Z|$ (i.e. for $Z$ replaced by $|Z|$ 
and in particular for $\alpha$ replaced by $\E_{\PQ_1}(|Z|)$). 
%\begin{theorem}
Let Condition  \ref{condn0} hold. Then, Condition \ref{condpqe} holds for $|Z|$ and from Theorem \ref{thvaruniq},
$\PQ^*$ as in Theorem \ref{thpqs} but for $Z$ replaced by $|Z|$, i.e. such that 
\begin{equation}\label{abszpq1}
\frac{d\PQ^*}{d\PQ_1}=\frac{|Z|}{\E_{\PQ_1}(|Z|)}, 
\end{equation}
is the unique probability $\PQ_2$ for which Condition \ref{condpqs} holds for $|Z|$.
The fact that Condition \ref{condpqs} holds for $|Z|$ for such a $\PQ^*$ is well-known, see e.g. Theorem 1.2 in Chapter V in \cite{asmussen2007stochastic},
but the uniqueness result is to our knowledge new. Furthermore, we have
\begin{equation}
L^*:=\I(Z\neq 0)\frac{\E_{\PQ_1}(|Z|)}{|Z|}=(\frac{d\PQ_1}{d\PQ^*})_{Z\neq 0}. 
\end{equation}
%or equivalently the unique distribution $\PQ_2$ for which
%Condition \ref{condDerivR} holds and we have equality in (\ref{varpqineq}).
%Note that from Theorem \ref{thvaruniq}, 
Note that from Condition \ref{condpqs} holding for $|Z|$ and (\ref{varazl}), $\PQ^*$ a.s. (or equivalently $\PQ_1$ a.s. on $Z\neq 0$) 
\begin{equation}\label{zlsmsq}
|Z|L^*= \E_{\PQ_1}(|Z|)=\sqrt{\E_{\PQ_1}(Z^2L^*)}. 
\end{equation}
We call such a $\PQ^*$ the optimal-variance IS distribution. Under
Condition \ref{condpqe} the optimal-variance IS distribution is also the zero-variance one. 
In some places in the literature our optimal-variance IS distribution is called simply the optimal IS distribution
(see e.g. page 127 in \cite{asmussen2007stochastic}). However, since as argued in Chapter \ref{secIneff}
it may be more optimal to minimize inefficiency constant than variance and the optimal-variance IS distribution does not need to lead
to the lowest inefficiency constant achievable via IS, calling it optimal may be misleading. 

\section{\label{secISIneff}Mean cost and inefficiency constant in IS}
Let $L=(\frac{d\PQ_1}{d\PQ_2})_{Z\neq 0}$ and 
let $C$ be a nonnegative (theoretical) cost variable on $\mc{S}_1$ for computing replicates of $ZL$ under $\PQ_2$. 
We shall consider $C$ to be the same for different $\PQ_2$ under consideration. 
%which seems to be a reasonable assumption for costs like internal simulation times of stochastic processes,
%but need not be reasonable for different costs (like e.g. number of steps of acceptance- rejection algorithm).
The mean cost under $\PQ_2$ is $\E_{\PQ_2}(C)$ and such a (theoretical) inefficiency constant is 
\begin{equation} 
\Var_{\PQ_2}(ZL)\E_{\PQ_2}(C)
\end{equation} 
(assuming that it is well-defined). 

Note that if the zero-variance IS distribution  $\PQ^*$ exists and the mean cost $\E_{\PQ^*}(C)$ is finite, 
then the inefficiency constant under $\PQ^*$ is zero. 

The below theorem provides an intuition why in our numerical experiments in Chapter \ref{secNumExp}, for some $a\in [0,\infty)$ and $s \in \R_+$,
for a nonincreasing function $f(x)=\I(x<s)+a$ and a strictly decreasing one $f(x)=\exp(-sx)$, and for $Z=f(C)$, 
we observed mean cost reduction after changing the initial distribution to a one in a sense closer to the respective zero-variance IS distribution $\PQ^*$. 
\begin{theorem}\label{thDecrCost}
Let $f:\mc{S}(\R) \to \mc{S}([0,\infty))$, 
$Z=f(C)$,  $\E_{\PQ_1}(Z)\in \R_+$, and $\E_{\PQ_1}(C)<\infty$. 
Let $\PQ^*$ be the zero-variance IS distribution. 
\begin{enumerate}
\item If $f$ is nonincreasing, then
\begin{equation}\label{ineqpqc}
\E_{\PQ^*}(C)\leq \E_{\PQ_1}(C),
\end{equation}
and if further for some $0\leq x_1<x_2<\infty$ we have  $f(x_1)>f(x_2)$, $\PQ_1(C\in[0,x_1])>0$, and $\PQ_1(C\in[x_2,\infty))>0$
(which is the case e.g. if $f$ is strictly decreasing and $C$ is not $\PQ_1$ a.s. constant),
then the inequality in (\ref{ineqpqc}) is sharp.
\item If $f$ is nondecreasing, then
\begin{equation}\label{ineqpqc2}
\E_{\PQ^*}(C)\geq \E_{\PQ_1}(C),
\end{equation}
and if further for some $0\leq x_1<x_2<\infty$ we have  $f(x_1)<f(x_2)$, $\PQ_1(C\in[0,x_1])>0$, and $\PQ_1(C\in[x_2,\infty))>0$, then 
the inequality in (\ref{ineqpqc2}) is sharp.
\end{enumerate} 
\end{theorem}
\begin{proof}
From (\ref{pqprzez}) we have
\begin{equation}\label{pqsc}
\E_{\PQ^*}(C) =\frac{\E_{\PQ_1}(f(C)C)}{\E_{\PQ_1}(f(C))}.
\end{equation}
For $C_1$ and $C_2$ being independent replicates of $C$ under $\PQ_1$, we have 
\begin{equation}\label{pq1fcc}
\E_{\PQ_1}(f(C)C)-\E_{\PQ_1}(C)\E_{\PQ_1}(f(C))=\frac{1}{2}\E_{\PQ_1}((f(C_1)-f(C_2))(C_1-C_2)),
\end{equation}
which is nonpositive if $f$ is nonincreasing and negative under the additional assumptions of point one, 
or nonnegative if $f$ is nondecreasing 
and positive under the additional assumptions of point two. From this and (\ref{pqsc}), the thesis easily follows. 
\end{proof}

\section{\label{secFamily}Parametric IS}
For some nonempty set $A$, let us consider a family 
$\PQ(b)$, $b \in A$, of probability distributions on $\mc{S}_1$. % such that for some $b_0 \in A$, $\PQ(b_0)=\PQ_1$.
Typically, we shall assume that for some $l\in \N_+$ 
\begin{equation}\label{abrl}
A \in\mc{B}(\R^l).
\end{equation}
Consider a function $L:A\times \Omega_1 \to \R$, for which we denote $L(b)=L(b,\cdot)$, $b \in A$.
If the following condition is fulfilled, then for each $b \in A$ 
one can perform IS using the IS distribution $\PQ(b)$ and density $L(b)$
as in Section \ref{secIS}.
\begin{condition}\label{condpqbllpq1}
It holds %$\PQ_1\ll_{Z\neq0}\PQ(b)$ and 
$L(b)=(\frac{d\PQ_1}{d\PQ(b)})_{Z\neq 0}$, $b \in A$.
\end{condition}
% \begin{condition}\label{condPQMes}
% We have (\ref{abrl}), and for each $B \in \mc{F}_1$, $b\to \PQ(b)(B)$ is measurable from $A$ to $\R$.
% \end{condition}
%For some nonempty set $A\in \mc{B}(\R^l)$ consider a family of distributions $\PQ(b)\ll_{Z\neq 0} \PQ_1$, $b \in A$. 
For $x_1$ and $x_2$ being two $\sigma$-fields, measurable spaces, or measures, by $x_1 \otimes x_2$ we denote their product
$\sigma$-field, measurable space, or measure respectively, while for $n \in \N_+$,
by $x_1^n$ we mean such an $n$-fold product of $x_1$. 
The following conditions will be useful further on. 

\begin{condition}\label{condLmes}
We have (\ref{abrl}) and $L$ is measurable from $\mc{S}(A)\otimes\mc{S}_1$ to $\mc{S}(\R)$.
\end{condition}
\begin{condition}\label{condxi}
We have (\ref{abrl}) and a probability $\PR_1$ on a measurable space $\mc{C}_1$ and
%us assume that there exists some measurable function 
%\begin{equation}\label{xidef} 
$\xi: \mc{C}_1 \otimes \mc{S}(A)\rightarrow \mc{S}_1$ 
%\end{equation} 
are such that for each $b \in A$
\begin{equation}\label{ximin}
\PQ(b)(B)=\PR_1(\xi(\cdot,b)^{-1}[B]),\quad B \in \mc{F}_1, 
\end{equation}
or equivalently, for each random variable $X \sim \PR_1$, $\xi(X,b)\sim \PQ(b)$, $b \in A$. % (see (16.16) \cite{billingsley1979}). 
\end{condition}
\begin{remark}
Let conditions \ref{condpqbllpq1}, \ref{condLmes}, and \ref{condxi} hold and let $b$ be some $A$-valued random variable, which 
can be e.g. some adaptively obtained IS parameter. 
Let  $\beta_i \sim \PR_1$, $i \in \N_+$, be i.i.d. and independent of $b$. 
Then, from Fubini's theorem it follows that the
random variables $\wh{\alpha}_n =\frac{1}{n}\sum_{i=1}^n(ZL(b))(\xi(\beta_{i},b))$, $n\in \N_+$, are unbiased and strongly consistent estimators of $\alpha$, 
i.e. $\E(\wh{\alpha}_n)=\alpha$, $n\in \N_+$, and a.s. $\lim_{n\to \infty}\wh{\alpha}_n=\alpha$. 
\end{remark}

In the further sections we shall often deal with families of distributions and densities satisfying the following condition. 
\begin{condition}\label{condB1} 
A set $B_1 \in \mc{F}_1$ is such that we have  $\PQ(b)\sim_{B_1} \PQ_1$  and 
$L(b)=(\frac{d\PQ_1}{d\PQ(b)})_{B_1}$,  $b \in A$. 
\end{condition} 
Let us formulate separately the special important case of the above condition. 
\begin{condition}\label{condpqpq1} 
Condition \ref{condB1} holds for $B_1=\Omega_1$, or equivalently $\PQ(b)\sim \PQ_1$ and $L(b)=\frac{d\PQ_1}{d\PQ(b)}$, $b \in A$. 
\end{condition} 
The following condition will be useful to avoid different technical problems like when dividing by $L$ or taking its logarithm. 
\begin{condition}\label{condg0} 
It holds $L(b)(\omega)>0$, $b \in A$, $\omega \in \Omega_1$. 
\end{condition} 
%Note that for $L$ as in Condition \ref{condB1} or condpqpq1.
% Note that if conditions \ref{condLmes} and \ref{condpqpq1} hold, then for each $B\in \mc{F}_1$ and $b\in A$ 
% \begin{equation} 
% \PQ(b)(B)=\E_{\PQ_1}(\I_BL^{-1}(b)), 
% \end{equation} 
% so Condition \ref{condPQMes} holds. 
% Note that $\PQ^* \sim \PQ_1$ only if $Z>0$, $\PQ_1$ a.s. 
% $\mc{A}(B_1)=\{\PQ_2: \PQ_2\sim_{B_1}\PQ_1 \}$. 
% such that there exists $b_0 \in A$ for which $\PQ(b_0)=\PQ_1$. 
% Consider densities $L(b)=\left(\frac{d\PQ_1}{d\PQ(b)}\right)_{B_1}$, $b\in A$, 
% and the following condition. 
%Consider the following condition. 
\begin{condition}\label{condZB} 
Condition \ref{condB1} holds, $\PQ_1(\{Z\neq 0\}\setminus B_1)=0$, and
$\PQ(b)(\{Z\neq 0\}\setminus B_1)=0$, $b\in A$.
\end{condition}
%Note that under Condition holds e.g. when $Z=\I_{B_1}Z$. 
\begin{remark}\label{rempqbznpq1}
Note that for $Z$ such that Condition \ref{condZB} holds, we have $\PQ(b) \sim_{Z \neq 0\cup B_1} \PQ_1$, $b \in A$, and 
\begin{equation} 
L(b)=(\frac{d\PQ_1}{d\PQ(b)})_{B_1\cup \{Z \neq 0\}},\quad b\in A, 
\end{equation} 
so that Condition \ref{condpqbllpq1} holds. 
\end{remark}
%Note also that Condition \ref{condpqpq1} implies Condition \ref{condZB} for each $Z$. 
% so that Condition \ref{condDerivR} holds for $L=L(b)$ and $\PQ_2 = \PQ(b)$ 
% and one can perform IS with distributions from $\mc{W}$. 
% \begin{lemma} 
% The below theorem characterizes random variables $Z$ on $\mc{S}_1$ 
% for which there exist optimal IS distributions in $\mc{W}$ under appropriate conditions. %Condition \ref{condB1}.
%and \ref{condZB}. 
\begin{definition}\label{defParamZero}
We say that $b^* \in A$ is a zero-variance (optimal-variance) IS parameter if Condition \ref{condpqe} (Condition \ref{condn0}) holds and 
$\PQ(b^*)$ is the zero-variance (optimal-variance) IS distribution. 
\end{definition}
Note that in the literature the name optimal IS parameter is sometimes used for the parameter minimizing the variance $b\in A\to \Var_{\PQ(b)}(ZL(b))$
of the IS estimator (see e.g. \cite{Lapeyre2011}), which may be not equal 
to an optimal-variance IS parameter in the sense of the above definition. 

The below theorem characterizes the random variables $Z$ as above for which there exists a zero-variance IS parameter,
under some of the above conditions. 
% under Condition \ref{condB1} and some 
% additional assumptions allowing to perform IS with respect to distributions from $\mc{W}$ as discussed above. 
\begin{theorem}\label{thexistqs} 
Let us assume Condition \ref{condB1}.
Then, Condition \ref{condZB} holds and there exists a zero-variance IS parameter $b_1$ (for which we denote $\PQ^*=\PQ(b_1)$), only 
if for some $b_2 \in A$, $\PQ(b_2)(B_1)=1$ and for some $\beta \in \R \setminus 0$, 
\begin{equation}\label{zperffrom}
Z = \I_{B_1}\I(L(b_2)\neq 0)\frac{\beta}{L(b_2)},\quad \PQ_1\text{ a.s.} 
\end{equation}
Furthermore, in the latter case we have  $\beta=\alpha$ and $\PQ(b_2)=\PQ^*$.
\end{theorem}
\begin{proof}
Let us first show the right implication. From Condition \ref{condZB} and $\PQ^*=\PQ(b_1)$ it follows for $b_2=b_1$ that 
$\PQ(b_2)(B_1)=\PQ^*(B_1)=\PQ^*(Z\neq 0\cap B_1)=\PQ^*(Z\neq 0) - \PQ(b_2)(\{Z\neq 0\}\setminus B_1)=1$. 
From (\ref{lsdef}), $\PQ^*$ a.s. $L(b_2)Z=\I(Z\neq 0)\alpha$, which from $\PQ^*\sim_{Z\neq 0} \PQ_1$ holds also $\PQ_1$ a.s. 
Thus, since from Condition \ref{condpqe} we have $\alpha \neq 0$, it holds $\PQ_1$ a.s. that if $Z\neq 0$ then also $L(b_2)\neq 0$. Therefore, 
we have $\PQ_1$ a.s. 
$Z= \I(Z\neq 0\wedge L(b_2)\neq 0)\frac{1}{L(b_2)}\alpha$. 
Thus, from Condition \ref{condZB}, 
\begin{equation}\label{zperffrom2}
Z= \I_{B_1}\frac{1}{L(b_2)}\I(Z\neq 0\wedge L(b_2)\neq 0)\alpha,\quad \PQ_1\text{ a.s. } 
\end{equation}
From Condition \ref{condB1}, $\PQ^*\sim_{B_1}\PQ_1$, and thus from Lemma \ref{lemsimA} and (\ref{pqprzez}),
$0=\PQ_1(\{L(b_2)= 0\}\cap B_1)=\PQ_1(\{Z= 0\}\cap B_1)$, 
so that from (\ref{zperffrom2}) and $\PQ_1(Z\neq 0)>0$ 
 %implying Condition \ref{condn0}
we have (\ref{zperffrom}) only for $\beta=\alpha$.
%L^*:=\I(Z\neq 0)\frac{\alpha}{Z} \PQ^* a.s.

For the left implication note that for $Z$ as in (\ref{zperffrom}) Condition \ref{condZB} holds.
Furthermore, from Condition \ref{condB1} and Lemma \ref{lemsimA}, 
\begin{equation}\label{pqb2b1l}
\PQ(b_2)(B_1\cap \{L(b_2)\neq 0\})= \PQ(b_2)(B_1).
\end{equation}
From (\ref{zperffrom}), (\ref{pqb2b1l}), and $\PQ(b_2)(B_1)=1$
\begin{equation}\label{alphaD}
\alpha= \E_{\PQ(b_2)}(ZL(b_2))=\beta\PQ(b_2)(B_1\cap \{L(b_2)\neq 0)=\beta\neq0, 
\end{equation}
so that Condition \ref{condpqe} holds. We have
\begin{equation}
\frac{d\PQ(b_2)}{d\PQ_1}= \I(B_1)\I(L(b_2)\neq 0)\frac{1}{L(b_2)}=\frac{Z}{\beta} = \frac{d\PQ^*}{d\PQ_1}, 
\end{equation}
where in the first equality we used  Condition \ref{condB1}, $\PQ(b_2)(B_1)=1$, and Lemma \ref{lemsimA},
in the second (\ref{zperffrom}), and in the last (\ref{alphaD}) and (\ref{pqprzez}). 
\end{proof}
\begin{remark}\label{remOptIS}
From the discussion in Section \ref{secIS}, the optimal-variance IS distribution for $Z$ is the zero-variance one for $|Z|$.
Thus, from the above theorem for $Z$ replaced by $|Z|$ we receive a characterization of variables $Z$ for which there exists an optimal-variance IS 
parameter under certain assumptions. 
\end{remark}

\chapter{The minimized functions and their estimators}\label{secMinFun}
\section{The minimized functions}\label{secCoeffDiv}
For some nonempty set $A$, consider a family of probability 
distributions as in Section \ref{secFamily} 
for which Condition \ref{condpqbllpq1} holds. 
Assuming Condition \ref{condg0} and that
%for some $b \in A$
\begin{equation}\label{cecond}
\E_{\PQ_1}((Z\ln(L(b)))_{-}) <\infty,\quad b \in A,
\end{equation}
we define a cross-entropy (function) $\ce:A\to \R \cup \{\infty \}$ as 
\begin{equation}\label{ceDef}
\ce(b)=\E_{\PQ_1}(Z\ln(L(b))),\quad b \in A  
\end{equation}
(see the discussion in Chapter \ref{secInt} regarding its name). 
\begin{remark}\label{remfdif}
Let us discuss how $\ce(b)$ is related to a certain $f$-divergence of the 
zero-variance IS distribution from  $\PQ(b)$. 
For some convex function $f:[0,\infty) \rightarrow \R$, the 
$f$-divergence $\dist(\PR_1,\PR_2)$ of a probability $\PR_2$ from another one $\PR_1$ such that $\PR_2 \ll \PR_1$ is given by the formula
\begin{equation}\label{dpqpr} 
\dist(\PR_1,\PR_2)=\E_{\PR_1}(f(\frac{d\PR_2}{d\PR_1})). 
\end{equation} 
Such an $f$-divergence is also known as Csisz\'ar $f$-divergence or Ali-Silvey distance \cite{pupa2011, Ali_1966,LieseV06}. 
From Jensen's inequality we have  $\dist(\PR_1,\PR_2)\geq f(1)$, and if $f$ is 
strictly convex then the equality in this inequality holds only if $\PR_1=\PR_2$. 
For example, for the strictly convex function $f(x) =x \ln(x)$ (which we assume to be zero for $x=0$), $\dist(\PR_1,\PR_2)$ 
is called Kullback-Leibler divergence or cross-entropy distance (of $\PR_2$ from $\PR_1$), while 
for $f(x) =(x^2-1)$, $\dist(\PR_1,\PR_2)$ is called Pearson divergence. 
For $\dist$ denoting the cross-entropy distance, 
let us assume Condition \ref{condpqe}, so that the zero-variance IS distribution $\PQ^*$ exists, $\PQ^* \ll \PQ(b)$, $b \in A$, and
\begin{equation}\label{dKul1}
\begin{split}
\dist(\PQ(b),\PQ^*)&=  \E_{\PQ(b)}(\frac{d\PQ^*}{d\PQ(b)}\ln(\frac{d\PQ^*}{d\PQ(b)}))\\
&= \E_{\PQ^*}(\ln(\frac{d\PQ^*}{d\PQ(b)}))\\
&=\E_{\PQ_1}(\frac{Z}{\alpha}(\ln(\frac{d\PQ^*}{d\PQ_1}) + \ln(L(b)))),\\
\end{split}
\end{equation}
where in the last equality we used (\ref{pqprzez}) and 
\begin{equation}\label{pql}
(\frac{d\PQ^*}{d\PQ(b)})_{Z\neq0}=\frac{d\PQ^*}{d\PQ_1}L(b).
\end{equation}
Assuming that $Z \geq 0$, we have 
\begin{equation}
\E_{\PQ_1}(Z\ln(\frac{d\PQ^*}{d\PQ_1}))=\E_{\PQ_1}(Z\ln(Z))- \alpha\ln(\alpha). 
\end{equation}
From $x\ln(x)\geq -e^{-1}$ we have $\E_{\PQ_1}(Z\ln(Z))\geq -e^{-1}$. Assuming further that
\begin{equation}\label{pq1zlnz}
\E_{\PQ_1}(Z\ln(Z))<\infty, 
\end{equation}
we receive from (\ref{dKul1}) that 
\begin{equation}\label{dKul}
\begin{split}
\dist(\PQ(b),\PQ^*)&= \alpha^{-1}(\E_{\PQ_1}(Z\ln(Z))- \alpha\ln(\alpha) + \ce(b)).\\
\end{split}
\end{equation}
If (\ref{dKul}) holds as above for each $b \in A$, then
$b\to\ce(b)$ and $b\to\dist(\PQ(b),\PQ^*)$ are positively linearly equivalent (see Chapter \ref{secInt}).
Note that from the discussion leading to formula (\ref{dKul}) and from $\dist(\PQ(b),\PQ^*)\geq 0$, 
a sufficient assumption for (\ref{cecond}) to hold is that we have $Z\geq 0$ and (\ref{pq1zlnz}).
\end{remark}

We define the mean square of the IS estimator as
\begin{equation}\label{msqDef}
\msq(b)=\E_{\PQ(b)}((ZL(b))^2)=\E_{\PQ_1}(Z^2L(b)),\quad b\in A,
\end{equation}
and such a variance as
\begin{equation}\label{varform}
\var(b)=\msq(b)-\alpha^2, \quad b\in A. 
\end{equation}
\begin{remark}\label{remPers}
Assuming that Condition \ref{condn0} holds, for $\PQ^*$ denoting the optimal-variance IS distribution as in Section \ref{secIS} and 
$\dist$ denoting the Pearson divergence as in Remark \ref{remfdif}, 
from (\ref{pql}) and (\ref{abszpq1}) we have for $b\in A$
\begin{equation}\label{distvar}
\begin{split}
\dist(\PQ(b),\PQ^*) &= \E_{\PQ(b)}((\frac{|Z|}{\E_{\PQ_1}(|Z|)}L(b))^2-1)\\
&= \frac{1}{(\E_{\PQ_1}(|Z|))^2}(\msq(b)-(\E_{\PQ_1}(|Z|))^2)\\   
&= \frac{1}{(\E_{\PQ_1}(|Z|))^2}(\var(b)+\alpha^2-(\E_{\PQ_1}(|Z|))^2).\\
\end{split}
\end{equation}
Thus, in such a case $b \in A\to \dist(\PQ(b),\PQ^*)$ is positively 
linearly equivalent to $\msq$ and $\var$.
\end{remark}
% functions from $A$ to $\R\cup \{\infty\}$ given by (\ref{msqDef}), (\ref{varform}), and (\ref{distvar}) are positively 
% linearly equivalent.

Let $C$ be some $[0,\infty]$-valued theoretical cost variable on $\mc{S}_1$. % called a (theoretical) cost variable.  
%and let us assume that Condition \ref{condpqbllpq1} holds for $Z=C$.
% Let us assume Condition \ref{condpqpq1},  
% cost variable as in Section \ref{secISIneff}, 
%as discussed in the Introduction %. Let us assume that $\PQ(b) \sim \PQ_1$, $b \in A$,
Let $c(b)=\E_{\PQ(b)}(C)$ be the mean cost under $\PQ(b)$, $b \in A$. 
%As discussed in Chapter \ref{secIneff}, 
%Consider the following condition. 
\begin{condition}\label{condicwelldef}
For each $b \in A$, it does not hold $c(b)=\infty$ and $\var(b)=0$, 
or $c(b)=0$ and $\var(b)=\infty$. 
\end{condition}
Assuming Condition \ref{condicwelldef},
we define a (theoretical) inefficiency constant as
\begin{equation}
\ic(b)=c(b)\var(b),\quad b \in A.
\end{equation}
Frequently, the proportionality constants $p_{\dot{C}}$ of 
the practical to the theoretical costs of the IS MC as in Chapter \ref{secIneff}
can be chosen the same for different IS parameters $b\in A$,
so that the practical and theoretical inefficiency constants are proportional and their minimization is equivalent. 

\section{\label{secEstMin}Estimators of the minimized functions} 
%TODO
Consider a family of probability 
distributions as in Section \ref{secFamily} 
and let us assume that conditions \ref{condpqbllpq1} and \ref{condLmes} hold. 
Consider a measurble function $f:\mc{S}(A) \to \mc{S}(\overline{\R})$ and for some $p \in \N_+$, consider 
\begin{equation}\label{dvcdef} 
\wh{\est}_n:\mc{S}(A)^2\otimes\mc{S}_1^n \rightarrow \mc{S}(\R),\quad n \in \N_p, 
\end{equation} 
called estimators of $f$, where $\wh{\est}_n(b',b)$ is thought of as an estimator of $f(b)$ under $\PQ(b')^n$, $b,b' \in A$, $n \in \N_p$. 
In all this work, for $b' \in A$, we denote $\PQ'=\PQ(b')$ and $L'=L(b')$. 
We say that some $\wh{\est}_n$ as above is an unbiased estimator of $f$  if 
\begin{equation}
f(b) = \E_{(\PQ')^n}(\wh{\est}_n(b',b)),\quad b',b \in A. 
\end{equation} 
Let us further in this section assume the following condition. % for the same $b'$ as above. 
\begin{condition}\label{condKappa} 
We have $b' \in \R^l$ and 
$\kappa_1,\kappa_2,\ldots,$ are i.i.d. $\sim \PQ'$ and $\wt{\kappa}_n=(\kappa_i)_{i=1}^n$, 
$n \in \N_+$. 
\end{condition} 
We call the estimators $\wh{\est}_n$, $n \in \N_p$,  strongly consistent for $f$ if 
for each $b',b \in A$, a.s. 
\begin{equation} 
\lim_{n\to \infty}\wh{\est}_n(b',b)(\wt{\kappa}_n)=f(b). 
\end{equation} 
For a function $Y$ on $\Omega_1$ (like e.g. $Z$ or $L$), 
we define such functions $Y_1,\ldots, Y_n$ on $\Omega_1^n$ by the formula 
\begin{equation}\label{Yidefs}
Y_i(\omega)=Y(\omega_i),\quad \omega=(\omega_{i})_{i=1}^n\in \Omega_1^n,
\end{equation}
and whenever $Y$ takes values in some linear space we denote
\begin{equation}\label{overdef}
\overline{(Y)}_n =\frac{1}{n}\sum_{i=1}^nY_i. 
\end{equation}
For the cross-entropy as in the previous section, assuming (\ref{cecond}), we have 
\begin{equation}
\ce(b)=\E_{\PQ'}(L'Z\ln(L(b))),
\end{equation}
so that for $n \in \N_+$, from Theorem \ref{thSLLNG}, its unbiased strongly consistent estimators are
\begin{equation}\label{ceEst}
\wh{\ce}_n(b',b)=\overline{(L'Z\ln(L(b)))}_n=\frac{1}{n}\sum_{i=1}^n L_i'Z_i\ln(L_i(b)).
\end{equation}
For mean square, we have  
\begin{equation}
\msq(b)=\E_{\PQ'}(Z^2L'L(b)), 
\end{equation}
so that for $n \in \N_+$, its unbiased strongly consistent estimators are
\begin{equation}\label{whmsqdef}
\wh{\msq}_n(b',b)= \overline{(Z^2L'L(b))}_n. % \frac{1}{n}\sum_{i=1}^nZ^2_iL_i'L_i.
\end{equation}
The above mean square estimators 
and estimators negatively linearly equivalent to 
the above cross-entropy estimators in the function of $b$ (see Chapter \ref{secInt}) have been considered before 
in the literature; see e.g. \cite{Rubinstein97optimizationof, Rubinstein_optim, Rubinstein_2004, Jourdain2009}. 
Thus, we call the above estimators well-known. 
We shall now proceed to define some new estimators. 
If $\PQ(b) \ll \PQ'$, then for variance, we have  for $n \in \N_2$
\begin{equation}\label{varEqu} 
\begin{split} 
\var(b) &= \Var_{\PQ(b)}(ZL(b))\\ 
&= \E_{(\PQ(b))^n}\left(\frac{1}{n(n-1) }\sum_{i<j \in \{1,\ldots,n\}} (Z_iL_i(b) - Z_jL_j(b))^2\right)\\ 
&= \frac{1}{n(n-1)}\E_{(\PQ')^n} 
\left(\sum_{i<j \in \{1,\ldots,n\}}\left(\frac{d\PQ(b)}{d\PQ'}\right)_i\left(\frac{d\PQ(b)}{d\PQ'}\right)_j (Z_iL_i(b) - Z_jL_j(b))^2\right).\\ 
\end{split} 
\end{equation} 
Let us further in this section assume conditions \ref{condpqpq1} and \ref{condg0}. %pqpq1
Then, $\frac{d\PQ(b)}{d\PQ'} = \frac{L'}{L(b)}$, and from (\ref{varEqu}), we have the
following unbiased estimators of $\var$ for $n \in \N_2$ 
 \begin{equation}\label{varest}
 \begin{split}
 \wh{\var}_n(b',b)&=\frac{1}{n(n-1)}\sum_{i<j \in \{1,\ldots,n\}} \frac{L_i'L_j'}{L_i(b)L_j(b)}(Z_iL_i(b) - Z_jL_j(b))^2\\
 &= \frac{1}{n(n-1)}\left(\sum_{i=1}^n\left(Z_i^2L_i'L_i(b)\sum_{j\in \{1,\ldots,n\}, j \neq i}\frac{L_j'}{L_j(b)}\right)
 -\sum_{i<j \in \{1,\ldots,n\}}2Z_iZ_jL_i'L_j'\right)\\
&=\frac{n}{n-1}\left(\wh{\msq}_n(b',b)\overline{\left(\frac{L'}{L(b)}\right)}_n - \overline{(ZL')}_n^2\right).
 \end{split} 
 \end{equation}
Thus, $b\to\wh{\var}_n(b',b)$ is positively linearly equivalent to the following estimator of mean square
\begin{equation}\label{msq2ndef}
%\begin{split}
\wh{\msq2}_{n}(b',b)= \wh{\msq}_n(b',b)\overline{\left(\frac{L'}{L(b)}\right)}_n,
%\end{split}
\end{equation}
which can be considered also for $n=1$. 
From the facts that from the SLLN, a.s. 
\begin{equation}
\lim_{n\to \infty} \overline{(ZL')}_n(\wt{\kappa}_n)= \E_{\PQ'}(ZL')=\alpha  
\end{equation}
and 
\begin{equation}\label{cons1}
\lim_{n\to \infty} \overline{\left(\frac{L'}{L(b)}\right)}_n(\wt{\kappa}_n)= \E_{\PQ'}\left(\frac{L'}{L(b)}\right)=1, 
\end{equation}
estimators $\wh{\msq2}_{n}$ and $\wh{\var}_n$ 
are strongly consistent for $\msq$ and $\var$ respectively. 
Let us further in this section assume that 
\begin{equation}\label{Cinfty0}
\PQ(b)(C=\infty)=0,\quad b \in A.  
\end{equation}
Then, strongly consistent and unbiased estimators of the mean cost $c$ are 
\begin{equation}\label{cest}
\wh{c}_n(b',b)=\I(i=1,\ldots,n)\overline{\left(\frac{L'}{L(b)}\I(C\neq \infty)C\right)}_n.
\end{equation}
Let us further in this section assume Condition \ref{condicwelldef}. Then, strongly consistent estimators of $\ic$ are for $n \in \N_2$,
\begin{equation}\label{icdef}
\wh{\ic}_n(b',b)=\wh{c}_n(b',b)\cdot\wh{\var}_n(b',b),
\end{equation}
% (by which we mean that $\wh{\ic}_n(b',b)=0$ whenever 
% we multiply infinity by zero on the right and similarly in (\ref{ic2Est})), 
which are in general not unbiased. For each $n \in \N_3$, defining helper unbiased estimators of variance for $k=1,\ldots,n$
\begin{equation}
\wh{\var}_{n,k}(b',b)=\frac{1}{(n-1)(n-2)}\left(\sum_{i<j \in \{1,\ldots,n\}\setminus\{k\}}\frac{L_i'L_j'}{L_i(b)L_j(b)}(Z_iL_i(b) - Z_jL_j(b))^2\right), 
\end{equation}
we have the following unbiased estimator of $\ic$ 
\begin{equation}\label{ic2Est}
\wh{\ic2}_{n}(b',b)=\frac{1}{n}\sum_{k=1}^n \left(\frac{L'}{L(b)}\I(C\neq \infty)C\right)_k\wh{\var}_{n,k}(b',b). 
\end{equation}

\chapter{\label{secExamp}Examples of parametrizations of IS} 
 In this chapter we introduce a number of parametrizations of IS, most
 of which shall be used in the theoretical reasonings or numerical experiments 
 in this work. 
\section{\label{secECM}Exponential change of measure} 
Exponential change of measure (ECM), also known as exponential tilting, is a popular 
method for obtaining a family of IS distributions from a given one. 
It has found numerous applications among others 
in IS for rare event simulation \cite{bucklew2004,asmussen2007stochastic} 
or for pricing derivatives in computational finance \cite{Jourdain2009, Lemaire2010}. 
% $\PQ_1$ as in the Introduction. 
%Most of results from this section are well-known, though we prove them for reader's convenience. 
%ECM has important applications among others in IS for rare event simulation, see e.g. \cite{bucklew2004, asmussen2007stochastic}. 
In this work by default all vectors (including gradients of functions)
are considered to be column vectors. % (however, gradient $\nabla f$ of a differentiable 
For some $l \in \N_+$, consider an $\R^l$-valued random vector $X$ on $\mc{S}_1$. We define the moment-generating function as
%\begin{equation} 
$b\in \R^l\to\Phi(b)=\E_{\PQ_1}(\exp(b^TX))$. 
%\end{equation}
Let $A$ be the set of all $b \in \R^l$ for which $\Phi(b) <\infty$.
%Such $\Phi$ is called moment-generating function. 
Note that $0 \in A$ and from the convexity of the exponential function, $A$ is convex. 
The cumulant generating function is defined as $\Psi(b)=\ln(\Phi(b))$, $b \in A$.
%From H\"{o}lder inequality it follows that $\Psi$ is convex on $A$. 
% \begin{condition}\label{condAInt}
% $A$ has nonempty interior.  
% \end{condition}
\begin{condition}\label{condtX2}
For each $b_1, b_2 \in A$ such that $b_1\neq b_2$, $(b_1-b_2)^TX$ is not $\PQ_1$ a.s. constant.
\end{condition}
%\begin{remark}
%\end{remark}
\begin{lemma}\label{lempsi}
$\Psi$ is convex on $A$ and it is strictly convex on $A$ only if Condition \ref{condtX2} holds.
\end{lemma}
\begin{proof} 
Let $b_1, b_2 \in A$ and $q_1,q_2 \in \R_+$ be such that $q_1+q_2 =1$. From H\"{o}lder's inequality 
\begin{equation}\label{phisum} 
\Phi(\sum_{i=1}^2q_ib_i)\leq \prod_{i=1}^2\Phi(b_i)^{q_i} 
\end{equation} 
and taking the logarithms  of the both sides  we receive 
\begin{equation}\label{psisum} 
\Psi(\sum_{i=1}^2q_ib_i)\leq \sum_{i=1}^2q_i\Psi(b_i).
\end{equation} 
Thus, $\Psi$ is convex. Equality in (\ref{phisum}) or equivalently in (\ref{psisum}) holds only if for some $a \in \R_+$, $\PQ_1$ a.s. 
$\exp(b_1^TX)=a\exp(b_2^TX)$ (see  page 63 in \cite{rudin1970}) or equivalently if for some $c \in \R$, $\PQ_1$ a.s. 
$(b_1-b_2)^TX=c$. $\Psi$ is strictly convex only if there do not exist $b_1, b_2 \in A$, $b_1\neq b_2,$ such that an
equality in (\ref{psisum}) holds, and thus only if Condition \ref{condtX2} holds. %$(b_1-b_2)^TX$ is not $\PQ_1$ a.s. constant, $b_1, b_2 \in A$, $b_1\neq b_2$,
%which for $A$ having nonempty interior is equivalent to Condition \ref{condtX}.
\end{proof}
\begin{condition}\label{condtX}
For each $t \in \R^l\setminus \{0\}$, $t^TX$ is not $\PQ_1$ a.s. constant. 
\end{condition}
Note that Condition \ref{condtX} implies Condition \ref{condtX2} and if $A$ has a nonempty interior then 
these conditions are equivalent. If $A$ contains some neighbourhood of zero, then $X$ has finite all mixed moments,
i.e. $\E(\prod_{i=1}^l|X_i|^{v_i})<\infty$, $v \in \N^l$. For $v \in \N^l$,
let us denote $\partial_v = \frac{\partial^{v_1+\ldots+v_l}}{\partial_{b_1}^{v_1}\ldots \partial_{b_l}^{v_l}}$.
\begin{condition}\label{condPartvm}
$A$ is open, $\Phi$ is smooth (i.e. infinitely continuously differentiable) on $A$, and for each $v \in N^l$ we have 
\begin{equation}\label{partvm}
\partial_v \Phi(b)=\E_{\PQ_1}(\partial_v\exp(b^TX))=\E_{\PQ_1}(\exp(b^TX)\prod_{i=1}^lX_i^{v_i}).\\ 
\end{equation}
\end{condition}
\begin{remark}\label{remPartv}
It is easy to show using inductively the mean value theorem and Lebesgue's dominated convergence theorem that 
Condition \ref{condPartvm} holds when $A=\R^l$ or when $\PQ_1([0,\infty)^l)=1$ and for some $\lambda>0$,  $A=(-\infty,\lambda)^l$.  
\end{remark}
We define the exponentially tilted family of probability distributions $\PQ(b)$, $b \in A$, corresponding to the above 
$\PQ_1$ and $X$ by the formula 
\begin{equation} 
\frac{d\PQ(b)}{d\PQ_1} = \exp(b^TX - \Psi(b)),\quad  b \in A. 
\end{equation} 
Note that $\PQ(0)=\PQ_1$ and 
\begin{equation}\label{lbecm} 
L(b):= \exp(-b^TX +\Psi(b)) = \frac{d\PQ_1}{d\PQ(b)}, \quad b \in A. 
\end{equation} 
Note that conditions \ref{condLmes}, \ref{condpqpq1}, and \ref{condg0} hold for the above $\PQ(b)$ and $L(b)$, $b \in A$. 
From Lemma \ref{lempsi}, for each $\omega \in \Omega_1$, $b\in A \to L(b)(\omega)$ is log-convex (and thus also convex)
and if Condition \ref{condtX2} holds, then it is strictly log-convex (and thus also strictly convex). 
Let us define means $\mu(b)=\E_{\PQ(b)}(X)$ and covariance matrices $\Sigma(b)=\E_{\PQ(b)}((X-\mu(b))(X-\mu(b))^T)$, for $b \in A$ for which 
they exist. Note that the functions $\Phi$, $\Psi$, $\Sigma$, and $\mu$ depend only on the law of $X$ under $\PQ_1$. 
If for some $b \in A$ it holds $\Sigma(b)\in \R^{l\times l}$, then we have 
$t^T\Sigma(b)t=\E_{\PQ(b)}((t^T(X-\mu(b)))^2)$, $t \in \R^l$, and thus $\Sigma(b)$ is positive definite only if 
Condition \ref{condtX} holds. 
When Condition \ref{condPartvm} holds, then we receive by direct calculation that $\nabla \Psi(b) = \mu(b)$ and $\nabla^2 \Psi(b)=\Sigma(b)$, $b \in A$. 

Let $U$ be an open subset of $\R^l$. 
The following well-known lemma is an easy consequence of the inverse function theorem. 
\begin{lemma}\label{lemInv} 
If $f:U \to \R^l$ is injective and differentiable with an
invertible derivative $Df$ on $U$, then $f$ is a diffeomorphism of the open sets $U$ and $f(U)$. 
\end{lemma} 
By $|\cdot|$ we denote the standard Euclidean norm.  
\begin{lemma}\label{lemConvFun} 
If $U$ is convex and a function $g:U \to \R$ is strictly convex and differentiable, then the function 
$b\in U\to \nabla g(b)$ is injective. 
\end{lemma}
\begin{proof}
If for some $b_1, b_2 \in U$, $b_1\neq b_2$, we had $\nabla g(b_1)=\nabla g(b_2)$, then for $v = \frac{b_2-b_1}{|b_2-b_1|}$ it would hold 
\begin{equation}
\left(\frac{dg(b_1+tv)}{dt}\right)_{t=0} =v^T\nabla g(b_1)= v^T\nabla g(b_2)=\left(\frac{dg(b_1+tv)}{dt}\right)_{t=|b_2-b_1|},
\end{equation}
which is impossible since $t \in  [0,|b_2-b_1|]\to g(b_1+tv)$ is strictly convex. % for $t \in [0,|b_2-b_1|]$ it must hold $b_1=b_2$.
\end{proof}
\begin{theorem}\label{thPsimu}
If conditions \ref{condtX} and \ref{condPartvm} hold, then 
$b\in A \to\mu(b)=\nabla \Psi(b)$ is a diffeomorphism of the open sets $A$ and $\mu[A]$.
\end{theorem}
\begin{proof}
From Condition \ref{condtX} and Lemma \ref{lempsi},
$\Psi$ is strictly convex. From Condition \ref{condPartvm},
$D\mu=\nabla^2\Psi=\Sigma$, which from \ref{condtX} and the above discussion 
is positive definite. Thus, for $U=A$ the thesis follows from Lemma \ref{lemConvFun}
for $g=\Psi$ and Lemma \ref{lemInv} for $f=\mu$. 
\end{proof}
% \end{lemma}
% which is greater than zero only if $t^TX$ is not $\PQ(b)$ a.s. constant
% or from (\ref{pqbpq1sim}), equivalently when it is not $\PQ_1$ a.s. constant.
% We thus receive that $\Sigma(0)$ is positive definite only if $\nabla^2 \Sigma(b)$ is positive definite for $b \in A$
% in which case
%For $b$ in the interior of $A$ we have  $\nabla_b \Psi(b)= \E_{\PQ(b)}(X)=\mu(b)$, 
%and $\nabla^2_b \Psi(b)= \Sigma(b)$, so that if $\Sigma(0)$ is positive definite, then 
%$\psi$ is strictly convex and thus $L(b)$ is strictly log-convex. 
%in all of which $\PQ_1=\PQ(0)$
%Consider the following condition.
% \begin{condition}\label{condIdX}
% $\mc{S}_1=(\Omega_1,\mc{F}_1)=(\R^l,\mc{B}(\R^l))$ and $X = \id_{\Omega_1}$. 
% \end{condition}
Some important special cases of ECM 
for $l=1$ are when $X$ has a binomial, Poisson, or gamma distribution under $\PQ(b)$, $b \in A$, 
while for general $l \in \N_+$ --- when $X$ has a multivariate normal 
distribution  (see page 130 in \cite{asmussen2007stochastic}). 
%In all these families $X$ has the same type of distribution under all $\PQ(b)$, $b \in A$, but with di
In all these cases,  %families, %$\mc{S}_1=(\Omega_1,\mc{F}_1)=(\R^l,\mc{B}(\R^l))$, $X = \id_{\Omega_1}$, and 
from Remark \ref{remPartv}, Condition 
\ref{condPartvm} holds. Furthermore, for the first three cases and 
non-degenerate multivariate normal distributions, Condition \ref{condtX} is satisfied 
and we have analytical formulas for $\mu^{-1}$. 
% For the special case of $\mc{S}_1=(\Omega_1,\mc{F}_1)=(\R^l,\mc{B}(\R^l))$ and $X = \id_{\Omega_1}$ the corresponding  
% exponentially tilted families $\mc{W}$ as above are simply the binomial, Poisson, gamma, or multivariate normal families. 
% Typically, when defining these families one assumes that $\mc{S}_1=(\Omega_1,\mc{F}_1)=(\R^l,\mc{B}(\R^l))$, $X = \id_{\Omega_1}$,
% and $\PQ_1$ is equal to binomial 
% Below we describe only the case of 
% gamma, Poisson, and multivariate normal distributions. 
In the gamma case, for some $\alpha,\lambda \in \R_+$, 
and $A=(-\infty,\lambda)$, for each $b \in A$, for $\lambda_b = \lambda -b$, under $\PQ(b)$, $X$ has a distribution with a density 
\begin{equation}
\frac{1}{\Gamma(\alpha)}\lambda_b^{\alpha}x^{\alpha-1}\exp(-\lambda_bx) 
\end{equation}
with respect to the Lebesgue measure on $(0,\infty)$. Furthermore, for each $b \in A$ it holds
$\Psi(b)= \alpha\ln(\frac{\lambda}{\lambda-b})$ and  
$\mu(b)=\frac{\alpha}{\lambda-b}$, and for each $x\in \mu[A]=\R_+$,
$\mu^{-1}(x)=\lambda - \frac{\alpha}{x}$.
In the Poisson case we have $A=\R$ and for some initial mean $\mu_0\in\R_+$, for each $b \in A$
we have $\mu(b) = \mu_0 \exp(b)$ and 
\begin{equation}
\PQ(b)(X=k) = \frac{\mu(b)^k}{k!}\exp(-\mu(b)), \quad k \in \N,
\end{equation}
i.e. $X \sim \Pois(\mu(b))$ under $\PQ(b)$. 
Furthermore, it holds %$L(b)=\exp(-bX + \mu_0(\exp(b)-1))$,
$\Psi(b)= \mu_0(\exp(b)-1)$, $b\in A$, $\mu^{-1}(x)=\ln(\frac{x}{\mu_0})$, $x \in \mu[A]= (0,\infty)$, 
and $\Sigma(b)=\mu(b)$, $b \in A$. In the multivariate normal case we have $A=\R^l$ and for 
$M\in \R^{l\times l}$ being some positive semidefinite covariance matrix and $\mu_0 \in \R^l$ some initial mean, 
for each $b \in A$, $\mu(b)=\mu_0 + Mb$ and under $\PQ(b)$,  $X\sim\ND(\mu(b),M)$. 
Moreover, it holds $\Psi(b)=b^T\mu_0 + \frac{1}{2}b^TMb$ and $\Sigma(b)=M$, $b \in A$. 
An important special case are non-degenerate normal distributions in which $M$ is positive definite, 
$\mu[A]=A$, and $\mu^{-1}(x)=M^{-1}(x-\mu_0)$, $x \in A$. In the standard multivariate normal case we have $M=I_l$ and $\mu_0=0$, so that $X\sim\ND(b, I_l)$
under $\PQ(b)$, $b\in A$. 

% For some $l'\in \N_+$ consider a matrix $D \in \R^{l\times l'}$. Let $\mc{W}$ be a family of exponentially tilted distributions 
% given by some $\PQ_1$ and $X$ as above with $A=\R^l$. Then, for $A'=\R^{l'}$, the family 
% $\mc{W}'=\{\PQ'(b')=\PQ(Db'): b'\in A'\}$ is an exponentially tilted family given by $\PQ_1'=\PQ_1$ and $X'=D^TX$. 
% From Remark \ref{remPartv}, Condition \ref{condPartvm} holds in the primed case. The primed 
% cumulant generating function $\Psi'$ fulfills $\Psi'(b')=\Psi(Db')$, 
% so that $\mu'(b')=\nabla \Psi'(b')=D^T\mu(Db')$, $b' \in A'$. 
% %if Condition \ref{condPartvm} holds in he unprimed case, then it also holds in primed case. 
% Condition \ref{condtX} holds in the primed case only if for each $t' \in \R^{l'}$, $t'\neq 0$, $t'^T(D^TX)=(Dt')^TX$ 
% is not a.s. constant, so assuming that this condition holds in the unprimed case, it holds in the primed one 
% only if columns of $D$ are linearly independent. %In such case for $x\in \mu'[A']$, 
% %Note that for $D$ with linearly independent columns, and 
% Note that if $X \sim \ND(0,I_l)$ under $\PQ_1$, then 
% $X' \sim \ND(0,D^TD)$ under $\PQ_1$. Such special case of ECM for $D$ with independent columns was used in IS for pricing 
% options using Black-Scholes model in \cite{Jourdain2009}. 
% above formulas for $\mu$ and $\mu'$ (or alternatively from (\ref{mund}) and the fact that under $\PQ_1$, $X' \sim \ND(0,D^TD)$),
% $\mu'^{-1}(x)=(D^TD)^{-1}x$, $x \in \mu'[A']=A'$. 

For an exponential tilting in which $A=\R^l$ we shall further need the following function defined for $a\in[0,\infty)$
\begin{equation}\label{fadef}
F(a)= \sup\{|\Psi(b)|:b\in \R^l,\ |b|\leq a\}. 
\end{equation}
For instance, in the multivariate standard normal case as above we have $\Psi(b)=\frac{|b|^2}{2}$ and thus
$F(a)=\frac{a^2}{2}$, while in the Poisson case $F(a)=\mu_0(\exp(a)-1)$. 
%\alpha\ln(\frac{\lambda}{\lambda-b})

\begin{remark}\label{remECMcost}
In some practical realizations of ECM,
the computation times on a computer needed to generate i.i.d. replicates of 
the IS estimator $ZL(b)$ under $\PQ(b)$ for different $b\in A$
are approximately equal to the same constant. %and do not depend on $b\in A$. %In particular it is approximately proportional to any constant theoretical cost variable $C$ 
% with the same proportionality constant for all $b\in A$. 
 %We shall assume such $C$ for ECM in several places 
This is typically the case e.g. when $X \sim \ND(0,I_l)$ under $\PQ_1$. 
In such a case one can often take the theoretical cost $C=1$.
%when explicitly stated. 
\end{remark}

\section{\label{secProd}IS for independently parametrized product distributions}
Let $n \in \N_+$. For each $i\in\{1,\ldots,n\}$, consider a probability distribution $\wt{\PQ}_{1,i}$ on
a measurable space $\mc{S}_{1,i}=(\Omega_{1,i},\mc{F}_{1,i})$, 
a nonempty set $A_i$, and parametric families of probabilities $\wt{\PQ}_i(b_i)$
and  densities $\wt{L}_i(b_i)=\frac{d\wt{\PQ}_{1,i}}{d\wt{\PQ}_i(b_i)}$, $b_i\in A_i$. 
Let us define the corresponding product measure $\PQ_1=\bigotimes_{i=1}^n\wt{\PQ}_{1,i}$, product parameter set
$A=\prod_{i=1}^nA_i$, and families of independently parametrized product probabilities $\PQ(b)=\bigotimes_{i=1}^n\wt{\PQ}(b_i)$
and  densities $L(b)=\prod_{i=1}^n\wt{L}_i(b_i)$, $b=(b_i)_{i=1}^n\in A$. 
Then, $\PQ(b) \sim \PQ_1$ and $L(b)=\frac{d\PQ_1}{d\PQ(b)}$, $b \in A$.

Let us further consider the special case of $\wt{\PQ}_i$ and $\wt{L}_i$ as above being the
exponentially tilted probabilities and densities given by some probabilities $\wt{\PQ}_{1,i}$ and random variables $\wt{X}_i$,
having moment-generating functions $\Phi_i$, and cumulant generating functions $\Psi_i$,  $i=1,\ldots,n$.
Then, $\PQ(b)$ and $L(b)$, $b\in A$, are the exponentially tilted probabilities and densities corresponding to the above
probability $\PQ_1$ and a random variable
$X(\omega)=(\wt{X}_i(\omega_i))_{i=1}^n$, $\omega\in \prod_{i=1}^n\Omega_{1,i}$, 
with a moment-generating function $\Phi(b)=\prod_{i=1}^n\Phi_i(b_i)$ and a cumulant generating function
$\Psi(b)=\sum_{i=1}^n\Psi_i(b_i)$. If Condition \ref{condtX} or \ref{condPartvm}
holds in the $i$th case for $i \in \{1,\ldots,n\}$, then such a condition holds  
also in the product case. If 
$\mu_i$ is the  mean function in the $i$th case, $i=1,\ldots,n$, then 
$\mu(b)=(\mu_i(b_i))_{i=1}^n$, $b= (b_i)_{i=1}^n\in A$, is such a mean function in the product case, and if all $\mu^{-1}_i$ exist, then  
for each $x=(x_i)_{i=1}^n \in \mu[A]=\prod_{i=1}^n\mu_i[A_i]$, $\mu^{-1}(x)=(\mu^{-1}_i(x_i))_{i=1}^n$. 

\section{IS for stopped sequences}\label{secStoppedGen}
\subsection{Change of measure for stopped sequences using a tilting process}\label{secStopped}
Let $\wt{\PU}_1$ be a probability measure on a measurable space $\wt{\mc{C}}=(\wt{E},\wt{\mc{E}})$,
let $\mc{C}=(E,\mc{E}):=\wt{\mc{C}}^{\N_+}$, let
$\eta=(\eta_i)_{i\in\N_+}=\id_E$ be the coordinate process on 
$E$, and let $\wt{\eta}_k=(\eta_i)_{i=1}^k$, $k \in \N_+$. %(see Section 36 in \cite{billingsley1979}).
Let $\PU$ be the unique probability measure on $\mc{C}$ such that
$\eta_1,\ \eta_2,\ \ldots$, are i.i.d. $\sim \wt{\PU}_1$ under $\PU$ (see Theorem 16, Chapter 9 in \cite{fristedt1997modern}). 
Let $\mc{F}_k=\sigma(\wt{\eta}_k)$, $k \in\N_+$, i.e. it is the natural filtration of $\eta$, 
and let $\mc{F}_0=\{\emptyset,E\}$, i.e. it is a trivial $\sigma$-field. 
For some $d \in \N_+$ and a nonempty set $B \in \mc{B}(\R^d)$, let conditions \ref{condLmes},
\ref{condpqpq1},  and \ref{condg0} hold for $A=B$, $\PQ_1 =\wt{\PU}_1$, and some probabilities $\PQ(b)$
and densities $L(b)$ denoted further as $\wt{\PU}(b)$ and $\wt{L}(b)$, $b \in B$.  
%such that for some $\wt{b_0} \in B$, $\wt{\PU}(\wt{b_0})=\wt{\PU}_1$.
Let $\kappa(b)=\wt{L}(b)^{-1}=\frac{d\wt{\PU}(b)}{d\wt{\PU}_1}$, $b \in B$.  
\begin{definition}\label{defJ}
We define $\mc{J}$ to be the set of all $\mc{S}(B)$-valued, $(\mc{F}_k)_{k\in \N}$-adapted stochastic processes $\lambda=(\lambda_k)_{k\in \N}$ on $\mc{C}$. 
%$(\mc{F}_k)_{k\in \N}$-adapted $B$-valued stochastic processes on $\mc{C}$ 
\end{definition} 
% We shall need the following Lemma 
% \begin{lemma} 
%  
% \end{lemma}
Processes $\lambda$ as in the above definition shall be called tilting processes. 
The following lemma follows from Lemma 7, Chapter 21 in \cite{fristedt1997modern}. See Definition 18,
Chapter 21 in \cite{fristedt1997modern} for the definition of Borel spaces. From Proposition 20 in that Chapter, $\mc{S}(B)$ is a Borel space. 
\begin{lemma}\label{lemMesPsi}
Let $\mc{\Psi}$ be a measurable space, $\mc{B}$ be a Borel space, $V$ be a $\Psi$-valued random variable, and $Y$ be a 
$\mc{B}$-valued, $\sigma(V)$-measurable random variable. Then, there exists a measurable function 
$f:\mc{\Psi}\to \mc{B}$ such that $Y=f(V)$.
\end{lemma}
Let further in this section $\lambda$ be as in Definition \ref{defJ}. 
From the above lemma there exist $h_0 \in B$ and $h_k:\wt{\mc{C}}^k \to \mc{S}(B)$, $k \in \N_+$, 
such that $\lambda_0 = h_0$ and $\lambda_k=h_k(\wt{\eta}_k)$, $k \in \N_+$, which let us further consider.  
% Note that such $\lambda$ are adapted to $(\mc{F}_k)_{k\in \N}$. If $\wt{\mc{C}}$ is a Borel space like $\mc{S}(\R^k)$ for some $k \in \N_+$, then $\mc{J}$ is
% simply the set of all $(\mc{F}_k)_{k\in \N}$-adapted $B$-valued stochastic processes on $\mc{C}$ (see Lemma 7 and Definition 18 in 
% Chapter 21 in \cite{fristedt1997modern}). 
% Let further in this section $\lambda$ 
% 
Let $\gamma_0=1$ and
\begin{equation}\label{gammandef}
\gamma_n=\prod_{k=0}^{n-1}\kappa(\lambda_{k})(\eta_{k+1}),\quad n \in \N_+. 
\end{equation}
For a nonempty set $T \subset [0,\infty)$ and
a filtration $\mc{G}_{t\in T}$ on a measurable space $(\Omega,\mc{G})$,  
let $\mc{G}_{\infty}:=\sigma(\bigcup_{t\in T}\mc{G}_t)$. A stopping time $\tau$ for $\mc{G}_{t\in T}$
is a $T \cup\{\infty\}$-valued random variable such that $\tau\leq t\in \mc{G}_t$, $t \in T$.
For such a $\tau$ one defines a $\sigma$-field 
\begin{equation}
\mc{G}_{\tau}=\{A\in \mc{G}_{\infty}: A\cap \{\tau \leq t\} \in \mc{G}_t,\ t \in T\}. 
\end{equation}
For $\tau$ being a stopping time for the filtration $(\mc{F}_k)_{k\in \N}$ as above it also holds
\begin{equation}\label{mcftau}
\mc{F}_{\tau}=\{B \in \mc{E}:B\cap\{\tau = n\} \in \mc{F}_n,\ n\in\N\}. 
\end{equation}
For a probability $\PS$ on $\mc{C}$ and such a $\tau$ 
we shall denote $\PS_{|\tau} = \PS_{|\mc{F}_{\tau}}$. 
Identifying each $n \in \N_+$ with a constant random variable we thus have $\PS_{|n} = \PS_{|\mc{F}_{n}}$.  
The following theorem is an easy consequence of Theorem 3, Chapter 22 in \cite{fristedt1997modern}. 
\begin{theorem}\label{thDiscrGirs} 
There exists a unique probability $\PV$ on $\mc{C}$ satisfying one of the following equivalent conditions.
\begin{enumerate}
 \item Under $\PV$, $\eta_1$ has density $\kappa(h_0)$ with respect to $\wt{\PU}_1$ and for each $k \in \N_+$, $\eta_{k+1}$
 has conditional density $\kappa(\lambda_k)$ with respect to $\wt{\PU}_1$ given $\mc{F}_k$ 
 (see Definition 14, Chapter 21 in \cite{fristedt1997modern}). 
 \item For each $n \in \N$, 
\begin{equation}\label{dqdp} 
 \frac{d\PV_{|n}}{d\PU_{|n}} =\gamma_n.
 \end{equation} 
\end{enumerate}
\end{theorem} 
% \begin{proof}
% Existence of a unique probability $\PV$ on $\mc{C}$ satisfying the first point
% follows from Theorem 3, Chapter 22 in \cite{fristedt1997modern}.  Thus, it is sufficient to prove that a probability $\PV$ on $\mc{C}$ satisfies the first point only if it satisfies 
% the second one. Assuming the first point, from Conditional 
% \begin{equation}
%  
% \end{equation} 
% \end{proof}
Let $\PV$ %$\gamma_n$, $n \in \N_+$, 
be as in the above theorem and let $\tau$ be a stopping time for $(\mc{F}_{n})_{n\in \N}$. 
\begin{lemma} \label{lemdpvdpu}
It holds 
\begin{equation}\label{tauinfsim}
\PV_{|\tau}\sim_{\tau <\infty} \PU_{|\tau}, 
\end{equation}
with
\begin{equation}\label{dpvdpu}
\I(\tau<\infty)\gamma_{\tau}=\left(\frac{d\PV_{|\tau}}{d\PU_{|\tau}}\right)_{\tau<\infty}. 
\end{equation} 
\end{lemma} 
\begin{proof} 
To prove (\ref{dpvdpu}) we notice that for each $B \in \mc{F}_{\tau}$ we have  
\begin{equation} 
\begin{split} 
\E_{\PU}(\I(B\cap \{\tau<\infty\}) \gamma_{\tau})&= \sum_{n=0}^{\infty}\E_{\PU}(\I(\{\tau =n\}\cap B) \gamma_n) \\ 
&= \sum_{n=0}^{\infty}\PV(\{\tau = n\}\cap B) = \PV(B\cap \{\tau<\infty\}),\\ 
\end{split} 
\end{equation} 
where in the second equality we used (\ref{mcftau}) and (\ref{dqdp}). 
%the fact that $\{\tau = n\}\cap B \in \mc{F}_n$, $n \in \N_+$. 
Now, (\ref{tauinfsim}) follows from (\ref{dpvdpu}) and Lemma \ref{lemsimA}. 
\end{proof} 
From the above lemma, if $\tau < \infty$ both $\PV$ and $\PU$ a.s., 
then $\PV_{|\tau}\sim \PU_{|\tau}$. 
In this work, a product over an empty set is considered to be equal one. 
For some  $\epsilon\in\R_+$, considered to avoid some technical problems 
as discussed above Condition \ref{condg0}, let us define 
\begin{equation}\label{Leps}
L=\I(\tau<\infty)\frac{1}{\gamma_\tau} + \epsilon \I(\tau=\infty) 
= \I(\tau<\infty)\prod_{k=0}^{\tau-1}\wt{L}(\lambda_{k}(\eta_{k+1}))+ \epsilon \I(\tau=\infty).
\end{equation}
Then, from Lemma \ref{lemdpvdpu} and the discussion in Section \ref{secBackDens} it holds 
\begin{equation}
L=\left(\frac{d\PU_{|\tau}}{d\PV_{|\tau}}\right)_{\tau<\infty}.  
\end{equation}
Let $Z$ be an $\R$-valued, $\mc{F}_\tau$-measurable random variable such that 
$\E_{\PU}(|Z|)<\infty$ (for short we shall also informally describe such a $Z$ as 
an $\R$-valued element of $L^1(\PU_{|\tau})$, 
see e.g. Chapter 20 in \cite{fristedt1997modern}). 
Let us assume that 
$\PU(Z\neq 0,\ \tau =\infty) =\PV(Z\neq 0,\ \tau =\infty)=0$, so that 
from (\ref{tauinfsim}),
$\PU_{|\tau} \sim_{Z\neq 0\vee\tau<\infty} \PV_{|\tau}$ and % from Lemma \ref{lemdpvdpu},
\begin{equation}\label{LRadon}
L=\left(\frac{d\PU_{|\tau}}{d\PV_{|\tau}}\right)_{Z\neq 0\vee\tau<\infty}.  
\end{equation}
Then, one can perform IS as in Section \ref{secIS} for $\PQ_1=\PU_{|\tau}$, $\PQ_2=\PV_{|\tau}$, and $L$ as above. 
Note that such a $\PQ_1$ is defined on $\mc{S}_1=(\Omega_1,\mc{F}_1)=(E,\mc{F}_\tau)$. 
\begin{remark}\label{rembettertau} 
Consider two stopping times $\tau_1,\tau_2$ for $(\mc{F}_n)_{n\in \N}$, such that 
$\tau_1\leq \tau_2$  and an $\R$-valued $Z \in L^1(\PU_{\tau_1})$ such that $\PU(Z\neq 0,\ \tau_2 =\infty) =\PV(Z\neq 0,\ \tau_2 =\infty)=0$. 
Then, we also have $Z \in L^1(\PU_{\tau_2})$ and $\PU(Z\neq 0,\ \tau_1 =\infty) =\PV(Z\neq 0,\ \tau_1 =\infty)=0$. 
Furthermore, denoting $L$ as in (\ref{Leps}) for $\tau=\tau_i$ as $L_{\tau_i}$, we have  
\begin{equation}\label{condZLtau}
\E_{\PV}(ZL_{\tau_2}|\mc{F}_{\tau_1})=ZL_{\tau_1}. 
\end{equation} 
Indeed, for each $D \in \mc{F}_{\tau_1}$ and $i=1,2$, from (\ref{LRadon}) it holds 
\begin{equation} %29
\E_{\PV}(ZL_{\tau_i}\I(D))=\E_{\PV}(ZL_{\tau_i}\I(D\cap\{Z\neq 0\}))=\E_{\PU}(Z\I(D\cap\{Z\neq 0\})). 
\end{equation} 
From (\ref{condZLtau}) and conditional Jensen's inequality we have  $\Var_{\PV}(ZL_{\tau_2})\geq \Var_{\PV}(ZL_{\tau_1})$,
i.e. using $\tau_1$ for IS as above leads to not higher variance than using $\tau_2$. Furthermore, $\E_{\PV}(\tau_1)\leq \E_{\PV}(\tau_2)$, 
so that, for the theoretical costs equal to the respective stopping times,
using $\tau_1$ also leads to not higher mean cost and inefficiency constant than $\tau_2$
(assuming that such constants are well-defined).
\end{remark}

\subsection{\label{secParamStopped}Parametrizations of IS for stopped sequences}
For some $l \in \N_+$ and a nonempty set $A \in \mc{B}(\R^l)$, let us consider a function 
\begin{equation}\label{lambdaAJ}
\lambda:A\to \mc{J}  
\end{equation}
(see Definition \ref{defJ}), 
called a parametrization of tilting processes. 
For each $b \in A$, let $\PV(b)$ and $\PV_{|\tau}(b)$ be given by $\lambda(b)$
similarly as $\PV$ and $\PV_{|\tau}$
are given by $\lambda$ in the unparametrized case in the previous section. %In particular, $\PV_{|\tau}(b_0)=\PU_{|\tau}$. 
Let $\PQ_1$ and $\mc{S}_1=(\Omega_1,\mc{F}_1)$ be as in the previous section. Let for each $b \in A$, $\PQ(b)=\PV_{|\tau}(b)$ 
and $L(b)$ be defined by formula (\ref{Leps}) 
but using $\lambda(b)$ in the place of $\lambda$.
Note that such an $L$ satisfies Condition \ref{condg0}. 

\begin{condition}\label{condbxmes}
For each $n \in \N$, $(b,x)\rightarrow \lambda_n(b)(x)$ is measurable from $\mc{S}(A)\otimes (E,\mc{F}_n)$ to 
$\mc{S}(B)$. % and for some $b_0\in A$, 
%$\lambda(b_0)=(\wt{b}_0)^{\N_+}$. 
\end{condition}

\begin{theorem}\label{thLmes}
Under Condition \ref{condbxmes}, Condition \ref{condLmes} holds for the  above $L$. 
\end{theorem} 

To prove the above theorem we will need the following lemmas. 

\begin{lemma}\label{lemGB}
 Let $\mc{G}$ be a $\sigma$-field, $A \in \mc{G}$, for some set $T$, 
 $C_t \in \mc{G}$, $t \in T$, and $\mc{C}=\sigma(C_t:t \in T)$. Then
 \begin{equation}\label{Acapsigma}
 A \cap \mc{C}:= \{A\cap C:C\in\mc{C}\}  \subset \sigma(A,\{A \cap C_t:t \in T\}).
 \end{equation}
\end{lemma}
 \begin{proof}
Let $\mc{A}=\sigma(A,\{A \cap C_t:t \in T\})$ and $\mc{D}=\{C \in \mc{C}: A\cap C \in \mc{A}\}$.  
It holds $\emptyset \in \mc{D}$ and  $C_t \in\mc{D}$, $t \in T$. If $B_i \in \mc{D}$, $i \in \N$, then 
$A\cap \bigcup_{i\in \N}B_i =\bigcup_{i\in \N}A\cap B_i \in \mc{A}$ and thus $\bigcup_{i\in \N}B_i \in \mc{D}$.
If $B\in \mc{D}$, then $A\cap B'= A \setminus(A \cap B) \in \mc{A}$ and thus $B' \in \mc{D}$. Thus, 
$\mc{D}= \mc{C}$ and we have  (\ref{Acapsigma}).
\end{proof}

\begin{lemma}\label{lemBi}
Let  $(B,\mc{B})$ be a measurable space, $I$ be a countable set,
$\mc{B}_i$ be a sub-$\sigma$-field of $\mc{B}$, $i \in I$, and 
$B_i \in \mc{B}_i$, $i\in I$, be such that $\bigcup_{i \in I}B_i=B$. 
Then, 
\begin{equation}
\mc{K}=\{C \subset B:\forall_{i \in I} C \cap B_i \in \mc{B}_i\}  
\end{equation}
is a sub-$\sigma$-field of $\mc{B}$. If further for each $i \in I$,
for some set $T_i$ and  
for some $C_{i,j}\in \mc{B}_i$, $j \in T_i$, it holds 
%\begin{equation}\label{mcbidef}
$\mc{B}_i= \sigma(C_{i,t}:\ t \in T_i)$,
%\end{equation}
then
\begin{equation}\label{Kgen}
\mc{K}= \sigma(B_i,\{C_{i,t}\cap B_i:t \in T_i\}:i \in I).   
\end{equation}
\end{lemma}
\begin{proof}
For each $C \in \mc{K}$ it holds $C \cap B_i \in \mc{B}$, $i \in I$, and thus
$C= \bigcup_{i\in I}C \cap B_i \in \mc{B}$.  
It holds $\emptyset \in \mc{K}$, for $A_i\in \mc{K}$, $i \in I$, 
$(\bigcup_{i \in \N} A_i) \cap B_j = \bigcup_{i \in \N} (A_i \cap B_j) \in \mc{B}_j$, $j \in I$,
and for $A \in \mc{K}$, $A'\cap B_i=B_i \setminus (B_i \cap A) \in \mc{B}_i$, $i \in I$,
so that $\mc{K}$ is a sub-$\sigma$-field of $\mc{B}$. %We have
%\begin{equation}\label{mckidef}
%B_i\cap\mc{B}_i=\{C \in \mc{B_i}:C \subset B_i\}, \quad i \in I. 
%\end{equation}
For $A \in \mc{K}$ we have $A = \bigcup_{i \in I} A \cap B_i$ and $A \cap B_i \in B_i\cap\mc{B}_i$, $i \in I$. 
Furthermore, $B_i\cap\mc{B}_i \subset \mc{K}$, $i \in I$. Thus,
%\begin{equation}\label{Kgen}
$\mc{K}= \sigma(B_i\cap\mc{B}_i: i \in I)$.
%\end{equation}
Therefore, (\ref{Kgen}) follows from the fact that from Lemma \ref{lemGB}
\begin{equation}
B_i\cap\mc{B}_i \subset \sigma(B_i,\{C_{i,t}\cap B_i:t \in T_i\}), \quad i \in I. 
\end{equation}
\end{proof}

\begin{lemma}\label{lemMesDFtau}
Let $(D,\mc{D})$ be a measurable space, $\mc{F}_n$, $n \in \N$, be a filtration in a measurable space $(\Omega,\mc{F})$, and $\tau$ be a stopping time for 
such a filtration. Then,
\begin{equation}\label{longProd}
\mc{D}\otimes \mc{F}_{\tau}= \{C \in D\times\Omega :\forall_{k\in \N\cup\{\infty\}} C \cap (D\times \{\tau =k\}) \in \mc{D}\otimes \mc{F}_k\}.
\end{equation}
\end{lemma}
\begin{proof}
Let us denote the right-hand side of \ref{longProd} as $\mc{K}$. Then, it is equal to such a
$\mc{K}$ from Lemma \ref{lemBi} for $B=D\times\Omega$, $\mc{B}=\mc{D}\otimes \mc{F}_{\infty}$,
$I=\N\cup\{\infty\}$, and for $B_i= (D\times \{\tau =i\})$ and $\mc{B}_i = \mc{D}\otimes \mc{F}_i$, $i \in I$, which let us further consider.  
From that lemma, 
\begin{equation}
\mc{K}= \sigma(C_1 \times (C_2\cap \{\tau =i\}):C_1\in \mc{D},\ C_2\in \mc{F}_{i},i \in I).   
\end{equation}
By definition, $\mc{D}\otimes \mc{F}_{\tau}=\sigma(C_1\times C_2:C_1 \in \mc{D},\ C_2 \in \mc{F}_{\tau})$. 
For each $C_1 \in \mc{D}$, $C_2 \in \mc{F}_{\tau}$, and $i\in I$ it holds
$(C_1\times C_2) \cap (D\times \{\tau =i\})=C_1\times (C_2\cap\{\tau =i\})\in \mc{D}\otimes \mc{F}_i$, so that 
$\mc{D}\otimes \mc{F}_{\tau}\subset \mc{K}$.  
For each $C_1\in \mc{D}$, $i \in I$, and $C_2\in \mc{F}_{i}$, it holds 
$C_1 \times (C_2\cap \{\tau =i\}) \in \mc{D}\otimes \mc{F}_{\tau}$, so that 
$\mc{K} \subset \mc{D}\otimes \mc{F}_{\tau}$. 
\end{proof}

Let us now provide a proof  of Theorem \ref{thLmes}. 
\begin{proof} 
From Condition \ref{condbxmes} and Condition \ref{condLmes} holding for $\wt{L}$, for $n \in \N$, for $\gamma_n(b)$ given by 
$\lambda(b)$, $b \in A$, in the way that $\gamma_n$ is given by $\lambda$ in the previous section, 
$(b,x)\to \gamma_n(b)(x)$ is measurable from $\mc{S}(A)\otimes (E,\mc{F}_n)$ to $\mc{S}(\R)$. 
Let $B \in \mc{B}(\R)$. For $n \in \N$ it holds 
\begin{equation} 
L^{-1}(B)\cap(A \times \{\tau=n\}) = \gamma_n^{-1}(B)\cap(A \times \{\tau=n\}) \in \mc{B}(A)\otimes \mc{F}_n. 
\end{equation} 
Furthermore, $L^{-1}(B)\cap(A \times \{\tau=\infty\})$ is equal to $A \times \{\tau=\infty\}$ if 
$\epsilon \in B$ and to $\emptyset$ otherwise. Thus, from Lemma \ref{lemMesDFtau}, $L^{-1}(B)\in \mc{B}(A)\otimes \mc{F}_\tau$. 
\end{proof}

\begin{condition}\label{condzntau} 
$\PQ_1(Z\neq0,\ \tau=\infty)=0$ and $\PQ(b)(Z\neq0,\ \tau=\infty)=0$, $b\in A$.  
\end{condition}
\begin{condition}\label{equivCond}
It holds  $\tau < \infty$, $\PQ_1$ a.s. and $\PQ(b)$ a.s., $b\in A$. 
\end{condition}

\begin{remark}\label{remztauequiv}
From (\ref{tauinfsim}), Condition \ref{condB1} is satisfied 
for $B_1 = \{\tau <\infty\}$. Thus, for such a $B_1$, Condition \ref{condZB} 
is equivalent to Condition \ref{condzntau}. In particular,  Condition \ref{condpqpq1} is implied by Condition  
\ref{equivCond}. 
\end{remark}

\begin{definition}\label{defLin}
Let $B=\R^d$, $A=\R^l$, and let
an $\R^{d\times l}$-valued 
process $\Lambda=(\Lambda_k)_{k\geq 0}$ on $\mc{C}$ 
be such that for each  $j\in\{1,\ldots,l\}$, 
$(((\Lambda_k)_{i,j})_{i=1}^d)_{k \in \N} \in \mc{J}$. 
Then, we define the corresponding linear parametrization $\lambda$ of 
tilting processes as in (\ref{lambdaAJ}) to be such that
\begin{equation}\label{lambdalambda}
\lambda_k(b) =\Lambda_kb,\quad  k \in \N, b \in A.
\end{equation}
\end{definition}
Note that for $\lambda$ as in the above definition 
Condition \ref{condbxmes} holds and we have $\PQ(0)=\PQ_1$. 

\subsection{Change of measure for Gaussian stopped sequences using a tilting process}\label{secGSS}
Let $\wt{\PU}_1=\ND(0,I_d)$, $X=\id_{\R^d}$, and let $\wt{\PU}(b)$ and $\wt{L}(b)$, $b \in B:=\R^d$, 
be the exponentially tilted distributions and densities corresponding to such $X$, $\PQ_1=\wt{\PU}_1$,  and $A=B$, as in Section \ref{secECM}. 
For such distributions and densities, let us consider the corresponding definitions for stopped sequences 
for some tilting process $\lambda \in \mc{J}$ and $h_k$, $k \in \N$, as in Section \ref{secStopped}.  
%for the above $\wt{\PU}(b)$ and  $\wt{L}(b)$, $b \in B$. 
In particular, $\kappa(b)(x)=\exp(-\frac{1}{2}|b|^2+b^Tx)$. 
Let $\dot{\eta}_{k} = \eta_{k} - \lambda_{k-1}$, $k \in \N_+$. 
The following theorem is a discrete version of Girsanov's theorem. 
\begin{theorem}\label{thNormGirs}
Under $\PV$, the random variables $\dot{\eta}_{k}$, $k \in \N$, are i.i.d. $\sim \ND(0,I_d)$.  
\end{theorem}
\begin{proof}
Writing $h_k$ in the place of $h_k(x_1,\ldots,x_k)$, $k\in \N_+$,
%Denoting $h_k=h_k(x)$, $x \in (\R^d)^k$, 
for each $n\in\N_+$ and $\Gamma \in (\mc{B}(\R^{d}))^n$ %, and $\wt{\Gamma} =\Gamma + \Lambda_n$  
\begin{equation} 
\begin{split} 
 \PV((\dot{\eta}_i)_{i=1}^n \in \Gamma)&=\E_{\PU}(\I((\dot{\eta}_i)_{i=1}^n \in \Gamma)\gamma_n)\\ 
 &= \int_{(\R^d)^n} \I((x_k- h_{k-1})_{k=1}^n \in \Gamma)\frac{1}{(2\pi)^{nd/2}}\exp(-\frac{1}{2}\sum_{k=1}^n (x_k - h_{k-1})^{2}) \,\mathrm{d}x_n\ldots \,\mathrm{d}x_1\\ 
 &= \int_{(\R^d)^n} \I(y \in \Gamma)\frac{1}{(2\pi)^{nd/2}}\exp(-\frac{1}{2}\sum_{k=1}^n y_k^{2})\,\mathrm{d}y_n\ldots \,\mathrm{d}y_1 \\ 
 &= \PU((\eta_i)_{i=1}^n \in \Gamma),
\end{split} 
\end{equation} 
where we used Fubini's theorem and a sequence of changes of variables $y_k(x_k) = x_k - h_{k-1}$, $k =1, \ldots,n$, 
each of which is a diffeomorphism with a Jacobian 1. 
\end{proof}
Let us consider a function $\pi:E\to E$ such that $\pi=(\dot{\eta}_i)_{i\in \N_+}$. 
Its inverse function $\pi^{-1}$ is given by the formula 
\begin{equation}\label{piinv}
\pi^{-1}= (\eta_k + \lambda_{k-1}(\pi^{-1}))_{k \in \N_+},
\end{equation}
or in more detail we have $\pi^{-1}=(\ddot{\eta}_i)_{i=1}^{\infty}$
for $\ddot{\eta}_i=\eta_i +\ddot{\lambda}_{i-1}$, $i \in \N_+$, where $\ddot{\lambda}_0=h_0$ and 
$\ddot{\lambda}_k=h_k((\ddot{\eta}_i)_{i=1}^k)$, $k\in\N_+$. 
Note that both $\pi$ and $\pi^{-1}$ are measurable from $\mc{U}_n:=(E,\mc{F}_n)$ to $\mc{U}_n$, $n\in\N$,
i.e. $\pi$ is an isomorphism of $\mc{U}_n$, $n \in \N$, and thus also of $\mc{U}_{\infty}:=(E,\mc{F}_{\infty})=\mc{C}$. From Theorem 
\ref{thNormGirs} we have  
$\PU(B)=\PV(\pi^{-1}[B]),\ B \in \mc{E}$, so that 
\begin{equation}\label{pupv}
\PU(\pi[B]) =\PV(B),\quad B \in \mc{E}. 
\end{equation} 
In particular, for each random variable $Y$ on $\mc{C}$ 
the distribution of $Y\pi^{-1}:=Y(\pi^{-1})$ under $\PU$ is the same as of $Y$ under $\PV$. % and 
\begin{remark}
For $\overrightarrow{\pi}$ denoting the image function of $\pi$, we have
\begin{equation}
\begin{split}
\overrightarrow{\pi}[\mc{F}_{\tau}] &= \{\pi[B]:\ B\in \mc{F}_\tau\}\\ 
&= \{\pi[B]:\ B \in \mc{F}_{\infty},\ B \cap \{\tau=k\} \in \mc{F}_k,\ k \in \N\cup\{\infty\} \}\\ 
&= \{C \in \mc{F}_{\infty}:\ \pi^{-1}[C] \cap \{\tau=k\} \in \mc{F}_k,\ k \in \N\cup\{\infty\}\}\\ 
&= \{C \in \mc{F}_{\infty}:\ C \cap \{\tau\pi^{-1}=k\} \in \mc{F}_k,\ k \in \N\cup\{\infty\}\}\\ 
&=\mc{F}_{\tau \pi^{-1}}, \\ 
\end{split} 
\end{equation} 
where in the fourth equality we used the fact that $\pi$ is an isomorphism of $\mc{U}_n$, $n\in \N\cup\{\infty\}$. 
In particular, if a random variable $Y$ on $\mc{C}$ is $\mc{F}_\tau$-measurable, then 
$Y\pi^{-1}$ is $\mc{F}_{\tau\pi^{-1}}$-measurable, i.e. it depends only on the information available until the time $\tau\pi^{-1}$. 
\end{remark}
For some parametrization $\lambda(b)$, $b\in A$, of tilting processes as in (\ref{lambdaAJ}),
let $\pi_b$ be given by $\lambda(b)$ in the way that $\pi$ is given by $\lambda$ above. 
Let further $\PQ(b)$ and $L(b)$, $b\in A$, correspond to such a parametrization
as in Section \ref{secParamStopped}, and let $\mc{S}_1=(\Omega_1,\mc{F}_1)$ and $\PQ_1$ be as in that section. 
%and $\mc{F}_1=\mc{F}_\tau$ be as in that section. 
% \begin{condition}\label{condxi}
% We have (\ref{abrl}) and a probability $\PR_1$ on a measurable space $\mc{C}_1$ and
% %us assume that there exists some measurable function 
% %\begin{equation}\label{xidef} 
% $\xi: \mc{C}_1 \otimes \mc{S}(A)\rightarrow \mc{S}_1$ 
% %\end{equation} 
% are such that for each $b \in A$
% \begin{equation}\label{ximin}
% \PQ(b)(B)=\PR_1(\xi(\cdot,b)^{-1}[B]),\quad B \in \mc{F}_1, 
% \end{equation}
% or equivalently, for each random variable $X \sim \PR_1$, $\xi(X,b)\sim \PQ(b)$, $b \in A$. % (see (16.16) \cite{billingsley1979}). 
% \end{condition}
Let $\xi: E \times A \to E$ be such that 
\begin{equation}\label{xigss}
\xi(\eta, b) = \pi_b^{-1}(\eta),\quad b \in A.
\end{equation}
% \begin{lemma}\label{lemmesxi}
% Under Condition \ref{condbxmes}, $\xi$ is measurable 
% from $\mc{C}\otimes\mc{S}(A)$ to $\mc{C}$.
% \end{lemma}
% \begin{proof}
% From the discussion below (\ref{piinv}) it easily follows that
% $\xi$ is measurable from $\mc{U}_n\otimes\mc{S}(A)$ to $\mc{U}_n$, $n \i n\N_+$,
%\end{proof}
\begin{theorem}
Under Condition \ref{condbxmes} and the above definitions, Condition \ref{condxi} holds
for $\mc{C}_1= \mc{C}$ and $\PR_1=\PU$.  
\end{theorem}
\begin{proof}
From  (\ref{piinv}) it follows by induction that
$\xi$ is measurable from $\mc{U}_n\otimes\mc{S}(A)$ to $\mc{U}_n$, $n \in\N$.
Thus, it is also measurable from 
$\mc{C}\otimes\mc{S}(A)$ to $\mc{C}$ and due to
$\mc{F}_1=\mc{F}_\tau\subset \mc{E}$, also to $\mc{S}_1$.
Furthermore, from (\ref{pupv}), 
\begin{equation}
\PU(\xi(\cdot, b)^{-1}[B])= \PV_{|\tau}(b)(B),\quad B \in \mc{F}_\tau, 
\end{equation}
i.e.  (\ref{ximin}) holds. 
\end{proof}

For each random variable $Y$ on $\mc{C}$ and $b\in A$, let us denote %like $\tau$ or $Z$, we shall denote  
\begin{equation}\label{ybdef}
Y^{(b)}= Y\pi_b^{-1}= Y(\xi(\cdot, b)).
\end{equation}
Note that from (\ref{piinv}) we have 
\begin{equation}\label{xietab}
(\xi(\eta, b))_k=\eta_k^{(b)} = \eta_k + \lambda_{k-1}(b)^{(b)},\quad k \in \N_+. 
\end{equation}
% \begin{equation}\label{piinv}
% \pi^{-1}= (\eta_k + \lambda_{k-1}(\pi^{-1}))_{k \in \N_+},
% \end{equation}
%It holds 
%$(\xi(\cdot, b))_k = \eta + \lambda(b)^{(b)}$  
For each $b \in A$ it holds
\begin{equation}\label{LLETGS}
L(b)=\I(\tau < \infty)\exp(\sum_{k=0}^{\tau-1} (\frac{1}{2}|\lambda_k(b)|^2 - \lambda_k(b)^T \eta_{k+1}))+\I(\tau =\infty)\epsilon.
\end{equation}
From (\ref{xietab}) and (\ref{LLETGS}), for each $b',b \in A$ we have
\begin{equation}\label{LBBP}
L(b)^{(b')}= \I(\tau^{(b')} < \infty)
\exp(\sum_{k=0}^{\tau^{(b')}-1} (\frac{1}{2}|\lambda_k(b)^{(b')}|^2 - (\lambda_k(b)^{(b')})^T (\eta_{k+1}+ \lambda_k(b')^{(b')})))+
\I(\tau^{(b')} =\infty)\epsilon
\end{equation}
 and in particular
\begin{equation}\label{LBB}
L(b)^{(b)}= \I(\tau^{(b)} < \infty)\exp(-\sum_{k=0}^{\tau^{(b)}-1} (\frac{1}{2}|\lambda_k(b)^{(b)}|^2 + (\lambda_k(b)^{(b)})^T \eta_{k+1}))+
\I(\tau^{(b)} =\infty)\epsilon. 
\end{equation}

\subsection{\label{secLETS}Linearly parametrized exponential tilting for stopped sequences} 
Let $\wt{\PU}_1$, $\wt{\mc{C}}=(\wt{E},\wt{\mc{E}})$, $\wt{X}$, $\wt{\PU}(b)$, $\wt{L}(b)$, 
$\wt{\Psi}(b)$, $b \in B=\R^d$, and $\wt{F}$ be 
as some $\PQ_1$, $\mc{S}_1=(\Omega_1,\mc{F}_1)$, $X$, $\PQ(b)$, $L(b)$, $\Psi(b)$, $b \in A=B$, 
and $F$ in the ECM setting in Section \ref{secECM}. 
Let $\lambda$ be a linear paramatrization of tilting 
processes corresponding to some $\Lambda$ as in Definition \ref{defLin} 
and consider the corresponding families of probabilities 
$\PQ(b)$ and densities  $L(b)$, $b \in A$, as in Section \ref{secParamStopped}. 
Note that we now have from (\ref{Leps}), for $U(b)=\I(\tau<\infty)\sum_{k=0}^{\tau-1}\wt{\Psi}(\lambda_k(b))$, $b \in A$, 
and $H = -\I(\tau<\infty)\sum_{k=0}^{\tau-1} (\wt{X}(\eta_{k+1}))^T\Lambda_k$, that 
\begin{equation}\label{lnlbLETS}
L(b)= \I(\tau<\infty)\exp(U(b)+ Hb) +\epsilon\I(\tau=\infty),\quad b\in A. 
\end{equation}
We shall call the above parametrization of IS the linearly parametrized exponentially tilted stopped sequences (LETS) setting.  
Its special case in which $\wt{\PU}_1=\ND(0,I_d)$ and $\wt{X}= \id_{\R^d}$ 
shall be called the linearly parametrized exponentially tilted Gaussian stopped sequences (LETGS) setting. Note that the LETGS setting
is a special case of the parametrized IS for Gaussian stopped sequences as in Section \ref{secGSS}. 
In the LETGS setting $H = -\I(\tau<\infty)\sum_{k=0}^{\tau-1} \eta_{k+1}^T\Lambda_k$ and for 
$G:= \I(\tau<\infty)\frac{1}{2}\sum_{k=0}^{\tau-1}\Lambda_k^T\Lambda_k$ 
we have  $U(b)=b^TGb$, $b\in A$, so that %formula \ref{LLETGS} reduces to 
\begin{equation}\label{lnlG}
\begin{split}
L(b)=%\I(\tau < \infty)\exp(\sum_{k=0}^{\tau-1} (\frac{1}{2}|\lambda_k(b)|^2 - \lambda_k(b)^T \eta_{k+1}))+\I(\tau =\infty)\epsilon\\ 
\I(\tau < \infty)\exp(b^TGb+Hb) + \I(\tau =\infty)\epsilon, \quad b\in A.\\  
\end{split}
\end{equation}
Furthermore, we have
$G^{(b')}=\I(\tau^{(b')}<\infty)\frac{1}{2}\sum_{k=0}^{\tau^{(b')}-1}(\Lambda_k^{(b')})^T\Lambda_k^{(b')}$ and 
\begin{equation}\label{hbp}
H^{(b')} = -\I(\tau^{(b')}<\infty)\sum_{k=0}^{\tau^{(b')}-1}(\eta_{k+1} + \Lambda_k^{(b')}b')^T\Lambda_k^{(b')}, 
\end{equation}
and formula (\ref{LBBP}) can be rewritten as 
\begin{equation}\label{lbbp2}
L(b)^{(b')}=\I(\tau^{(b')} < \infty)\exp(b^TG^{(b')}b+H^{(b')}b) + \I(\tau^{(b')} =\infty)\epsilon.\\   
\end{equation}
%\label{LBBP}
\begin{remark}\label{remInfLb}
Note that in the LETGS setting, on $\tau <\infty$ we have 
\begin{equation}\label{infLb}
%\begin{split}
\inf_{b\in \R^l} \ln(L(b)) \geq \sum_{k=1}^{\tau}\inf_{y\in \R^d} (\frac{1}{2}|y|^2 - \eta_{k}^Ty) 
= -\frac{1}{2}\sum_{k=1}^{\tau}|\eta_{k}|^2  \in \R.\\
%\end{split}
\end{equation}
 \end{remark}
 \begin{remark}
  In our numerical experiments performing IS for computing expectations of functionals of an Euler scheme in the LETGS setting, 
  the simulation times were roughly proportional to the replicates of $\tau$ under $\PQ(b)$. 
  Thus, on several occasions in this work when dealing with the 
  LETGS setting we shall consider the theoretical cost $C=s\tau$ for some $s \in \R_+$.
 \end{remark}

\begin{remark}\label{remConstMat} 
Consider the special case of the LETS setting in which $\Lambda$ is a sequence of constant matrices and 
$\tau=n \in \N_+ $ is deterministic. 
% Returning to the general LETS setting, consider its special case for $\Lambda$ being a sequence of constant matrices and
% $\tau=n \in \N $ being deterministic. %Let the matrix $D\in\R^{mn\times l}$ consist of consecutive rows of matrices 
Then, for the above $\PQ(b)$ and $L(b)$, $b\in A$, 
a family of probabilities $\PQ'(b), b\in A$, on $\wt{\mc{C}}^{n}$ 
such that $\PQ'(b)(\wt{\eta}_n[C])=\PQ(b)(C)$, $C \in \mc{F}_n$, $b\in A$, 
and $L':A\times\wt{E}^n\to \R$ such that $L'(b)(\wt{\eta}_n)=L(b)$, 
are the exponentially tilted families of probabilities and densities 
corresponding to $\PQ_1':=\wt{\PU}_1^n$ 
and $X'(\omega):=\sum_{i=1}^n\Lambda_{i-1}^T\wt{X}(\omega_i)$, 
$\omega=(\omega_i)_{i=1}^n \in \wt{E}^n$, as in Section \ref{secECM}. 
Note that for each random variable $Y'$ on $\wt{\mc{C}}^n$, $Y=Y'(\wt{\eta}_n)$ is an $\mc{F}_n$-measurable random variable
with the same distribution  under $\PQ(b)$ as of $Y'$ under $\PQ(b)'$, $b \in A$. 
% $Y$ under $\PQ(b)$, $b \in A$. 
% Note that from Lemma \ref{lemMesPsi}, 
% for each $\mc{F}_n$-measurable random variable $Y$ with values in some Borel space $\mc{B}$, 
% there exists $Y':\wt{\mc{C}}^n \to \mc{B}$ such that $Y=Y'(\wt{\eta}_n)$. Note also that the distribution of such $Y'$ under $\PQ(b)'$ is the same as of 
% $Y$ under $\PQ(b)$, $b \in A$. 
% Furthermore, there is a one-to-one correspondence between the above $\PQ(b)$, $L(b)$, $b \in A$, and $Y$, 
% and $\PQ(b)'$, $L'(b)$, $b \in A$, and $Y'$. 
Note also that if further $\tau=1$ and $\Lambda_0=I_d$, then $\PQ_1':=\wt{\PU}_1$, $L'(b)=\wt{L}(b)$ 
$\PQ'(b)=\wt{\PU}(b)$, $b \in A$, and $X'=\wt{X}$. 
\end{remark}

% Furthermore, if $\wt{X} \sim \ND(0,I_d)$ under $\wt{\PU}_1$, then $X' \sim \ND(0, \sum_{i=0}^{n-1}\Lambda_{i}^T\Lambda_{i})$ 
% under $\PQ_1'$. 
%For $\wt{\PU}_1\sim \ND(0,I_d)$ it 
% Returning to the general LETS setting, consider its special case for $\Lambda$ being a sequence of constant matrices and
% $\tau=n \in \N $ being deterministic. Let the matrix $D\in\R^{mn\times l}$ consist of consecutive rows of matrices 
% $\Lambda_0, \Lambda_1,\ldots,\Lambda_{n-1}$ and a mapping $g_n:\mc{C}\to \wt{\mc{C}}^n$ be such that  $g_n(\eta)=\wt{\eta}_n$.
% Then, for the above constructed family $\{\PQ(b): b\in A\}$, 
% the family of distributions $\{\PQ'(b):b\in A\}$ on $\wt{\mc{C}}^{n}$ 
% such that $\PQ'(b)(g_n[C])=\PQ(b)(C)$, $C \in \mc{F}_n$, $b\in A$, 
% is equal to the exponentially tilted family of distributions $\mc{W}'$ corresponding to $\PQ_1=\wt{\PU}_1^n$, 
% $X((\omega_i)_{i=1}^n)=(\wt{X}(\omega_i))_{i=1}^n$, and the above $D$, as at the end of Section \ref{secECM}. 
% Note that when $\tau=1$ and $\Lambda_1=I_d$, then such constructed $\mc{W}'$ coincides with 
% the exponentially tilted family $\{\wt{\PU}(b):b\in A\}$ used for its construction. 

\section{\label{secBrown}IS for a Brownian motion up to a stopping time}
Let us now briefly discuss IS for computing expectations 
of functionals of a Brownian motion up to a stopping time. %It will be useful 
%for obtaining some intuitions for IS for computation of Euler scheme expectations in the further sections. 
For some $d \in \N_+$, let $B=(B_t)_{t \geq 0}$ be the coordinate process on the Wiener space $\mc{C}([0,\infty),\R^d)$, whose measurable space let us denote as $\mc{W}$.
Let $(\wt{\mc{F}}_t)_{t\geq 0}$ be the natural filtration of $B$.
Let $\wt{\PU}$ be the unique probability on $\mc{W}$ for which $B$ is a $d$-dimensional Brownian motion 
(see Chapter 1, Section 3 in \cite{revuz1999continuous}). 
For a probability $\PS$ on $\mc{W}$ and a stopping time $\tau$ for $(\wt{\mc{F}}_t)_{t\geq 0}$, we denote $\PS_{|\tau}=\PS_{|\wt{\mc{F}}_\tau}$. 
From Girsanov's theorem, if $(\wt{\lambda}_t)_{t\geq 0}$ is a predictable locally square-integrable $\R^d$-valued process on $\mc{W}$ for which
\begin{equation}
\wt{\gamma}_t=\exp\left(\int_{0}^t \wt{\lambda}_s^T \, \mathrm{d}B_s - \frac{1}{2}\int_{0}^t |\wt{\lambda}_s|^2 \mathrm{d}s\right), \quad t \geq 0,
\end{equation}
is a martingale under $\wt{\PU}$ (for which e.g. Novikov's condition suffices), then from Kolmogorov's extension theorem there exists 
a unique measure $\wt{\PV}$ on $\mc{W}$ such that $\frac{d\wt{\PU}_{|t}}{d\wt{\PV}_{|t}}=\wt{\gamma}_t$, $t \geq 0$. 
Furthermore, 
\begin{equation}\label{wtBsol}
\wt{B}_t = B_t-\int_{0}^t\wt{\lambda}_s\, \mathrm{d}s,\quad t\geq 0,
\end{equation}
is a Brownian motion under $\wt{\PV}$. 
From Proposition 1.3, Chapter 8 in \cite{revuz1999continuous}, for a stopping time $\wt{\tau}$ for $(\wt{\mc{F}}_t)_{t \geq 0}$,
we have 
$\I(\wt{\tau}<\infty)\wt{\gamma}_{\wt{\tau}}=\left(\frac{d\wt{\PV}_{|\wt{\tau}}}{d\wt{\PU}_{|\wt{\tau}}}\right)_{\wt{\tau}<\infty}$ 
and thus $\wt{L}=\I(\wt{\tau}<\infty)\frac{1}{\wt{\gamma}_{\wt{\tau}}}=\left(\frac{d\wt{\PU}_{|\wt{\tau}}}{d\wt{\PV}_{|\wt{\tau}}}\right)_{\wt{\tau}<\infty}$, 
similarly as in the discrete case. 
Thus, if for some $\R$-valued $\wt{Z} \in L^1(\wt{\PU}_{|\wt{\tau}})$ we have 
$\wt{\PU}$ and $\wt{\PV}$ a.s. that $\wt{\tau}=\infty$ implies 
$\wt{Z}=0$, then 
$\wt{\PU}_{|\wt{\tau}} \sim_{Z\neq 0} \PV_{|\wt{\tau}}$ and we can perform IS 
for computing $\E_{\wt{\PU}}(\wt{Z})$ analogously as in the discrete case. 
For adaptive IS, for some $l \in \N_+$, we can use e.g. 
linear parametrization $\wt{\lambda}_t(b) = \wt{\Lambda}_tb$, $b\in A:=\R^l$ of tilting processes for some  
$\R^{d\times l}$-valued predictable
process $(\wt{\Lambda}_t)_{t \geq 0}$ with locally square integrable coordinates.

Due to the fact that the sequence $(B_{k+1}-B_k)_{k \in \N}$ has i.i.d. $\sim \ND(0,I_d)$ coordinates
under $\wt{\PU}$, under appropriate identifications
the LETGS setting can be viewed as a special discrete case of the IS for Brownian motion with a linear parametrization of tilting processes
as above. %We do not check 
In the further sections we focus mainly on the discrete case, both
for simplicity and due to it having important numerical applications. 
However, many of our reasonings can be generalized to the Brownian case. 

\section{\label{secEulSDE}IS for diffusions and Euler schemes}
%TODO approximate diffusion by Euler scheme
Let us  use the notations for IS for a
Brownian motion from the previous section.
Let us consider  Lipschitz functions 
$\mu:\mc{S}(\R^m) \rightarrow \mc{S}(\R^m)$ 
and $\sigma:\mc{S}(\R^m) \rightarrow \mc{S}(\R^{m \times d})$. 
%let us consider functions $\mu:\R^m \rightarrow \R^m$ and $\sigma:\R^m \rightarrow \R^{m \times d}$ 
Then, there exists a unique strong solution $Y$ of the SDE 
\begin{equation}\label{diff} 
dY_t=\mu(Y_t)dt + \sigma(Y_t)dB_t, \quad Y_0=x_0
\end{equation} 
(see e.g. Section 5.2 in \cite{karatzas1991brownian}). 
Such a $Y$ is called a diffusion, $\mu$ a drift, and $\sigma$ a diffusion matrix. 
For $\wt{\tau}$ being a stopping time for $(\wt{\mc{F}}_t)_{t\geq 0}$ (like e.g. some hitting time of $Y$ of an appropriate set) 
and some $\R$-valued $\wt{Z} \in L^1(\wt{\PU}_{|\wt{\tau}})$, one can be interested in estimating 
\begin{equation} 
\wt{\phi}(x_0) = \E_{\wt{\PU}}(\wt{Z}). 
\end{equation} 
%Let $\eta_1, \eta_2, \ldots$ be i.i.d. $\sim \ND(0,\I_d)$ for probability $\PU$ 
%considered on their natural filtration $(\mc{F}_n)_{n\geq 0}$. 
A popular way of discretizing $Y$, especially in many dimensions, is by using an Euler scheme $X=(X_k)_{k\in \N}$ 
with a time step $h\in\R_+$, which, for some $\eta_1,\eta_2,\ldots,$ i.i.d. $\sim \ND(0,I_m)$ and some starting point $x_0 \in \R^m$, fulfills   
$X_{0}=x_0$ and 
\begin{equation} 
X_{k+1}=X_k+h\mu(X_k)+\sqrt{h}\sigma(X_k)\eta_{k+1},\quad k \in \N.  
\end{equation} 
We shall sometimes need a time-extended version $X'$ of such an $X$, defined in the below remark. 
\begin{remark}\label{remExt} 
For an Euler scheme $X$ as above, $X'=(X_k,kh)_{k \in \N}$ is also an Euler scheme, 
in the definition of which, in the place of $m$, $x_0$, $\mu$ and $\sigma$, we use
$m'=m+1$, $x_0'=(x_0,0)$, as well as $\mu':\R^{m'}\to \R^{m'}$ and  $\sigma': \R^{m'}\rightarrow \R^{m'\times d}$ 
such that for each $x\in \R^m$ and $t \in \R$ we have 
$\mu'(x,t)=(\mu(x),1)$, $\sigma_{i,j}'(x,t)=\sigma_{i,j}(x)$,  $i \leq m$, and  $\sigma_{m',j}'(x,t)=0$, $j \in \{1,\ldots, d\}$. 
% Let $\wh{D}=D \times (-\infty,T)$, and let $\wt{\tau}'$ and $\tau'$ be the exit times of $Y'$ and $X'$ of $\wh{D}$ respectively. 
% Then, for  $g(x,t)=\I(x \in D'\wedge t\leq T)+a,\ x\in \R^{m}$, $t \in \R$, 
% $\wt{Z}'= g(Y_{\wt{\tau}'}')$, and $Z'= g(X_{\tau'}')$, we have $\wt{Z}=\wt{Z}'$ and $Z=Z'$. 
 \end{remark}
%TODO approximation results. 
%and $X_{0}=x_0$. 
%using notations as in Section \ref{secStopped}, we define as an $\R^m$-valued process such that 
%We shall consider such schemes for arbitrary measurable $\sigma$ and $\mu$. 
Let further $\eta_i$, $i \in \N_+$, be as in Section \ref{secStopped} for $\wt{\PU}_1=\ND(0,I_d)$,
so that $X$ as above is an Euler scheme under $\PU$ as in that section. 
As discussed further on, in some cases, for a sufficiently small $h$, 
for an appropriate stopping time $\tau$ for $(\mc{F}_n)_{n\geq 0}$ and an appropriate $Z \in  L^1(\PU_{|\tau})$, 
$\wt{\phi}(x_0)$ can be approximated well using 
\begin{equation} 
\phi(x_0)=\E_{\PU}(Z). 
\end{equation} 
% In some special cases one can prove that $\phi(x_0)\rightarrow \wt{\phi}(x_0)$ for $h\rightarrow 0$. 
% For instance, for $\wt{\tau}=\tau=t$ deterministic and $U=g(Y_{t})$, $Z=g(X_{n})$, $h = t/n$, and $n\rightarrow \infty$, this holds 
% under certain assumptions on $g$ and diffusion coefficients and is known as weak convergence of Euler scheme 
% at time $t$ with respect to an appropriate class of functions (see \cite{kloeden1992numerical}). 
% See \cite{Gobet_2010} for such convergence results for certain random variables defined using hitting times of sets. 
%See discussion further on. 
%Using continuous Euler scheme may be more optimal, 
%Let us now move on to describing how IS can be used for estimating $\phi(x)$. 
%Let $\tau$ be the hitting time of $X$ of the complement of some open set $D$. 
%It has been shown in some special cases that for $h\rightarrow 0$ and $D=D(h)$ becoming close to $K$ 
%the value of $\phi(x)$ converges to that of $\alpha(x)$. In particular for 
%$\tau$ deterministic - TODO read now 
%(see Proposition 1.3, from Chapter 8 in Yor for alternative proof). 
% In particular, for each $Z \in L^1(\PR_\tau)$,
% \begin{equation}
% \E(F\I(\tau<\infty))=\E_{\PQ}(Z \gamma_{\tau}(\tau<\infty)). 
% \end{equation}.
% Furthermore, we receive that 
%For the special case of Euler scheme, 

For some function $r:\mc{S}(\R^m)\rightarrow \mc{S}(\R^d)$, called an IS drift, let us consider a tilting process 
$\lambda_k = \sqrt{h}r(X_k)$, $k \in \N$. Then, for 
\begin{equation}\label{wtmu}
\wt{\mu} = \mu+\sigma r 
\end{equation}
and $\dot{\eta}_k$, $k \in \N$, as in Section \ref{secGSS}, we have  
\begin{equation}
X_{k+1} = X_k+h\wt{\mu}(X_k)+\sqrt{h}\sigma(X_k)\dot{\eta}_{k+1},\quad k \in \N,
\end{equation}
so that from Theorem \ref{thNormGirs}, $X$ is an Euler scheme under $\PV$ with a drift $\wt{\mu}$. 
%and starting at the same point.  
As discussed in Section \ref{secGSS}, the distribution of $X$ under $\PV$ is the same as of $\wh{X}:=X(\pi^{-1})$ under $\PU$. 
Since  $\dot{\eta}_i=\eta_i\pi$, we have $\dot{\eta}_i\pi^{-1}=\eta_i$, $i \in \N_+$, so that $\wh{X}$ satisfies
$\wh{X}_{0}=x_0$ and
\begin{equation}\label{whx}
\wh{X}_{k+1} = \wh{X}_k+h\wt{\mu}(\wh{X}_k)+\sqrt{h}\sigma(\wh{X}_k)\eta_{k+1},\quad k \in \N,
\end{equation}
i.e. it is also an Euler scheme with a drift $\wt{\mu}$, but this time under $\PU$. 
%, i.e. it is an Euler scheme with a new drift $\wt{\mu}$
%and starting at the same point. 

For a nonempty set $A \in \mc{B}(\R^l)$, let us consider a parametrization 
$r: A\rightarrow\{f:\R^m\rightarrow \R^d\}$ of IS drifts, such that $(b,x)\rightarrow r(b)(x)$
is measurable from $\mc{S}(A)\otimes \mc{S}(\R^m)$ to $\mc{S}(\R^d)$, and let
$\wt{\mu}(b)= \mu+\sigma r(b)$, $b \in A$. 
Consider a parametrization $\lambda:A\rightarrow\mc{J}$ of tilting processes such that
\begin{equation}\label{lambdakb}
\lambda(b)=(\lambda_k(b))_{k\in \N}=(\sqrt{h}r(b)(X_k))_{k\in \N},\quad b \in A.
\end{equation}
Note that Condition \ref{condbxmes} holds for such a parametrization. 
Note also that, using notation (\ref{ybdef}), from (\ref{whx}) we have
\begin{equation}\label{Xkb}
X_{k+1}^{(b)} = X_k^{(b)}+h\wt{\mu}(b)(X_k^{(b)})+\sqrt{h}\sigma(X_k^{(b)})\eta_{k+1},\quad k \in \N. 
\end{equation}
Let us now describe the linear case of the above parametrization, leading to IS in the special case of the LETGS setting. We take $A=\R^l$ and
for some functions $\wt{r}_i:\mc{S}(\R^m) \rightarrow \mc{S}(\R^d)$, $i=1, \ldots,l$, called IS basis functions, 
we set 
%be a process $\wt{\lambda}_{i,k}=\sqrt{h}r_i(X_k)$ and 
%and $\wt{\lambda}=(\wt{\lambda}_i)_{i=1}^l$  
\begin{equation}\label{rbxdef}
r(b)(x)=\sum_{i=1}^l b_i \wt{r}_i(x),\quad b \in \R^l,\ x \in \R^m.  
\end{equation} 
Let $\Theta:\R^m \rightarrow \R^{d\times l}$ be such that for $i=1,\ldots,l$ and $j =1,\ldots,d$
\begin{equation}\label{thetadef}
\Theta_{j,i}(x)=\sqrt{h}(\wt{r}_{i})_j(x).
\end{equation}
Then, a process $\Lambda$ leading to $\lambda(b)$ given by (\ref{lambdalambda}) 
and such that (\ref{lambdakb}) holds, can be defined as 
% Then, the process $\Lambda$ leading to $\lambda(b)$ given by (\ref{lambdalambda}) 
% and such that (\ref{lambdakb}) holds is given by 
%and define and $\wt{\mu}(b)= \mu+\sigma r(b)$. 
%$\the(b)(x)=\sqrt{h}r(b)(x) = \theta(x) b$, so that 
%for some $r_i:\R^m \rightarrow \R^m$, $i=1, \ldots,l$ , so that for $\lambda_i=\sqrt{h}r_i(X_k)$ we have  
\begin{equation} \label{lambdatheta}
\Lambda_k = \Theta(X_k),\quad k \in \N.
\end{equation} 
%We have $\Lambda_k^{(b)}=\Theta(X_k^{(b)})$, $b \in A$. 
% \begin{remark}
% Further on we shall be also interested in using time-dependent drifts
% . They can 
% \end{remark}

An example of a stopping time $\tau$ for $(\mc{F}_k)_{k\geq 0}$ is  
an exit time of $X$ of some $D\in \mc{B}(\R^m)$, that is 
$\tau = \inf\{k \in \N:X_k \notin D\}$, 
for which we have $\tau^{(b)} = \inf\{k \in \N:X_k^{(b)} \notin D\}$, $b \in A$.
%Note that if $x_0 \in D$, then  $\tau\geq1$.

\begin{theorem}\label{thNTau}
Let us consider some linear parametrization of IS drifts as above. 
Let $\tau$ be the exit time of $X$ of $D\in \mc{B}(\R^m)$ such that $x_0 \in D$,
let $B \subset \R^l$ be nonempty, and let there exist $v \in \R^m$, $v \neq 0$, such that 
\begin{equation}\label{m1def}
M_1:= \sup_{x,y \in D} |v^T (x-y)|<\infty
\end{equation}
and
\begin{equation}\label{m2def}
M_2:=\sup_{x \in D,b \in B}|v^T\wt{\mu}(b)(x)|< \infty.
\end{equation}
For some $i \in \{1,\ldots,d\}$, let there exist $\delta_i\in\R_+$ and $\delta_j\in [0,\infty)$,
$j \in \{1,\ldots,d\}$, $j \neq i$, such that 
$|(v^T \sigma(x))_i| \geq \delta_i$ and $|(v^T \sigma(x))_j| \leq \delta_j$, $j \neq i$, $x\in D$. 
Let $M=M_1+hM_2$ and consider the following random conditions on $\mc{C}$ for $k \in \N_+$
\begin{equation}
q_k(\omega) =(|\eta_{k,i}(\omega)\delta_i| > \frac{M}{\sqrt{h}}+ |\sum_{j \in \{1,\ldots,d\},\ j\neq i}\eta_{k,j}(\omega)\delta_j|), \quad \omega \in E.
\end{equation}
Then, a random variable $\wh{\tau}$ on $\mc{C}$ such that
%\begin{equation}
$\wh{\tau}(\omega)= \inf\{k \in \N_+:q_k(\omega)\}$, $\omega \in E$,
 %\end{equation}
fulfills 
%\begin{equation}
$\tau^{(b)}\leq \wh{\tau}$, $b \in B$.
%\end{equation}
Under $\PU$, the variable $\wh{\tau}$ has a geometric distribution with a parameter $q=\PU(q_1)$, that is
%\begin{equation}\label{geomDistr}
$\PU(\wh{\tau}=k)=q(1-q)^{k-1}$, $k \in \N_+$.  
% \sum_{k=s+1}q(1-q)^{k-1}
%\end{equation}
 \end{theorem}
 \begin{proof}
 Let $b \in B$. From (\ref{Xkb}), for each $k \in \N_+$ 
 \begin{equation}
 \begin{split}
 (X_{k}^{(b)} \notin D)\wedge (X_{k-1}^{(b)} \in D)  &=
 (\sqrt{h}\sigma(X_{k-1}^{(b)})\eta_k \notin D - X_{k-1}^{(b)}-h\wt{\mu}(b)(X_{k-1}^{(b)})) \wedge (X_{k-1}^{(b)} \in D)\\
 %X_{k+1}^{(b)} = X_k^{(b)}+h\wt{\mu}(b)(X_k^{(b)})+\sqrt{h}\sigma(X_k^{(b)})\eta_{k+1},\quad k \in \N. 
 &\Leftarrow (\sqrt{h}v^T\sigma(X_{k-1}^{(b)})\eta_k \notin v^T (D - X_{k-1}^{(b)} -h\wt{\mu}(b)(X_{k-1}^{(b)}))) \wedge (X_{k-1}^{(b)} \in D)\\
 &\Leftarrow (\sqrt{h}|v^T\sigma(X_{k-1}^{(b)})\eta_k| > M) \wedge (X_{k-1}^{(b)} \in D)\\
 &\Leftarrow q_k\wedge (X_{k-1}^{(b)} \in D).
 %|\wt{\eta}_{k,i}\delta_i| > \frac{M}{\sqrt{h}}+ |\sum_{j\neq i}\wt{\eta}_{k,j}\delta_j| \wedge X_{k-1}^b \in D\\
 %&\Leftarrow |\wt{\eta}_{k,i}|>\frac{\frac{M}{\sqrt{h}} +\sum_{j\neq i}\delta_j}{\delta_i}| 
 % |\sum_{j\neq i}\wt{\eta}_{k,j}\delta_j| \wedge X_{k-1}^b \in D\\
 \end{split} 
 \end{equation}
 Thus, $q_k \Rightarrow (X_{k-1}^{(b)} \notin D \vee X_{k}^{(b)} \notin D) \Rightarrow \tau^{(b)} \leq k$. 
 For $\omega \in E$ such that $\wh{\tau}(\omega) < \infty$ it holds  $q_{\wh{\tau}(\omega)}(\omega)$, 
 and thus $\tau^{(b)}(\omega) \leq \wh{\tau}(\omega)$. 
 \end{proof}
 \begin{remark}\label{remCond}
 Note that if %for some $K \subset \R^l$ 
 for each $b \in A$ the assumptions of Theorem \ref{thNTau} 
 hold for $B=\{b\}$, then from $\tau^{(b)}$ having the same distribution under $\PU$
 as $\tau$ under $\PV(b)$, we receive that $\tau$ has all finite moments under $\PV(b)$, $b \in  A$, and 
 in particular Condition \ref{equivCond} holds.
 \end{remark}
% %if its assumptions hold for $B$ equal $B_k$
 We say that a matrix- or vector-valued function $f$ is uniformly bounded 
 on some subset $B$ of its domain if 
for some arbitrary vector or matrix norm $||\cdot||$ we have $\sup_{x \in B}||f(x)||<\infty$. 
 \begin{remark}\label{remNTau} 
 Note that (\ref{m1def}) holds for $D$ bounded and arbitrary $v \in \R^m$. Furthermore, if for some $v \in \R^m$,  $v^T \mu$, $v^T\sigma$, and $\Theta$ are uniformly 
 bounded on $D$ (which holds e.g. when they are continuous on $\R^l$ and $D$ is bounded) then  (\ref{m2def}) holds 
 for each bounded $B$. 
 \end{remark}

\section{\label{secSomeISFun}Zero-variance IS for diffusions}
To provide an intuition when the variance of the IS estimator of the expectation a functional of an Euler scheme
can be small, let us 
briefly describe a situation when its diffusion counterpart 
has zero variance. See 
Section 4 in \cite{Glynn_2012} for details. 
% If $\lambda_t$ is a predictable $\R^m$-valued process for which
% \begin{equation}
% \Gamma_t=\exp(\int \psi_s dB_s - \frac{1}{2}\int_{0}^t \lambda_s^2 ds), \quad t \geq 0
% \end{equation}
% is a martingale (for which e.g. the Novikov's criterion suffices) then from Kolmogorov's consistency theorem there exists 
% a unique measure $\PQ$ on the canonical space of Brownian motion such that $\frac{d\PQ_t}{d\PR_t}=\Gamma_t$, $t \geq 0$.
% From Girsanov theorem, under $\PQ$, $\wt{B}_t = B_t-\int_{0}^t\lambda_sds$ is a Brownian motion.
% Taking $\kappa_s= \sigma r(X_s)$, we receive 
% \begin{equation} 
% dY_{t} = \wt{\mu}(Y_t)dt+\sigma(Y_t)d\wt{B}_t,\quad Y_{0}=x_0. 
% \end{equation} 
% From Proposition 1.3 from Chapter 8 in \cite{revuz1999continuous}, for a stopping time $\tau$ we have that 
% $\frac{d\PQ_{|\wt{\tau}}}{d\PR_{|\wt{\tau}}}=\Gamma_{\wt{\tau}}$ on $\wt{\tau} < \infty$ similarly as in the discrete case above. 
Using notations as in the previous section, for $\wt{\tau}$ being the hitting time of $Y$ a boundary of an open set $D$ such that $x_0 \in D$, 
as well as for an appropriate $g:\R^m\rightarrow \R$ and 
$\beta:\R^m\rightarrow \R$, consider 
\begin{equation}\label{USDE} 
\wt{Z}=\I(\wt{\tau}<\infty)g(Y_{\wt{\tau}})\exp(\int_{0}^{\wt{\tau}} \beta(Y_s)\,\mathrm{d}s). 
\end{equation} 
If there exists an appropriate function $u:\R^m\rightarrow \R$, such that for $Lu =\Tr(\sigma\sigma^T)\Delta u + \mu^T\nabla u+ \beta u$ we have  
\begin{equation}\label{PDEIn} 
Lu(x)=0, \quad x \in D, 
\end{equation} 
and 
\begin{equation}\label{PDEBord} 
u(x)=g(x), \quad x \in \partial D, 
\end{equation} 
then, from the Feynman-Kac theorem, $\wt{\phi}(x_0) = u(x_0)$. 
Under certain assumptions, including $u(x)>0$, $x\in D$, it can be proved (see Theorem 4 in \cite{Glynn_2012}) that for $r$ equal to 
\begin{equation}\label{rstar} 
r^*:= \frac{\sigma^T\nabla u}{u}= \sigma^T\nabla(\ln(u)), 
\end{equation} 
for the IS for a Brownian motion as in Section \ref{secBrown} with $\wt{\lambda}_t= r^*(Y_t)$, 
we have  $\wt{Z}\wt{L}=\wt{\phi}(x_0)$, $\wt{\PV}$ a.s., i.e. the IS estimator for the diffusion 
case has zero variance. Furthermore, from (\ref{wtBsol}) 
\begin{equation}\label{Ynewequ} 
dY_t =(\mu + \sigma r^*)(Y_t)dt + \sigma(Y_t)d\wt{B}_t, \quad Y_0=x_0. 
%dY_{t} = -\nabla \wt{V}(Y_t)dt + \sqrt{2\epsilon} d \wt{B}_t, \quad Y_{0}=x_0,
\end{equation} 
For $\tau$ being the exit time of $X$ of some set $B$, a possible Euler scheme counterpart of (\ref{USDE}) is 
\begin{equation} \label{zform} 
Z= \I(\tau <\infty)g(X_{\tau})\exp(\sum_{k=0}^{\tau-1}h\beta(X_k)).  
\end{equation} 
Under appropriate assumptions for such a $Z$ we have 
\begin{equation}\label{phiwtphi} 
\lim_{h\to 0}\phi(x_0)= \wt{\phi}(x_0) 
\end{equation}
for $B=D$; see \cite{Gobet2000,Gobet2010}. 
Furthermore, in \cite{Gobet2010} it was proved that in some situations the rate of convergence in (\ref{phiwtphi}) 
can be increased by taking as $B$ an appropriately shifted $D$. 
Further on for $\tau$ as above and $Z$ as in (\ref{zform}) we shall assume that $B=D$, 
but one can easily modify the below reasonings to consider the shifted set instead. 
It seems intuitive that for some such $Z$, for $r$ close to $r^*$, and for small $h$, we can receive 
low variance of the Euler scheme IS estimator $ZL$. This intuition shall be confirmed in our numerical 
experiments in Chapter \ref{secNumExp}. 
%Note that for such $Z$, $Z^{(b)}=\I(\tau^{(b)} <\infty)g((X^{(b)})_{\tau^{(b)}})\exp(\sum_{k=0}^{\tau^{(b)}-1}h\beta(X_k^{(b)}))$. 

\section{\label{secImpSpec}Some examples of expectations of functionals of diffusions and Euler schemes} 
We shall now discuss several examples of expectations of functionals of diffusions and their 
Euler scheme counterparts. As discussed in Chapter \ref{secInt}, these expectations can be of interest among others in molecular dynamics, and 
their Euler scheme counterparts were estimated in our numerical experiments described in Section \ref{secNumExp}. 
% As discussed in Chapter \ref{secInt}, these quantities can be of interest among others in molecular dynamics, and 
% their Euler scheme counterparts 
In the first two examples, for diffusions 
we consider the expectations $\wt{\phi}(x_0)=\E_{\wt{\PU}}(\wt{Z})$ for some $\wt{Z}$ as in (\ref{USDE}),
and for the corresponding Euler schemes we consider $\phi(x_0)=\E_{\PU}(Z)$ for the variable $Z$ as in (\ref{zform}). 
%which can be of interest in the field of molecular dynamics. 
% As discussed in Chapter \ref{secInt}, these quantities can be of interest among others in molecular dynamics, and 
% their Euler scheme counterparts 
% $\phi(x_0)=\E_{\PU}(Z)$ for $Z$ as in (\ref{zform}) shall be estimated in our numerical experiments in Chapter \ref{secNumExp}. 
%We shall assume that $\wt{\tau}<\infty$, $\wt{\PU}$ a.s. and $\tau<\infty$, $\PU$ a.s.
%and for simplicity we shall not write factors $I(\wt{\tau} <\infty)$ and $I(\tau <\infty)$ appearing in $\wt{Z}$ and $Z$.
In the first example, for some $p\in \R_+$ we take $\beta(x) = -p$ and  $g(x)=1$, $x \in \R^m$, so that
%\begin{equation}
$\wt{Z}=\exp(-p \wt{\tau})\I(\wt{\tau}<\infty)$ 
%\end{equation}
and $Z=\exp(-p h \tau)\I(\tau<\infty)$. %, so that Z^{(b)}=\exp(-p h \tau)\I(\tau<\infty)
%$F=\exp(-\sigma \tau)$ and 
The quantities $\wt{\mgf}(x_0):=\wt{\phi}(x_0)$ and 
$\mgf(x_0):=\phi(x_0)$ for this case are called the moment-generating functions (MGFs) of $\wt{\tau}$ and $h\tau$ respectively.
%\begin{equation} 
%\wt{\mgf}(x_0)=\E_{\wt{\PU}}(\I(\wt{\tau}<\infty)\exp(-\sigma \wt{\tau})) 
%\end{equation} 
%or $\mgf(x_0)=\E_U(\I(\tau<\infty)\exp(-\sigma h\tau)$ for the Euler scheme case. 
%(we further denote the Euler scheme 
%counterparts of certain quantities for diffusions with a tilde). 
Let us consider some
$a \in \R$, called an added constant.
For the second example let us assume that 
\begin{equation}\label{wtputaufin}
\PU(\tau<\infty)=\wt{\PU}(\wt{\tau}<\infty)=1
\end{equation}
and let $D' =\R^m \setminus D = A \cup B$ for two closed disjoint sets $A$ and $B$ from $\mc{B}(\R^m)$. Let 
$\beta(x)=0$, $x\in \R^m$, $g(x) =a+1$, $x \in B$, and $g(x)=a$, $x \in A$. %denoting the corresponding $U$ as $U_a$ 
%and (respectively $Z$ as $Z_a$), we 
We receive $\wt{Z}=\I(\wt{\tau}<\infty)(a+ \I(Y_{\wt{\tau}}\in B))$ and $\wt{\phi}(x_0)$ equal to 
%\begin{equation}
$\wt{q}_{AB,a}(x_0):=a+ \wt{\PU}(Y_{\wt{\tau}}\in B))$, 
%\end{equation}
which we shall call a translated committor. 
For the added constant $a=0$, we denote $\wt{q}_{AB,a}(x_0)$ simply as 
$\wt{q}_{AB}(x_0)$ and call it a committor. In the Euler scheme case we consider analogous definitions but with omitted tildes and with $X$ in the place of $Y$. 
Committors are of interest for instance when computing the 
reaction rates and characterizing the reaction mechanisms of dynamic processes; see \cite{HartmannEntr14,Prinz_2011, Allen_2009}. 
%one compute the other. 
% For the Euler scheme case we have $Z=a+ \I(X_{\tau}\in B)$. We 
% get $\wt{\phi}(x_0)$ equal to 
% \begin{equation} 
% \wt{q}_{AB,a}(x_0)=a+ \PR(X_{\wt{\tau}}\in B)), 
% \end{equation} 
% where $\wt{q}_{AB}(x_0):=q_{AB,0}(x_0)$ 
%For $\wt{\tau}$ and $\tau$ being the exit times of $X$ and $Y$ of some open set $D$ as above, l

For the third example, for some $D$, $X$, $\tau$, and $\wt{\tau}$ as in Section \ref{secSomeISFun}, as well  
for some $T \in \R_+$, let us now consider $\wt{Z}=\I(\wt{\tau} \leq T) +a$, 
%\begin{equation} 
$\wt{p}_{T,a}(x_0)= \E_{\wt{\PU}}(\wt{Z})$, 
%\end{equation} 
and $\wt{p}_T(x_0)=\wt{p}_{T,0}(x_0)$, while for the Euler scheme case 
%\begin{equation}
$Z=\I(h\tau \leq T) +a$, 
%\end{equation}
$p_{T,a}(x_0)=\E_{\PU}(Z)$, and $p_{T}(x_0)=p_{T,0}(x_0)$. %, so that $Z^{(b)}=\I(h\tau^{(b)} < T) +a$. 
Note that for 
\begin{equation}\label{taup}
\tau'= \tau\wedge \left\lfloor \frac{T}{h}\right\rfloor 
\end{equation}
it holds $Z= \I(X_{\tau'}\in D') +a$. 
%this time $\wt{Z}$ and $Z$ do not have forms as in (\ref{USDE}) and (\ref{zform}). 
Note also that for the time-extended process $X'$  corresponding to the above $X$ as in Remark \ref{remExt}, such a $\tau'$ is the exit time of $X'$ of
\begin{equation}\label{whd}
\wh{D}=D \times [0,h \left\lfloor \frac{T}{h}\right\rfloor). 
\end{equation}
Such a $\tau'$ is the stopping time which we shall further consider by default for IS in the LETGS setting for computing $p_{T,a}(x_0)$. 
% However, they 
% can be rewritten in such forms for some time-extended processes as discussed in the below remark. %$\wt{Z}'$ of form as in (\ref{USDE}) 
%under $\wt{\PU}$, and analogously for its Euler scheme counterpart. 
% \begin{remark}\label{remExt} 
% Consider $Y'=(Y_t,t)_{t \in [0,\infty)}$ and  $X'=(X_k,kh)_{k \in \N}$, which are a diffusion and a corresponding Euler scheme 
% in the definition of which, in the place of $m$, $x_0$, $\mu$ and $\sigma$, one uses 
% $m'=m+1$, $x_0'=(x_0,0)$, as well as $\mu':\R^{m'}\to \R^{m'}$ and  $\sigma': \R^{m'}\rightarrow \R^{m'\times d}$ 
% such that for each $x\in \R^m$ and $t \in \R$ we have 
% $\mu'(x,t)=(\mu(x),1)$, $\sigma_{i,j}'(x,t)=\sigma_{i,j}(x)$,  $i \leq m$, and  $\sigma_{m',j}'(x,t)=0$, $j \in \{1,\ldots, d\}$. 
% Let $\wh{D}=D \times (-\infty,T)$, and let $\wt{\tau}'$ and $\tau'$ be the exit times of $Y'$ and $X'$ of $\wh{D}$ respectively. 
% Then, for  $g(x,t)=\I(x \in D'\wedge t\leq T)+a,\ x\in \R^{m}$, $t \in \R$, 
% $\wt{Z}'= g(Y_{\wt{\tau}'}')$, and $Z'= g(X_{\tau'}')$, we have $\wt{Z}=\wt{Z}'$ and $Z=Z'$. 
% \end{remark}
% Note that for $\tau'$ as in the above remark we have  $\tau'= \tau \wedge \lceil\frac{T}{h}\rceil$. 
A possible alternative would be to use $\tau$, which, 
as discussed in Remark \ref{rembettertau}, would lead to not lower variance and mean cost for the cost variables 
equal to the respective stopping times. 

\begin{remark}\label{remSuffphiwtphi}
Sufficient assumptions for (\ref{phiwtphi}) to hold for the MGFs and translated committors as above can be derived e.g. from the
discussion in Section 4 in \cite{Gobet2010} (along with appropriate convergence rates in it), while for 
%\lim_{h\to 0}\phi(x_0)= \wt{\phi}(x_0) 
\begin{equation}\label{ptax0}
\lim_{h\to 0} p_{T,a}(x_0)=\wt{p}_{T,a}(x_0)
\end{equation}
--- from reasonings analogous as in Section 1.2 of \cite{Gobet2000}. 
\end{remark}

Let $\wh{\psi}_a$ be an unbiased estimator of $\psi_a$ equal to 
$q_{AB,a}(x_0)$ or $p_{T,a}(x_0)$, i.e. $\E(\wh{\psi}_a)=\psi_a$.
Then, the translated estimator $\wh{\psi}_{a,0}=\wh{\psi}_a-a$ is an unbiased estimator of $\psi_0$ equal to $q_{AB}(x_0)$ or $p_{T}(x_0)$ respectively,
and $\Var(\wh{\psi}_{a,0})=\Var(\wh{\psi}_a)$. 
The reason why we are considering such translated estimators of $\psi_0$ for nonzero added constants $a$
is that using these estimators in the adaptive IS procedures 
in our numerical experiments as discussed in Chapter \ref{secNumExp} led to lower variances and inefficiency constants than for $a=0$. 

% As discussed in Section \ref{secNumExp},
% in some of our numerical experiments, applying 
% as above $\wh{\psi}_a$ IS estimators $ZL(b)$ using some adaptive IS parameters $b$, % from multi-stage minimization of inefficiency constant, 
% for the corresponding translated estimators $\wh{\psi}_{a,0}$ of $\psi_0$ we received lower estimates of variances and inefficiency constants than  
% when using 
% for such direct IS estimators $\wh{\psi}_0$.  
% for Euler scheme expectations for some such $a$, the variances and inefficiency constants of 
% such translated IS estimators 
% were lower than of the direct IS estimators $\wh{\psi}_0$ of $\psi_0$ (see Section \ref{secNumExp}). 
%The reason for considering translated estimators $\wh{\psi}_{a,0}$ with $a > 0$ as 
%above, is that 
% A possible intuition behind this effect is provided in below Remark \ref{someahigh} (though note that this remark 
% applies to diffusions and our experiments were carried out for Euler scheme). 
%for a possible intuition why choosing $a=0$ may lead to worse results. 

Note that we have  $q_{AB}(x_0)+q_{BA}(x_0)=1$ and similarly for the diffusion case, so that if $\wh{q}$ is an
unbiased estimator of one of the quantities $q_{AB}(x_0)$ or $q_{BA}(x_0)$, then $1-\wh{q}$ is such an estimator of the other quantity
with the same variance and inefficiency constant. Therefore, given an estimator $\wh{q}_{AB}$ of $q_{BA}(x_0)$
and $\wh{q}_{AB}$ of $q_{BA}(x_0)$, it seems reasonable to compute both quantities as above using the estimator leading to a
lower inefficiency constant. 

% \begin{remark}\label{someahigh} abov
% % % Supremum norm of error of approximation of 
% % % can be used for bounding variance of the IS estimator.
%   Note that if for some $x_0 \in \partial D$ we have $u(x_0)=0$, as is the case for $u= q_{AB}$ and $x_0\in \partial A$), 
%   then $F(x)\to \infty$ as $x\to x_0$ 
%   and $|r^*(y_n)|\to \infty$ for some $y_n \to x_0$. Thus, when approximating $r^*$ by a linear combination of functions bounded on some
%   neighbourhood of some such $x_0$, we will always receive infinite error of such approximation in supremum norm. 
%   This need not to be the case when $u(x)>0$, $x \in \partial D$, as is the case for  $u=q_{AB,a}$ for $a>0$  
%   (see Figure \ref{figVCom}). 
% %  This needs not to be the case 
% % % %TODO refer to Hartmann
% % % This does not happen for such $F$ and $r^*$
% % % for $q_{AB,a}$ for $a>0$.  
%  \end{remark}

\section{Diffusion in a potential}\label{secDiffusion}
%Consider a differentiable potential $V:\R^m\longmapsto \R$ such that $\nabla V$ is Lipschitz.
We define a diffusion $Y$ in a differentiable potential $V:\R^m\longmapsto \R$ and corresponding to a temperature $\epsilon\in\R_+$ to be a unique strong solution of
\begin{equation}\label{equPot}
dY_t = -\nabla V(Y_t)dt + \sqrt{2\epsilon} d B_t,\quad Y_{0}=x_0,
\end{equation}
assuming that such a solution exists, which is the case e.g. if $\nabla V$ is Lipschitz. 
%(see e.g. \cite{HartmannShuette2012}). 
% Note that for general IS drifts $r$ we have from (\ref{wtmu}) %for some parametrization of IS basis functions $r(b)$, $\mu(b)$
% \begin{equation}\label{wtmu}
% \wt{\mu} = -\nabla V+\sqrt{2\epsilon} r 
% \end{equation}
% In particular, for $r= \nabla U$ it holds $\wt{\mu}= \nabla\wt{V}$ for 
For such a diffusion, under appropriate assumptions as in Section \ref{secSomeISFun}, 
an IS drift  (\ref{rstar}) leading to a zero-variance IS estimator and probability $\wt{\PV}$ is
\begin{equation}\label{rstar2}
 r^*=\sqrt{2\epsilon}(\nabla \ln(u)). 
\end{equation}
Let $F=-\epsilon \ln(u)$, $C_0 \in \R$, and let us define an optimally-tilted potential
\begin{equation}\label{wtV}
V^*=V+2F + C_0.
\end{equation}
Then, (\ref{Ynewequ}) can be rewritten as
\begin{equation}
dY_{t} = -\nabla V^*(Y_t)dt + \sqrt{2\epsilon} d \wt{B}_t, \quad Y_{0}=x_0. 
\end{equation}
Thus, under $\wt{\PV}$, $Y$ is a diffusion in potential $V^*$. 

\section{\label{secSpecNumExp}The special cases considered in our numerical experiments}
Let $D: = (a_1,a_2)= (-3.5,3.5)$. 
Consider a smooth potential $V:\R\to\R$ such that
\begin{equation}\label{3wpot}
V(x)= \frac{1}{200}(0.5x^6 - 15x^4 + 119x^2 + 28x + 50),\quad x \in D, 
\end{equation}
and $\nabla V$ is Lipschitz. Such a $V$ restricted to $D$ 
is shown in Figure \ref{pot3W}.
\begin{figure}[h]
\includegraphics[width=0.5\textwidth]{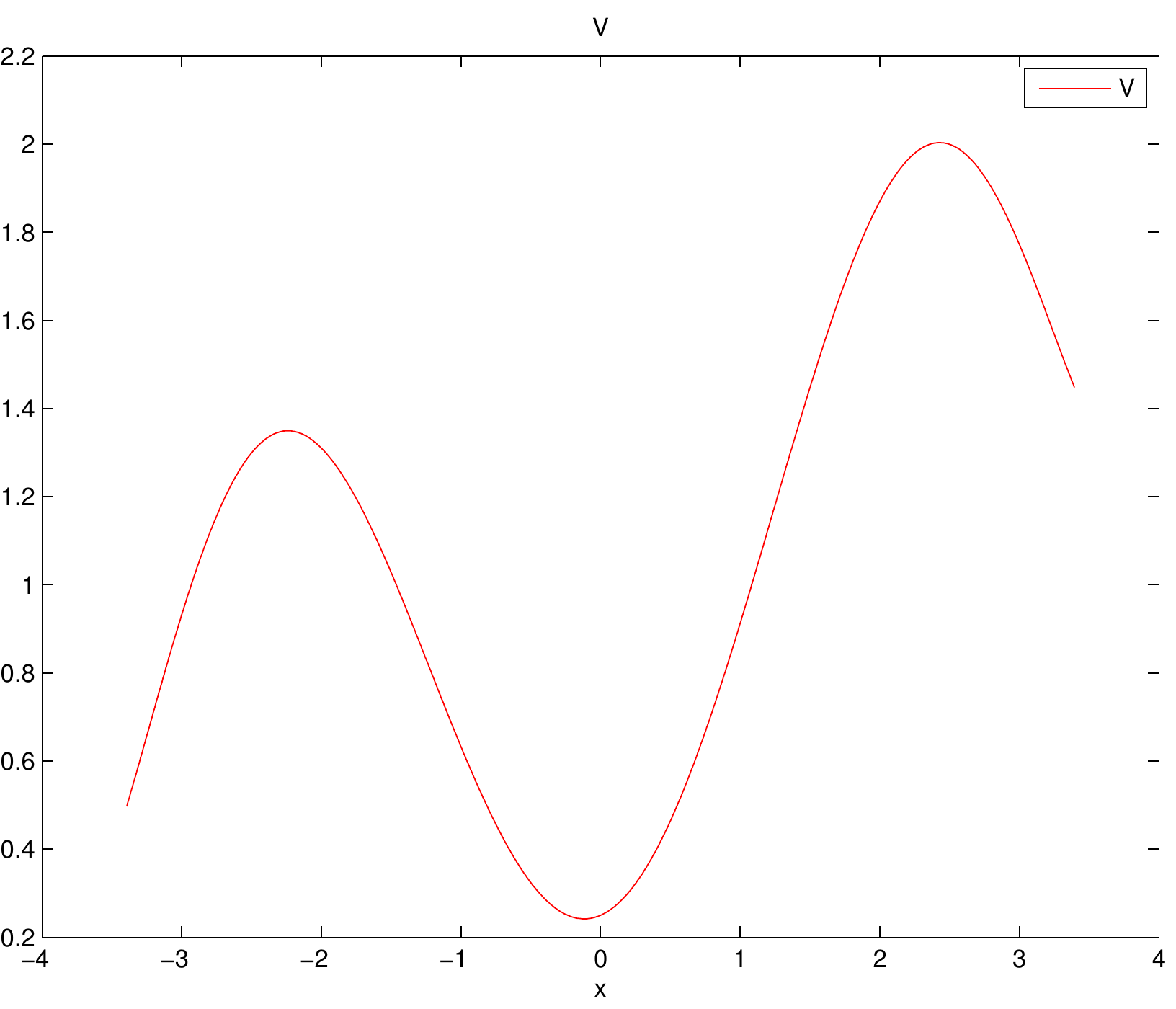}
\caption{The three-well potential given by (\ref{3wpot}) on $D$.}
\label{pot3W}
\end{figure}
For a temperature $\epsilon=0.5$, consider a diffusion $Y$ in such a potential starting at some $x_0 \in D$.
Let $\wt{\tau}$ be the hitting time of $Y$ of the boundary of $D$. 
Let $A=(-\infty,a_1)$, $B=(a_2,\infty)$, and let $\wt{q}_{1,a}=\wt{q}_{AB,a}$ and $\wt{q}_{2,a}=\wt{q}_{BA,a}$ 
(see Section \ref{secImpSpec}), 
which for $a=0$ will be denoted simply as $\wt{q}_1$ and $\wt{q}_2$,  
and analogously in the Euler scheme case in which the tildes are omitted. 
Let us also consider $\wt{\mgf}$ and $\mgf$ for $p = \wt{p}:=0.1$. 
We computed approximations of such $\wt{q}_i(x)$ and $\wt{\mgf}(x)$ in the function of $x$ using finite difference 
discretizations of PDEs given by (\ref{PDEIn}) and (\ref{PDEBord}). The results are shown in figures
\ref{figCom} and \ref{figMGF}. %\ref{coms3W} and \ref{mgf3W}. 
In figures \ref{figVCom} and \ref{figVMGF} we show
approximations of the optimally tilted potentials (\ref{wtV}) for the MGF and committors $\wt{q}_{i,a}$ for $a=0$ and $a=\wt{a}:=0.05$, $i=1,2$. 
%\begin{condition}

\begin{figure}%
\centering
\subfloat[]{\label{figCom}\includegraphics[width=0.47\textwidth]{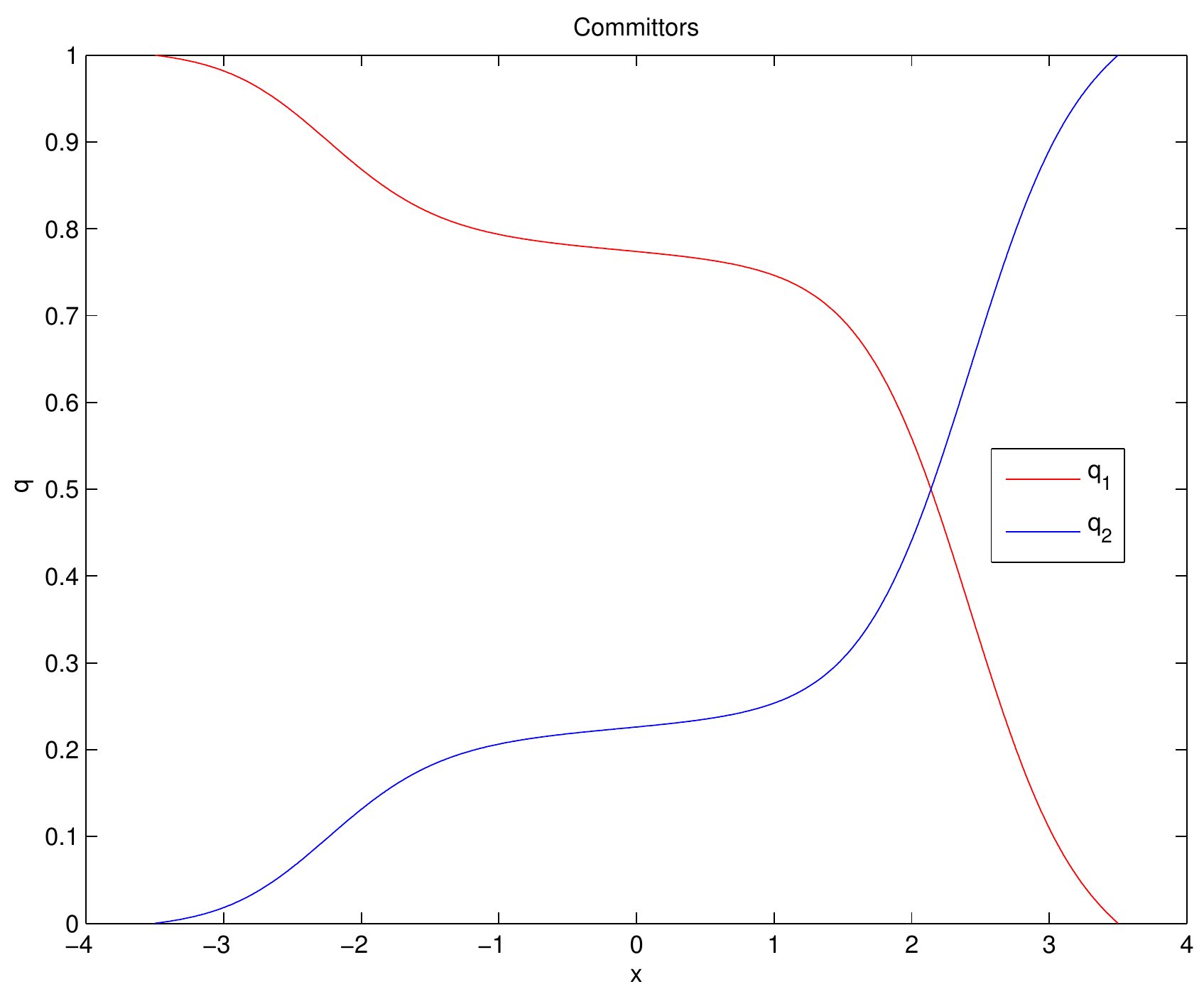}}
\qquad  
\subfloat[]{\label{figMGF}\includegraphics[width=0.45\textwidth]{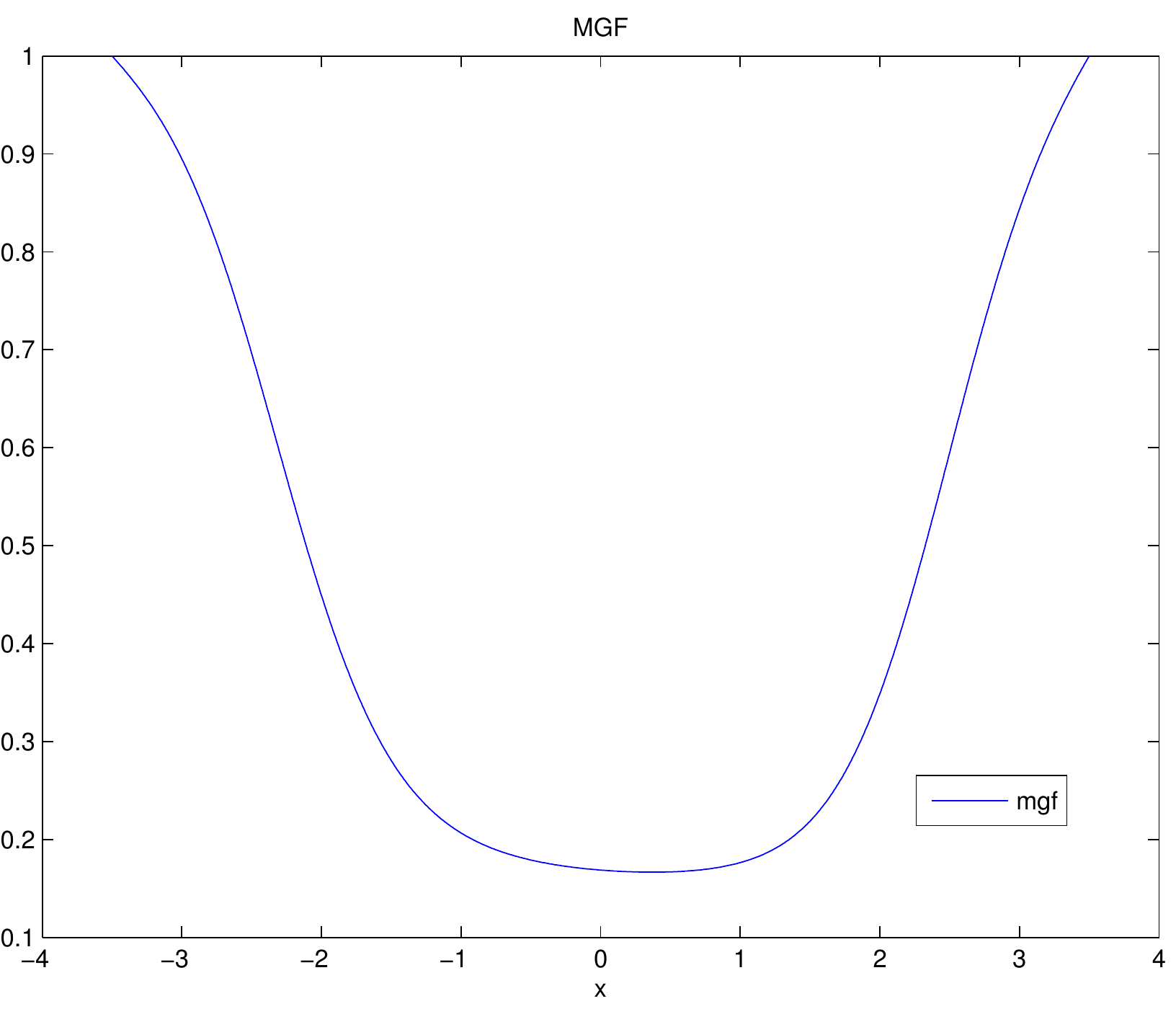}} 
\caption{The committors and MGF as in the main text.}
%\label{comMGF3W}
\end{figure}

\begin{figure}
\centering
\subfloat[]{\label{figVCom}
\includegraphics[width=0.45\textwidth]{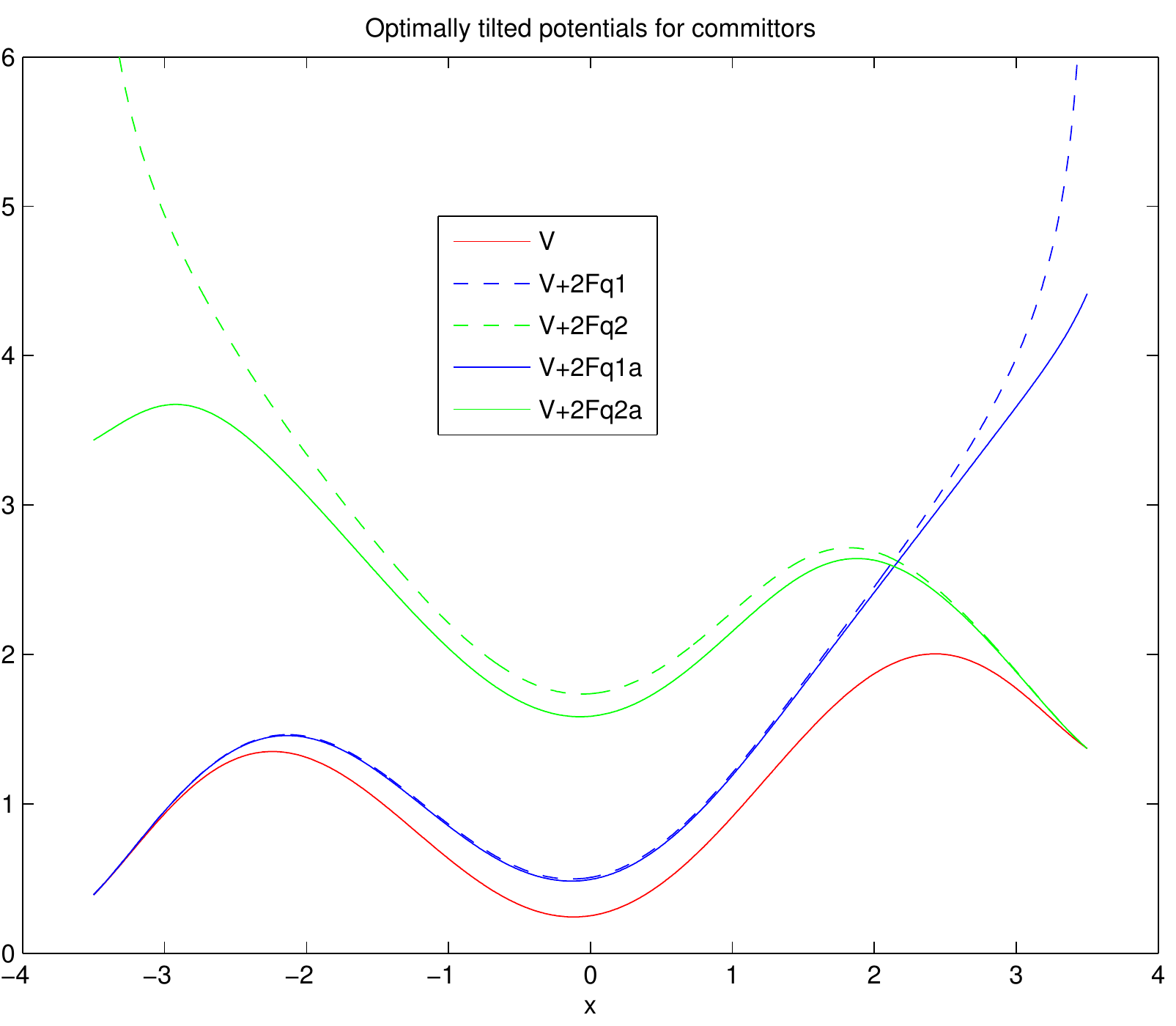}} 
\qquad
\subfloat[] {\label{figVMGF}\includegraphics[width=0.45\textwidth]{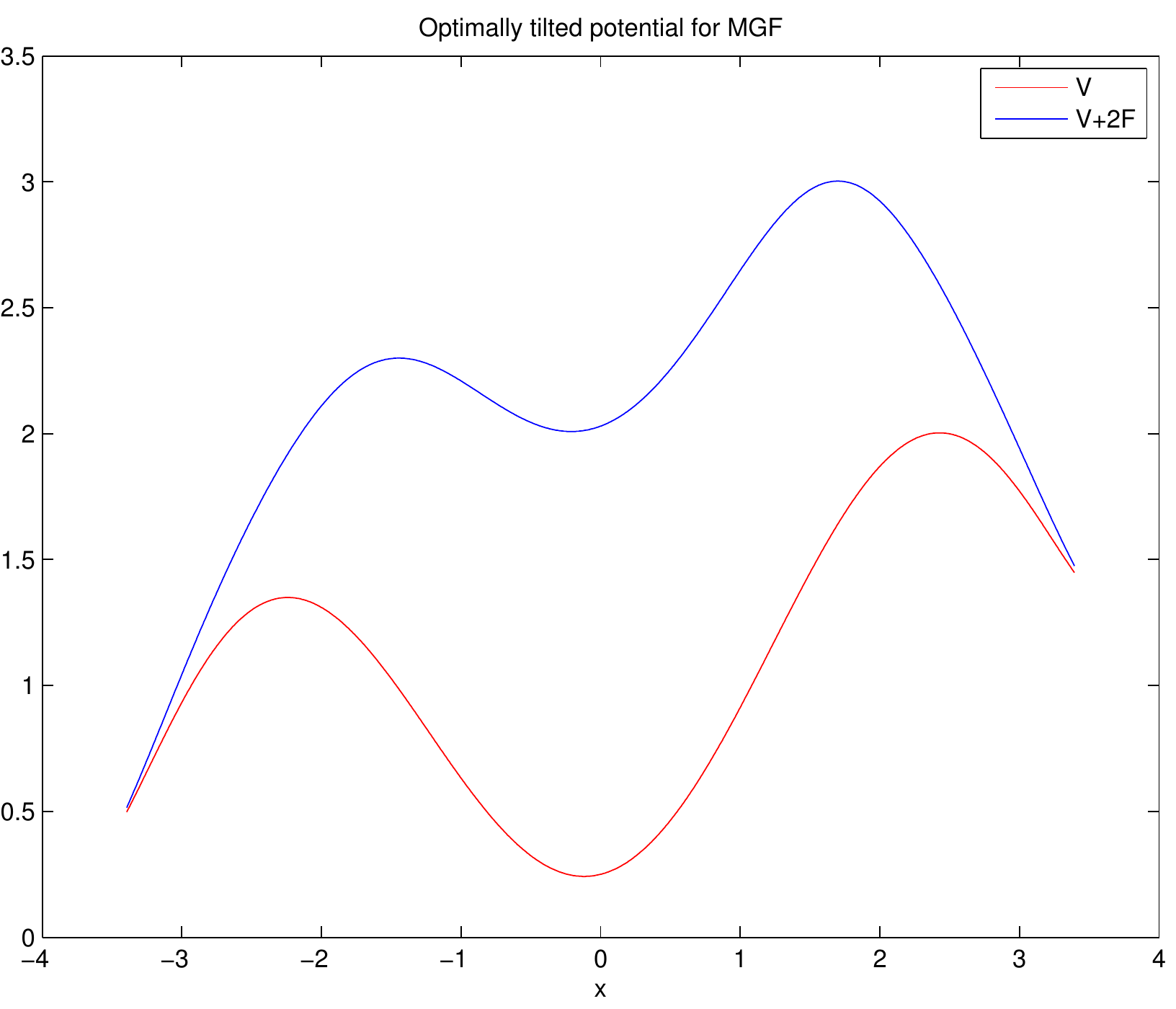}}
\caption{Optimally tilted potentials as in (\ref{wtV}) for the MGF and committors. $Fqi$ and $Fqia$ 
are the functions $F$ as in Section \ref{secDiffusion} for the $i$th committor for $a=0$ and $a=\wt{a}$ respectively. 
The constant $C_0$ for the $i$th committor for $i\in \{1,2\}$ was chosen so that the tilted potential is 
equal to the original potential in point $a_i$ and for the MGF --- so that these potentials are equal in both $a_1$ and $a_2$.} 
\label{comMGF3WOpt}
\end{figure}

In our experiments we considered an Euler scheme $X$ with a time step $h=0.01$ corresponding to the above diffusion $Y$
starting at $x_0=0$. We focused on estimating  $\mgf(x_0)$ for $p=\wt{p}$, 
$q_{i}(x_0)$ for $i=1,2$, and $p_{T}(x_0)$ for $T=10$. 
For some $M\in \N_2$,  $\wt{a}_1 = -3.6$, $\wt{a}_2=3.6$,  $\wt{d}=\frac{\wt{a}_2-\wt{a}_1}{M-1}$, 
and $p_i= \wt{a}_1 +(i-1)\sigma$, $i \in\{1,\ldots,M\}$, consider Gaussian functions 
\begin{equation}\label{riexp} 
\wt{r}_i(x)=\frac{1}{\sqrt{\epsilon}}\exp(-\frac{(x-p_i)^2}{\wt{d}^2}), \quad i =1,\ldots,M. 
\end{equation} 
In our experiments we used a linear parametrization of IS drifts as in Section \ref{secEulSDE}. 
For each estimation problem we used as the IS basis functions the above Gaussian functions for $M=10$. 
For estimating $p_{T,a}(x_0)$, considering a time-extended Euler scheme as 
in Remark \ref{remExt} corresponding to the above $X$, we additionally performed experiments using $2M$ time-dependent IS basis functions 
\begin{equation}\label{rtimedep}
\wh{r}_i(x,t)= \wt{r}_i(x),\ \wh{r}_{M+i}(x,t)= t^p\wt{r}_i(x),\quad i=1,\ldots,M,  
\end{equation}
for different $p \in \N_+$, and for $M=5$ and $M=10$. 
See Section \ref{secSimpleMin} and Chapter \ref{secNumExp} for further details on our numerical experiments. 
Note that since the above $\wt{r}_i$ are continuous and $D$ is bounded,  
from remarks \ref{remCond}, \ref{remNTau}, 
and Theorem \ref{thNTau}, in which one can take $v=1$ and $\delta_1=\sqrt{2\epsilon}$, it follows that Condition \ref{equivCond} holds %in the LETGS setting 
when estimating the MGF and committors as above. %in the LETGS setting 
Using further the fact that $\E_{\wt{\PU}}(\wt{\tau})<\infty$ (which follows e.g. from Lemma 7.4 in \cite{karatzas1991brownian}), (\ref{wtputaufin}) holds. 
Furthermore, from Remark \ref{remSuffphiwtphi}, we have (\ref{phiwtphi}) for the MGF and tanslated committors, and (\ref{ptax0})
for the exit probability. 
\chapter{Some properties of the minimized functions and their estimators}\label{secDiv}
In this chapter we discuss various properties of the functions and their estimators from Chapter \ref{secMinFun}
for some parametrizations of IS from the previous chapter. These properties 
will be useful when proving the convergence and asymptotic properties of certain minimization methods of such estimators further on. 
\section{\label{secCEECM}Cross-entropy and its estimators in the ECM setting}
Let us consider the ECM setting as in Section \ref{secECM}. We have
\begin{equation}\label{ceEstECM}
\wh{\ce}_n(b',b) = \overline{(ZL'(\Psi(b) -bX))}_n= \Psi(b)\overline{(ZL')}_n -b^T \overline{(ZL'X)}_n.
\end{equation}
Let us assume Condition \ref{condPartvm}. Then,
\begin{equation}\label{nablabce}
\nabla_b\wh{\ce}_n(b',b) = \nabla\Psi(b)\overline{(ZL')}_n -\overline{(ZL'X)}_n
\end{equation}
and
\begin{equation}\label{whce}
\nabla^2_b\wh{\ce}_n(b',b) = \nabla^2\Psi(b)\overline{(ZL')}_n.
\end{equation}
Let us further assume Condition \ref{condtX}, so that $\nabla^2\Psi$ is positive definite. 
Then, from (\ref{whce}), $b\to\wh{\ce}_n(b',b)(\omega)$ has a positive definite Hessian and thus it is strictly convex
only for $\omega \in \Omega_1^n$ and $b'\in A$ such that 
\begin{equation}\label{gzlp}
\overline{(ZL')}_n(\omega)>0.
\end{equation} 
Furthermore, $b^*_n \in A$ is the unique minimum point of  $b\to \wh{\ce}_n(b',b)(\omega)$ only if (\ref{gzlp}) holds and
\begin{equation}\label{nabcebs} 
\nabla_b\wh{\ce}_n(b',b^*_n)(\omega)=0
\end{equation} 
(where by $\nabla_b\wh{\ce}_n(b',b^*_n)(\omega)$ we mean $\nabla_b(\wh{\ce}_n(b',b)(\omega))_{b=b^*_n}$).
Assuming (\ref{gzlp}),  from 
(\ref{nablabce}), (\ref{nabcebs}) holds only if %, (\ref{nabcebs}) holds only if 
\begin{equation}
\mu(b^*_n)=\nabla\Psi(b^*_n) =\frac{\overline{(ZL'X)}_n}{\overline{(ZL')}_n}(\omega),
\end{equation}
or from Theorem \ref{thPsimu} only if 
$\frac{\overline{(ZL'X)}_n}{\overline{(ZL')}_n}(\omega) \in \mu[A]$
and
\begin{equation}\label{bsceECM}
b^*_n=\mu^{-1}\left(\frac{\overline{(ZL'X)}_n}{\overline{(ZL')}_n}(\omega)\right).
\end{equation}
Let us assume that 
\begin{equation}\label{zxifin}
\E_{\PQ_1}(|ZX_i|)<\infty, \quad i=1,\ldots, l. 
\end{equation}
Due to $X$ having finite all mixed moments, from H\"{o}lder's inequality, (\ref{zxifin})
holds e.g. when $\E_{\PQ_1}(|Z|^{p})<\infty$ for some $p>1$. For the cross-entropy we then have
\begin{equation}\label{cebform}
\ce(b) = \alpha\Psi(b)-b^T\E_{\PQ_1}(ZX),\quad b \in A.
\end{equation}
Thus, analogously as for the cross-entropy estimator above,  $\ce$ has a positive definite Hessian everywhere only if $\alpha>0$,
and $\ce$ has a unique minimum point only if $\alpha> 0$ and
\begin{equation}\label{muepqzx}
\frac{\E_{\PQ_1}(ZX)}{\alpha}\in \mu[A], 
\end{equation}
in which case such a point is
\begin{equation}\label{ceminbs}
b^*=\mu^{-1}\left(\frac{\E_{\PQ_1}(ZX)}{\alpha}\right). 
\end{equation}
\begin{remark}\label{remCEneg}
Note that we can receive analogous conditions as above for the cross-entropy and its estimator to have negative definite Hessians
or have unique maximum points by replacing $Z$ by $-Z$ (and thus also $\alpha$ by $-\alpha$) 
in the above conditions. The formulas for the maximum points remain the same 
as for the minimum points above. With some exceptions, in the further sections we shall focus on the minimization of cross-entropy and its estimators and will
be interested in checking the conditions as in the main text above. 
However, we can analogously perform their 
maximization, or jointly optimization if we consider alternatives of the above conditions. 
\end{remark}

\section{\label{secCondLETS}Some conditions in the LETS setting}
Let $||\cdot||_\infty$ denote the supremum norm induced by the standard Euclidean norm $|\cdot|$. 
Consider the LETS setting as in Section \ref{secLETS}. For each real matrix-valued process $Y=(Y_k)_{k\in \N}$ on $\mc{C}$ and $B \in \mc{E}$, let 
us define 
\begin{equation}
||Y||_{\tau,B,\infty} = \esssup_{\PU}(\I_B\I(0<\tau<\infty)\max(||Y_0||_{\infty}, \ldots,||Y_{\tau-1}||_{\infty})),
\end{equation}
which for $B=\Omega_1$ is denoted simply as $||Y||_{\tau,\infty}$.
Let $S$ be an $\overline{\R}$-valued random variable on $\mc{S}_1$.
Further on in this work we will often assume the following conditions.
\begin{condition}\label{condBound1}
It holds
\begin{equation}\label{Rlambda}
R:=||\Lambda||_{\tau,S\neq0,\infty}<\infty. 
\end{equation} 
\end{condition}
\begin{condition}\label{condBound2}
A number $s \in \N_+$ is such that
\begin{equation}\label{Staufin}
\I(S\neq 0)\tau \leq s.  
\end{equation}
\end{condition}
Note that conditions \ref{condBound1} or \ref{condBound2} hold for each possible random variable $S$ as above only if 
they hold for some $S$ such that $S(\omega)\neq0$, $\omega \in \Omega_1$, that is only if
\begin{equation}\label{supLambdafin}
||\Lambda||_{\tau,\infty}<\infty 
\end{equation}
for Condition \ref{condBound1}, or 
\begin{equation}\label{tauFin}
\tau \leq s \in \N_+
\end{equation}
for Condition \ref{condBound2}.
\begin{remark}\label{condImplzntau}
Note that Condition \ref{condBound2} implies Condition \ref{condzntau} for $Z=S$, %and thus also Condition \ref{condpqbllpq1},
while (\ref{tauFin}) implies Condition \ref{equivCond}. %and thus also Condition \ref{condpqpq1}.
\end{remark} 
For each real matrix-valued function $f$ on $\R^m$ and $B \subset \R^m$, 
let us denote $||f||_{B,\infty}=\sup_{x \in B}||f||_{\infty}$. If $\tau$ is the exit time of an Euler scheme 
$X$ of a set $D$ such that $X_0\in D$, then for $\Lambda$ is as in  (\ref{lambdatheta}) we have 
%\begin{equation} 
$||\Lambda||_{\tau,\infty}\leq ||\Theta||_{D,\infty}$. 
%\end{equation} 
In particular, if 
\begin{equation}\label{thetafin} 
||\Theta||_{D,\infty}<\infty, 
\end{equation} 
then we have (\ref{supLambdafin}). Note that from (\ref{thetadef}), 
(\ref{thetafin}) is equivalent to $||\wt{r}_i||_{D,\infty}<\infty$, $i=1, \ldots,l$. 
In particular, (\ref{thetafin}) and thus also (\ref{supLambdafin}) 
hold in our numerical experiments as discussed in Section \ref{secSpecNumExp}, both when using the time-independent 
and time-dependent IS basis functions, where in the time-dependent case by $\wt{r}_i$ we mean $\wh{r}_i$ as in Section \ref{secSpecNumExp} and we consider 
$D$ equal to $\wh{D}$ as in (\ref{whd}), $X$ equal to $X'$ as in Remark \ref{remExt}, and $\tau$ equal to $\tau'$ as in (\ref{taup}). 

Let us discuss how one can enforce (\ref{tauFin}) if it is initially not fulfilled, as is the case for the translated committors and the MGF in our numerical experiments.
Analogous reasonings as below can be applied also to more general stopped sequences or processes than in the LETS setting. 
For some $s\in \N_+$ and $z_s \in \R$, 
instead of $\tau$ and $Z$ we can consider their terminated versions $\tau_s = \tau\wedge s$ and $Z_s=\I(\tau \leq s)Z +z_s\I(\tau > s)$ 
and focus on computing $\alpha_s=\E_{\PU}(Z_s)=\E_{\PU}(\I(\tau \leq s)Z) +z_s\PU(\tau > s)$ rather than $\alpha=\E_\PU(Z)$.  
If $\PU(\tau=\infty)=0$, or 
$\PU(Z\neq 0,\ \tau=\infty)=0$ and $\lim_{s\to \infty}z_s=0$, then $\PU$ a.s. $Z_s \rightarrow Z$, so that 
assuming further that $\limsup_{s\to \infty}|z_s|<\infty$, from $|Z_s|\leq |Z| +|z_s|$  and Lebesgue's dominated convergence theorem, 
%\begin{equation}
$\lim_{s\rightarrow\infty}\alpha_s=\alpha$.
%\end{equation}
Thus, in such a case, for a sufficiently large $s$ we will make arbitrarily small absolute error when  approximating $\alpha$ by $\alpha_s$.
Let us provide some upper bounds on this error. 
If $\esssup_{\PU}(|Z -z_s|\I(\tau> s))\leq M_s \in [0,\infty)$, then 
\begin{equation}\label{errEstim}
|\alpha - \alpha_s|=|\E_{\PU}((Z-z_s)\I(\tau> s))| \leq M_s\PU(\tau> s).
\end{equation}
For the MGF example from Section \ref{secImpSpec} we can take 
$z_s=M_s=\frac{1}{2}\exp(-ph(s+1))$, while for the translated committors we can choose $z_s=a+\frac{1}{2}$ and $M_s= \frac{1}{2}$. 
The quantity $\PU(\tau > s)$ % $\PR(\tau< s)$ which can be used to compute it, 
can be estimated using IS from the same simulations as used to estimate 
$\alpha_s$ or in a separate IS MC procedure. 
Alternatively, if we have $\tau \leq \wh{\tau}$ for some random variable $\wh{\tau}$ with a known distribution, 
we can use the inequality $\PU(\tau> s) \leq \PU(\wh{\tau}> s)$ to bound the right side of (\ref{errEstim}) from above.
For instance, if $\wh{\tau}$ has a geometric distribution with a parameter $q$ (see Theorem \ref{thNTau} for a situation in which this may occur), then 
we have $\PU(\wh{\tau}> s)= (1-q)^{s}$ and thus $|\alpha - \alpha_s|\leq M_s(1-q)^{s}$.

\section{Some conditions in the LETGS setting}\label{secSomeCondLETGS}
Let us discuss some conditions and random conditions in the LETGS setting, which, as we shall discuss in the further sections, 
turn out to be necessary for the existence of the unique minimum points of cross-entropy, mean square, and their estimators in this setting. 
Let $Z$ be an $\R$-valued $\mc{S}_1$-measurable random variable (where $\mc{S}_1=(E,\mc{F}_{\tau})$). 
\begin{definition}\label{defAb} 
For $b \in A=\R^l$, we define a random condition $A_b$ on $\mc{S}_1$ as follows 
\begin{equation} 
A_b = (Z\neq0,\ 0<\tau<\infty,\text{ and there exists } k \in \N,\ k<\tau, \text{ such that } \lambda_k(b) \neq 0). 
\end{equation} 
\end{definition} 
\begin{lemma}\label{lemABL} 
If $A_b$ does not hold and $Z\neq 0$, then for each $a \in \R^l$ and $t \in \R$
\begin{equation}\label{Latb} 
L(a+tb)=L(a). 
\end{equation} 
\end{lemma} 
\begin{proof} 
From (\ref{LLETGS}), when $\tau=0$ then the both sides of (\ref{Latb}) are equal to $1$ and
when $\tau=\infty$ --- to $\epsilon$.
If $A_b$ does not hold, $Z\neq 0$, and
$0<\tau<\infty$, then for each $0\leq k<\tau$ we have $\lambda_k(b)=0$, and thus for each $a \in \R^l$ and $t \in \R$, 
$\lambda_k(a+tb) =\lambda_k(a)+t\lambda_k(b)=\lambda_k(a)$, so that (\ref{Latb}) also follows from (\ref{LLETGS}). 
\end{proof}  

\begin{lemma}\label{lemEquiv}
For $n \in \N_+$, the following random conditions on $\mc{S}_1^n$ are equivalent.
\begin{enumerate}
\item For each $b\in \R^l,\ b \neq 0$, there exists  $i\in\{1,\ldots, n\}$ such that 
$(A_b)_i$ holds (where we use the notation as in (\ref{Yidefs})). 
\item For some (equivalently, for each) random variable $K$ on $\Omega$ which is positive on $Z \neq 0$, 
%\begin{equation}
$\overline{(\I(Z\neq 0)GK)}_n$ %=\frac{1}{n}\sum_{i=1}^nK_i\I(s_i)G_i
%\end{equation}
is positive definite.
\item It holds $N:= \sum_{i=1}^n \I(Z_i\neq0,\ \tau_i<\infty)\tau_i>0$. Let 
a matrix $B \in \R^{(dN)\times l}$ be such that 
for each $i \in \{1,\ldots,n\}$ such that $0<\tau_i<\infty$ and $Z_i\neq0$, for each $k \in \{0,\ldots,\tau_i-1\}$ and $j\in \{1,\ldots,d\}$
the $\sum_{v=1}^{i-1}\I(Z \neq 0,\ \tau_v<\infty)\tau_v +kd +j$th row of $B$
is equal to the $j$th row of $(\Lambda_{k})_i$. Then, the columns of $B$ are linearly independent. 
\end{enumerate}
\end{lemma}
\begin{proof}
The fact that the second point above is a random condition follows from Sylvester's criterion.
The equivalence of the first two conditions follows from the fact that for each $b\in\R^l$
\begin{equation}
b^T\overline{(K\I(Z\neq 0)G)}_nb=\frac{1}{2n}\sum_{i=1}^n(\I(0<\tau<\infty,\ Z\neq 0)K\sum_{k=0}^{\tau-1}|\lambda_k(b)|^2)_i
\end{equation}
and the equivalence of the first and last condition is obvious. 
\end{proof}
\begin{definition}\label{defLemEquiv}
We define $r_n$ to be one of the equivalent random conditions in Lemma \ref{lemEquiv}. 
\end{definition}

\begin{lemma}\label{lemCond1}
The below three conditions are equivalent. 
\begin{enumerate} 
\item For each $b \in \R^l$, $b \neq 0$, we have  $\PQ_1(A_b)>0$. 
\item For each $b \in \R^l$, from $\PQ_1(A_b)=0$ it follows that $b=0$. 
\item Let $\wt{\Lambda}_j=((\Lambda_{k,i,j})_{i=1}^d)_{k\in \N} \in \mc{J}$, $j=1,\ldots,l$  (see Definition \ref{defJ}). Let 
$\sim$ be a relation of equivalence on $\mc{J}$ such that for $\psi_1, \psi_2 \in \mc{J}$, $\psi_1\sim\psi_2$, 
only if $\PQ_1$ a.s. if $0<\tau<\infty$ and $Z \neq 0$ then 
$\psi_{1,i}=\psi_{2,i}$, $i=0,\ldots,\tau-1$. Then, the equivalence classes 
$[\wt{\Lambda}_1]_\sim,\ldots,[\wt{\Lambda}_l]_{\sim}$ are linearly independent 
in the linear space 
$\mc{J}/\mathord\sim$ of equivalence classes of $\sim$, defined in a standard way
(i.e. the operations in such a linear space are defined by using in them in the place of the equivalence classes their arbitrary members and then taking the 
equivalence class of the result).  
\end{enumerate}
\end{lemma}
\begin{proof}
The equivalence of the first two conditions is obvious. 
The equivalence of the last two conditions follows from the fact that, using notations as in the third condition, for $b \in \R^l$,
$\sum_{i=1}^lb_i \wt{\Lambda}_i$ is equal to the zero in $\mc{J}/\mathord\sim$ only if $\PQ_1(A_b)=0$. 
\end{proof}

\begin{condition}\label{cond1}
We define the condition under consideration to be one of the conditions from Lemma \ref{lemCond1}.  
\end{condition}

% A_b \equiv (0<\tau<\infty, Z\neq0,\text{ and there exists } k \in \N,\ k<A_b \equiv (0<\tau<\infty, Z\neq0,\text{ and there exists } k \in \N,\ k<\tau, \text{ such that } \lambda_k(b)=\Lambda_k b \neq 0). 
% \tau, \text{ such that } \lambda_k(b)=\Lambda_k b \neq 0). 

\begin{remark}\label{remOtherPU}
Note that for a probability $\PS \sim_{\tau<\infty} \PQ_1$ we have  
$\PS(A_b)>0$ only if  $\PQ_1(A_b)>0$, so that 
Condition \ref{cond1} holds only if it holds for such a $\PS$ in the place of $\PQ_1$. 
%A_b \equiv (0<\tau<\infty, Z\neq0,\text{ and there exists } k \in \N,\ k<\tau, \text{ such that } \lambda_k(b)=\Lambda_k b \neq 0). 
\end{remark}
\begin{remark}\label{remPQ1AEquiv}
Note that $\PQ_1(A_b)>0$ only if for some $l\in \N_+$ and $k \in \N$, $k <l$, we have  $\PQ_1(Z \neq 0,\ \tau=l,\ \lambda_k(b)\neq0)>0$.  
\end{remark}

% \begin{remark}
% Note that $A_b\in \mc{F}_{\tau}$ and thus $\PU(A_b)=\PQ_1(A_b)$. Furthermore, $\PQ_1(A_b)>0$ only if 
% %A_b \equiv (0<\tau<\infty, Z\neq0,\text{ and there exists } k \in \N,\ k<\tau, \text{ such that } \lambda_k(b)=\Lambda_k b \neq 0). 
% for some $l\in \N_+$ and $k \in \N$, $k <l$, we have  $\PQ_1(Z \neq 0,\ \tau=l,\ \lambda_k(b)\neq 0)$. 
% \end{remark}
%We shall provide some useful criteria for this condition to hold in the next section.
\begin{lemma}\label{lemPosK}
Let for some probability $\PS \sim_{\tau<\infty}\PQ_1$, a
random variable $K$ on $\mc{S}_1$ be $\PS$ a.s. positive on $Z\neq0$, and let
$\I(Z\neq 0)KG$ have $\PS$-integrable entries. Then, $\E_{\PS}(\I(Z\neq0)KG)$ is positive definite only if Condition \ref{cond1} holds. 
\end{lemma} 
\begin{proof} 
%It follows from the fact that from Condition \ref{cond1}, 
For each $b\in \R^l$, $b \neq 0$,
\begin{equation}
b^T\E_{\PS}(\I(Z\neq 0)KG)b=\frac{1}{2}\E_{\PS}(\I(Z\neq 0,\ 0<\tau<\infty)K\sum_{k=0}^{\tau-1}|\lambda_k(b)|^2)
\end{equation}
is greater than zero only if $\PS(A_b)>0$, so that from Remark \ref{remOtherPU} we receive the thesis.
\end{proof}

Let $\Sym_n(\R)$ denote the subset of $\R^{n\times n}$ consisting of symmetric matrices, 
and let $m_n:\Sym_n(\R) \to \R$ be such that for $A \in \Sym_n(\R)$, $m_n(A)$ is equal to the lowest eigenvalue of $A$,
or equivalently
\begin{equation}
m_n(A)=\inf_{x \in \R^n,\ |x|=1} x^TAx.  
\end{equation}

\begin{lemma}\label{lemLipschMin}
$m_n$ is Lipschitz from $(\Sym_n(\R),||\cdot||_{\infty})$ to $(\R,|\cdot|)$ with a Lipschitz constant $1$.
\end{lemma}
\begin{proof}
For $A, B \in \Sym_n(\R)$ and $x \in \R^n$, $|x|=1$, we have 
$x^TAx= x^TBx+x^T(A-B)x$, so that 
\begin{equation}
x^TBx - ||A-B||_{\infty} \leq x^TAx\leq x^TBx + ||A-B||_{\infty} 
\end{equation}
and thus 
\begin{equation}
m_n(B) - ||A-B||_{\infty} \leq m_n(A)\leq m_n(B)+ ||A-B||_{\infty} 
\end{equation}
and 
\begin{equation}
|m_n(B)-m_n(A)|\leq ||A-B||_{\infty}. 
\end{equation}
\end{proof}
\begin{lemma}\label{lemPosDef}
If the entries of some matrices $M_n\in \Sym_{l}(\R)$, $n \in \N_+$, 
converge to the respective entries of a positive definite symmetric matrix $M\in \R^{l\times l}$, then for a sufficiently large 
$n$, $M_n$ is positive definite.
\end{lemma}
\begin{proof}
This follows from the fact that 
$A \in \Sym_{l}(\R)$ is positive definite only if $m_l(A)>0$, and from Lemma \ref{lemLipschMin}, $\lim_{n\to \infty}m_l(M_n)=m_l(M)$.
\end{proof}
\begin{theorem}\label{thPos}
If Condition \ref{cond1} holds, then under Condition \ref{condKappa}, 
a.s. for a sufficiently large $n$,
$r_n(\wt{\kappa}_n)$ holds for $r_n$ as in Definition \ref{defLemEquiv}. 
In particular, a.s.
%\begin{equation}
$\lim_{n\rightarrow \infty }\PR(r_n(\wt{\kappa}_n))= 1$. 
%\end{equation}
%as $n\rightarrow \infty$.
\end{theorem}
\begin{proof}
Let $K=\exp(-\max_{i,j=1,\ldots,d}|G_{i,j}|)$. Then, $K>0$ and the entries of the matrix $\I(Z\neq 0)GK$ are bounded and thus $\PQ'$-integrable. 
Thus, from Lemma \ref{lemPosK} for $\PS=\PQ'$, %(note that from (\ref{tauinfsim}) we have $\PQ'\sim_{\tau<\infty}\PQ_1$), 
$\E_{\PQ'}(\I(Z\neq 0)KG)$ is positive definite. %\ref{} 
Let $A_n=\overline{(\I(Z\neq 0)KG)}_n(\wt{\kappa}_n)$. From the SLLN, a.s.
\begin{equation}
\lim_{n\rightarrow \infty}A_n=\E_{\PQ'}(\I(Z\neq 0)KG).%=\E_{\PU}(\I(s)\wt{K}G).
\end{equation} 
Thus, from Lemma \ref{lemPosDef}, a.s. $A_n$ is positive definite for a sufficiently large $n$ and 
the thesis follows from the second point of Lemma \ref{lemEquiv}. 
\end{proof}

\section{Discussion of Condition \ref{cond1} in the Euler scheme case}\label{secCond}
Let us consider IS for an Euler scheme with a 
linear parametrization of IS drifts, discussed in Section \ref{secEulSDE} below formula (\ref{Xkb}). 
In this section we shall reformulate Condition \ref{cond1} and provide some sufficient assumptions for it to hold in such a case. 
%We shall use notations as in Section \ref{secEulSDE}. 
% \begin{equation}\label{Abdef} 
% A_b \equiv 0<\tau<\infty, Z\neq0,\text{ and there exists } k \in \N,\ k<\tau, \text{ such that } \lambda_k(b)=\Lambda_k b \neq 0. 
% \end{equation} 
% \begin{equation} 
% \begin{split} 
% \PU(A_b)&=\sum_{l=1}^{\infty}\PU(\{Z\neq0,\ \tau=l\}\cap \bigcup_{k=0}^{l-1}\lambda_k(b) \neq 0)\\ 
% &= \sum_{l=1}^{\infty}\PU(\tau=l, Z\neq0,\text{ and there exists } k \in \N,\ k<l, \text{ such that } \lambda_k(b)=\Lambda_k b \neq 0).\\ 
% \end{split} 
% \end{equation} 

Let us define a measure $\nu$ on $\mc{S}(\R^m)$ to be such that %for each $B \in \mc{B}(\R^m)$ 
for each $B \in \mc{B}(\R^m)$ 
\begin{equation}\label{nudef} 
\begin{split} 
\nu(B)&= \E_{\PU}(\I(Z\neq 0,\ 0<\tau<\infty)\sum_{k=0}^{\tau-1}\I(X_k\in B))\\ 
&= \sum_{l\in \N_+}\sum_{k=0}^{l-1}\PU(Z\neq 0,\ \tau=l,\ X_k\in B)\\ 
&= \sum_{k=0}^{\infty}\PU(Z\neq 0,\ k< \tau<\infty,\ X_k\in B).\\ 
\end{split} 
\end{equation} 
\begin{remark}\label{remnu}
From the second line of (\ref{nudef}), Remark \ref{remPQ1AEquiv}, and (\ref{lambdatheta}), $\PQ_1(A_b)=\PU(A_b)=0$ is equivalent to 
\begin{equation}\label{mutheta}
\nu(\{\Theta b \neq 0\})=0
\end{equation}
(where $\{\Theta b \neq 0\}=\{x\in \R^m:\Theta(x) b \neq 0\}$).  
\end{remark}
\begin{remark}
Let for each $i \in \{1,\ldots,l\}$, $\wt{\Theta}_i:\R^m \rightarrow \R^d$ be the $i$th column of $\Theta$ and $[\wt{\Theta}_i]_{\approx}$ be the class of 
equivalence of $\wt{\Theta}_i$ with respect to the relation $\approx$ of equality $\nu$ a.e. on the set $\mc{K}$ of
measurable functions from $\mc{S}(\R^m)$ to $\mc{S}(\R^d)$. Then, from Remark \ref{remnu}, Condition \ref{cond1} is equivalent to 
$[\wt{\Theta}_i]_{\approx}$, $i=1,\ldots,l$, being linearly independent in the linear space $\mc{K}/\mathord\approx$ defined in a standard way. 
\end{remark}

Let us assume that $m =n+1$ for some $n \in \N_+$. 
Consider the following condition concerning the IS basis functions $\wt{r}_i:\R^m \to \R^d$, $i=1,\ldots,l$, as in Section \ref{secEulSDE}. 

\begin{condition}\label{condrform}
For some $m_1,m_2 \in \N_+$, functions $g_{1,i}:\R \to \R$, $i =1,\ldots,m_1$, 
and $g_{2,i}:\R^n \to \R^d$, $i =1,\ldots,m_2$, are such that
for $K_1=\{kh:k\in\{1,\ldots,m_1\}\}$, $g_{1,i|K_1}$, $i =1,\ldots,m_1$, are linearly independent and
for each $i \in\{1,\ldots,m_1\}$, for some open set $K_{2,i}\subset \R^n$, 
$g_{2,j|K_{2,i}}$, $j =1,\ldots,m_{2}$, are continuous and linearly independent.   
Furthermore, we have $l=m_1m_2$, and denoting $\pi(i,j)= m_2(i-1)+j$, 
for each $x \in \R^n$ and $t \in \R$ we have  $\wt{r}_{\pi(i,j)}(x,t) = g_{1,i}(t)g_{2,j}(x)$,  $i =1,\ldots,m_1$, $j =1,\ldots,m_2$. 
\end{condition}

\begin{remark}\label{remrfulfill}
As the functions $g_{1,i}$ as in the above condition one can take for example polynomials $g_{1,i}(t)=t^{i-1}$, $i =1,\ldots,m_1$. 
For $m_1=2$ one can also use $g_{1,1}(t)=1$ and $g_{1,2}(t)=t^p$ for some $p \in \N_+$. For $n=1$ 
and arbitrary nonempty open sets $K_{2,i}\subset \R$, $i=1,\ldots,m_2$, as the functions  $g_{2,i}$ in the above condition one can take 
e.g. polynomials analogously as above or Gaussian functions $g_{2,i}(x)=a_i\exp(\frac{(x-p_i)^2}{s})$ 
for some $a_i\in \R \setminus\{0\}$, $s\in \R_+$,  and $p_i \in R$ different 
for different $i$
(the linear independence of such Gaussian functions on each open interval can be proved by an analogous reasoning as in \cite{egreg2014}). 
In particular, for such $K_{2,i}$, Condition \ref{condrform} holds 
for the functions $\wt{r}_i$ equal to $\wh{r}_i$ as in (\ref{rtimedep}) 
or equal to $\wh{r}_i$ such that $\wh{r}_i(x,t)= \wt{r}_i(x)$, $x \in \R^n$, $t \in \R$, for $\wt{r}_i$ as in (\ref{riexp}), 
where in the first case $m_1=2$, in the second $m_1=1$, and in both cases 
$n=1$ and $m_2$ is equal to $M$ as in Section \ref{secSpecNumExp}. 
\end{remark}

Let $\lambda$ denote the Lebesgue measure on $\R^n$ and $\delta_x$ --- the Dirac measure centred on $x$.
\begin{theorem}\label{thRestr} %\label{remRestr} 
If Condition \ref{condrform} holds and
\begin{equation}\label{lambdadeltall}
\lambda \otimes\delta_{ih}\ll_{K_{2,i}\times\{ih\}} \nu,\quad i=1,\ldots,m_1,  
\end{equation}
then Condition \ref{cond1} holds. 
\end{theorem}
\begin{proof} 
Let $b \in \R^l$ be such that $\PU(A_b)=0$. Then, from Remark \ref{remnu}, 
$\nu(\{\Theta b \neq 0\})=0$ and thus for $i =1,\ldots,m_1$, 
$\nu(\{(x,ih):x \in K_{2,i},\ \Theta(x,ih) b \neq 0\})=0$ and  from (\ref{lambdadeltall})
\begin{equation}\label{lambdaxK}
\lambda(\{x \in K_{2,i}:\Theta(x,ih) b \neq 0\})=0.  
\end{equation}
From (\ref{thetadef}) and Condition \ref{condrform} we have for $x \in \R^n$ and $t \in \R$
\begin{equation}\label{Thetaxtb}
\Theta(x,t) b= \sqrt{h}\sum_{j=1}^{m_1}\sum_{k=1}^{m_2}b_{\pi(j,k)}g_{1,j}(t)g_{2,k}(x).
\end{equation}
Denoting for $i=1,\ldots,m_1$ and  $k=1,\ldots,m_2$
\begin{equation}\label{ajdef}
a_{i,k} =\sum_{j=1}^{m_1}b_{\pi(j,k)}g_{1,j}(ih),   
\end{equation}
we thus have $\Theta(x,ih) b= \sqrt{h}\sum_{k=1}^{m_2}a_{i,k}g_{2,k}(x)$, $x \in \R^n$,
and from (\ref{lambdaxK}),  
\begin{equation}
\lambda(\{x\in K_{2,i}: \sum_{k=1}^{m_2}a_{i,k}g_{2,k}(x)\neq0\})=0.  
\end{equation}
Thus, for $i=1,\ldots,m_1$, 
from the continuity and linear independence of $g_{2,k|K_{2,i}}$ , $k=1,\ldots,m_2$, we 
have $a_{i,k}=0$, $k= 1,\ldots,m_2$. Therefore, from (\ref{ajdef}), for $k=1,\ldots,m_2$, 
$\sum_{j=1}^{m_1}b_{\pi(j,k)}g_{1,j|K_1} =0$, so that from the linear independence of $g_{1,j|K_1}$, $j=1,\ldots, m_1$, we have $b=0$.
\end{proof}

Let us assume the following condition. 
\begin{condition}\label{condformprim} 
We have $\sigma_{m,i}=0$, $i =1,\ldots,d$, $\mu_m=1$, $(x_{0})_m=0$, and 
$\wt{\sigma}:\R^m \to \R^{n \times d}$ is such that 
$\wt{\sigma}_{i,j}=\sigma_{i,j}$, $i =1,\ldots,n$, $j =1,\ldots,d$. 
\end{condition}
Note that it now holds %$X_{k,m}=kh$, $k \in \N$. 
for $\wt{X}_k= (X_{k,i})_{i=0}^n$, $k \in \N$, that 
\begin{equation}\label{xkprob}
X_k = (\wt{X}_k, kh),\quad k \in \N. 
\end{equation}
For $x\in \R^n$ and $k \in \N$ for which $\wt{\sigma}(x,kh)$ has linearly independent rows, let $Q_k(x) = (h\wt{\sigma}(x,kh) \wt{\sigma}(x,kh)^T)^{-1}$, 
and for $y \in \R^n$, let
\begin{equation}
\rho_k(x,y)=\frac{\sqrt{\det(Q_k(x))}}{(2\pi)^{\frac{m}{2}}}\exp((y- x -h\mu(x))^TQ_k(y- x -h\mu(x))).
\end{equation}
%so that $\rho_k(x,\cdot)$ is the transition density of $\wt{X}=(\wt{X}_k)_{k\in \N}$ in $x$ from its $k$th to its $k+1$st state. 

 \begin{theorem}\label{thbc}
 Let $k\in \N_+$ and sets $B_1,B_2,\ldots,B_k,C \in \mc{B}(\R^n)$ have positive Lebesgue measure. Let $\PU$ a.s. the fact that
 $\wt{X}_i \in B_i,\ i=1,\ldots,k$ and  $\wt{X}_{k+1} \in C$ imply 
that $Z\neq 0$ and $k<\tau<\infty$. 
%  \begin{equation}\label{implthbc}
% ()\Rightarrow (Z\neq 0\wedge k<\tau<\infty).  
% \end{equation}
 Let further 
 $\wt{\sigma}(x,t)$ have independent rows for each $(x,t) \in \{x_0\} \cup \bigcup_{i=1}^kB_i\times \{ih\}$. Then, %$\nu \ll_K \mu$. 
\begin{equation}
\lambda \otimes\delta_{ih}\ll_{B_{i}\times\{ih\}} \nu,\quad i=1,\ldots,k.   
\end{equation}
 \end{theorem}
 \begin{proof}
It follows from the fact that for each $j\in \{1,\ldots,k\}$, for each $D\subset \mc{B}(B_j)$ such that $\lambda(D)>0$, 
for $\wt{D}=\prod_{i=1}^{j-1}B_i\times D\times \prod_{i=j+1}^{k}B_i\times C$ we have
%\begin{equation}\label{lambdadeltall} 
%$\lambda \otimes\delta_{ih}\ll_{K_2\times\{ih\}} \nu,\quad i\in\{1,\ldots,m_1\}$. 
%\end{equation}
\begin{equation}
\begin{split}
0 &< \int_{\wt{D}}\prod_{i=0}^k\rho_i(x_i,x_{i+1})\,\mathrm{d}x_1\,\mathrm{d}x_2\ldots\,\mathrm{d}x_{k+1}\\
&= \PU((\wt{X}_i)_{i=1}^{k+1}\in \wt{D})\\
&=\PU(Z\neq 0,\ k<\tau<\infty,\ (\wt{X}_i)_{i=1}^{k+1}\in \wt{D})\\
&\leq  \PU(Z\neq 0,\ j<\tau<\infty,\ \wt{X}_j \in D)\\
&\leq \nu(D\times \{jh\}),\\
\end{split}
\end{equation}
where in the last line we used (\ref{xkprob}) and the last line of (\ref{nudef}). 
\end{proof}

\begin{remark}\label{remImpl}
Let us consider the problems of estimating an MGF $\mgf(x_0)$, a translated committor $q_{AB,a}(x_0)$, and a translated 
exit probability by a given time 
$p_{T,a}(x_0)$ as in Section \ref{secImpSpec} for $a \neq -1$. As $x_0$, $\mu$, and $\sigma$ fulfilling the above Condition 
\ref{condformprim} let us consider $x_0'$, $\mu'$, and $\sigma'$ as in Remark \ref{remExt}, and 
as an Euler scheme $X$ in the LETGS setting as above let us consider the time-extended process $X'$ as in that remark. 
Note that the process $X$ as in Section \ref{secImpSpec} is now equal to the above $\wt{X}$. 
Let $k \in \N_+$. Then for $D$ and $Z$ corresponding to the above expectations as in Section \ref{secImpSpec}, assuming that 
$C\in\mc{B}(B)$ in the case of estimation of the translated committor, or $C\in \mc{B}(D')$ for the MGF or the exit probability,  
and additionally $T\geq h(k+1)$ in the case of the exit probability, for each $B \in \mc{B}(D)$ 
we have that $\wt{X}_i \in B,\ i=1,\ldots,k$ and  $\wt{X}_{k+1} \in C$ implies that $Z\neq 0$ and $\tau =k+1$ 
(where for the exit probability rather than $\tau$ we mean $\tau'$ as in (\ref{taup})). 
This holds also for $Z$ and $\tau$ replaced by their terminated versions $Z_s$ and $\tau_s$ 
for $s\in \N_+$, $s>k$, as in Section \ref{secCondLETS}. 
\end{remark}

From remarks \ref{remrfulfill} and \ref{remImpl} and theorems \ref{thRestr} and \ref{thbc}
it follows that Condition \ref{cond1} holds in all the cases considered in our numerical experiments as in Section 
\ref{secSpecNumExp} if for the case of the exit probability before 
a given time we assume that $T\geq h(m_1+1)$ for $m_1$ depending on the basis functions used as 
in Remark \ref{remrfulfill}. 
Furthermore, this condition also holds in such terminated cases as in Section \ref{secCondLETS} for each $s\in \N_+$, $s > m_1$.

\section{\label{secCeLETGS}Cross-entropy and its estimators in the LETGS setting}
Consider the LETGS setting. From (\ref{lnlG}),
\begin{equation}\label{whcebLETGS}
\begin{split}
\wh{\ce}_n(b',b) &= \overline{ZL'(b^TGb +Hb + \I(\tau=\infty)\ln(\epsilon))}_n\\
&=b^T\overline{(ZL'G)}_nb+\overline{(ZL'H)}_nb+ \ln(\epsilon)\overline{(ZL'\I(\tau=\infty))}_n,\\ 
\end{split}
\end{equation}
so that 
\begin{equation}
\nabla_b\wh{\ce}_n(b',b)= 2\overline{(ZL'G)}_nb+\overline{(ZL'H)}_n.
\end{equation}
Thus, $b\rightarrow\wh{\ce}_n(b',b)(\omega)$ has a unique minimum point $b^*_n\in A$ only for 
$\omega \in \Omega_1^n$ for which $\overline{(ZL'G)}_n(\omega)$ is positive definite, in which case
for $A_n(b'):=2\overline{(ZL'G)}_n$ and $B_n(b'):=-\overline{(ZL'H)}_n$ we have 
\begin{equation}\label{bsdefLETGS}
b^*_n= (A_n(b'))^{-1}(\omega)B_n(b')(\omega).
\end{equation}
Note that if $Z\geq 0$ then from the second point of Lemma \ref{lemEquiv} for $K=ZL'$, for each
$\omega \in \Omega_1^n$, $\overline{(ZL'G)}_n(\omega)$ is positive definite only if $r_n(\omega)$ holds (see Definition \ref{defLemEquiv}). 
%We shall need the following condition to hold.
\begin{condition}\label{condInt1}
$ZG$ and $ZH$ have $\PQ_1$-integrable (equivalently, $\PU$-integrable) entries.% i.e. $\E_{\PQ_1}(|ZG_{i,j}|)<\infty$
%and $\E_{\PQ_1}(|ZH_i|)<\infty$, $i, j \in \{1,\ldots,l\}$. 
\end{condition}
\begin{lemma}\label{lempInt} 
Let Condition \ref{condBound1} hold for $S=Z$, let for some $p>1$, $\E_{\PU}(|Z|^p)<\infty$, and let 
for some $s \in \N_+$ and a random
variable $\wh{\tau}$ with a geometric distribution under $\PU$ with a parameter $q\in (0,1]$ it hold
\begin{equation}\label{tauleqwt}
\I(Z\neq 0)\tau \leq \wt{\tau}:=s+\wh{\tau}.  
\end{equation}
Then, for each $1\leq u<p$, we have $\E_{\PU}(|ZH_i|^u)<\infty$ and
$\E_{\PU}(|ZG_{i,j}|^u)<\infty$, $i, j \in \{1,\ldots,l\}$. In particular, Condition \ref{condInt1} holds. 
\end{lemma}
 \begin{proof} 
Let  $1\leq u<p$.  
 For $r \in (u,\infty)$ such that $\frac{u}{r}+\frac{u}{p}=1$, using H\"{o}lder's inequality and 
 (\ref{tauleqwt}) we have 
 \begin{equation}
 \E_{\PU}(|Z||G||_{\infty}|^u)\leq \E_{\PU}((|Z|\tau \frac{1}{2}R^2)^u) \leq (\frac{1}{2}R^{2})^u(\E_{\PU}(|Z|^{p}))^{\frac{u}{p}}(\E_{\PU}(\wt{\tau}^{r}))^{\frac{u}{r}} < \infty
 \end{equation}
and
 \begin{equation}
 \E_{\PU}((|Z||H||_{\infty}|)^u) \leq \E_{\PU}((|Z| R\sum_{k=1}^{\tau}|\eta_{k}|)^u)
 \leq R^u(\E_{\PU}(Z^{p}))^{\frac{u}{p}}(\E_{\PU}((\sum_{k=1}^{\wt{\tau}}|\eta_{k}|)^r))^{\frac{u}{r}}.\\
 \end{equation}
 Furthermore,
 \begin{equation}
 \E_{\PU}((\sum_{k=1}^{\wt{\tau}}|\eta_{k}|)^r) = \sum_{l=1}^{\infty}\E_{\PU}(\I(\wt{\tau}=l)(\sum_{k=1}^{l}|\eta_{k}|)^r)
 \leq \sum_{l=1}^\infty l^{r-1}\E_{\PU}(\I(\wt{\tau}=l)\sum_{k=1}^{l}|\eta_{k}|^r)\\
 \end{equation} 
and from Schwarz's inequality, 
 \begin{equation}
 \E_{\PU}(\I(\wt{\tau}=l)\sum_{k=1}^{l}|\eta_{k}|^r)\leq \PU(\wt{\tau}=l)^{\frac{1}{2}}(\E_{\PU}((\sum_{k=1}^{l}|\eta_{k}|^r)^2))^{\frac{1}{2}}.  
 \end{equation}
It holds $\PU(\wt{\tau}=s+k) = q(1-q)^{k-1}$, $k \in \N_+$, and 
$\E_{\PU}((\sum_{k=1}^{l}|\eta_{k}|^r)^2)= l \E_{\PU}(|\eta_{1}|^{2r} + l(l-1)(\E_{\PU}(|\eta_{1}|^{r}))^2$. 
The thesis easily follows from the above formulas. 
\end{proof} 
Note that (\ref{tauleqwt}) in the above lemma holds e.g. for $s=0$ and $\wh{\tau}$ as in 
Theorem \ref{thNTau} if the assumptions of this theorem hold for $B=\{0\}$, or for $\wh{\tau}=0$ 
for $\tau$ being an arbitrary stopping time terminated at $s$ as in Section \ref{secCondLETS}. 

Let us assume conditions \ref{condzntau} and \ref{condInt1}. 
Then, from (\ref{lnlG}) we receive the following formula for the cross-entropy  
\begin{equation} \label{cebLETGS}
%\begin{split}
\ce(b)=\E_{\PQ_1}(Z\ln(L(b)))=b^T\E_{\PQ_1}(ZG)b+\E_{\PQ_1}(ZH)b. 
%&=\E_{\PQ'}L'Z\sum_{k=0}^{\tau}(\frac{1}{2}\lambda^2_{k}-\lambda_{k}(\wt{\eta_{k+1}}'+\lambda_k')).\\
%\end{split}
\end{equation} 
%, i.e. for each $b \in R^l$
%\begin{equation}
%b^T\E(G)b=\E((\overline{\lambda}^Tb)^2) >0.
%\end{equation}
Let 
%\begin{equation}\label{wtadef}
$\wt{A}=2\E_{\PQ_1}(ZG) =\nabla^2\ce(b)$, $b \in \R^l$,
%\end{equation}
and 
%\begin{equation}\label{wtbdef}
$\wt{B}=-\E_{\PQ_1}(ZH)$.
%\end{equation}
Then, we have 
%\begin{equation}
$\nabla \ce(b) = \wt{A}b-\wt{B}$.
%\end{equation}
Thus, if $\E_{\PQ_1}(ZG)$ is positive definite, then %so is $\nabla^2\ce$
$\ce$ has a unique point $b^*\in A$, satisfying 
\begin{equation}\label{zlg}
b^*=\wt{A}^{-1}\wt{B}. 
\end{equation}
If $Z\geq 0$, then from Lemma \ref{lemPosK} for $K=Z$, $\E_{\PQ_1}(ZG)$
is positive definite only if Condition \ref{cond1} holds. 
Remark \ref{remCEneg} applies also to the above discussion in the LETGS setting. % of cross-entropy and its estimators in the LETGS setting. 

\section{\label{secGen}Some properties of expectations of random functions} 
Some of the below theorems are modifications or 
slight extensions of well-known results; 
see the appendix of Chapter 1 in \cite{Shapiro2003}. 

Let $l \in \N_+$ and $A\in \mc{B}(\R^l)$ be nonempty. A function $f:A\to \overline \R$ is said to be lower semicontinuous in a point $b \in A$ if 
$\liminf_{x \to b}f(x)\geq f(b)$, and it is said to be lower semicontinuous if it is lower semicontinuous in each $b \in A$. 
\begin{lemma}\label{lemMin}
A lower semicontinuous function $f:A \to \overline{\R}$ such that $f> -\infty$ (i.e. $f(b)> -\infty$, $b \in A$)
attains a minimum on each nonempty compact set $K \subset A$  (where such a minimum may be equal to infinity). 
\end{lemma}
 \begin{proof}
Let $m=\inf_{b \in K}f(b)$ and let $a_n \in K$, $n \in \N_+$, be such that $\lim_{n\rightarrow \infty} f(a_n)=m$. 
Consider a subsequence $(a_{n_k})_{k \in \N_+}$ of $(a_n)_{n\in \N_+}$, converging to some $b^* \in K$.
Then, from the lower semicontinuity of $f$,  $m=\liminf_{k\rightarrow \infty} f(a_{n_k}) \geq f(b^*)$, so that $f(b^*)=m$.
\end{proof}

\begin{condition}\label{condHlow}
A (random) function $h:\mc{S}(A)\otimes(\Omega,\mc{F})\to \mc{S}(\overline{\R})$ 
is such that a.s. 
$b\ in A\to h(b):=h(b,\cdot)$ is lower semicontinuous and 
\begin{equation}\label{lowerBound}
\E(\sup_{b \in A}(h(b)_{-}))<\infty. 
\end{equation}
For such a $h$ we denote $b\in A\to f(b):=\E(h(b))$. 
\end{condition}

% Note that (\ref{lowerBound}) holds when $h$ is nonnegative,  in which case $(h(b,Y))_{-}=0$, $b \in A$.
% The following lemma is well-known, see Proposition 14, Chapter 1 in \cite{Shapiro2003}.
\begin{lemma}\label{lemSemi}
Assuming  Condition \ref{condHlow}, we have  $f>-\infty$ and $f$ is 
lower semicontinuous on $A$.
\end{lemma} 
\begin{proof}
From (\ref{lowerBound}), $f>-\infty$. For each $b \in A$ and 
  $a_n \in A$, $n \in \N$, such that $\lim_{n\to \infty }a_n=b$,
  from Fatou's lemma (which can be used thanks to (\ref{lowerBound})) and the a.s. lower semicontinuity of $b\rightarrow h(b)$,
  \begin{equation}
  \liminf_{n\to \infty}f(a_n) \geq \E(\liminf_{n \rightarrow \infty} h(a_{n})) \geq f(b).
  \end{equation} 
  \end{proof} 

Let further in this section $A \subset \R^l$ be open. % and let $f:A\to \overline{\R}$. 
For $x \in A$, let 
%\begin{equation}\label{dxdef} 
$d_x=\inf_{y\in A'} |y-x|$. 
%\end{equation} 
For a sequence $x_n \in A$, $n \in \N_+$, let us write $x_n\uparrow A$  if
$\max (\frac{1}{d_{x_n}},\ |x_n|)\to \infty$ as $n \to \infty$, i.e. $x_n$ 
in a sense tries to leave $A$.
For $a \in\overline{\R}$ and $f:A\to \overline{\R}$, let us denote by
$\lim_{x \uparrow A} f(x)=a$ the fact that $\lim_{n\to \infty}f(x_n)= a$ 
whenever $x_n \uparrow A$. 

\begin{condition}\label{condInf}
A lower semicontinuous function $f:A \to \overline{\R}$ fulfills $f>-\infty$
and $\lim_{x\uparrow A}f(x)= \infty$. 
\end{condition}

\begin{condition}\label{condStrict}
Condition \ref{condInf} holds, the set $B$ on which $f$ is finite is nonempty and convex, and $f$ is strictly convex on $B$.
\end{condition}

\begin{lemma}\label{lemMinPoint}
Under Condition \ref{condInf}, $f$ attains a minimum on $A$ and if Condition \ref{condStrict}
holds, then the corresponding minimum point $b^*$ is unique and $f(b^*)<\infty$.
\end{lemma}
\begin{proof}
Under Condition \ref{condInf}, 
for a sufficiently large $M>0$, for a compact set 
\begin{equation}\label{Kdef}
K=\{b\in A: |b| \leq M,\ d_b \geq \frac{1}{M}\}, 
\end{equation}
we have  $\inf_{b\in A}f(b)=\inf_{b \in K} f(b)$. 
From Lemma \ref{lemMin} there exists a minimum point $b^* \in K$ of $f$ on $K$ and thus also on $A$.
Under Condition \ref{condStrict} we have $f(b^*)<\infty$ and the uniqueness of 
$b^*$ follows from the strict convexity of $f$. 
\end{proof}

\begin{lemma}\label{lemFat}
Assuming Condition \ref{condHlow}, if with positive probability 
$\lim_{b\uparrow A}h(b)= \infty$, then $\lim_{b\uparrow A}f(b)= \infty$. 
\end{lemma} 
\begin{proof}
For $a_k \uparrow A$, from Fatou's lemma
\begin{equation}
\liminf_{k\rightarrow\infty}f(a_k)\geq 
\E(\liminf_{k\rightarrow \infty}h(a_k)) =\infty.  
\end{equation}
\end{proof}

\begin{lemma}\label{lemPropConv}
Under Condition \ref{condHlow}, let us assume that $A$ is convex, 
for some $b_0 \in A$,  $f(b_0) <\infty$, 
and a.s. $b\rightarrow h(b)$ is 
convex. Then, $f$ is convex on the convex nonempty set $B \subset A$ 
on which it is finite. 
If further 
with positive probability $b\rightarrow h(b) \in\R$ is strictly convex and 
$\lim_{b\uparrow A}h(b)= \infty$,
%satisfies Condition \ref{condStrict},
then $f$ satisfies Condition \ref{condStrict}. 
\end{lemma}
\begin{proof}
The (strict) convexity of $f$ and the convexity of $B$ easily follow from 
$f(b)=\E(h(b))$. 
The remaining points of Condition \ref{condStrict} 
follow from lemmas \ref{lemSemi} and \ref{lemFat}. 
\end{proof}

\section{\label{secMsqGen}Some properties of mean square and its estimators}
Let us consider the mean square function and its estimators 
as in sections \ref{secCoeffDiv} and \ref{secEstMin} (under appropriate assumptions as in these sections). 
\begin{condition}\label{condLConv}
$A$ is convex and $b \in A\to L(b)(\omega)\in \R$ is 
convex and continuous, $\omega \in \Omega_1$. 
\end{condition} 
%\begin{remark}\label{remConv}
From (\ref{whmsqdef}), if Conditon \ref{condLConv} holds, 
then $b\to \wh{\msq}_n(b',b)$ is convex and continuous (for each $b'\in A$ and when evaluated on each 
$\omega\in \Omega_1^n$).   
%\end{remark}
\begin{definition}\label{pmsqdef}
For $A$ open and convex, let the random condition $p_{\msq}$ on $\mc{S}_1$ hold only if $Z\neq 0$, 
$b\to L(b)\in \R_+$ is strictly convex, and $\lim_{b\uparrow A} L(b)= \infty$.  
\end{definition}
\begin{remark}\label{remMsqEst}
If Condition \ref{condLConv} holds and for some $n \in \N_+$, $\omega=(\omega_i)_{i=1}^n \in \Omega_1^n$,
and $i \in \{1,\ldots,n\}$, $p_{\msq}(\omega_i)$ holds, then for each $b' \in A$,
$b\to\wh{\msq}_n(b',b)(\omega)\in \R$ is strictly convex, 
continuous, and $\lim_{b\uparrow A}\wh{\msq}_n(b',b)(\omega)= \infty$. 
\end{remark}
It holds 
\begin{equation}\label{msq2sum}
\wh{\msq2}_{n}(b',b)=\frac{1}{n^2}\sum_{i=1}^n\left(Z_i^2L_i'L_i(b)\sum_{j\in\{1,\ldots,n\}, j\neq i}\frac{L_j'}{L_j(b)}\right)
+\frac{1}{n}\overline{((ZL')^2)}_n.
\end{equation}
%From (\ref{varest}) or (\ref{msq2ndef}),
Thus, $b\to\wh{\var}_{n}(b',b)$ and $b\to\wh{\msq2}_{n}(b',b)$
are positively linearly equivalent to $b\to f_{\var,n}(b',b)$ for
\begin{equation}\label{fnvardef}
f_{\var,n}(b',b)= \sum_{i=1}^n\left(Z_i^2L_i'\sum_{j\in \{1,\ldots,n\}, j \neq i}\frac{L_j'L_i(b)}{L_j(b)}\right),\quad b',b \in A.
\end{equation}
%Consider the following condition. 
\begin{condition}\label{condconVVar}
For each $\omega_1,\omega_2 \in \Omega_1$, $b\to \frac{L(b)(\omega_1)}{L(b)(\omega_2)}$ is convex. 
\end{condition}
Note that if Condition \ref{condconVVar} holds, then $b\to f_{\var,n}(b',b)$ is convex
and thus so are $b\to \wh{\var}_n(b',b)$ and $b\to \wh{\msq2}_{n}(b',b)$. 

\begin{remark}\label{remStrategy}
Let us assume Condition \ref{condKappa} and that $b \in A\to L(b)(\omega)\in \R$ is continuous, $\omega \in \Omega_1$. 
Then, for each $n \in \N_+$, from $\wh{\msq}_n$ being nonnegative, Condition \ref{condHlow} holds for 
$h(b,\cdot)= \wh{\msq}_n(b',b)(\wt{\kappa}_n)$, for which $f=\msq$ in that condition. 
%and $\wh{\var}_n$ are nonnegative, then, assuming Condition \ref{condKappa}, 
% one can attempt to apply Lemma \ref{lemPropConv} for some $n \in \N_+$, 
% $Y=\wt{\kappa}_n$,  $f=\msq$, and $h(b,\cdot)= \wh{\msq}_n(b',b)(\cdot)$.  
% Under Condition \ref{condHlow}, let us assume that $A$ is convex, for some $b_0 \in A$,  $f(b_0) <\infty$, 
% and a.s. $b\rightarrow h(b,Y)$ is 
% convex. Then, $f$ is convex on the convex nonempty set $B \subset A$ on which it is finite. If further 
% with positive probability $b\rightarrow h(b,Y)$ is strictly convex and $\lim_{b\uparrow A}h(b,Y)= \infty$,
%satisfies Condition \ref{condStrict},
%then $f$ satisfies Condition \ref{condStrict}. 
%or alternatively for $n \in \N_2$, $Y$ as above, $f=\var$, and $h(b,\cdot)= \wh{\var}_n(b',b)(\cdot)$. 
\end{remark}
\begin{lemma}\label{lempmsq}
Under Condition \ref{condLConv}, if  $\PQ_1(p_{\msq})>0$ and for some $b\in A$, $\msq(b)<\infty$,
then $f=\msq$ satisfies Condition \ref{condStrict}.  
\end{lemma}
\begin{proof}
It follows from Remark \ref{remMsqEst}, Remark \ref{remStrategy} for $n=1$, and Lemma \ref{lemPropConv}.
%for $h(b,\cdot)= \wh{\msq}_1(b',b)(\cdot)=Z^2L'L(b)$ we receive the following lemma.  
\end{proof}

%However, in general the corresponding 
%minimum point may not be unique. 
% Let us assume that $\var(b)\to \infty$ 
% as $b$ tries to leave $A$. 

\section{\label{secMsqECM}Mean square and its estimators in the ECM setting}
Let us consider the ECM setting as in Section \ref{secECM} for $A$ open. As discussed there, 
for each $\omega\in\Omega_1$, $b\to L(b)(\omega)$ is convex, and under Condition 
\ref{condtX}, $b\to L(b)(\omega)$ is strictly convex. % which can be used in Lemma \ref{lemStrictMsq}
Thus, under Condition \ref{condtX}, for each $\omega \in \Omega_1$, $p_{\msq}(\omega)$ holds (see Definition (\ref{pmsqdef})) only if $Z(\omega)\neq0$ and 
$\lim_{b \uparrow A}L(b)(\omega)= \infty$. 
%to prove (strict) convexity of $\msq$. 
Note that for $X$ having a non-degenerate normal distribution under $\PQ_1$, %or for a
%one created from it using some matrix $D$ with linearly independent columns, it holds 
for each $\omega \in \Omega_1$, $\lim_{|b| \to \infty}L(b)(\omega)= \infty$,
so that $p_{\msq}$ holds only if $Z\neq 0$. 
For $X$ having the distribution of a product of $n$ exponentially tilted distributions from the gamma family under $\PQ_1$, we
have $\PQ_1(X \in \R_+^n)=1$, and for $\omega \in \Omega_1$ such that $X(\omega) \in \R_+^n$, 
we have $L(b)(\omega)\to \infty$ as $b \uparrow A$. 
Thus, for such an $\omega$, $p_{\msq}(\omega)$ holds only if $Z(\omega)\neq 0$, and the condition $\PQ_1(p_{\msq})>0$, appearing in Lemma \ref{lempmsq},
reduces to $\PQ_1(Z\neq0)>0$. 
For $X$ having a Poisson distribution under $\PQ_1$, we have $\lim_{|b| \to \infty}L(b)(\omega)= \infty$ 
when $X(\omega) \in \N_+$, 
but not when $X(\omega)=0$. Thus, in such a case $p_{\msq}$ holds when $Z\neq0$ and $X \in \N_+$,
but not when $X=0$, and we have $\PQ_1(p_{\msq})>0$ only if $\PQ_1(X\in \N_+,\ Z\neq0)>0$.

\begin{remark}\label{remMsqECM}
Let us assume Condition \ref{condPartvm}. Then, for each $n\in\N_+$ and $b'\in A$, %and $\omega\in \Omega_1^n$, 
%$b\to \wh{\msq}_n(b',b)(\omega)$ is smooth and 
\begin{equation}
\nabla_b \wh{\msq}_n(b',b)= \overline{(Z^2L'(\mu(b)-X)L(b))}_n 
\end{equation}
and 
\begin{equation}
\nabla^2_b \wh{\msq}_n(b',b)= \overline{(Z^2L'(\Sigma(b)+(\mu(b)-X)(\mu(b)-X)^T)L(b))}_n.  
\end{equation}
Let us further in this remark assume Condition \ref{condtX}, so that $\Sigma(b)$ is positive definite. Then, for
$\omega\in \Omega_1^n$ such that $Z(\omega_i)\neq 0$ for some $i\in\{1,\ldots,n\}$,
$\nabla^2_b \wh{\msq}_n(b',b)(\omega)$ is positive definite for each $b',b\in A$. Indeed, in such a case for each $v \in \R^l\setminus\{0\}$ we have
\begin{equation}
\begin{split}
v^T\nabla^2_b \wh{\msq}_n(b',b)v&= v^T\Sigma(b)v\overline{(Z^2L'L(b))}_n +\overline{(Z^2L'L(b)((\mu(b)-X)^Tv)^2)}_n\\
&\geq v^T\Sigma(b)v\overline{(Z^2L'L(b))}_n>0.\\
\end{split}
\end{equation}
\end{remark}

%\end{remark}
Note that Condition \ref{condconVVar} holds for ECM since for each $\omega_1,\omega_2 \in \Omega_1$ and $b \in A$ we have
\begin{equation}
\frac{L(b)(\omega_1)}{L(b)(\omega_2)}=\exp(b^T(X(\omega_2)-X(\omega_1))).
\end{equation}
In particular, as discussed in the previous section, the estimators of variance and the new estimators of mean square are convex.  
For each $n \in \N_+$, $\omega \in \Omega_1^n$, and $i,j \in \{1,\ldots,n\}$, let us denote 
\begin{equation}
v_{j,i}(\omega)=X(\omega_{j})- X(\omega_i). 
\end{equation}
For each $n \in \N_2$, let a function $g_{\var,n}:A\times \R^l\times \Omega_1^n \to \R$ be such that for each
$b' \in A$, $b \in \R^l$, and $\omega \in \Omega_1^n$
\begin{equation}\label{gvardef}
g_{\var,n}(b',b)(\omega)=\sum_{i=1}^n\left((Z^2L')(\omega_i)\sum_{j\in \{1,\ldots,n\}, j \neq i}L'(\omega_j)\exp(b^Tv_{j,i}(\omega))\right).
\end{equation} 
Note that for each $b'$ and $\omega$ as above, $b\in \R^l\to g_{\var,n}(b',b)(\omega)$ is convex and 
\begin{equation}\label{gvarfvar}
g_{\var,n}(b',b)(\omega)=f_{\var,n}(b',b)(\omega),\quad b \in A. 
\end{equation}
For $A=\R^l$, we have $g_{\var,n}=f_{\var,n}$,
but in some cases, like for the gamma family of distributions as in  Section \ref{secECM},
we have $A \neq \R^l$ and $f_{\var,n}$ is only a restriction of $g_{\var,n}$. For each $b'$, $b$, and $\omega$ as above, it holds
\begin{equation}
\nabla_b g_{\var,n}(b',b)(\omega)= \sum_{i=1}^n\left((Z^2L')(\omega_i)\sum_{j\in \{1,\ldots,n\}, j \neq i}L'(\omega_j)v_{j,i}(\omega)\exp(b^Tv_{j,i}(\omega))\right)
\end{equation} 
and 
\begin{equation}\label{nabla2g}
\nabla^2_b g_{\var,n}(b',b)(\omega)= \sum_{i=1}^n\left((Z^2L')(\omega_i)\sum_{j\in \{1,\ldots,n\}, j \neq i}L'(\omega_j)v_{j,i}(\omega)v_{j,i}(\omega)^T\exp(b^Tv_{j,i}(\omega))\right).
\end{equation} 

Let $n \in \N_2$ and $\omega \in \Omega_1^n$. Let $D(\omega) \in \R^{l \times n^2}$ be a
matrix whose $(i-1)n+j$th column, $i, j \in\{1,\ldots,n\}$, is equal to 
$\I(Z\neq 0)(\omega_i)v_{j,i}(\omega)$. % (which is equal to $0$ for $i=j$).  

\begin{lemma}\label{thgvarstrict}
If $D(\omega)$ has linearly independent rows, then for each $b \in \R^l$ and $b' \in A$, 
$\nabla^2_bg_{\var,n}(b',b)(\omega)$ is positive definite.
\end{lemma}
\begin{proof}
If $D(\omega)$ has linearly independent rows, then for each $t \in \R^l$, $t\neq 0$, there exist $i,j \in \{1,\ldots,n\}$, $i\neq j$,
 such that $t^T\I(Z\neq 0)(\omega_i)v_{j,i}(\omega)\neq 0$, so that from (\ref{nabla2g}), 
\begin{equation}
t^T\nabla^2_bg_{\var,n}(b',b)(\omega)t \geq (Z^2L')(\omega_i)L'(\omega_j)(t^Tv_{j,i}(\omega))^2\exp(b^Tv_{j,i}(\omega))> 0.
\end{equation}
\end{proof}

Let for each $k \in \{1,\ldots,n\}$ a matrix $\wt{D}(k,\omega)\in \R^{l \times (n-1)}$ have the consecutive columns equal to
$v_{k,j}(\omega)$ for $j=1,2,\ldots,k-1,k+1,k+2,\ldots,n$. 

\begin{lemma}\label{lemDindep}
$D(\omega)$ has linearly independent rows only if for some $i \in \{1,\ldots,n\}$, 
$Z(\omega_i)\neq 0$, and 
for some (equivalently, for each) $k \in \{1,\ldots,n\}$, $\wt{D}(k,\omega)$  has linearly independent rows. 
\end{lemma}
\begin{proof}
If $Z(\omega_i)=0$, $i=1,\ldots,n$, then $D(\omega)$ has zero rows, 
so that they are linearly dependent. 
The dimensions of the linear spans of the columns and vectors 
of a matrix are the same, so that the matrices 
$D(\omega)$ and $\wt{D}(k,\omega)$, $k \in \{1,\ldots,n\}$,
have linearly independent rows only if the dimension of 
the linear span of their columns is equal to $l$. 
Thus, the thesis follows from the easy to check fact that 
the linear span $V$ of the columns of 
$\wt{D}(k,\omega)$ for different $k\in \{1,\ldots,n\}$ is the same and if 
$Z(\omega_i)\neq 0$ for some $i \in \{1,\ldots,n\}$, 
then the linear span of the columns of $D(\omega)$ is equal to $V$.
\end{proof}

%TODO postive definite.
For a vector $v \in \R^m$, by $v\leq 0$ we mean 
that its coordinates are nonpositive. 
\begin{theorem}\label{thgvarInf}
If the system of linear inequalities
\begin{equation}\label{sysIneq} 
D^T(\omega)b\leq 0,\quad b \in \R^l, 
\end{equation}
has only the zero solution, then for each $b' \in A$, 
$g_{\var,n}(b',b)(\omega) \to \infty$ as $|b|\to \infty$ 
and $\nabla^2_bg_{\var,n}(b',b)(\omega)$ 
is positive definite, $b \in \R^l$. 
%and $b\to g_{\var,n}(b',b)(\omega)$ has positive hessian. 
If $b$ is a solution of (\ref{sysIneq}), 
then for each $b'\in A$, $a \in \R^l$, and $t \in [0,\infty)$, we have
\begin{equation}\label{gnvarsm}
g_{\var,n}(b',a+tb)(\omega)\leq g_{\var,n}(b',a)(\omega). 
\end{equation}
\end{theorem}
\begin{proof}
For $b\in \R^l$ for which (\ref{sysIneq}) holds, for $i\in \{1,\ldots,n\}$ such that  $Z(\omega_i)\neq 0$,
for $j\in \{1,\ldots,n\}$, $i\neq j$, we have $b^Tv_{j,i}(\omega) \leq 0$, so that (\ref{gnvarsm})
follows from (\ref{gvardef}). 
Let further (\ref{sysIneq}) have only the zero solution. 
Then, $D(\omega)$ has linearly independent rows, and thus the
positive definiteness of $\nabla^2_bg_{\var,n}(b',b)(\omega)$ follows from Lemma \ref{thgvarstrict}. 
Consider a function $b \in \R^l\to f(b):=\max \{b^Tv_{j,i}(\omega):Z(\omega_i)\neq 0,\ i, j \in \{1,\ldots,n\},\ i\neq j\}$. 
Then, for each $b\in \R^l$, $b\neq0$, it holds $f(b)>0$. Thus, from the continuity of $f$
we have  $0<\delta :=  \min\{f(b):|b|=1\}$ and for
$0<a:=\min\{(Z^2L')(\omega_i)L'(\omega_j): Z(\omega_i)\neq 0,\ i, j \in \{1,\ldots,n\},\ i\neq j\}$, from (\ref{gvardef}) 
\begin{equation}
g_{\var,n}(b',b)(\omega) \geq a\exp(\delta |b|) \to \infty
\end{equation}
as $|b| \to \infty$.  
\end{proof}

There exist numerical methods for finding the set of solutions of (\ref{sysIneq})
and in particular for checking if it has only the zero solution; 
see \cite{Kolbig1979}. 
% Thanks to the below Theorem \ref{thzgd}, if
% \begin{equation}\label{zomegai}
% Z(\omega_i)\neq 0,\quad i =1\ldots,n,
% \end{equation}
% then checking if (\ref{sysIneq}) has only the zero solution is particularly simple. 

\begin{theorem}\label{thzgd}
Let us assume that
\begin{equation}\label{zomegai}
Z(\omega_i)\neq 0,\quad i =1\ldots,n.
\end{equation}
Then, (\ref{sysIneq}) has only the zero solution only 
if $D(\omega)$ has linearly independent rows, which from Lemma \ref{lemDindep}
holds only if for some (equivalently, for each) $k \in \{1,\ldots,n\}$, 
$\wt{D}(k,\omega)$  has linearly independent rows. 
\end{theorem}
\begin{proof}
Assuming (\ref{zomegai}), for $b \in \R^l$, $D^T(\omega)b\leq 0$
holds only if 
\begin{equation}
v_{i,j}(\omega)^Tb\leq 0,\quad i, j \in\{1,\ldots,n\}.  
\end{equation}
Since $v_{i,j}(\omega)=-v_{j,i}(\omega)$, this holds only if 
$v_{i,j}(\omega)^Tb= 0$, $i, j \in\{1,\ldots,n\}$, i.e.
only if $D^T(\omega)b=0$. 
\end{proof}

\section{\label{secStrong}Strongly convex functions and $\epsilon$-minimizers}
For some nonempty $A \subset \R^l$, 
consider a function $f:A \to \overline{\R}$. 
%$\epsilon>0$, and a function $\phi:S \to \R$, 
For some $\epsilon\geq 0$, we say that $x^* \in A$ is an $\epsilon$-minimizer of $f$, if
\begin{equation}
f(x^*) \leq \inf_{x\in A}f(x) + \epsilon. 
\end{equation}
Consider a convex set $S\subset A$, such that $A$ is a neighbourhood of $S$ (i.e. $S$ is contained in some open set $D \subset A$).
Then, $f$ is said to be strongly convex on $S$ (where we do not mention $S$ if it is equal to $A$) with a (strong convexity) constant $m>0$,
if $f$ is twice differentiable on $S$ and for each $b \in \R^l$ and $x \in S$ 
\begin{equation} 
b^T\nabla^2f(x)b \geq m|b|^2. 
\end{equation} 
Let us discuss some properties of strongly convex functions $f$ on $S$ as above (see Section 9.1.2. in  \cite{Boyd_2004} for more details). 
It is well known that  $f$ as above is strictly convex on $S$, and from Taylor's theorem 
it easily follows that for $x,y \in S$ 
\begin{equation}\label{fxfy}
f(y)\geq f(x)+(\nabla f(x))^T(y-x)+\frac{m}{2}|y-x|^2.
\end{equation}
In particular, $f(y) \to \infty$ as $|y| \to \infty$, $y \in S$. %, so that if $S=A=\R^l$, then Condition \ref{condStrict} holds. 
Furthermore, if $\nabla f(x)=0$, then 
\begin{equation}\label{strongfb}
f(y)\geq f(x)+\frac{m}{2}|y-x|^2.
\end{equation}
Thus, $x$ is a unique minimum point of $f_{|S}$ only if $\nabla f(x)=0$.
The right-hand side of (\ref{fxfy}) in the function of $y\in \R^l$ is minimized by $\wt{y}=x- \frac{1}{m}\nabla f(x)$, and thus we have 
\begin{equation}\label{fxfy2}
f(y)\geq f(x)+(\nabla f(x))^T(\wt{y}-x)+\frac{m}{2}|\wt{y}-x|^2= f(x)-\frac{1}{2m}|\nabla f(x)|^2.
\end{equation}
Let $f$ have a unique minimum point $b^*\in S$. Then, from (\ref{fxfy2}) for $y=b^*$, for $x \in S$ we have
\begin{equation}
f(x)\leq f(b^*)+\frac{1}{2m}|\nabla f(x)|^2. 
\end{equation}
In particular, each $x\in S$ is a $\frac{1}{2m}|\nabla f(x)|^2$-minimizer of $f$.
%Note also that from 
% Let $f$ have unique minimum point $b^*\in S$. 

\section{\label{secMsqLETGS}Mean square and its well-known estimators in the LETGS setting}
Let us in this section consider the LETGS setting.
\begin{theorem}\label{thMsqStr} 
Let $b',b\in \R^l$, $n\in \N_+$, and $\omega \in\Omega_1^n$. Then, 
$b\in \R^l\to f(b):=\wh{\msq}_n(b',b)(\omega)$ is convex and 
if $r_n(\omega)$ holds (see Definition \ref{defLemEquiv}), then 
$f$ is strongly convex. If $r_n(\omega)$ does not hold, then there exists $b \in \R^l \setminus \{0\}$ such that 
\begin{equation}\label{fatb}
f(a+tb)=f(a),\quad a \in \R^l,\ t \in \R.
\end{equation}
\end{theorem} 
\begin{proof} 
It holds
\begin{equation}
\nabla^2 f(b)=\overline{(Z^2L'(2G+ (2Gb +H)(2Gb +H)^T)L(b))}_n(\omega),
\end{equation}
which is positive semidefinite, so that $f$ is convex. 
If $r_n(\omega)$ does not hold, then from the first point of Lemma \ref{lemEquiv}
there exists $b \in \R^l$ such that for each $i\in 1,\ldots,n,$, $A_b(\omega_i)$ does not hold, so that 
from Lemma \ref{lemABL} and (\ref{whmsqdef}) we receive (\ref{fatb}). 
Let us assume that $r_n(\omega)$ holds. Then, from the second point of Lemma \ref{lemEquiv} for $K=Z^2L'$, 
the matrix $M:=\overline{(Z^2L'G)}_n(\omega)$ is positive definite. 
%From Theorem \ref{thPos} for $K=Z^2L'$ and above $s$ the matrix $M_n=\overline{(Z^2L'G)}_n(\wt{\kappa}_n)$ 
%is positive definite a.s. for a sufficiently large $n$. % and with probability tending to $1$ as $n\rightarrow \infty$. 
Let $m_1>0$ be such that $b^TMb \geq m_1 |b|^2$, $b \in \R^l$. 
For each $i \in \{1,\ldots,n\}$ such that $\tau(\omega_i)<\infty$, from Remark \ref{remInfLb} we have 
\begin{equation}
m_{2,i}:=\inf_{b \in \R^l}L(b)(\omega_i) =\exp(\inf_{b\in \R^l} \ln(L(b)(\omega_i))) \in \R_+.
\end{equation}
% \begin{equation}\label{infLb}
% %\begin{split}
% \inf_{b\in \R^l} \ln(L(b)(\omega_i) \geq \sum_{k=0}^{\tau(\omega_i)-1}\inf_{y\in \R^d} (\frac{1}{2}|y|^2 - \eta_{k+1}^T(\omega_i)y) 
% = -\frac{1}{2}\sum_{k=1}^{\tau(\omega_i)-1}|\eta_{k}(\omega_i)|^2  \in \R.\\
% %\end{split}
%\end{equation}
Let $m_2 = \min\{m_{2,i}:i\in\{1,\ldots,n\},\ \tau(\omega_i)<\infty\}$. Then, $m_2\in \R_+$ and for each $a, b \in \R^l$ we have 
\begin{equation}
% \begin{split}
 a^T\nabla^2 f(b)a \geq 2a^T\overline{(Z^2L'GL(b))}_n(\omega)a 
 \geq 2m_2a^T\overline{(Z^2L'G)}_n(\omega)a 
 \geq 2m_1m_2|a|^2. 
%\end{split}
\end{equation}
\end{proof}
%TODO For the variance estimator in the LETGS setting in general
%  $b\to\var(b',b)(\omega)$ need not be convex, it is easy to provide a counter-example for $\tau(\omega) =2$

\begin{theorem}\label{thWhMsq}
If conditions \ref{condKappa} and \ref{cond1} hold, then a.s. for a sufficiently large $n$, 
$b\rightarrow\wh{\msq}_n(b',b)(\wt{\kappa}_n)$ is strongly convex. In particular, 
the probability of this event converges to one as $n\to \infty$.
\end{theorem}
\begin{proof}
It follows directly from theorems \ref{thPos} and \ref{thMsqStr}.
\end{proof}

\begin{theorem}\label{thMsq}
Let Condition \ref{condzntau} hold. 
If $\msq(b_0)<\infty$ for some $b_0\in \R^l$, then $\msq$ is convex on the convex nonempty set $B$ on which it is finite and
if further Condition \ref{cond1} holds, then $f=\msq$ satisfies Condition \ref{condStrict} (in particular,
from Lemma \ref{lemMinPoint}, it has a unique minimum point). 
% then $\msq$ is strictly convex on $B$, $\msq(b)\rightarrow \infty$ as $|b|\rightarrow\infty$,
% and $\msq$ has a unique minimum point $b^*$. 
If Condition \ref{cond1} does not hold, then there exists $b \in \R^l$, $b\neq 0$, such that 
\begin{equation}\label{msqaptb}
\msq(a+tb)= \msq(a),\quad a \in \R^l,\ t \in \R. 
\end{equation}
\end{theorem}
\begin{proof}
The first part of the thesis follows from theorems \ref{thMsqStr} and 
\ref{thWhMsq}, the properties of strongly convex functions discussed
in Section \ref{secStrong}, Remark \ref{remStrategy}, and, under Condition \ref{condKappa},
from Lemma \ref{lemPropConv} for 
$h(b,\cdot)=\wh{\msq}_n(b',b)(\wt{\kappa}_n)$ for a sufficiently large $n$.
If Condition \ref{cond1} does not hold, then there exists $b\in \R^l$, $b\neq 0$, such that $\PQ_1(A_{b})=0$, for which 
(\ref{msqaptb}) follows from Lemma \ref{lemABL} and formula (\ref{msqDef}). 
\end{proof} 
% If a.s. $b\rightarrow h(b,Y)$ is 
% convex, then $f$ is convex on the convex nonempty set $B \subset A$ on which it is finite. If further 
% with positive probability $b\rightarrow h(b,Y)$ is strictly convex and 

\section{\label{secDiff}Smoothness of functions in  the LETS setting}
In this section we provide some sufficient conditions for the smoothness and for certain properties of the derivatives of functions defined 
in Section \ref{secCoeffDiv}. Unless stated otherwise, we consider the LETS setting, which 
contains the LETGS setting as a special case (see Section \ref{secLETS}). 
From Remark \ref{remConstMat}, the ECM setting for $A=\R^l$ can be identified 
with the LETS setting for $\tau=1$ and $\Lambda_0=I_l$, so that it is easy 
to modify the below theory to deal also with such an ECM setting. 
\begin{condition}\label{condlemYmore}
A measurable function $S:\mc{S}_1\to \mc{S}(\overline{\R})$ is such that conditions \ref{condBound1} and \ref{condBound2} hold and for each $\theta \in (\R^d)^{s}$ 
\begin{equation}\label{zbeta} 
\E_{\PU}(|S|\exp(\sum_{i=1}^s\theta_i^T\wt{X}(\eta_i)))<\infty. 
\end{equation}
\end{condition}
Note that Condition \ref{condlemYmore} implies that $S$ is $\PU$-integrable. 
\begin{remark}\label{remCounter}
In the special case which can be identified with the ECM setting 
for $A=\R^l$ as discussed above,  Condition \ref{condBound1} holds for $R=1$ 
and Condition \ref{condBound2} holds for $s=1$. Thus, 
for some $S:\mc{S}_1\to \mc{S}(\overline{\R})$, a counterpart of Condition \ref{condlemYmore}
in the ECM setting for $A=\R^l$ reduces to demanding that
\begin{equation}\label{yen}
\E_{\PU}(|S|\exp(\theta^TX))<\infty,\quad \theta\in \R^l. 
\end{equation} 
\end{remark}

\begin{remark}\label{remHold}
Since for each $s \in \N_+$ and $\theta \in (\R^d)^s$, 
$\E_{\PU}(\exp(\sum_{i=1}^s\theta_i^T\wt{X}(\eta_i)))
=\exp(\sum_{i=1}^s\wt{\Psi}(\theta_i))<\infty$, 
from H\"{o}lder's inequality, (\ref{zbeta}) holds if we have $\E_{\PU}(|S|^q)<\infty$ for some $q\in (1,\infty)$. 
\end{remark}

\begin{condition}\label{condUN}
We have $t,s \in \N_+$ and $f:(\mc{S}(\R^l))^t\otimes \mc{C}\to \mc{S}(\R)$ 
is such that for each $M\in \R_+$, 
for some $N \in \N_+$, $\phi \in ((\R^d)^{s})^N$, 
and $u \in \R_+^N$, we have $\PU$ a.s. 
\begin{equation}\label{zbound}
\sup_{b\in (\R^l)^t:|b_i|\leq M,\ i=1,\ldots,t}|f(b)|
\leq \sum_{i=1}^Nu_i\exp(\sum_{j=1}^s\phi_{i,j}^T\wt{X}(\eta_j)) 
\end{equation}
(where $f(b)=f(b,\cdot)$).
\end{condition}

\begin{remark}\label{remCondUN}
Let Condition \ref{condUN} hold and $S$ satisfy Condition \ref{condlemYmore} (for the same $s$). Let $M\in \R_+$ and consider the 
corresponding $N$, $\phi$, and $u$ as in Condition \ref{condUN}. Then, for each $\theta \in (\R^d)^{s}$
\begin{equation}\label{partialBound}
\E_{\PU}(\sup_{b\in (\R^l)^t:|b_i|\leq M,\ i=1,\ldots,t}|Sf(b)\exp(\sum_{j=1}^s\theta_j^T\wt{X}(\eta_j))|) 
\leq \sum_{i=1}^Nu_i\E_{\PU}(S\exp(\sum_{j=1}^s(\phi_{i,j}+\theta_j)^T\wt{X}(\eta_j))) <\infty.
\end{equation}
In particular, $\E_{\PU}(\sup_{b\in (\R^l)^t:|b_i|\leq M,\ i=1,\ldots,t}|Sf(b)|)<\infty$. Furthermore, from the above $M\in \R_+$ being arbitrary, 
for each $b \in (\R^l)^t$, $Sf(b)$ satisfies Condition \ref{condlemYmore}. 
\end{remark}

In this work we assume that $x^0=1,\ x \in \R$. 
\begin{theorem}\label{thzb} 
Let conditions \ref{condBound1} and \ref{condBound2} hold, let $t \in \N_+$, $r \in \R^t$, 
$w\in \N^{t\times l}$, $u \in \N^{t\times s\times d}$, $z \in \N^{t\times s\times d\times l}$, 
$y \in \N_+^{t\times s}$, $v \in \prod_{m=1}^t\prod_{i=1}^s(\N^l)^{y_{m,i}}$, and 
$q \in \prod_{m=1}^t\prod_{i=1}^s\N^{y_{m,i}}$. Let for each $b\in (\R^l)^t$ 
\begin{equation}\label{zbdef}
\begin{split}
f(b)&=\I(S\neq0)|\prod_{m=1}^t( 
L(b_m)^{r_{m}}\prod_{i=1}^lb_{m,i}^{w_{m,i}}\prod_{i=1}^{\tau\wedge s}(\prod_{j=1}^d(\wt{X}_j(\eta_i)^{u_{m,i,j}}\\
&\cdot\prod_{k=1}^l (\Lambda_{i-1})_{j,k}^{z_{m,i,j,k}})
\prod_{j=1}^{y_{m,i}}(\partial_{v_{m,i,j}}\wt{\Psi}(\lambda_{i-1}(b_m)))^{q_{m,i,j}}))|. 
\end{split}
\end{equation}
Then, Condition \ref{condUN} holds for such an $f$ (for the same $t$ and $s$ as above). 
\end{theorem}
\begin{proof}
Let $M \in [0,\infty)$ and $g(x)=e^x +e^{-x}$. For $p \in \R$ and $b \in \R^l$, $|b|\leq M$, 
from (\ref{lnlbLETS}) we have  $\PU$ a.s.
\begin{equation}\label{lbound}
\I(S\neq 0)L^p(b)\leq K(p):=\exp(|p|s\wt{F}(RM))\prod_{i=1}^s\prod_{j=1}^dg(pRM\wt{X}_j(\eta_i)). 
\end{equation}
Let for $x \in [0,\infty)$ and $a\in \N^l$
\begin{equation}
U_a(x)=1+\sup_{b\in\R^l,\ |b|\leq x}|\partial_a\wt{\Psi}(b)|,
\end{equation}
which is finite thanks to Remark \ref{remPartv}. 
Then, for each $b \in (\R^l)^t$, $|b_i|\leq M$, $i=1, \ldots,m$, we have $\PU$ a.s. 
\begin{equation}\label{zbound2}
\begin{split}
|f(b)|&\leq \I(S\neq0)\prod_{m=1}^t(K(r_{m})M^{\sum_{i=1}^lw_{m,i}} \prod_{i=1}^{s}(\prod_{j=1}^d(g(\wt{X}_j(\eta_i))^{u_{m,i,j}}
(1+R)^{\sum_{k=1}^lz_{m,i,j,k}})\\
&\cdot \prod_{j=1}^{y_{m,i}}U_{v_{m,i,j}}(RM)^{q_{m,i,j}})). 
\end{split}
\end{equation} 
The right-hand side of (\ref{zbound2}) can be rewritten to have the form as the right-hand side of (\ref{zbound}). 
\end{proof}

% \I(S\neq0)\prod_{m=1}^t(K(p_{1,m}p_{2,m})M^{\sum_{i=1}^lw_{m,i}} \prod_{i=1}^{s}\prod_{j=1}^d((2\cosh(\wt{X}_j(\eta_i)))^{u_{m,i,j}}(1+R)^{\sum_{k=1}^lv_{m,i,j,k}})).
% \end{equation} 
%Thus for $b \in (\R^l)^t$, $|b_i|\leq M$, $i=1, \ldots,m$, and $z$ as above, $\PU$ a.s. 
%\end{theorem}
% \begin{proof}
% \begin{equation}\label{lbound}
% I(S\neq 0)L(b)^p\leq K(p):=\exp(ps\wt{F}(RM))\prod_{i=1}^s\prod_{j=1}^d(2\cosh(pRM\wt{X}_j(\eta_i))),\quad b \in \R^l. 
% \end{equation}
% Thus for $b \in (\R^l)^t$, $|b_i|\leq M$, $i=1, \ldots,m$, and $z$ as above, $\PU$ a.s. 
% \begin{equation}\label{zbound}
% |z(b)|\leq \I(S\neq0)\prod_{m=1}^t(K(p_{1,m}p_{2,m})M^{\sum_{i=1}^lw_{m,i}} \prod_{i=1}^{s}\prod_{j=1}^d((2\cosh(\wt{X}_j(\eta_i)))^{u_{m,i,j}}(1+R)^{\sum_{k=1}^lv_{m,i,j,k}})).
% \end{equation} 
% \sup_{b\in (\R^l)^t:|b_i|\leq M,\ i=1,\ldots,t}|h_{p_1,p_2}(b_i)|\leq \sum_{i=1}^NU_i\exp(\sum_{j=1}^s\theta_{i,j}^T\wt{X}(\eta_i)).
% \end{proof}
%\begin{theorem}\label{thYmore}

\begin{theorem}\label{thYmore}
If conditions \ref{condBound1} and \ref{condBound2} hold, 
then for each $t \in \N_+$,  $p_1,p_2 \in \R^t$ such that $p_{2,i}\geq 1$, $i=1,\ldots,t$, 
$M\in \R_+$, $v\in (\N^l)^t$, and for 
$h_{p_1,p_2}(b,\omega):=(\I(S\neq 0)\prod_{i=1}^t|\partial_{v_i}(L^{p_{1,i}}(b_i))|^{p_{2,i}})(\omega)$, 
$b \in (\R^l)^t$, $\omega \in E$, Condition \ref{condUN} holds for $f=h_{p_1,p_2}$. 
\end{theorem}
\begin{proof}
Since for $p_3 \in (\N_+)^t$ such that $p_{3,i}\geq p_{2,i}$, $i=1,\ldots,t$, we have 
\begin{equation}
|h_{p_1,p_2}(b)|\leq \I(S\neq 0)\prod_{i=1}^t(1+|\partial_{v_i}(L^{p_{1,i}}(b_i))|)^{p_{3,i}}, 
\end{equation}
it is sufficient to prove the above theorem for $p_2 \in \N_+^t$. In such a case 
$h_{p_1,p_2}(b)$ is a linear combination of a finite number of variables as in (\ref{zbdef}). % for $r_i = p_{1,i}p_{2,i}$. 
Thus, the thesis follows from Theorem \ref{thzb}. 
\end{proof}

\begin{theorem}\label{thDiffLp}
If Condition \ref{condlemYmore} holds, then for each $p\in\R$, for
$g(b)=SL^{p}(b)$, $b\in \R^l\to f(b)=\E_{\PU}(g(b)) \in \R$
is smooth and we have 
%\begin{equation}\label{partialdiff}
$\partial_vf(b)=\E_{\PU}(\partial_vg(b))$, $v \in \N^l$, $b \in \R^l$.
%\end{equation}
\end{theorem}
\begin{proof}
It follows from Theorem \ref{thYmore} for $p_1=p$ and $p_2=1$ and from Remark \ref{remCondUN}
by induction over $\sum_{i=1}^lv_i$
using mean value and Lebesgue's dominated convergence theorems. 
\end{proof}

\begin{theorem}\label{thDiff} 
If Condition \ref{condlemYmore} holds 
\begin{enumerate}
\item for $S=1$, then $1=\E_{\PU}(L^{-1}(b))$ and for each $v \in \N^l$, $v\neq0$, 
$\E_{\PU}(\partial_v(L^{-1}(b)))=0$, $b \in \R^l$,
\item for $S=Z^2$, then $\msq$ is smooth
and $\partial_v\msq(b)=\E_{\PU}(Z^2\partial_vL(b))$, $b \in \R^l$, $v\in \N^l$,
\item for $S=C$, then $b\to c(b)$
is smooth and $\partial_v c(b)=\E_{\PU}(C \partial_v(L^{-1}(b)))=\E_{\PU}(\I(C\neq \infty)C \partial_v(L^{-1}(b)))$, $b \in \R^l$, $v\in \N^l$,
%is infinitely continuously differentiable with $\partial_v c(b)=\E_{\PU}(C \partial_v(L^{-1}(b)))$, $v\in \N^l$. 
\item for $S$ equal to $Z^2$ and $C$, then $\ic$ is smooth. 
\end{enumerate}
\end{theorem}
\begin{proof}
The first three points follow from Theorem \ref{thDiffLp} and from the fact that due to remarks \ref{condImplzntau} and
\ref{remztauequiv},
we have $1=\E_{\PU}(L(b)^{-1})$, $\msq(b)=\E_{\PU}(Z^2L(b))$, and $c(b)=\E_{\PU}(CL^{-1}(b))$ respectively, $b \in \R^l$,
and in the third point additionally (\ref{Cinfty0}). 
The last point is a consequence of points two and three.
\end{proof}

\begin{theorem}\label{thECMPos}
In the ECM setting for $A=\R^l$, let us assume that $\PQ_1(Z\neq0)>0$, conditions \ref{condtX} and \ref{condPartvm} hold, and we have (\ref{yen}) for
$S=Z^2$. Then, $\nabla^2\msq(b)$ exists and is positive definite, $b \in \R^l$. 
\end{theorem}
\begin{proof}
From a counterpart of Theorem \ref{thzb} and Remark \ref{remCondUN} for ECM, 
$W=Z^2L(b)(\mu(b) -X)(\mu(b) -X)^T$ has integrable entries.
Thus, from the second point of a counterpart of Theorem \ref{thDiff} for ECM
\begin{equation}
\begin{split}
\nabla^2\msq(b)&=\E_{\PQ_1}(Z^2L(b)(\nabla^2\Psi(b)+(\mu(b) -X)(\mu(b) -X)^T))\\
&=\nabla^2\Psi(b)\msq(b) + \E_{\PQ_1}(W).\\
\end{split}
\end{equation}
For $v \in \R^l$, $v^T\E_{\PQ_1}(W)v=\E_{\PQ_1}(Z^2L(b)((\nabla\Psi(b) -X)^Tv)^2)$, so that $\E_{\PQ_1}(W)$ is positive semidefinite.
Thus, the thesis follows from the fact that as discussed in Section \ref{secECM}, $\nabla^2\Psi(b)$ is positive definite and from
$\msq(b)\in \R_+$, $b \in \R^l$. 
\end{proof}

\begin{theorem}\label{thLETGSStrong}
In the LETGS setting, if Condition \ref{condlemYmore} holds for $S=Z^2$ 
and Condition \ref{cond1} holds, then for a positive definite matrix 
\begin{equation}\label{Mdef}
M=\E_{\PQ_1}(2GZ^2\exp(-\frac{1}{2}\sum_{i=1}^{\tau}|\eta_i|^2)) \in \R^{l\times l},  
\end{equation}
$\nabla^2\msq(b)-M$ is positive semidefinite, $b\in \R^l$. In particular, $\msq$ is strongly convex
with a constant $m$ equal to the lowest eigenvalue of $M$.
\end{theorem}
\begin{proof}
From the second point of Theorem \ref{thDiff}, we have
\begin{equation}
\nabla^2 \msq(b)=\E_{\PQ_1}(Z^2(2G+ (2Gb +H)(2Gb +H)^T)L(b)).
\end{equation}
From Theorem \ref{thzb} and Remark \ref{remCondUN}, $Z^2G$ and 
 $W:=Z^2(2Gb +H)(2Gb +H)^T)L(b)$ have $\PQ_1$-intergrable entries,
and from $\I(Z\neq 0)|\exp(-\frac{1}{2}\sum_{i=1}^{\tau}|\eta_i|^2)|\leq 1$, so does $Z^2G\exp(-\frac{1}{2}\sum_{i=1}^{\tau}|\eta_i|^2)$. 
Furthermore, $v^T\E_{\PQ_1}(W)v=\E_{\PQ_1}(Z^2L(b)((2Gb +H)^Tv)^2)$, $v \in \R^l$, i.e. $\E_{\PQ_1}(W)$ is positive semidefinite. 
From Lemma \ref{lemPosK} for $K=2Z^2\exp(-\frac{1}{2}\sum_{i=1}^{\tau}|\eta_i|^2)$, 
$M$ is positive definite. Furthermore, 
from Remark \ref{remInfLb}, for each $v \in \R^l$,  $v^T\E_{\PQ_1}(2GZ^2L(b))v\geq v^TMv$, and thus also 
$v^T(\E_{\PQ_1}(2GZ^2L(b))-M)v\geq 0$. 
\end{proof}

\section{Some properties of inefficiency constants} 
Let us consider the inefficiency constant function and its estimator as in sections \ref{secCoeffDiv} and \ref{secEstMin}. 
\begin{condition}\label{condcgegc0} 
It holds $\inf_{b \in A}c(b) = c_{min}\in\R_+$. 
\end{condition} 
\begin{condition}\label{condcgegCmin} 
For some $C_{min}\in\R_+$ we have $C(\omega) \geq  C_{min}$, $\omega \in \Omega_1$. 
\end{condition} 
Note that Condition \ref{condcgegCmin} implies Condition \ref{condcgegc0} for $c_{min}=C_{min}$.
\begin{remark}
Note that in the Euler scheme case as in Section \ref{secEulSDE}, for 
$\tau$ being the exit time of the scheme of a set $D$ such that $x_0 \in D$,
for $s \in \R_+$ and $C=s\tau$,  Condition \ref{condcgegCmin} holds for $C_{min}=s$. 
%in the Euler scheme case as in Section \ref{secEulSDE} for $s \in \R_+$, $C=s\tau$ and 
\end{remark}
Under Condition \ref{condcgegc0} we have $\ic\geq c_{min}\var$ and thus if further $A$ is open and 
$\lim_{b\uparrow A}\var(b)= \infty$, then $\lim_{b\uparrow A}\ic(b)=\infty$. 
Note also that if $c$ and $\var$ are lower semicontinuous (which from Lemma \ref{lemSemi} holds e.g. 
if $b\to L(b)(\omega)$ is continuous, $\omega \in \Omega_1$) then $\ic$  %(or if we extend it to $A$ by zero in case 
% when we multiply zero by infinity in its definition), 
is lower semicontinuous as well. Thus, if further $A$ is open and  $\lim_{b\uparrow A}\ic(b)=\infty$, 
then from Lemma \ref{lemMinPoint}, $\ic$ attains a minimum on $A$.
\begin{remark}\label{remicvar}
Let us assume that $\var$ has a unique minimum point $b^* \in A$.
If for some $b\in A$ it holds $\ic(b)<\ic(b^*)$, then $b\neq b^*$ and thus 
$\var(b^*)<\var(b)$, so that we must have $c(b)<c(b^*)$. 
Note that if $\var$, $c$, and $\ic$ are differentiable 
(some sufficient assumptions for which were discussed in Section \ref{secDiff}), then
a sufficient condition for the existence of $b\in A$ such that $\ic(b)<\ic(b^*)$ is that 
$\nabla \ic(b^*)\neq 0$. Since $\nabla\var(b^*)=0$, 
we have $\nabla \ic(b^*)=\var(b^*)\nabla c(b^*)$, so that
$\nabla \ic(b^*)\neq 0$ only if $\var(b^*)\neq 0$ and $\nabla c(b^*)\neq 0$.
\end{remark}
\begin{remark}\label{remicpos}
Let $c(b)>0 $, $b\in A$, let $\var$ have a unique minimum point $b^*$, 
and let $\var(b^*)=0$ and $c(b^*)<\infty$. 
Then, $b^*$ is also the unique minimum point of $\ic$ and we have  $\ic(b^*)=0$. 
If further $A$ is open,
$\msq$ and $c$ are twice continuously differentiable, 
and $\nabla^2\msq(b^*)$ is positive definite, then from
\begin{equation}
\nabla^2 \ic(b) =(\nabla^2c(b))\var(b) + c(b)\nabla^2\msq(b) + (\nabla c(b))(\nabla \msq(b))^T+(\nabla \msq(b))(\nabla c(b))^T,
\end{equation}
we have
\begin{equation}\label{nicb}
\nabla^2 \ic(b^*)=c(b^*)\nabla^2\msq(b^*),
\end{equation}
and thus $\nabla^2 \ic(b^*)$ is positive definite. 
\end{remark}

\chapter{\label{secMin}Minimization methods of estimators and their convergence properties} 
\section{\label{secSimpleMin}A simple adaptive IS procedure used in our numerical experiments} 
Let us describe a simple framework of adaptive IS 
via minimization of estimators of various functions from Section \ref{secEstMin}, 
shown in Scheme \ref{schm1}. A special case 
of this framework was used in the numerical experiments in this work. 
In the further sections we discuss 
some modifications of this framework which ensure suitable convergence 
and asymptotic properties of the minimization results of the estimators. 

Consider some estimators $\wh{\est}_k$, $k \in \N_p$, as in (\ref{dvcdef}). 
Let $b_0\in A$, $k\in \N_+$, $n_i \in \N_+ $, $i=1,\ldots,k$, and $N=n_{k+1}\in\N_+$. 
Let $b_i$, $i=1,\ldots,k$, be some $A$-valued random variables, defined in Scheme \ref{schm1}.
Let us assume Condition \ref{condxi} and
let for  $i=1,\ldots,k+1$ and $j=1,\ldots,n_i$,
$\beta_{i,j}$ be i.i.d. $\sim \PR_1$ and $\chi_{i,j}=\xi(\beta_{i,j},b_{i-1})$. 
Let us denote $\wt{\chi}_i=(\chi_{i,j})_{j=1}^{n_i}$, $i=1,\ldots,k+1$
\floatname{algorithm}{Scheme}
\begin{algorithm}
% enter the algorithm environme
\begin{algorithmic}                    % enter the algorithmic environment
\FOR{ $i:=1$ \TO $k$} 
\STATE %Set $b_i=f_i(b_{i-1}, \wt{\chi}_i)$. 
Minimize $b\rightarrow\wh{\est}_{n_i}(b_{i-1},b)(\wt{\chi}_i)$, e.g. using exact formulas or 
some numerical minimization method started at $b_{i-1}$. Let $b_i$ be the minimization result. 
\ENDFOR 
\STATE Approximate $\alpha$ with 
\begin{equation}\label{zln} 
\overline{(ZL(b_k))}_N(\wt{\chi}_{k+1}).
\end{equation}
\end{algorithmic}
\caption{A scheme of adaptive IS}          % give the algorithm a caption
\label{schm1}                           % and a label for \ref{} commands later in the document
\end{algorithm}
For $k=1$ we call the inside of the loop in Scheme \ref{schm1} single-stage minimization (SSM) and denote $b'=b_0$, while 
for $k>1$ we call this whole loop multi-stage minimization (MSM). 
%In some of our experiments we carried out only SSM. 

Let us now consider the LETGS setting and $\xi$ as in (\ref{xigss}).
%Section \ref{secGSS}. 
% We performed experiments using the LETGS setting and considering $\xi$ as in Section \ref{secGSS}, which let us now assume. 
Then, using the notation (\ref{ybdef}), (\ref{zln}) is equal to 
\begin{equation} 
\frac{1}{N}\sum_{i=1}^N(ZL(b_k))^{(b_k)}(\beta_{k+1,i}). 
\end{equation} 
From the discussion in Section \ref{secCeLETGS}, if 
$A_{n_i}(b_{i-1})(\wt{\chi}_i)=\frac{1}{n_i}\sum_{j=1}^{n_i}2(ZL(b_{i-1})G)^{(b_{i-1})}(\beta_{i,j})$ 
is positive definite, then for $B_n(b_{i-1})(\wt{\chi}_i)=-\frac{1}{n_i}\sum_{j=1}^{n_i}(ZL((b_{i-1}))H)^{(b_{i-1})}(\beta_{i,j})$, 
the unique minimum point $b_i$ of $b\to \wh{\ce}_{n_i}(b_{i-1},b)(\wt{\chi}_i)$ is given by the formula 
$b_i= (A_{n_i}(b_{i-1})(\wt{\chi}_i))^{-1}B_n(b_{i-1})(\wt{\chi}_i)$. 
Thus, in such a case, finding $b_i$ reduces to solving a linear system of equations. 
For $\est$ replaced by $\msq$, $\msq2$, or $\ic$, the %one can easily derive appropriate formulas 
functions $b\to \wh{\est}_{n_i}(b_{i-1},b)(\wt{\chi}_i)$ can be minimized using some numerical minimization methods which can 
utilise some formulas for their exact derivatives. 
Let us only provide formulas for such derivatives used in our numerical experiments. 
%For instance for $\est=\msq$ we have %Let us provide some such formulas. We have
% \begin{equation}
% \wh{\msq}_{n_i}(b_{i-1},b)(\wt{\chi}_i)=\frac{1}{n_i}\sum_{j=1}^{n_i}(Z^2L(b_{i-1})L(b))^{(b_{i-1})}, 
% \end{equation}
It holds
\begin{equation}
\nabla_b\wh{\msq}_{n_i}(b_{i-1},b)(\wt{\chi}_i)=\frac{1}{n_i}\sum_{j=1}^{n_i}(Z^2L(b_{i-1})(2Gb + H) L(b))^{(b_{i-1})}(\beta_{i,j}),
\end{equation}
\begin{equation}
\nabla^2_b\wh{\msq}_{n_i}(b_{i-1},b)(\wt{\chi}_i)=\frac{1}{n_i}\sum_{j=1}^{n_i}(Z^2L(b_{i-1})(2G + (2Gb+H)(2Gb+H)^T) L(b))^{(b_{i-1})}(\beta_{i,j}),
\end{equation}
%For $\est=\msq$ we have
% \begin{equation}
% \wh{\msq2}_{n_i}(b_{i-1},b)(\wt{\chi}_i)=\wh{\msq}_{n_i}(b_{i-1},b)(\wt{\chi}_i)\frac{1}{n_i}\sum_{j=1}^{n_i}(\frac{L(b_{i-1})}{L(b)})^{(b_{i-1})}(\beta_{i,j}), 
% \end{equation}
\begin{equation}
\begin{split}
\nabla_b\wh{\msq2}_{n_i}(b_{i-1},b)(\wt{\chi}_i)&=\nabla_b\wh{\msq}_{n_i}(b_{i-1},b)(\wt{\chi}_i)
 \frac{1}{n_i}\sum_{j=1}^{n_i}(\frac{L(b_{i-1})}{L(b)})^{(b_{i-1})}(\beta_{i,j}) \\
 & - \wh{\msq}_{n_i}(b_{i-1},b)(\wt{\chi}_i)\frac{1}{n_i}\sum_{j=1}^{n_i}((2Gb + H)\frac{L(b_{i-1})}{L(b)})^{(b_{i-1})}(\beta_{i,j}), 
\end{split}
 \end{equation}
% The formulas for higher derivatives 
% such derivatives for $\est$ replaced by $\msq2$ and $\ic$ are slightly more complicated but still very
% easy to derive, so we do not write them here.
%One can easily derive such formulas also for \est replaced . 
%It is easy to derive such formulas for other estimators.  
% For
%  \begin{equation}
%   \wh{\msq2}_{n_i}(b_{i-1},b)(\wt{\chi}_i)=\wh{\msq}_{n_i}(b_{i-1},b)(\wt{\chi}_i)\frac{1}{n_i}\sum_{j=1}^{n_i}(\frac{L(b_{i-1})}{L(b)})^{(b_{i-1})}(\beta_{i,j}) 
%  \end{equation}
%  \begin{equation}
%  \nabla_b\wh{\msq2}_{n_i}(b_{i-1},b)(\wt{\chi}_i)=\nabla_b\wh{\msq}_{n_i}(b_{i-1},b)(\wt{\chi}_i)
%  \frac{1}{n_i}\sum_{j=1}^{n_i}\frac{L(b_{i-1})}{L(b)}^{(b_{i-1})}(\beta_{i,j}) 
%  - \wh{\msq}_{n_i}(b_{i-1},b)(\wt{\chi}_i)\frac{1}{n_i}\sum_{j=1}^{n_i}(2Gb + H)\frac{L(b_{i-1})}{L(b)}{(b_{i-1})}(\beta_{i,j})
%  \end{equation}
  and for 
  \begin{equation}
  \wh{\var}_{n_i}(b_{i-1},b)(\wt{\chi}_i) = 
  \frac{n_i}{n_i-1}(\wh{\msq2}_{n_i}(b_{i-1},b)(\wt{\chi}_i)-(\frac{1}{n_i}\sum_{i=1}^{n_i}(ZL(b_{i-1}))^{(b_{i-1})}(\beta_{i,j}))^2), 
  \end{equation}
%   \begin{equation}
%   \wh{\ic}_{n_i}(b_{i-1},b)(\wt{\chi}_i)=\wh{\var}_{n_i}(b_{i-1},b)(\wt{\chi}_i)\frac{1}{n_i}
%   \sum_{j=1}^{n_i}(\frac{L(b_{i-1})}{L(b)}^{(b_{i-1})}C_i)(\beta_{i,j}),
%   \end{equation}
  %&=\left(\wh{\msq2}_n(b',b) - \overline{(ZL')}_n^2\right).
  %\left(\wh{\msq2}_n(b',b) - \overline{(ZL')}_n^2\right)
  \begin{equation}
\begin{split}
  \nabla_b \wh{\ic}_{n_i}(b_{i-1},b)(\wt{\chi}_i) &=\frac{1}{n_i-1}\sum_{j=1}^{n_i}(\frac{L(b_{i-1})}{L(b)}C)^{(b_{i-1})}(\beta_{i,j}) 
  \nabla_b\wh{\msq2}_{n_i}(b_{i-1},b)(\wt{\chi}_i)\\
  &-\frac{1}{n_i}\sum_{j=1}^{n_i}((2Gb + H)\frac{L(b_{i-1})}{L(b)}C)^{(b_{i-1})}(\beta_{i,j})
  \wh{\var}_{n_i}(b_{i-1},b)(\wt{\chi}_i).  
\end{split}
  \end{equation}
Formulas for the second derivatives of $\wh{\msq2}_{n_i}(b_{i-1},\cdot)$ and $\wh{\ic}_{n_i}(b_{i-1},\cdot)$ can also be easily computed
and used in minimization algorithms, but we did not apply them in our experiments. 
When evaluating the above expressions one can take advantage of formulas %from sections \ref{secGSS} and \ref{secLETS}, like  
(\ref{LBB}), (\ref{hbp}), and (\ref{lbbp2}). 

\section{\label{secSLLN}Helper strong laws of large numbers}
In this section we provide various SLLNs needed further on. 
The following uniform SLLN is well-known; see Theorem A1, Section 2.6 in \cite{rubinstein1993discrete}.
\begin{theorem}\label{thStUnif}
Let $Y$ be a random variable with values in a measurable space $\mc{S}$, let
$V \subset \R^l$ be nonempty and compact, and let $h:\mc{S}(V)\otimes \mc{S} \to \mc{S}(\R)$ be
such that a.s. $x\to h(x,Y)$ is continuous and $\E(\sup_{x \in V}|h(x,Y)|)<\infty$.
Then, for $Y_1,Y_2,\ldots$ i.i.d. $\sim Y$, a.s. as $n \to \infty$,
$x\in V\to \frac{1}{n}\sum_{i=1}^n h(x,Y_i)$  converges uniformly to a continuous 
function $x\in V\to \E(h(x,Y))$.
\end{theorem}
For each $p \in [1,\infty]$ and $\overline{\R}$-valued 
random variable $U$, let $|U|_p$ denote the norm of $U$ in $L^p(\PR)$. 
We have the following well-known generalization of H\"{o}lder's inequality which follows from it by induction. 
\begin{lemma}\label{lemHoldGen}
Let $n \in \N_+$, let $U_i$,  $i=1,\ldots,n$, be $\overline{\R}$-valued random variables, and let 
$q_i \in [1,\infty]$, $i=1,\ldots,n$, be such that $\sum_{i=1}^n \frac{1}{q_i}=1$. Then, it holds
\begin{equation}
|\prod_{i=1}^nU_i|_1 \leq \prod_{i=1}^n|U_i|_{q_i}.
\end{equation}
\end{lemma}
Let further in this section $n_i \in \N_+$, $i \in \N_+$, and $r \in \N_+$. 
To our knowledge, the SLLNs that follow are new. 

\begin{theorem}\label{thSLLN}
Let $M_i \geq 0$, $i \in \N_+$, be such that 
\begin{equation}\label{sumsuff}
\sum_{i=0}^\infty\frac{M_i}{n_i^r}<\infty.
\end{equation}
Consider $\sigma$-fields $\mc{G}_i\subset \mc{F}$, $i \in \N_+$, and $\overline{\R}$-valued random variables 
$\psi_{i,j}$, $j=1\ldots,n_i$, $i \in \N_+$, which are conditionally independent given $\mc{G}_i$ for the same $i$ and 
different $j$, and we have $\E(\psi_{i,j}^{2r})\leq M_i<\infty$ and
$\E(\psi_{i,j}|\mc{G}_i)=0$. Then, for
%\begin{equation}
$\wh{a}_i =\frac{1}{n_i}\sum_{j=1}^{n_i}\psi_{i,j}$, $i \in \N_+$,
%\end{equation}
we have a.s.
%\begin{equation}
$\lim_{n\rightarrow \infty}\wh{a}_n= 0$.
%\end{equation}
\end{theorem}
\begin{proof}
From the Borel-Cantelli lemma it is sufficient to prove that for each $\epsilon >0$ 
\begin{equation}\label{bcto}
\sum_{i=1}^{\infty}\PR(|\wh{a}_i|>\epsilon) < \infty.
\end{equation}
From Markov's inequality we have 
\begin{equation}
\PR(|\wh{a}_i|>\epsilon) \leq \frac{\E(\wh{a}_i^{2r})}{\epsilon^{2r}},
\end{equation}
so that it is sufficient to prove that 
\begin{equation}\label{sumeai2r}
\sum_{i=1}^{\infty}\E(\wh{a}_i^{2r}) < \infty. 
\end{equation}
Let us consider separately the easiest to prove case of $r=1$. % this is particularly easy to prove as
We have for $i\in\N_+$, and $j,l\in\{1,\ldots,n_i\}$, $j\neq l$, from the conditional independence
\begin{equation}\label{epsijlG}
\E(\psi_{i,j}\psi_{i,l}|\mc{G}_i)=\E(\psi_{i,j}|\mc{G}_i)\E(\psi_{i,l}|\mc{G}_i)=0, 
\end{equation}
and thus
\begin{equation}\label{epsijl}
\E(\psi_{i,j}\psi_{i,l})=\E(\E(\psi_{i,j}\psi_{i,l}|\mc{G}_i))=0.
\end{equation}
Thus, for $i\in \N_+$
\begin{equation}\label{ai2Mn}
%\begin{split}
\E(\wh{a}_i^2)=\frac{1}{n_i^2}(\sum_{j=1}^{n_i}\E(\psi_{i,j})^2 + \sum_{j<l \in \{1,\ldots,n\}}2\E(\psi_{i,j}\psi_{i,l}))
\leq \frac{M_i}{n_i}.
%\end{split}
\end{equation}
Now, (\ref{sumeai2r}) follows from (\ref{sumsuff}).

For general $r\in \N_+$, denoting $J_i = \{v \in \N^{n_i}: \sum_{j=1}^{n_i}v_j=2r\}$ and for $v \in J_i$, ${2r \choose v} =\frac{(2r)!}{\prod_{j=1}^{n_i} v_j!}$, 
we have for $v \in J_i$, from Lemma \ref{lemHoldGen} for $n=n_i$, $U_j=\psi_{i,j}^{v_j}$, %(which is assumed to be $1$ for $v_j=0$), 
and $\frac{1}{q_i}=\frac{v_i}{2r}$,
\begin{equation}\label{in1}
\E(|\prod_{j=1}^{n_i}\psi_{i,j}^{v_j}|)\leq \prod_{j=1}^{n_i}(\E(\psi_{i,j}^{2r}))^{\frac{v_j}{2r}} \leq M_i. 
\end{equation}
Thus,
\begin{equation} 
\E(\wh{a}_i^{2r})=\frac{1}{n_i^{2r}}\sum_{v \in J_i}{2r \choose v} \E(\prod_{j=1}^{n_i}\psi_{i,j}^{v_j})<\infty.
\end{equation}
For $v \in J_i$ such that $v_k =1$ for some $k \in \{1,\ldots,n_i\}$,
denoting $\psi_{i,\sim k} = \prod_{j\in\{1,\ldots, n_i\},j\neq k}\psi_{i,j}^{v_j}$, we have
that $\psi_{i,k}$ and $\psi_{i,\sim k}$ are conditionally independent. Furthermore, 
from Lemma \ref{lemHoldGen} for $n=n_i$, $U_j=\psi_{i,j}^{v_j}$ for $j\neq k$, $U_k=1$,
and $\frac{1}{q_i}=\frac{v_i}{2r}$, we have 
\begin{equation}
\E(|\psi_{i,\sim k}|) \leq \prod_{j\in\{1,\ldots,n_i\}, j\neq k}(\E(\psi_{i,j}^{2r}))^{\frac{v_j}{2r}}<\infty. %(\E(\psi_{i,k}^{2r}))^{\frac{v_i}{}}
\end{equation}
Thus,
\begin{equation}
\E(\prod_{j=1}^{n_i}\psi_{i,j}^{v_j}|\mc{G}_i)= \E((\E(\psi_{i,k}|\mc{G}_i)\E(\psi_{i,\sim k})|\mc{G}_i))=0,
\end{equation}
and $\E(\prod_{j=1}^{n_i}\psi_{i,j}^{v_j})=0$.
Therefore, for $\wt{J}_i= \{v \in (\N\setminus\{1\}) ^{n_i}: \sum_{j=1}^{n_i}v_j=2r\}$ we have
\begin{equation}\label{wtjsum}
\E(\wh{a}_i^{2r})=\frac{1}{n_i^{2r}}\sum_{v \in \wt{J}_i}{2r \choose v} \E(\prod_{j=1}^{n_i}\psi_{i,j}^{v_j}). 
\end{equation}
Note that for $v \in \wt{J}_i$ and $p(v):=|\{j\in\{1,\ldots,n_i\}: v_j\neq 0\}|$, it holds  $p(v)\leq r$, 
and thus for $J'_i= \{v \in \{0,2,3,\ldots,2r\}^{n_i}: p(v)\leq r \}$, we have $\wt{J}_i\subset J_i'$.  
Therefore,
\begin{equation}\label{in2}
|\wt{J}_i|\leq|J_i'|\leq {n_i \choose  r}(2r)^r\leq n_i^r(2r)^r.
\end{equation}
Furthermore, 
\begin{equation}\label{in3}
{2r \choose v}\leq (2r)!.
\end{equation}
From (\ref{wtjsum}), (\ref{in1}), (\ref{in2}), and (\ref{in3})
\begin{equation}\label{lastin}
\E(\wh{a}_i^{2r})\leq \frac{M_i}{n_i^{r}}(2r)!(2r)^r.
\end{equation}
Inequality (\ref{sumeai2r}) follows from (\ref{lastin}) and (\ref{sumsuff}).
\end{proof}

%is a nonempty set as in the previous sections. 
Let $l \in \N_+$, let $A \in \mc{B}(\R^l)$ be nonempty, and let a family of probability distributions $\PQ(b)$, $b \in A$, be as in Section \ref{secFamily}. 
%probability distributions on $\mc{S}_1$. Consider following conditions. 
Let $b_i$, $i \in \N$, be $A$-valued random variables.
\begin{condition}\label{condKi}
Nonempty sets $K_i \in \mc{B}(A)$, $i \in \N_+$,  are such that a.s. for a sufficiently large $i$, 
$b_i \in K_i$. 
\end{condition}

\begin{condition}\label{condChi}
For each $i \in \N_+$, $\chi_{i,j}$, $j =1,\ldots,n_i$, are conditionally independent given $b_{i-1}$
and have conditional distribution $\PQ(v)$ given $b_{i-1}=v$ (see page 420 in \cite{fristedt1997modern} or  
page 15 in \cite{Ikeda1981} for a definition of a conditional distribution).
It holds $\wt{\chi}_i=(\chi_{i,j})_{j=1}^{n_i}$, $i \in \N_+$. 
% Equivalently, $\wt{\chi}_i=(\chi_{i,j})_{j=1}^{n_i}$ has conditional distribution $\PQ(v)^{n_i}$ 
% given $b_{i-1}=v$. 
\end{condition}
Condition \ref{condChi} is implied by the following one. 
\begin{condition}\label{condBetaxi}
Condition \ref{condxi} holds and 
for each $i \in \N_+$, $\beta_{i,j}\sim\PR_1$, $j =1,\ldots,n_i$, are independent 
and independent of $b_{i-1}$. Furthermore, $\chi_{i,j}=\xi(\beta_{i,j},b_{i-1})$,
$j =1,\ldots,n_i$,
and $\wt{\chi}_i=(\chi_{i,j})_{j=1}^{n_i}$, $i \in \N_+$. 
\end{condition}

Let us further in this section assume conditions \ref{condKi} and \ref{condChi}.
\begin{condition}\label{condhLLN}
A function $h:\mc{S}(A)\otimes \mc{S}_1 \to \mc{S}(\R)$ is such that
for each $v\in A$, $\E_{\PQ(v)}(h(v,\cdot))=0$, and for 
\begin{equation}\label{mineq}
M_{i}=\sup_{w \in K_{i-1}}\E_{\PQ(w)}(h(w,\cdot)^{2r}),\quad  i \in \N_+, 
\end{equation} 
(\ref{sumsuff}) holds.
\end{condition}

\begin{theorem}\label{thSLLNh} 
Under Condition \ref{condhLLN}, for 
\begin{equation}
\wh{b}_i=\frac{1}{n_i}\sum_{k=1}^{n_i}h(b_{i-1},\chi_{i,k}),\quad i \in \N_+,
\end{equation}
% \begin{equation}\label{mini} 
% \sum_{k=1}^{\infty}\frac{M_i}{n_i^r} <\infty,
% \end{equation}
we have a.s.
\begin{equation}\label{bitozero}
\lim_{i\to\infty}\wh{b}_i =0.
\end{equation} 
\end{theorem} 
\begin{proof} 
Let for $i \in \N_+$, $h_i:A\times \Omega_1 \to \R$ be such that for each $x \in \Omega_1$,
$h_i(v,x)= h(v,x)$ when $v \in K_{i-1}$ and $h_i(v,x)=0$ when $v \in A \setminus K_{i-1}$. For 
\begin{equation}
\wh{a}_i=\frac{1}{n_i}\sum_{k=1}^{n_i}h_{i}(b_{i-1},\chi_{i,k}),
\end{equation}
from Condition \ref{condKi} we have a.s.
$\wh{b}_i- \wh{a}_i=\I(b_{i-1} \notin K_{i-1})\wh{b}_i\to 0$
as $i \to \infty$. Thus, to prove (\ref{bitozero}) it is sufficient to prove that a.s. 
\begin{equation}\label{liai0}
\lim_{i\to \infty}\wh{a}_i=0.  
\end{equation}
Let $\psi_{i,j}=h_{i}(b_{i-1},\chi_{i,j})$, $i\in \N_+$, $j=1,\ldots,n_i$. 
From the conditional Fubini's theorem (see Theorem 2, Section 22.1 in \cite{fristedt1997modern})
\begin{equation}
\begin{split}
\E(\psi_{i,j}^{2r})&=\E((\E_{\PQ(v)}(h_i^{2r}(v,\cdot)))_{v=b_{i-1}})\\
&=\E((\I(v\in K_{i-1})\E_{\PQ(v)}(h^{2r}(v,\cdot)))_{v=b_{i-1}})\leq M_i.
\end{split}
\end{equation}
Furthermore, $\psi_{i,j}$, $j=1,\ldots,n_i$ are conditionally independent given $\mc{G}_i:=\sigma(b_{i-1})$, and  
from some well-known properties of conditional distributions (see Definition 1, Section 23.1 in \cite{fristedt1997modern}), we have
\begin{equation}
\begin{split}
\E(\psi_{i,j}|\mc{G}_i)&=(\E_{\PQ(v)}(h_i(v,\cdot)))_{v=b_{i-1}}\\
&=(\I(v \in K_{i-1})\E_{\PQ(v)}(h(v,\cdot)))_{v=b_{i-1}}=0.   
\end{split}
\end{equation}
Thus, (\ref{liai0}) follows from Theorem \ref{thSLLN}.
\end{proof}

\begin{theorem}\label{thSLLNS}
If $g:\mc{S}(A)\otimes \mc{S}_1 \to \mc{S}(\R)$ is such that 
$f(v):=\E_{\PQ(v)}(g(v,\cdot)) \in \R$, $v \in A$, and for
\begin{equation}
P_{i}=\sup_{v \in K_{i-1}}\E_{\PQ(v)}(g(v,\cdot)^{2r}),\quad i\in \N_+,
\end{equation} 
we have
\begin{equation}\label{pini} 
\sum_{i=1}^{\infty}\frac{P_i}{n_i^r} <\infty, 
\end{equation}
then Condition \ref{condhLLN} holds for $h(v,y)=g(v,y)-f(v)$, $v \in A$, $y\in \Omega_1$. 
\end{theorem}
\begin{proof}
Clearly, $\E_{\PQ(v)}(h(v,\cdot))=0$, $v \in A$. Furthermore, for $v \in A$
\begin{equation}
\begin{split}
\E_{\PQ(v)}(h(v,\cdot)^{2r})&\leq \E_{\PQ(v)}(|g(v,\cdot)| + |f(v)|)^{2r}) \\
&\leq 2^{2r-1}\E_{\PQ(v)}(g(v,\cdot)^{2r} + f(v)^{2r}) \\
&\leq 4^{r}\E_{\PQ(v)}(g(v,\cdot)^{2r}), \\
\end{split}
\end{equation}
where in the second inequality we used the fact that 
$\frac{a+b}{2}\leq (\frac{a^p+b^p}{2})^{\frac{1}{p}}$, $a,b \in[0,\infty)$, $p \in [1,\infty)$, and 
in the last inequality we used conditional Jensen's inequality. 
Thus, $M_i \leq 4^{r} P_i$ and (\ref{sumsuff}) follows from (\ref{pini}). 
\end{proof}

%Consider the following condition. 
\begin{condition}\label{condPi}
Condition \ref{condpqbllpq1} holds for $Z$ replaced by some $S\in L^1(\PQ_1)$ and for
\begin{equation}\label{dip}
P_i=\sup_{v\in K_{i-1}} \E_{\PQ(v)}((SL(v))^{2r})=\sup_{v\in K_{i-1}} \E_{\PQ_1}(S^{2r}L(v)^{2r-1}),\quad i\in \N_+,
\end{equation}
we have (\ref{pini}). 
% \begin{equation}\label{pini2}
% \sum_{i=1}^{\infty}\frac{P_i}{n_i^r} <\infty.
% \end{equation} 
\end{condition}

\begin{theorem}\label{thSLLNc}
Under Condition \ref{condPi}, a.s.  %$c = \E_{\PQ_1}(S)$ 
%$\wh{c}_k=\overline{(SL(b_{k-1}))}_{n_k}(\wt{\chi}_k)$, $k \in \N_+$, we have a.s.
\begin{equation}
\lim_{k\to \infty}\overline{(SL(b_{k-1}))}_{n_k}(\wt{\chi}_k)=\E_{\PQ_1}(S). 
\end{equation}
\end{theorem}
\begin{proof}
This follows from Theorem \ref{thSLLNS} for $g(v,y)=(SL(v))(y)$, $v \in A$, $y\in \Omega_1$, in which $f(v)= \E_{\PQ_1}(S)$, $v\in A$,
as well as from Theorem \ref{thSLLNh}.
\end{proof}

For each $\overline{\R}$-valued random variable $Y$ on $\mc{S}_1$ and $q \geq 1$, let $||Y||_q=\E_{\PQ_1}(|Y|^q)^{\frac{1}{q}}$. % denote the $L^q(\PQ_1)$ norm of $Y$.
\begin{lemma}\label{lemHold}
Let $p, q \in [1,\infty]$ be such that 
$\frac{1}{p}+\frac{1}{q}=1$, let $S \in L^{2rp}(\PQ_1)$, let Condition \ref{condpqbllpq1} hold for $Z=S$,
and let for 
\begin{equation}\label{ridef}
R_i=\sup_{v\in K_{i-1}}||\I(S\neq 0)L(v)^{2r-1}||_{q},\quad i \in \N_+,
\end{equation}
it hold
\begin{equation}\label{rini}
\sum_{i=1}^{\infty}\frac{R_i}{n_i^r} <\infty.
\end{equation}
Then, Condition \ref{condPi} holds. 
\end{lemma}
\begin{proof}
From H\"{o}lder's inequality $\E_{\PQ_1}(S^{2r}L(v)^{2r-1}) \leq ||S^{2r}||_p||\I(S\neq 0)L(v)^{2r-1}||_q$, so that 
for $P_i$ as in (\ref{dip}) we have $P_i \leq ||S^{2r}||_pR_i$. 
Thus, from  (\ref{rini}), (\ref{pini}) holds for such $P_i$. 
\end{proof} 
The following uniform SLLN can be thought of as a multi-stage version of Theorem \ref{thStUnif} and some reasonings in its below proof
are analogous as in the proof of the latter in Theorem A1, Section 2.6 in \cite{rubinstein1993discrete}. 

\begin{theorem}\label{thSLLNEstUnif}
Let $V \subset \R^l$ be a nonempty compact set and let $h:\mc{S}_1\otimes \mc{S}(V)\rightarrow \mc{S}(\overline{\R})$ 
be such that for $\PQ_1$ a.e. 
$\omega \in \Omega_1$, $b\to h(\omega,b)$ is continuous. 
Let $Y(\omega)= \sup_{b \in V}|h(\omega,b)|$, $\omega \in \Omega_1$, 
and let Condition \ref{condPi} hold for $S=Y$. 
Then,  a.s. as $k\rightarrow \infty$, $b\in V\to\wh{a}_k(b):=\frac{1}{n_k}\sum_{i=1}^{n_k}h(\chi_{k,i},b)L(b_{k-1})(\chi_{k,i})$ 
converges uniformly to a continuous function $b \in V\rightarrow a(b):=\E_{\PQ_1}(h(\cdot,b)) \in \R$.
\end{theorem}
\begin{proof}
Obviously,
\begin{equation}\label{hleqS}
|h(\omega,b)| \leq Y(\omega),\quad \omega \in \Omega_1,\ b \in V,
\end{equation}
and for each $b \in K_0$, for $P_1$ as in (\ref{dip}) for $S=Y$, 
\begin{equation}
\E_{\PQ_1}(Y) = \E_{\PQ(b)}(YL(b))\leq (\E_{\PQ(b)}((YL(b))^{2r}))^{\frac{1}{2r}}\leq P_1^{\frac{1}{2r}} <\infty. 
\end{equation}
Thus, for each $v \in V$ and $v_k \in V$, $k \in \N_+$, such that $\lim_{k\to\infty}v_k = v$, from Lebesgue's 
dominated convergence theorem and $\PQ_1$ a.s. continuity of $b\rightarrow h(\cdot,b)$,
\begin{equation}
\lim_{k\rightarrow \infty}a(v_k)=\E_{\PQ_1}(\lim_{k\rightarrow \infty}h(\cdot,v_k))=a(v) \in\R. 
\end{equation}
Thus, $a$ is finite and continuous on $V$.
Let $\epsilon >0$. From the uniform continuity of $a$ on $V$, let $\delta>0$ be such that 
\begin{equation}\label{axay}
|a(x)-a(y)|<\epsilon,\quad x,y \in V,\ |x-y|<\delta.
\end{equation}
For each $y \in V$ and $n \in \N_+$, let $B_{n,y}=\{x\in V: |x-y|\leq\frac{1}{n}\}$, and let for each $\omega \in \Omega_1$
\begin{equation}
r_{n,y}(\omega) =\sup\{|h(\omega,x) - h(\omega,y)|:x \in B_{n,y}\}. 
\end{equation}
For $\PQ_1$ a.e. $\omega$ for which $h(\omega,\cdot)$ is continuous, $\lim_{n\rightarrow \infty}r_{n,y}(\omega)=0$, $y \in V$. Furthermore, 
\begin{equation}\label{rk2s}
r_{n,y}(\omega) \leq 2Y(\omega), \quad \omega \in \Omega_1,\ n \in \N_+,\ y\in V,
\end{equation}
so that from Lebesgue's dominated convergence theorem
\begin{equation}\label{lkepq}
\lim_{n\rightarrow \infty}\E_{\PQ_1}(r_{n,y})=\E_{\PQ_1}(\lim_{n\rightarrow \infty}r_{n,y})=0, \quad y\in V.
\end{equation}
Thus, for each $y \in V$ there exists $n_y\in \N_+$, $n_y >\frac{1}{\delta}$, such that 
\begin{equation}\label{qrne6}
\E_{\PQ_1}(r_{n_y,y}) <\epsilon,
\end{equation}
for which let us denote $W_y=B_{n_y,y}$. For each $x, y \in V$
\begin{equation}
|\wh{a}_k(x)- \wh{a}_k(y)|\leq\frac{1}{n_k}\sum_{i=1}^{n_k}L(b_{k-1})(\chi_{k,i})|h(\chi_{k,i},x) - h(\chi_{k,i},y)|,
\end{equation}
so that for each $y \in V$
\begin{equation}\label{wwka}
\sup_{x \in W_y}|\wh{a}_k(x)- \wh{a}_k(y)|\leq \frac{1}{n_k}\sum_{i=1}^{n_k}L(b_{k-1})(\chi_{k,i})r_{n_y,y}(\chi_{k,i}).
\end{equation}
From (\ref{rk2s}), Condition \ref{condPi} holds for $S=r_{n_y,y}$, so that
from Theorem \ref{thSLLNc}, the right-hand side of (\ref{wwka}) converges a.s. to $\E_{\PQ_1}(r_{n_y,y})$ as $k\rightarrow \infty$. 
Thus, from (\ref{qrne6}), for each $y \in V$, a.s. for a sufficiently large $k$, 
\begin{equation}
\sup_{x \in W_y}|\wh{a}_k(x)- \wh{a}_k(y)|<\epsilon.
\end{equation}
The family $\{W_y, y \in V\}$ is a cover of $V$. 
From the compactness of $V$ there exists a finite set of points $y_1, \ldots, y_m \in V$ 
such that $\{W_{y_i}: i=1,\ldots,m\}$ is a cover $V$, and a.s. for a sufficiently large $k$ we have 
\begin{equation}\label {akwakv}
\sup_{x \in W_{y_i}}|\wh{a}_k(x)- \wh{a}_k(y_i)|<\epsilon,\quad i =1,\ldots,m.
\end{equation}
From (\ref{hleqS}), for each $x \in V$, Condition \ref{condPi} holds for $S=h(\cdot,x)$, so that
from  Theorem \ref{thSLLNc}, for each $x \in V$, 
a.s. $\lim_{k\rightarrow \infty}\wh{a}_k(x)=a(x)$. Thus, a.s. for a sufficiently large $k$
\begin{equation}\label{akviakvj}
|\wh{a}_k(y_i) -  a(y_i)|<\epsilon, \quad i=1,\ldots,m.
\end{equation}
Therefore, a.s. for a sufficiently large $k$ for which (\ref{akwakv}) and (\ref{akviakvj}) hold, 
for each $x \in V$, for some $i \in \{1,\ldots,m\}$ such that $|y_i -x|<\delta$,  
\begin{equation}
|\wh{a}_k(x)- a(x)|\leq |\wh{a}_k(x)- \wh{a}_k(y_i)| + |\wh{a}_k(y_i)- a(y_i)|+ |a(y_i)-a(x)| < 3\epsilon.
\end{equation}
\end{proof}

\section{\label{secUnifEst}Locally uniform convergence of estimators} 
In this section we apply the SLLNs from the previous section 
to provide sufficient conditions for the single- and multi-stage a.s. locally uniform 
convergence of various estimators from Section \ref{secEstMin} as well as their derivatives to the corresponding functions and their 
derivatives. Such a convergence will be needed 
when proving the convergence and asymptotic properties of the minimization results of these estimators in the further sections. 
% and final estimates of $\alpha$ 
%for proving convergence and asymptotic properties of minimization methods. 
%needed to 
%prove convergence and asymptotic properties of minimization results and final estimates of $\alpha$ in adaptive IS
%methods discussed in further sections. 
By $\rightrightarrows$ we denote uniform convergence.
For some $A \subset \R^l$, we say that functions $f_n:A\to \R$, $n \in \N_+$,
converge locally uniformly to some function $f:A\to \R$, which we denote as $f_n \overset{loc}{\rightrightarrows} f$, if for each compact set 
$K\subset A$, $f_{n|K}\rightrightarrows f_{|K}$, i.e. $f_n$ converges to $f$ uniformly on $K$.
%We will need the following lemma. 
\begin{lemma}\label{thCompact} 
Let $l,m \in \N_+$, let $D\subset \R^l$ be nonempty and compact, let functions $f:D\rightarrow \R^m$ and $s: \R^m \rightarrow \R$ be continuous,  
and for some $f_n:D\to\R^m$, $n \in\N_+$, let $f_n \rightrightarrows f$. Then, $s(f_n) \rightrightarrows s(f)$. 
If further $s_n:\R^m\to\R$, $n \in \N_+$, are such that  $s_n\overset{loc}{\rightrightarrows}s$, then $s_n(f_n)\rightrightarrows s(f)$. 
\end{lemma}
\begin{proof}
For $M= \sup_{x\in D}|f(x)|<\infty$ let $K=\overline{B}_l(0,M+1)$,
and let $\epsilon >0$. Since $s$ is uniformly continuous on $K$, let us choose 
$0<\delta<1 $ such that 
$|s(x)-s(y)|<\epsilon$ when $|x-y| <\delta$, $x, y \in K$. 
Let $N \in \N_+$ be such that 
for $n \geq N$, $|f_n(x)-f(x)|< \delta$, $x \in D$. Then, for $n\geq N$ we have  $|s(f_n(x)) - s(f(x))|<\epsilon$, $x \in D$. 
Let further $M \in \N_+$, $M \geq N$, be such that for $n \geq M$, $|s_n(y)-s(y)|<\epsilon$, $y \in K$. Then, for $n \geq M$
and $x \in D$
\begin{equation}
|s_n(f_n(x))-s(f(x))| \leq |s_n(f_n(x))- s(f_n(x))|+|s(f_n(x))  -s(f(x))|< 2\epsilon. 
\end{equation}
\end{proof}

Until dealing with the cross-entropy estimators at the end of this section, 
we shall consider 
%work in 
the LETS setting. Similarly as in Section \ref{secDiff}, 
this will allow us to cover the special case of the LETGS setting
and it is straightforward to modify the below theory 
to deal with the ECM setting for $A=\R^l$. 

\begin{theorem}\label{thunifDiff}
Assuming Condition \ref{condKappa}, if Condition \ref{condlemYmore} holds
\begin{enumerate}
\item for $S=1$, then  a.s. (as $n\to \infty$)
$b\to\overline{(L'\partial_v(L^{-1})(b))}_n(\wt{\kappa}_n)$ converges locally uniformly 
to $0$ for $v \in \N^l\setminus \{0\}$ and to $1$ for $v=0$, 
\item for $S=Z^2$, then a.s.
$b\to\partial_v\wh{\msq}_n(b',b)(\wt{\kappa}_n)= \overline{(Z^2L'\partial_vL(b))}_n(\wt{\kappa}_n)\overset{loc}{\rightrightarrows}\partial_v\msq$, $v \in \N^l$,
\item for $S=C$, then a.s. $b\to\partial_v \wh{c}_n(b',b)(\wt{\kappa}_n)\overset{loc}{\rightrightarrows}\partial_v c$, $v \in \N^l$,
\item both for $S$ equal to $Z^2$ and $1$, then
a.s. $b\to\partial_v \wh{\msq2}_{n}(b',b)(\wt{\kappa}_k)\overset{loc}{\rightrightarrows}\partial_v\msq$ and
$b\to\partial_v \wh{\var}_n(b',b)(\wt{\kappa}_n)\overset{loc}{\rightrightarrows}\partial_v\var$, $v \in \N^l$,
\item for $S=C$, $S=Z^2$, and $S=1$, then a.s. $b\to \partial_v \wh{\ic}_n(b',b)(\wt{\kappa}_n)\overset{loc}{\rightrightarrows}\partial_v \ic$, $v \in \N^l$. 
\end{enumerate}
\end{theorem}
\begin{proof}
The first three points follow from such points of Theorem \ref{thDiff}, 
Theorem \ref{thYmore} for $p_2=1$ and appropriate $p_1$, Remark \ref{remCondUN},
and from Theorem \ref{thStUnif} (note that from Condition \ref{condlemYmore} for $S=C$ we have such a condition for $S=\I(C\neq \infty)C$).
The fourth point follows from the first two points, the fact that a.s. 
$\overline{(ZL')}_n(\wt{\kappa}_n)\to \alpha$, the last line in  (\ref{varest}), (\ref{msq2ndef}), and Lemma \ref{thCompact}. 
The fifth point follows from points three, four, and Lemma \ref{thCompact}. 
\end{proof}

Let us further in this section assume the following condition. 
\begin{condition}\label{condLiKi}
$A=\R^l$, $r \in \N_+$, for each $i \in \N_+$, $n_i\in \N_+$, and for each
$i \in \N$, $L_i \in [0,\infty)$ and
$K_i= \{b \in \R^l: |b|\leq L_i\}$. 
\end{condition}
Consider the following conditions.
\begin{condition}\label{condLi}
For each $a_1,a_2\in \R_+$
\begin{equation}\label{bnini}
\sum_{i=1}^{\infty} \frac{\exp(a_1\wt{F}(a_2L_{i-1}))}{n_i^r}<\infty.
\end{equation}
\end{condition}
\begin{condition}\label{condLinf}
$\lim_{i\to \infty}L_i =\infty$.
\end{condition}
\begin{remark}\label{remniLi}
Let us discuss possible choices of $n_i$ and $L_i$ such that conditions \ref{condLi} and \ref{condLinf} 
hold for each $r \in \N_+$, in some special cases of the LETS setting. Let $A_1\in \N$, $A_2\in\N_+$,  $m\in \N_2$, $0<\delta<1$, and $B_1, B_2 \in \R_+$. 
Consider $\wt{F}(x)=\frac{x^2}{2}$, which corresponds to $\wt{X}$ having multivariate standard normal distribution under $\wt{\PU}_1$
(see sections \ref{secECM} and \ref{secLETS}).  
Then, one can take $n_i=A_1 +A_2m^i$ and  $L_i=(B_1+ B_2 (i+1)^{1-\delta})^{\frac{1}{2}}$, or alternatively  
% Indeed, we then have 
% b_i:=(\frac{\exp(a_1\wt{F}(a_2L_{i-1}))}{n_i^r})^{\frac{1}{i}}=
%Alternatively, for such constants one can take 
$n_i=A_1 + A_2 i!$ and $L_i=(B_1+ B_2 (i+1))^{\frac{1}{2}}$.
For some $a_1,a_2 \in \R_+$, denoting $b_i=\frac{\exp(a_1\wt{F}(a_2L_{i-1}))}{n_i^r}$, $i \in \N_+$,
in the first case we have $\lim _{i\to \infty}b_i^{\frac{1}{i}}=\frac{1}{m^r}<1$
and in the second case, using Stirling's formula, we have $\lim _{i\to \infty}b_i^{\frac{1}{i}}=0$. Thus, in both cases
(\ref{bnini}) follows from Cauchy's criterion. 
For $\wt{F}(x)=\mu_0(\exp(x)-1)$, which corresponds to the Poisson case with initial mean $\mu_0$, one can take e.g. %$n_i$ 
$L_i= B_1\ln(B_2+\ln(i+1))$ and some $n_i$ as for the normal case above. 
%$n_i=A_1 +A_2m^i$ or $n_i=A_1 +A_2(i!)$. 
%  See Section \ref{secAsympMSM} for discussion of influence of such choices of $n_i$ on asymptotic 
%  properties of minimization results of some estimators. 
\end{remark}
%Let $s\in \N_+$. 
\begin{lemma}\label{lemIneqL} 
If conditions \ref{condBound1} and \ref{condBound2}  hold, then for each $p \in [0,\infty)$ and $b \in \R^l$
\begin{equation}\label{inequqq1}
\E_{\PQ_1}(\I(S\neq0)L(b)^p)\leq\exp(s(p\wt{F}(R|b|)+\wt{F}(Rp|b|))).\\ %\exp(2T\wt{F}(Rq|b|)).\\
\end{equation}
In particular, if $p\geq 1$ then
\begin{equation}\label{inequqq2}
\E_{\PQ_1}(\I(S\neq0)L(b)^p)\leq\exp(2ps\wt{F}(Rp|b|))).\\ %\exp(2T\wt{F}(Rq|b|)).\\
\end{equation}
% \begin{equation}\label{inequqq2}
% ||I(S\neq0)L(b)||_q\leq \exp(s(\wt{F}(R|b|)+q^{-1}\wt{F}(Rq|b|)))\leq  \exp(2s\wt{F}(Rq|b|)). 
% \end{equation}
\end{lemma}
\begin{proof}
From (\ref{lnlbLETS}) we have $\PQ_1$ a.s. that if $S\neq 0$ (and thus from Condition \ref{condBound2}, $\tau\leq s$) then 
\begin{equation}
\begin{split}
L(b)^p&=\exp(p(U(b) +Hb))\\
&= \exp(pU(b)+U(-bp))\frac{1}{L(-bp)}\\
&\leq\exp(s(p\wt{F}(R|b|)+\wt{F}(Rp|b|)))\frac{1}{L(-bp)}.\\
\end{split}
\end{equation}
Now (\ref{inequqq1}) follows from 
$\E_{\PQ_1}(\I(S\neq0)\frac{1}{L(-bp)}) =\PQ(-bp)(S\neq0)\leq 1$. 
\end{proof}

%Consider the following condition. 
\begin{condition}\label{condSpr}
Conditions \ref{condBound1} and \ref{condBound2} hold and for some $p\in(1,\infty)$, $S \in L^{2rp}(\PQ_1)$.
\end{condition}

\begin{theorem}\label{thSLLNsp}
If conditions \ref{condLi} and \ref{condSpr} hold, then 
Condition \ref{condPi} holds (for the same $S$, $p$, $r$, $K_i$, and $n_i$ as in these conditions and Condition \ref{condLiKi}). 
\end{theorem} 
\begin{proof} 
From Lemma \ref{lemHold} it is sufficient to check that for $q$ as in that lemma 
corresponding to $p$ from Condition \ref{condSpr}, and
for $R_i$ as in (\ref{ridef}), we have (\ref{rini}). 
% It is sufficient to check that the assumptions of Lemma \ref{lemHold} are satisfied for $q$ as in that lemma 
% corresponding to $p$ from Condition \ref{condSpr}.  
From (\ref{inequqq2}) in Lemma \ref{lemIneqL}, it holds %for $R_i$ as in (\ref{ridef}), it holds
\begin{equation} 
\begin{split} 
R_i&=\sup_{b\in K_{i-1}}||\I(S\neq 0)L(b)^{2r-1}||_{q} \\ 
&\leq\sup_{b \in K_{i-1}}\exp(2s(2r-1)\wt{F}(Rq(2r-1)|b|))\\ 
&\leq \exp(2s(2r-1)\wt{F}(Rq(2r-1)L_{i-1})),\\ 
\end{split} 
\end{equation} 
so that (\ref{rini}) follows from Condition \ref{condLi} for $a_1= 2s(2r-1)$ and $a_2= Rq(2r-1)$. 
\end{proof} 

\begin{theorem}\label{thUnifMult} 
If conditions \ref{condKi}, \ref{condChi}, and \ref{condLi} hold
and Condition \ref{condSpr} holds for $S=U$ (that is for $S$ denoted as $U$), then for each $w \in\R$ and $v \in \N^l$, 
a.s. as $k\rightarrow \infty$, 
$b\in \R^l\to\wh{f}_k(b):=\overline{(UL(b_{k-1})\partial_v(L(b)^w))}_{n_k}(\wt{\chi}_k)$ 
converges locally uniformly to $b\in\R^l\to f(b):=\E_{\PQ_1}(U\partial_v(L(b)^w))\in\R$. 
\end{theorem} 
\begin{proof} 
Let $M\in \R_+$, $V = \{x \in\R^l: |x|\leq M\}$, 
$h(\omega,b)= U(\omega)\partial_v(L(b)^w)(\omega)$, $\omega \in \Omega_1$, $b \in V$, 
and 
%\begin{equation}
$W(\omega)=\sup_{b \in V} |h(\omega,b)|$, $\omega \in \Omega_1$. 
For some $1<p'<p$, from 
Remark \ref{remHold} for $S=U^{2rp'}$ and $q = \frac{p}{p'}$, 
Condition \ref{condlemYmore} holds for such an $S$. Thus, from
Theorem \ref{thYmore} for $p_1=w$ and $p_2=2rp'$ and from Remark \ref{remCondUN}, we have 
\begin{equation}
\E_{\PQ_1}(W^{2rp'})=\E_{\PQ_1}(\sup_{|b|\leq M}(U\partial_v(L(b)^w))^{2rp'})<\infty.  
\end{equation}
Furthermore, if $W\neq 0$ then also $U\neq 0$, so that Condition \ref{condSpr} holds 
for $S=W$ and $p=p'$. Thus, from theorems \ref{thSLLNEstUnif} and \ref{thSLLNsp} we receive that a.s. 
$b\to\wh{f}_k(b)$ converges to $f$ uniformly on $V$. 
\end{proof} 

\begin{theorem}\label{thunifDiffMult} 
Let conditions  \ref{condKi}, \ref{condChi}, and \ref{condLi} hold. If Condition \ref{condSpr} holds
\begin{enumerate} 
\item then a.s. $\overline{(L(b_{k-1})S)}_{n_k}(\wt{\chi}_k)$ converges to $\E_{\PQ_1}(S)$ (as $k\to \infty$), 
\item for $S=1$, then a.s. $b\to\overline{(L(b_{k-1})\partial_v(L(b)^{-1}))}_{n_k}(\wt{\chi}_k)$ converges locally uniformly 
to $0$ for $v \in \N^l\setminus \{0\}$ and to $1$ for $v=0$, 
\item for $S=Z^2$, then a.s. 
$b\to\partial_v\wh{\msq}_{n_k}(b_{k-1},b)(\wt{\chi}_k)\overset{loc}{\rightrightarrows}\partial_v\msq$, $v \in \N^l$,
\item for $S=C$, then $b\to\partial_v \wh{c}_{n_k}(b_{k-1},b)(\wt{\chi}_k)\overset{loc}{\rightrightarrows} \partial_v c$,
 $v \in \N^l$,
\item both for $S=1$ and $S=Z^2$, and if $n_k \in \N_2$, $k\in \N_+$,  then 
a.s. $b\to\partial_v \wh{\msq2}_{n_k}(b_{k-1},b)(\wt{\chi}_k)\overset{loc}{\rightrightarrows}\partial_v\msq$
and $b\to\partial_v \wh{\var}_{n_k}(b_{k-1},b)(\wt{\chi}_k)\overset{loc}{\rightrightarrows}\partial_v\var$, $v \in \N^l$,  
\item for $S=C$, $S=Z^2$, and $S=1$, and if $n_k \in \N_2$, $k\in \N_+$, then a.s. 
$b\to \partial_v \wh{\ic}_{n_k}(b_{k-1},b)(\wt{\chi}_k)\overset{loc}{\rightrightarrows}\partial_v \ic$,  $v \in \N^l$. 
\end{enumerate}
\end{theorem} 
\begin{proof} 
The first point follows directly from Theorem \ref{thUnifMult} for $v=0$ and $w=0$.
Points two to four follow from Theorem \ref{thUnifMult} and points one to three of Theorem \ref{thDiff} 
(note that from Condition \ref{condSpr} for $S=C$ we have such a condition for $S=\I(C\neq \infty)C$).
The fifth point follows from point one for $S=Z$ as well as points two, three, and Lemma \ref{thCompact}, similarly as in the proof of 
the fourth point of Theorem \ref{thunifDiff}. The sixth point follows from points four, five, and Lemma \ref{thCompact}. 
\end{proof}

Let us now discuss single- and multi-stage locally uniform convergence of the cross-entropy 
estimators, for which we shall consider the ECM and LETGS settings separately. 
\begin{theorem}\label{thECCEMunif} 
In the ECM setting, let us assume Condition \ref{condKappa} and that we have (\ref{zxifin}). Then, a.s. 
\begin{equation}\label{zlkappato} 
\overline{(ZL')}_n(\wt{\kappa}_n) \to \alpha 
\end{equation} 
and 
\begin{equation}\label{zlxkappato} 
\overline{(ZL'X)}_n(\wt{\kappa}_n)\to \E_{\PQ_1}(ZX). 
\end{equation} 
Assuming further Condition \ref{condPartvm}, we have a.s. 
\begin{equation}\label{cerrloc}
b\to \partial_v\wh{\ce}_n(b',b)(\wt{\kappa}_n) \overset{loc}{\rightrightarrows}\partial_v\ce,\quad v \in \N^l. 
\end{equation} 
\end{theorem} 
\begin{proof} 
Formulas (\ref{zlkappato}) and (\ref{zlxkappato}) follow from the SLLN. 
Under Condition \ref{condPartvm}, from (\ref{ceEstECM}) and (\ref{cebform}) we have for $v \in \N^l$
\begin{equation}\label{vcebform}
\partial_v\ce(b)-\partial_v\wh{\ce}_n(b',b)(\wt{\kappa}_n) = 
\partial_v\Psi(b)(\alpha-\overline{(ZL')}_n(\wt{\kappa}_n))-(\partial_vb^T)(\E_{\PQ_1}(ZX)-\overline{(ZL'X)}_n(\wt{\kappa}_n)).
\end{equation} 
Thus, for each compact $K\subset A$, from (\ref{zlkappato}) and (\ref{zlxkappato}),
\begin{equation}
\begin{split}
\sup_{b\in K}|\partial_v\ce(b)-\partial_v\wh{\ce}_n(b',b)(\wt{\kappa}_n)|
&\leq \sup_{b\in K}|\partial_v\Psi(b)|(\alpha-\overline{(ZL')}_n(\wt{\kappa}_n))\\
&+\sup_{b\in K}|\partial_vb^T|(\E_{\PQ_1}(ZX)-\overline{(ZL'X)}_n(\wt{\kappa}_n))\to 0. 
\end{split}
\end{equation} 
\end{proof}

\begin{theorem}\label{thECCEMunifMult}
In the ECM setting, let us assume that $A=\R^l$, conditions \ref{condKi}, \ref{condChi},
and \ref{condLi} hold, and for some $s>2$ we have $Z\in L^{rs}(\PQ_1)$. Then, a.s.
\begin{equation} \label{zlpchi}
\lim_{k\to \infty}\overline{(ZL(b_{k-1}))}_{n_k}(\wt{\chi}_k)= \alpha, 
\end{equation} 
\begin{equation}\label{zlpx} 
\lim_{k\to \infty}\overline{(ZXL(b_{k-1}))}_{n_k}(\wt{\chi}_k)= \E_{\PQ_1}(ZX), 
\end{equation} 
and assuming further Condition \ref{condPartvm}, a.s.
\begin{equation}\label{pvceunif}
b\to \partial_v\wh{\ce}_{n_k}(b_{k-1},b)(\wt{\chi}_k)\overset{loc}{\rightrightarrows}\partial_v\ce,\quad v \in \N^l. 
\end{equation}
\end{theorem}
\begin{proof}
From H\"older's inequality, for each $2<q<s$ we have $\E_{\PQ_1}(|ZX_i|^{rq}) <\infty$, $i =1,\ldots,l$. 
Thus, 
(\ref{zlpchi}) and (\ref{zlpx})
follow from the counterpart of the first point of Theorem  \ref{thunifDiffMult} for ECM for $S=Z$
and $S=ZX$ respectively and (\ref{pvceunif}) can be proved similarly
as (\ref{cerrloc}) in Theorem \ref{thECCEMunif}.
\end{proof}

\begin{theorem}\label{thLETGSCEMunif}
In the LETGS setting, let us assume conditions \ref{condKappa} and \ref{condInt1}. Then, a.s. 
\begin{equation}\label{zglkappato}
\overline{(ZGL')}_n(\wt{\kappa}_n) \to \E_{\PQ_1}(ZG),
\end{equation}
\begin{equation}\label{zhlkappato}
\overline{(ZHL')}_n(\wt{\kappa}_n) \to \E_{\PQ_1}(ZH),
\end{equation}
and
\begin{equation}\label{ceLETGSloc}
b\to \partial_v\wh{\ce}_n(b',b)(\wt{\kappa}_n)\overset{loc}{\rightrightarrows}\partial_v\ce, \quad v \in \N^l. 
\end{equation}
\end{theorem}
\begin{proof}
Formulas (\ref{zglkappato}) and (\ref{zhlkappato}) follow from the SLLN. 
From (\ref{whcebLETGS}) and (\ref{cebLETGS}),
$\partial_v\ce(b)$ and $\partial_v\wh{\ce}_n(b',b)(\wt{\kappa}_n)$ can be nonzero only for $v \in \N^l$ such that 
$\sum_{i=1}^lv_i\leq 2$. It is easy to check that for such a $v$, from (\ref{zglkappato}) and (\ref{zhlkappato}), a.s.
\begin{equation}
\partial_v\ce(b)-\partial_v\wh{\ce}_n(b',b)(\wt{\kappa}_n) = 
\partial_v(b^T(\E_{\PQ_1}(ZG) - \overline{(ZGL')}_n(\wt{\kappa}_n))b +(\E_{\PQ_1}(ZH)- \overline{(ZHL')}_n(\wt{\kappa}_n))b)\overset{loc}{\rightrightarrows} 0.
\end{equation} 
\end{proof}

\begin{theorem}\label{thLETGSCEMunifMult}
In the LETGS setting, let us assume conditions \ref{condBound1}, \ref{condBound2}, \ref{condKi}, \ref{condChi}, %\ref{condLiKi}, 
and \ref{condLi}, and that for some $p>2$ we have $Z\in L^{rp}(\PQ_1)$. 
Then, a.s. 
\begin{equation}\label{zgchi}
 \overline{(ZGL(b_{k-1}))}_{n_k}(\wt{\chi}_k) \to \E_{\PQ_1}(ZG),
\end{equation}
\begin{equation}\label{zhchi}
 \overline{(ZHL(b_{k-1}))}_{n_k}(\wt{\chi}_k) \to \E_{\PQ_1}(ZH),
\end{equation}
and %b^T\E_{\PU}(ZG)b+\E_{\PU}(ZH)b
\begin{equation}\label{ceLETGSunif}
b\to \partial_v\wh{\ce}_{n_k}(b_{k-1},b)(\wt{\chi}_k)\overset{loc}{\rightrightarrows}\partial_v\ce,\quad v \in \N^l. 
\end{equation}
\end{theorem}
\begin{proof}
From Lemma \ref{lempInt}, for each $2<u<p$ and $i \in \{1,\ldots,l\}$, we have 
$\E_{\PQ_1}(|ZH_i|^{ru}) <\infty$ and $\E_{\PQ_1}(|ZG_{i,j}|^{ru}) <\infty$.  
Thus, (\ref{zgchi}) and (\ref{zhchi}) follow 
from the first point of Theorem  \ref{thunifDiffMult} for $S=ZH_i$ and $S=ZG_{i,j}$ respectively, and 
(\ref{ceLETGSunif}) can be proved similarly as (\ref{ceLETGSloc}) in Theorem \ref{thLETGSCEMunif}. 
\end{proof}

\section{\label{secESSM}Exact minimization of estimators} 
% In this section we describe procedures of exact minimization 
% of random functions (EM) and their special cases called exact single- 
In this section we define exact single- and 
multi-stage minimization methods of estimators, abbreviated as ESSM and EMSM. 
We also discuss the possibility of their application
to the minimization of the cross-entropy estimators in the ECM and LETGS settings. 

Let $T \subset \R_+$ be unbounded and for some $l \in \N_+$, let $B\in \mc{B}(\R^l)$ be nonempty. 
The ESSM and EMSM methods can be viewed
as special cases of the following abstract method for exact minimization 
of random functions, which we call EM. 
In EM we assume the following condition. 
\begin{condition}\label{condftdef}
For each $t\in T$ we are given a function
$\wh{f}_t:\mc{S}(B) \otimes (\Omega,\mc{F})  \to \mc{S}(\R)$, a set $G_t\in \mc{F}$, 
and a $B$-valued random variable $d_t$. 
Random variable $\wh{f}_t(b,\cdot)$ is denoted shortly as $\wh{f}_t(b)$.
\end{condition}
Furthermore, it assumed that for each $t \in T$ and $\omega \in  G_t$, $d_t(\omega)$ is the unique minimum point of 
$b\to \wh{f}_t(b,\omega)$. 

Let us now define ESSM and EMSM. 
For some nonempty set $A\in \mc{B}(\R^l)$, $A\subset B$, and $p\in \N_+$, let us consider functions 
\begin{equation}\label{estnDef}
\wh{\est}_n:\mc{S}(A)\otimes \mc{S}(B) \otimes \mc{S}_1^n \to \mc{S}(\R), \quad n \in \N_p.
\end{equation}
For $B=A$ these can be some estimators as in (\ref{dvcdef}).  
%but we can also take e.g. $B=\R^l$ and $\wh{\est}_n= g_{\var,n}$ for the function $g_{\var,n}$
%defined in (\ref{gvardef}), in which case we may have $A\neq B$.
We shall further often need the following condition. 
\begin{condition}\label{condMesDn}
For each $n \in \N_p$, a set %\begin{equation}\label{dnf1o}
$D_n \in \mc{B}(A) \otimes \mc{F}_1^n$
%\end{equation}
is such that for each $(b',\omega) \in D_n$, the function $b\in B\to\wh{\est}_n(b',b)(\omega)$ has a unique minimum point, denoted as $b_n^*(b',\omega)$.
\end{condition}
% In some cases the following stronger condition will hold.
% %Consider the following condition which is stronger than Condition \ref{condMesDn}. 
% \begin{condition}\label{condMeswtDn}
% Condition \ref{condMesDn} holds for $D_n=A \times \wt{D}_n$ for some $\wt{D}_n \in \mc{F}_1^n$, $n \in \N_p$.
% % $n \in \N_p$, for some $\wt{D}_n \in \mc{F}_1^n$, 
% % for each $(b',\omega) \in D_n:=A \times \wt{D}_n$, 
% % $b\in B\rightarrow\wh{\est}_n(b',b)(\omega)$ has a unique minimum point $b_n^*(b',\omega)$. 
% \end{condition}

In ESSM and EMSM we assume the following condition. 
\begin{condition}\label{condmesbs} 
Condition \ref{condMesDn} holds and for each $n \in \N_p$, 
for $\mc{F}_n':=\{D_n\cap D: D \in\mc{B}(A)\otimes\mc{F}_1^n\}$, 
the function $(b',\omega)\to b_n^*(b',\omega)$ is measurable from $\mc{S}'_n=(D_n, \mc{F}_n')$ to $\mc{S}(B)$. 
\end{condition}%TODO 
%The following lemma provides a useful criterion for Condition \ref{condmesbs} 
%to hold. %, which will be sufficient for our needs. See Chapter 1, Section 5.3 for theory 
%which can be used to ensure Condition \ref{condmesbs} under different assumptions. 

In ESSM we also assume Condition \ref{condKappa} and the following condition. 
%consider a filtration 
% \begin{equation}\label{mcfkssmdef}
% \mc{F}_k=\sigma(\wt{\kappa}_k),\quad k\in \N_+. 
% \end{equation}
%We additionally assume the following condition. 
\begin{condition}\label{condT}
$\N\cup\{\infty\}$-valued random variables $N_t$, $t \in T$, are such that a.s. (\ref{ntoinfty}) and (\ref{nt0}) hold. 
%   assuming Condition   \ref{condKappa}, we can consider $N_t$ defined by (\ref{nt1}) or (\ref{nt2}) 
%   but for  $C_i=C(\kappa_i)$, $i \in \N_+$, 
%   while in MSM methods, for some $[0,\infty)$-valued random variables $M_i$ modelling costs of minimization algorithms 
%   in the $i$th stage (we can set $M_i=0$ if we do not want to consider them), $i \in \N_+$, 
%   under Condition \ref{condChi} we can take e.g. 
%   \begin{equation}\label{taucpinf} 
%   N_t=\inf\{k \in \N:\sum_{i=1}^{k}(M_i+\sum_{j=1}^{n_k}C(\chi_{i,j})) \geq t\} 
%   \end{equation} 
%   or 
%   \begin{equation}\label{taucpsup} 
%   N_t=\sup\{k \in \N:\sum_{i=1}^{k}(M_i+\sum_{j=1}^{n_i}C(\chi_{i,j})) \leq t\}. 
%   \end{equation}
% %Let us denote $d$ as in (\ref{ptau}) for $\tau$ replaced by some such  $\tau_t$ as $a_t$.
%Note that if we have a.s. (\ref{nt0}) and (\ref{ntoinfty})
\end{condition} 
\begin{remark}\label{remCondT}
In Condition \ref{condT} one can take e.g. $T=\N_+$ and $N_k=k$, $k \in \N_+$. Alternatively, 
one can take $T=\R_+$ 
and for some nonnegative random variable $U$ on $\mc{S}_1$, 
% variables $N_t$ in Condition \ref{condT} shall be interpreted as the number of simulations made corresponding to the computation
%budget $t$. For instance, 
$N_t$ can be given by formula (\ref{nt1}) or (\ref{nt2}) but for $C_i=U(\kappa_i)$, $i \in \N_+$ 
(i.e. for $S_n=\sum_{i=1}^nU(\kappa_i)$, $n \in \N_+$).
In such cases sufficient conditions for (\ref{ntoinfty}) and (\ref{nt0}) to hold  a.s. were discussed in Chapter \ref{secIneff}. 
% For instance, such $U$ can be some theoretical cost variable, fulfilling 
% $\dot{U}=p_UU$ for some $p_U\in \R_+$ and
% a practical cost variable $\dot{U}$ 
% for generating some replicates (e.g. of $Z$ and $L$) under $\PQ'$ and doing helper computations 
% needed for later estimator minimization. 
%Such $U$ and $\dot{U}$ 
% are defined analogously as such costs $C$ and $\cdot{C}$ of an MC step above and shall be called cost variables of a step of SSM. 
% In such case, $N_t$ as above are reasonable choices of numbers of steps of SSM to perform if we want to spend approximate 
% total (theoretical) budget $t$. Definition (\ref{nt1}) ensures that we do not exceed 
% the budget $t$. Under definition (\ref{nt2}) we let ourselves 
% finish the last computation started before the budget $t$ is exceeded and thus we do not waste 
% the computational effort already invested in it. 
For instance, such an $U$ can be some theoretical cost variable, fulfilling 
$\dot{U}=p_{\dot{U}}U$ for some $p_{\dot{U}}\in \R_+$ and 
an practical cost variable $\dot{U}$ 
for generating some replicates (e.g. of $Z$) under $\PQ'$ and doing some 
helper computations needed for the later estimator minimization. 
Such $U$ and $\dot{U}$ are defined analogously as such costs $C$ 
and $\dot{C}$ of an MC step in Chapter \ref{secIneff} 
and shall be called the cost variables of a step of SSM. 
In such a case, some $N_t$ as above can be interpreted as the number of steps of SSM 
corresponding to an approximate theoretical budget $t$. 
Often one can take $U=C$, as is the case in our numerical experiments. 
%then $p_{\dot{U}}$ will be typically higher than an analogous constant $p_{\dot{C}}$ for $C$ as in the Introduction. 
%and it may depend on the estimator being minimized. 
%For instance, in our numerical experiments in the LETGS setting for practical costs being
%computation times, in which both $U$ and $C$ can be taken equal to the stopping time $\tau$. 
% Such theoretical costs $U$ should be chosen in such a way, that for different minimization methods that we are comparing the 
% they should be proportional to 
% It is reasonable to take such $U$ greater or equal than the cost variable $C$ for generating the replicates under $\PQ'$
% for MC estimation, and we can take $C=U$ if we decide not to model the cost of additional computations needed for
% estimator minimization. 
%Such $U$ may be taken equal to the cost variable $C$ for performing a single simulation for MC
% or be larger than it if we want to model the cost of additional computations needed for
% estimator minimization. % as we shall discuss in Section \ref{secTwo}) 
%$U$ and $C$ to be proportional to some approximate practical costs with the same proportionality constant), 
\end{remark}
For each $t \in T$, in ESSM we %$d_t$ in SSM is defined 
define $d_t$ to be a $B$-valued random variable such that on the event 
%$N_t= k\in \N_p$ and
%equal to the unique minimum point $b^*_k(b',\wt{\kappa}_k)$ of the function $b\rightarrow\wh{\est}_k(b',b)(\wt{\kappa}_k)$ when
\begin{equation}\label{kappaDk}
G_t:=\{(N_t= k\in \N_p) \wedge ((b',\wt{\kappa}_{k})\in D_{k})\}, 
\end{equation}
we have 
\begin{equation}\label{ckdef}
d_t = b_{k}^*(b',\wt{\kappa}_{k}).
\end{equation}
%(equivalently  $\wt{\kappa}_k \in \wt{D}_k$ if Condition \ref{condMeswtDn} holds)
%otherwise $d_t= h_t$ for some $B$-valued random variable $h_t$. 
%$\mc{F}_k$-measurable random variable $h_k$. 
On $G_t'=\Omega\setminus G_t$ one can set e.g. $d_t=b'$, $t \in T$. 
%From Condition \ref{condmesbs}, such $d_t$ are random variables. 
%$\mc{F}_k$-measurable, $k \in \N_+$. 

In EMSM we assume that conditions \ref{condKi} and \ref{condChi} hold 
for $n_k\in \N_p$, $k \in \N_+$. 
% We define a filtration
%  \begin{equation}\label{mcfkmsmdef} 
%  \mc{G}_k=\sigma(b_i:i<k;\wt{\chi}_i: i \leq k),\quad k\in \N. 
%  \end{equation} 
Furthermore, for each $k \in \N_+$, $d_k$ is a $B$-valued random variable such that on the event 
\begin{equation}\label{chiDk} 
G_k:=\{(b_{k-1},\wt{\chi}_k)\in D_{n_k}\} 
\end{equation} 
we have
\begin{equation}\label{akdef}
d_k = b_{n_k}^*(b_{k-1},\wt{\chi}_k).  
\end{equation}
%and otherwise $d_k= h_k$ for some $B$-valued random variable $h_k$ 
On $G_k'$ one can set e.g. $d_k=b_0$ or $d_k=b_{k-1}$. 
%some $\mc{G}_k$-measurable random variable $g_k$ 
%From Condition \ref{condmesbs}, $a_k$ are random variables. %$\mc{G}_k$-measurable, $k \in \N_+$. 
\begin{remark}\label{remIdentEM}
ESSM and EMSM are special cases of EM for the respective $G_t$ and $d_t$ as above, 
in ESSM for $\wh{f}_t(b,\omega)=\I(N_t=k \in \N_p)\wh{\est}_k(b',b)(\wt{\kappa}_k(\omega))$,
while in EMSM for $T=\N_+$ 
and $\wh{f}_k(b,\omega)=\wh{\est}_{n_k}(b_{k-1},b)(\wt{\chi}_k(\omega))$. 
%$b \in B$, $b'\in A$, $\omega \in \Omega$, $t \in T$, $k \in \N_+$.
\end{remark}

%The fact that $a_k$ are $\mc{G}_k$-measurable follows from Condition \ref{condmesbs}. 
%given by Condition \ref{condbkdef} 
In EMSM the variables $b_k$, $k \in \N$, satisfying Condition \ref{condKi} can be defined in various ways.
An important possibility is when we are given some $K_0$-valued random variable $b_0$, 
and $b_k$, $k \in \N_+$, are as in the below condition. % and the comment following it. %TODO, in which case $b_k$ is $\mc{G}_k$-measurable. 
\begin{condition}\label{condbkdef}
For each $k \in \N_+$, if %(\ref{chiDk}) holds and 
$d_k\in K_k$, then $b_k =d_k$, and 
otherwise $b_k =r_k$ for some $K_k$-valued %$\mc{G}_k$-measurable $K_k$-valued 
random variable $r_k$. 
\end{condition}
%The fact that $b_k$ fulfilling the above condition are $\mc{G}_k$-measurable, $k \in \N_+$, follows from 
%Condition \ref{condmesbs}. 
Note that if $K_k\subset K_{k+1}$, $k\in\N$, then for each $k \in \N_+$, in the above condition we can take e.g.
$r_k=b_0$ or $r_k=b_{k-1}$.  %or  $r_k = \phi_k(d_k)$ for some $\phi_k:\mc{S}(B)\to \mc{S}(K_k)$. 

% ESSM and EMSM are special cases of EM for respective $G_t$ and $d_t$ as above, 
% in ESSM for $\wh{f}_t(b,\omega)=\I(N_t=k \in \N_p)\wh{\est}_k(b',b)(\wt{\kappa}_k(\omega))$,
% while in EMSM for $T=\N_+$ 
% and $\wh{f}_k(b,\omega)=\wh{\est}_{n_k}(b_{k-1},b)(\wt{\chi}_k(\omega))$. 

Consider some function $f: A\to \R$ and let $b^* \in A$ be its unique minimum point. 
We will be interested in verifying when some of the below conditions hold for EM methods, 
like ESSM and EMSM under the identifications as 
in Remark \ref{remIdentEM}, or for some other methods defined further on.  
%ESSM, EMSM, and other methods discussed later.  

\begin{condition}\label{condKappaDk}
Almost surely for a sufficiently large $t\in T$, $G_t$ holds. 
\end{condition}
\begin{condition}\label{condESSM1}
It holds a.s. $\lim_{t\to \infty }d_t = b^*$.   
\end{condition}
\begin{condition}\label{condESSM1f}
It holds a.s. $\lim_{t\to \infty }\wh{f}_{t}(d_t) = f(b^*)$. 
%(in particular a.s. $N_t \in \N_p$ for a sufficiently large $t$ so that the limit is well-defined a.s.).
\end{condition}

% \begin{condition}\label{condKappaDk}
% A. s. for a sufficiently large $t\in T$, (\ref{kappaDk}) holds.
% \end{condition}
% \begin{condition}\label{condESSM1}
% It holds a.s. $\lim_{t\to \infty }c_t = b^*$.   
% \end{condition}
% 
% \begin{condition}\label{condESSM1f}
% It holds a.s. $\lim_{t\to \infty }\wh{\est}_{N_t}(b',c_t)(\wt{\kappa}_{N_t}) = f(b^*)$
% (in particular a.s. $N_t \in \N_p$ for a sufficiently large $t$ so that the limit is well-defined a.s.).
% \end{condition}
% 
% \begin{condition}\label{condKappaDk}
% A. s. for a sufficiently large $k$, (\ref{chiDk}) holds. 
% \end{condition}
% 
% \begin{condition}\label{condESSM1}
% It holds a.s. $\lim_{k\to \infty }a_k = b^*$. 
% \end{condition}
% \begin{condition}\label{condESSM1f}
% I holds a.s. $\lim_{k\to \infty } \wh{\est}_{n_k}(b_{k-1},a_k)(\wt{\chi}_k) = f(b^*)$. 
% \end{condition}
%For different special cases of ESSM and its modifications discussed further on, we 
%will try to answer questions like when the below conditions hold. 
% the sequences $(c_k)$, $(a_k)$, and $(b_k)$ converge a.s. to the unique minimum points of 
% the corresponding divergence coefficients. 

Consider the following condition. 
\begin{condition}\label{condKiA} 
$A$ is open and $K_i \in \mc{B}(A)$, $i\in \N$, are such that 
for each compact set $D \subset A$, for a sufficiently large $i$, $D\subset K_i$. 
\end{condition}

Note that if conditions \ref{condLiKi} and \ref{condLinf} hold, 
then Condition \ref{condKiA} holds. 

\begin{remark}\label{remcondKiA}
For EMSM let us assume conditions \ref{condbkdef} and \ref{condKiA} (for the same sets $K_i$). 
Then, if for some compact set $D\subset A$ a.s. 
$d_k\in D$ for a sufficiently large $k$  (which happens e.g. if a.s. $d_k \to b^*$ and $D$ is some compact neighbourhood 
of $b^*$), then a.s. for a sufficiently large $k$, $d_k= b_k$.  
In particular, if additionally Condition \ref{condESSM1} or \ref{condESSM1f} holds for EMSM 
then such a condition holds also for $d_k$ replaced by $b_k$.
\end{remark}

% The following conditions will be useful for proving asymptotic properties of the IS parameters (see Condition \ref{condfndiff}). 
% Let us assume that $A$ is open and $f:A\to \R$ is twice continuously differentiable. 
% \begin{condition}\label{condESSM2} 
% As $n \to \infty$, $b\to\nabla_b\wh{\est}_n(b',b)(\wt{\kappa}_n) \to \nabla f$ and
% $b\to\nabla_b^2\wh{\est}_n(b',b)(\wt{\kappa}_n) \to \nabla^2 f$ locally uniformly.
% \end{condition}
% \begin{condition}\label{condEMSM2}
% As $k \to \infty$, $b\to\nabla_b\wh{\est}_{n_k}(b_{k-1},b)(\wt{\chi}_k) \to \nabla f$
% $b\to\nabla_b^2\wh{\est}_{n_k}(b_{k-1},b)(\wt{\chi}_k) \to \nabla^2f$ locally uniformly.
% \end{condition}

Let us now describe how ESSM and EMSM can be used 
for $\wh{\est}_n=\wh{\ce}_n$ in the ECM and LETGS settings. 
Let us first consider ECM as in sections \ref{secECM} and \ref{secCEECM}, 
assuming conditions \ref{condtX} and \ref{condPartvm}, as well as that we have 
(\ref{zxifin}), $\alpha>0$, and  (\ref{muepqzx}). Then, 
from the discussion in Section \ref{secCEECM}, %$f=\ce$ 
%satisfies Condition \ref{condfNice}, 
Condition \ref{condMesDn} holds for 
\begin{equation} 
D_n=\{(b',\omega)\in A\times \Omega_1^n:\overline{(ZL')}_n(\omega)>0 \wedge 
\frac{\overline{(ZL'X)}_n}{\overline{(ZL')}_n}(\omega) \in \mu[A]\}, 
\end{equation} 
and from formula (\ref{bsceECM}), Condition \ref{condmesbs} holds. In ESSM, from 
(\ref{zlkappato}) and (\ref{zlxkappato}) in Theorem \ref{thECCEMunif} 
as well as from $\alpha >0$, a.s. for a sufficiently large $n$ we have $\overline{(ZL')}_n(\wt{\kappa}_n)>0$ and a.s. 
\begin{equation}\label{zlxto} 
\frac{\overline{(ZL'X)}_n}{\overline{(ZL')}_n}(\wt{\kappa}_n)\to\frac{\E_{\PQ_1}(ZX)}{\alpha}. 
\end{equation} 
Thus, using further (\ref{muepqzx}), the fact that $\mu[A]$ is open, 
and Condition \ref{condT}, 
a.s. for a sufficiently large $t$, $G_t$ as in (\ref{kappaDk}) holds (i.e. Condition \ref{condKappaDk} holds for ESSM), 
in which case 
\begin{equation}\label{cnmuzlx} 
d_t=\mu^{-1}\left(\frac{\overline{(ZL'X)}_k}{\overline{(ZL')}_k}(\wt{\kappa}_k)\right). 
\end{equation} 
From Condition \ref{condT}, (\ref{ceminbs}), (\ref{zlxto}), (\ref{cnmuzlx}), and the continuity of $\mu^{-1}$, Condition \ref{condESSM1} holds. 
For EMSM let us additionally make the assumptions as in Theorem \ref{thECCEMunifMult}. 
Then, from (\ref{zlpchi}) and (\ref{zlpx}) in that theorem, 
by similar arguments as above for ESSM, conditions \ref{condKappaDk} and \ref{condESSM1} hold for EMSM. 

Consider now the LETGS setting and, using the notations as in Section \ref{secCeLETGS}, let us assume that 
Condition \ref{condInt1} holds and $\wt{A}$ is positive definite. 
From the discussion in that section, Condition \ref{condMesDn} holds for 
$D_n=\{(b',\omega)\in A\times \Omega_1^n: A_n(b')(\omega)\text{ is positive definite}\}$, 
which, for $Z\geq0$, fulfills $D_n=A \times \{\omega \in \Omega_1^n:r_n(\omega)\}$. 
From formula (\ref{bsdefLETGS}), Condition \ref{condmesbs} holds. 
In ESSM, from the SLLN a.s. $A_n(b')(\wt{\kappa}_n)\to \wt{A}$ and $B_n(b')(\wt{\kappa}_n) \to \wt{B}$. 
Thus, from Lemma \ref{lemPosDef} and Condition \ref{condT}, 
a.s. for a sufficiently large 
$t$, $N_t=n \in \N_+$ and $A_{n}(b')(\wt{\kappa}_{n})$ 
is positive definite (i.e.  Condition \ref{condKappaDk} holds), in which case 
$d_t=(A_n(b')(\wt{\kappa}_n))^{-1}B_n(b')(\wt{\kappa}_n)$. Thus, from (\ref{zlg}), 
Condition \ref{condESSM1} holds. 
% \begin{equation}\label{zgchi}
%  \overline{(ZG)}_{n_k}(\wt{\chi}_k) \to \E_{\PQ_1}(ZG),
% \end{equation}
% \begin{equation}\label{zhchi}
%  \overline{(ZH)}_{n_k}(\wt{\chi}_k) \to \E_{\PQ_1}(ZH),
% \end{equation}
For EMSM, let us make the additional assumptions as in Theorem \ref{thLETGSCEMunifMult}, 
so that from (\ref{zgchi}), a.s. $A_{n_k}(b_{k-1})(\wt{\chi}_k) \to \wt{A}$, and from (\ref{zhchi}), a.s. 
$B_{n_k}(b_{k-1})(\wt{\chi}_k) \to \wt{B}$. 
% In particular, from Lemma \ref{lemPosDef} a.s. $A_{n_k}(b_{k-1})(\wt{\chi}_k)$ is positive definite for a sufficiently large 
% $k$, in which case $a_k=(A_{n_k}(b_{k-1})(\wt{\chi}_{k}))^{-1}B_{n_k}(b_{k-1})(\wt{\chi}_k)$. 
Then, analogously as for ESSM above, conditions \ref{condKappaDk} and \ref{condESSM1} hold for EMSM. 
%TODOcondESSM1f

\section{\label{secMeans}Helper theorems for proving the convergence properties 
of minimization methods with gradient-based stopping criteria} 
% The following theorems will be useful for proving convergence of single- and multi-stage minimization methods 
% \begin{condition}\label{condHlow}
% Function 
% $h:\mc{S}(A)\otimes(\Omega,\mc{F})\to \mc{S}(\overline{\R})$ 
% is such that  a.s. 
% $b\to h(b):=h(b,\cdot)$ is lower semicontinuous on $A$ and 
% \begin{equation}\label{lowerBound}
% \E(\sup_{b \in A}(h(b)_{-}))<\infty. 
% \end{equation}
% For such $h$ we define $b\in A\to f(b)=\E(h(b))$. 
% \end{condition}
\begin{condition}\label{condH2}
For a random variable $Y$ with values in a measurable space $\mc{S}$ and a nonempty set $A\in \mc{B}(\R^l)$, a function
$r:\mc{S}(A)\otimes\mc{S}\to \mc{S}(\overline{\R})$ is such that Condition \ref{condHlow} holds for 
$h(b,\cdot)=r(b,Y(\cdot))$, $b \in A$. 
\end{condition}

For $a \in \R^l$ and $\epsilon\in\R_+$, 
we define a sphere $S_l(a,\epsilon)=\{x\in \R^l:|x-a|=\epsilon\}$, a ball $B_l(a,\epsilon)=\{x\in \R^l:|x-a|<\epsilon\}$, 
and a closed ball $\overline{B}_l(a,\epsilon)=\overline{B_l(a,\epsilon)}=\{x\in \R^l:|x-a|\leq\epsilon\}$. 
The proof of the below lemma uses a similar
reasoning as in the proof of consistency of M-estimators in Theorem 5.14 in \cite{Vaart}.

\begin{lemma}\label{lemSuff}
Let Condition \ref{condH2} hold, $Y_1, Y_2,\ldots$ be i.i.d. $\sim Y$, 
$b\in A \to\wh{f}_n(b):= \frac{1}{n}\sum_{i=1}^nr(b,Y_i)$, $n \in \N_+$,
$K \subset A$ be a nonempty compact set, and $m$ be the minimum of $f$ on $K$ 
(which exists due to lemmas \ref{lemMin} and \ref{lemSemi}). 
Then, for each $a\in (-\infty,m)$, a.s. for a sufficiently large $n$, $\wh{f}_n(b) > a$, $b \in K$. 
\end{lemma} 
\begin{proof} 
Let $U_{i}=B_l(0,i^{-1})$, $i \in \N_+$. % for some $r_i \to 0$ (e.g. $r_i=i^-1$). %be open balls in $\R^l$ with center zero and positive
%radii converging to zero. 
From the  a.s. lower semicontinuity of $b\to r(b,Y)$, for each $v \in K$, for 
$g_{l,v}(x)=\inf_{b\in \{v+U_l\}\cap A} r(b,x)$, we have a.s. $g_{l,v}(Y) \uparrow r(v,Y)$ as $l \to \infty$. 
Thus, from the monotone convergence theorem, 
%\begin{equation}
$\E(g_{l,v}(Y))\uparrow f(v)$ as $l\to \infty$, $v \in K$. 
%\end{equation}
In particular, $\E(g_{l,v}(Y))>a$ for $l \geq l_v$ 
for some $l_v \in \N_+$, $v \in K$. 
The family $\{D_v:=v+U_{l_v}:v \in K\}$ is a cover of $K$. 
From the compactness of $K$, let $\{D_{v_1},\ldots,D_{v_m}\}$ be its finite subcover. 
Then, from the generalized SLLN in Theorem \ref{thSLLNG} 
(which can be used thanks to (\ref{lowerBound})),
\begin{equation}
\inf_{b \in K} \wh{f}_n(b)\geq \min_{k\in\{1,\ldots,m\}} \frac{1}{n}\sum_{i=1}^ng_{l_{v_k},v_k}(Y_i) 
\overset{\text{a.s.}}{\to}
\min_{k\in\{1,\ldots,m\}} \E(g_{l_{v_k},v_k}(Y)) >a.  
\end{equation}
\end{proof}

\begin{lemma}\label{lemProdBel}
Let Condition \ref{condH2} hold for $r$ equal 
to some nonnegative $r_1$ and $r_2$, for the same $Y$ and $A$. Let 
$g(b)=\E(r_1(b,Y))$ and $\E(r_2(b,Y))=1$, $b \in A$, 
let $Y_1, Y_2,\ldots$ be i.i.d. $\sim Y$, 
and let $b\in A\to \wh{f}_{i,n}(b):= \frac{1}{n}\sum_{j=1}^nr_i(b,Y_j)$, $i=1,2$, 
and $b\in A\to \wh{g}_n(b):=\wh{f}_{1,n}(b)\wh{f}_{2,n}(b)$, $n\in \N_+$. 
Let $K \subset A$ be a nonempty compact set and $m$ be the minimum of $g$ on $K$. 
Then, for each $a \in(-\infty,m)$, a.s. for a sufficiently large $n$, 
\begin{equation}\label{gnbab} 
\wh{g}_n(b) >a,\quad b \in K. 
\end{equation} 
\end{lemma} 
\begin{proof} 
It holds $\wh{g}_n(b)\geq 0$, $n \in \N_+$, $b \in A$, 
so that it is sufficient to consider the case when $m>0$ and 
$0<a<m$. Let $a<d<m$. 
Then, from Lemma \ref{lemSuff}, a.s. for a sufficiently large $n$, $\wh{f}_{1,n}(b) > d$ 
and $\wh{f}_{2,n}(b) > \frac{a}{d}$, $b \in K$, in which case (\ref{gnbab}) holds. 
\end{proof}

\begin{condition}\label{condAf}
We have $b^*\in \R^l$ and $A\in \mc{B}(\R^l)$ is a neighbourhood of $b^*$. A function 
$f:A\to \overline{\R}$, $f>-\infty$, is lower semicontinuous and $b^*$ is  %be strictly convex on the convex set on which it is finite
its unique minimum point (in particular, $f(b^*)<\infty$).  
\end{condition}

\begin{condition}\label{condConvBel}
Condition \ref{condAf} holds and $B\subset\R^l$ is such that $A \subset B$. Functions $f_n:B \to \R$, $n \in \N_+$, fulfill 
\begin{equation}\label{limnfnbsfbs}
\lim_{n\to \infty}f_n(b^*)=f(b^*).
\end{equation}
Furthermore, for each compact set $K\subset A$, for $m$ equal to the minimum of $f$ on $K$, for each $a < m$, for a sufficiently large $n$, 
$f_n(x)>a$, $x \in K$. %(note that this holds e.g. when $f_{n|K} \overset{loc}{\rightrightarrows}f$). 
\end{condition} 
\begin{remark}\label{remsuffBelUnif} 
Let Condition \ref{condAf} hold, $B\subset\R^l$, $A \subset B$, and 
$f_n:B \to \R$, $n \in \N_+$ be such that $f_{n|A} \overset{loc}{\rightrightarrows}f$. Then, 
Condition \ref{condConvBel} holds.
\end{remark} 
 
\begin{remark}\label{remIneqs}
Let us assume Condition \ref{condConvBel}, 
let $\epsilon \in\R_+$ be such that $\overline{B}_l(b^*,\epsilon) \subset A$ and let 
$c$ be the minimum of $f$ on $S_l(b^*,\epsilon)$. From the uniqueness of the minimum 
point $b^*$ of $f$, we have $c>f(b^*)$. Let $\delta\in\R_+$ be such that  $c>f(b^*) +\delta$. Then, for a sufficiently large $n$
\begin{equation}\label{wfnb}
f_n(b)\geq f(b^*) +\delta, \quad b \in S_l(b^*,\epsilon)
\end{equation}
and
\begin{equation}\label{fnb}
f_n(b^*) \leq f(b^*) + \frac{\delta}{2}.
\end{equation}
\end{remark}

\begin{theorem}\label{convMinGenFun}
Let us assume that Condition \ref{condConvBel} holds for a convex $B$ and for $f_n$, $n\in \N_+$, which are convex and continuous.
%, and for a sufficiently large $n$, $f_n$ are strictly convex
%and continuous.
Then, for a sufficiently large $n$, $f_n$ possesses a minimum point $a_n \in B$. Furthermore,
\begin{equation}\label{limanbs}
\lim_{n\to \infty} a_n=b^*
\end{equation}
and
\begin{equation}\label{fnannb}
\lim_{n\to \infty} f_n(a_n)=f(b^*).
\end{equation}
%Let $\epsilon_n\in\R_+$, $n \in \N_+$, $\lim_{n\to \infty}\epsilon_n=0$. 
If further $B$ is open, $f_n$, $n\in\N_+$, are differentiable on $B$, and a sequence $b_n \in B$, $n\in\N_+$, 
is such that $\lim_{n\to \infty}|\nabla f_n(b_n)|=0$, then
\begin{equation}\label{limbnbs}
\lim_{n\to \infty} b_n=b^*
\end{equation}
and 
\begin{equation}\label{limfnbnfb}
\lim_{n\to \infty} f_n(b_n)=f(b^*). 
\end{equation}
\end{theorem}
\begin{proof}
Let us consider $\epsilon, \delta\in\R_+$ as in Remark \ref{remIneqs}. 
From this remark, let $N\in\N_+$ be such that for $n>N$ we have
(\ref{wfnb}) and (\ref{fnb}). Then, for $n >N$,   
for each $b \in B$ such that $|b-b^*|\geq \epsilon$, from the convexity of $f_n$ 
\begin{equation}
\begin{split}\label{fnbfnb}
f_n(b) - f_n(b^*) &\geq \frac{|b-b^*|}{\epsilon}(f_n(b^* + \epsilon\frac{b-b^*}{|b-b^*|})-f_n(b^*))\\
&\geq \frac{|b-b^*|\delta}{2\epsilon} >0.\\ 
\end{split}
\end{equation}
For $n>N$, from (\ref{fnbfnb}) and the continuity of $f_n$,  %\begin{equation}\label{xDfnx}
% $\inf_{x \in B}f_n(x)= \min_{x \in B_l(b^*,\epsilon)}f_n(x)$,
% %\end{equation}
$f_n$ has a minimum point $a_n$ fulfilling 
\begin{equation}\label{bnb}
|a_n-b^*|< \epsilon.
\end{equation}
This proves (\ref{limanbs}). For $n > N$, from (\ref{fnb}) and $f_n(a_n)\leq f_n(b^*)$ we have
\begin{equation}\label{fnanfbs}
f_n(a_n) \leq f(b^*) + \frac{\delta}{2}.
\end{equation}
From Condition \ref{condConvBel}, for some $N_1>N$, for $n >N_1$,
\begin{equation}\label{wfnbd} 
f_n(b) \geq f(b^*)-\frac{\delta}{2}, \quad b \in \overline{B}_l(b^*,\epsilon).
\end{equation} 
Thus, for $n>N_1$, from (\ref{bnb}), (\ref{wfnbd}), 
and (\ref{fnanfbs}), we receive that $|f_n(a_n)-f(b^*)|\leq \frac{\delta}{2}$. 
Since we could have selected $\delta$ arbitrarily small, we receive (\ref{fnannb}). 

Let $B$ be open and $f_n$ be differentiable. Then, for $b \in B$ such that $b \neq b^*$, for
$v= \frac{b-b^*}{|b-b^*|}$, from the convexity of $f_n$
\begin{equation}\label{nablafnb}
|\nabla f_n(b)| \geq \nabla_v f_n(b) \geq  \frac{f_n(b)-f_n(b^*)}{|b-b^*|}.
\end{equation}
Thus, for each $b \in B$ for which $|b - b^*|\geq \epsilon$,
for $n>N$, from (\ref{nablafnb}) and (\ref{fnbfnb})
\begin{equation}\label{nablafn}
|\nabla f_n(b)| \geq \frac{\delta}{2\epsilon}.
\end{equation}
Let $N_2>N$ be such that for $n>N_2$ 
\begin{equation}\label{nablafbn}
|\nabla f_n(b_n)|< \frac{\delta}{2\epsilon}. 
\end{equation}
Then, from (\ref{nablafn}), for $n>N_2$
\begin{equation}\label{bnbse}
|b_n -b^*| < \epsilon,
\end{equation}
which proves (\ref{limbnbs}). For $n>N_1\vee N_2$ we have
\begin{equation}
\begin{split}
\frac{\delta}{2}&= \frac{\delta}{2\epsilon}\epsilon> |\nabla f_n(b_n)||b_n-b^*|  \geq f_n(b_n)-f_n(b^*)\\
&\geq f_n(b_n) - f(b^*) - \frac{\delta}{2} \geq -\delta, 
\end{split}
\end{equation}
where in the first inequality we used (\ref{nablafbn}) and (\ref{bnbse}), in the second (\ref{nablafnb}), in the third
%f_n(b^*) \leq f(b^*) + \frac{\delta}{2}.
(\ref{fnb}), and in the last one (\ref{wfnbd}). 
Thus, in such a case 
\begin{equation} 
\delta\geq f_n(b_n) - f(b^*)\geq - \frac{\delta}{2}, 
\end{equation} 
which proves (\ref{limfnbnfb}). 
\end{proof} 

\begin{lemma}\label{lemStrConvA}
Let $A \subset \R^l$ be open. If a twice continuously differentiable function $f:A\to \R$ has a positive definite Hessian on $A$, 
then for each convex $U \subset A$ such that for some compact $K \subset A$, $U \subset K$, 
$f$ is strongly convex on $U$. If further $\lim_{x\uparrow A}f(x)= \infty$, then for each $x_0 \in A$ as such a $U$ one can take the sublevel set 
$S=\{x \in A:f(x)\leq f(x_0)\}$. 
\end{lemma} %is smooth with positive definite Hessian, and $\lim_{b\uparrow B}g(b)= \infty$. 
\begin{proof} 
From Lemma \ref{lemLipschMin}, $b\in A\to m_l(\nabla^2f(b))$ is continuous and thus $f$ is strongly convex on $U$ 
with a constant $\inf_{x\in K}m_l(\nabla^2f(x))>0$. From the convexity of $f$, $S$ as above is convex. 
Furthermore, if $\lim_{x\uparrow A}f(x)= \infty$, 
then for a sufficiently large $M$, for a compact set $K$ as in (\ref{Kdef}) we have $S\subset K$. 
\end{proof}

\section{\label{secminGrad}Minimization of estimators with gradient-based stopping criteria} 
In this section we define single- and multi-stage minimization methods of estimators 
with gradient-based stopping criteria, 
abbreviated as GSSM and GMSM respectively. We also discuss the possibility of their application to the minimization of 
the well-known mean square estimators in the LETGS setting and 
both the well-known and the new mean square estimators in the ECM setting. 

Consider some sets $T$ and $B$ as in Section \ref{secESSM} and let additionally such a $B$ be open. 
GSSM and GMSM are special cases of the following 
minimization method of random functions with gradient-based stopping criteria, abbreviated as GM. 
In GM we assume Condition \ref{condftdef}. Furthermore, we assume that 
$b\to \wh{f}_t(b,\omega)$ is differentiable, $t\in T$, $\omega \in \Omega_1$, and that we are given
$[0,\infty]$-valued random variables $\epsilon_t$, $t \in T$, such that a.s. 
\begin{equation}\label{epstozero} 
\lim_{t\to \infty}\epsilon_t=0
\end{equation}
and
%instead of Condition \ref{condEM} we assume the following one. 
%\begin{condition}\label{condGM}
\begin{equation}\label{nablawhfdt}
|\nabla_b \wh{f}_t(d_t(\omega),\omega)| \leq \epsilon_t(\omega),\quad \omega \in  G_t,\ t \in T.
\end{equation}
% GSSM is defined as GM in which Condition \ref{condKappa} holds and 
% ESSM and EMSM are special cases of EM for respective $G_t$ and $d_t$ as above,
% for ESSM for $\wh{f}_t(b,\omega)=\I(N_t=k \in \N_p)\wh{\est}_k(b',b)(\wt{\kappa}_k(\omega))$, while  for EMSM for  $T=\N_+$, 
% and $\wh{f}_k(b,\omega)=\wh{\est}_{n_k}(b_{k-1},b)(\wt{\chi}_k(\omega))$.
% Then, in GM for some functions $f:B\times\Omega_1 \to \R$ such that 
% $b\to f_t(b,\omega)$ is differentiable, $\omega \in \Omega$, and for some $G_t \in \mc{F}_1$, for $\omega \in G_t$ we have
% \begin{equation}
% \nabla_b f 
% \end{equation}
% $\omega \in G_t$
% Let us assume the following Condition. 
%  \begin{condition}\label{uniqGrad}
%  For $\omega\in \Omega$ and $t\in T$, $b\to f_t(b,\omega)$ is differentiable and for $\omega \in G_t$, 
%  $b\to \wh{f}_t(b,\omega)$ has a unique minimum point $c$ which is also the unique point 
%  in $B$ such that $\nabla_b \wh{f}_t(c,\omega)=0$. 
%  %and a unique point in $B$ such that $\nabla_b\wh{f}_t(b,\omega)=0$.
% \end{condition}
%A GM procedure is further defined in the same way as such EM procedure in Section \ref{secESSM}, except that %
%we assume that 
%$b\to f_t(b,\omega)$ is differentiable, $\omega\in \Omega$, $t\in T$, and
%\end{condition}
%From Condition \ref{uniqGrad}, for $\epsilon_t=0$, GM becomes EM. 
% GSSM and GMSM are defined as special cases of GM, 
% in the first case for
We shall further need the following conditions and lemmas.
\begin{condition}\label{condMesDiff}
Condition \ref{condMesDn} holds, for each $n \in \N_p$ and $(b',\omega)\in D_n$, 
$b \in B\to \wh{\est}_n(b',b)(\omega)$ is differentiable, 
and $b_n^*(b',\omega)$ is equal to the unique point  $c\in B$ such that 
%\begin{equation}
$\nabla_b \wh{\est}_n(b',c)(\omega)=0$. 
%\end{equation}
\end{condition}

\begin{lemma}\label{lemMesDiff} 
Condition \ref{condMesDiff} implies Condition \ref{condmesbs}. 
\end{lemma}
\begin{proof} 
A function $(b,(b',\omega))\in B\times D_n  \to g(b,(b',\omega)) :=\nabla_b \wh{\est}_n(b',b)(\omega)$ is measurable 
from $\mc{S}(B) \otimes (D_n,\mc{F}'_n)$ to $\mc{S}(\R^l)$ and for each $D \in \mc{B}(B)$, 
$(b^*_n)^{-1}(D)=\{(b',\omega) \in D_n: \text{there exists }c \in D, \text{ such that } g(c,(b',\omega))=0\}$ 
is a projection of $g^{-1}(0)\cap (D \times D_n) \in \mc{B}(B)\otimes\mc{F}'_n$ 
onto the second coordinate. Thus, $(b_n^*)^{-1}(D)\in \mc{F}'_n$, $D \in \mc{B}(B)$. 
\end{proof} 

\begin{condition}\label{condSmooth} 
The set $B$ is convex. Furthermore, for each $n \in \N_p$, a set %\begin{equation}\label{dnf1o} 
$\wt{D}_n \in \mc{F}_1^n$ is such that for each $b'\in A$ 
and $\omega \in \wt{D}_n$, $b \in B\to g(b):=\wh{\est}_n(b',b)(\omega)$ 
is smooth with a positive definite Hessian on $B$, and $\lim_{b\uparrow B}g(b)= \infty$. 
%, and for each $x_0 \in B$, $g$ it is strongly convex on the sublevel set $S= \{x\in B:g(x)\leq g(x_0)\}$. 
\end{condition} 

\begin{lemma}\label{lemCondSmooth} 
Condition \ref{condSmooth} implies Condition \ref{condMesDiff} 
for $b_n^*(b',\omega)$ as in Condition \ref{condMesDn} and $D_n=A\times \wt{D}_n$, $n \in \N_p$.
\end{lemma} 
\begin{proof} 
It follows from lemmas \ref{lemConvFun} and \ref{lemMinPoint}. 
\end{proof} 

Except for some differences mentioned below, we define GSSM and GMSM in the same way as ESSM and EMSM in the previous section. 
The first difference is that in GSSM and GMSM we additionally assume that Condition \ref{condMesDiff} holds for 
$B$ as above and we consider $[0,\infty]$-valued random variables $\epsilon_t$, $t \in T$, such that a.s. (\ref{epstozero}) holds. 
%assume the following condition. % \ref{condMesDiff} holds. %For some unbounded set $T$ as in Condition \ref{condT} 
Furthermore, in GSSM, for $t \in T$, on $G_t$ as in (\ref{kappaDk}), 
instead of (\ref{ckdef}) we require that $|\nabla_b\wh{\est}_{k}(b',d_t)(\wt{\kappa}_{k}(\omega))|\leq \epsilon_t$, 
while in GMSM, for $k \in \N_+$, on $G_k$ as in (\ref{chiDk}), instead of (\ref{akdef}) we 
require that $|\nabla_b\wh{\est}_{n_k}(b_{k-1},d_k)(\wt{\chi}_k)|\leq \epsilon_k$. 

Note that GSSM and GMSM are special cases of GM under the identifications as in Remark \ref{remIdentEM}. 
Such identifications shall be frequently considered below. 
From Lemma \ref{lemMesDiff}, for $\epsilon_t=0$, $t \in T$, 
GSSM and GMSM become special cases of ESSM and EMSM respectively. %, and in general 

\begin{remark}\label{remRealizGM}
Let us discuss how one can construct the variables $d_t$, $t \in T$, in GSSM and GMSM, 
assuming that the other variables as above are given. Let $t \in T$. 
From Assumption \ref{condMesDiff}, on an arbitrary event $A_t$ contained in the appropriate $G_t$ as above, like 
$A_t=G_t$ or $A_t=G_t\cap \{\epsilon_t =0\}$, we can take in GSSM $d_t= b_k^*(b',\wt{\kappa}_k)$ and in GMSM 
$d_t= b_{n_t}^*(b_{t-1},\wt{\chi}_t)$. Note that from Lemma \ref{lemMesDiff},
in both these cases $d_t$ is measurable on $A_t$. 
Unfortunately, in the examples discussed below such $d_t(\omega)$, $\omega \in A_t$, typically cannot be found in practice. 
Let now $\omega \in \Omega$ be such that $\epsilon_t(\omega)>0$. Then, under some additional assumptions on 
%\begin{equation}\label{gdef}
$b\to g(b):=\wh{f}_t(b,\omega)$, 
%\end{equation}
% \begin{equation}\label{gdef}
% b\to g(b):=\wh{f}_t(b,\omega), 
% \end{equation} 
$d_t(\omega)$ in GSSM or GMSM can be a result of some globally convergent 
iterative minimization method (i.e. one in which the gradients in the subsequent points converge to zero), 
minimizing $g$, started at $x_0$ equal to $b'$ in GSSM or $b_{t-1}(\omega)$ in GMSM, 
and stopped in the first point $d_t(\omega)$ in which (\ref{nablawhfdt}) holds. 
As such an iterative method one can potentially use the damped Newton method, for the global and quadratic convergence of 
which it is sufficient if $g$ is strongly convex 
on the sublevel set $S:=\{x\in B: g(x)\leq g(x_0)\}$, $g$ is twice continuously differentiable on some open neighbourhood of such an $S$, and 
the second derivative of $g$ is Lipschitz on $S$ (see Section 9.5.3 in \cite{Boyd_2004}). 
From Lemma \ref{lemStrConvA}, such assumptions 
hold in the above discussed GSSM and GMSM methods if 
Condition \ref{condSmooth} holds, we consider the corresponding $D_n$, $n \in \N_p$,  as in Lemma \ref{lemCondSmooth}, 
and we have $\omega \in G_t$. 
See \cite{Boyd_2004} and \cite{nocedal2006numerical} for some other examples 
of globally convergent minimization methods requiring typically weaker assumptions. 
In Remark \ref{remDetGM} below we discuss a situation when one can perform some minimization method 
of a $g$ as above for each $\omega \in \Omega$ such that $\epsilon(\omega)>0$. 
For most iterative minimization methods, including the damped Newton method, if the same method is used for each 
$\omega$ in some event $B_t$ contained in $\{\epsilon_t>0\}$, 
then the fact that the resulting $d_t$ is measurable on $B_t$ 
follows from the definition of the method. 
On $G_t'$ one can define $d_t$ in similar ways as for ESSM or EMSM in the previous section.
% If $\epsilon_t(\omega) >0$, $\omega\in \Omega$, then to avoid checking if $G_t$ holds (i.e. if $\omega \in G_t$), in some cases 
% one can perform minimization of $g$ as above 
% for each $\omega \in \Omega$. See Remark \ref{remDetGM} for a discussion of a situation in which this is possible. 
\end{remark} 

\begin{condition}\label{condsuffstrict} 
The above set $B$ is an open convex neighbourhood of some $b^* \in \R^l$, and  
$\epsilon\in\R_+$ is such that   $\overline{B}_l(b^*,\epsilon) \in B$. 
Furthermore, for some  $\wh{\est}_n$ as in (\ref{estnDef}) for some $n \in \N_p$, 
$b' \in A$, and  $\omega \in \Omega_1^n$ are such that 
\begin{equation}\label{slestg} 
\inf_{b \in S_l(b^*,\epsilon)}\wh{\est}_n(b',b)(\omega)> \wh{\est}_n(b',b^*)(\omega). 
\end{equation} 
\end{condition} 
The following remark will be useful for proving the convergence properties of the GM methods in the below examples. 
\begin{remark}\label{remsuffstrict} 
Consider the LETGS setting. Then, from Theorem \ref{thMsqStr}, if 
Condition \ref{condsuffstrict} holds for $\wh{\est}_n=\wh{\msq}_n$, 
then for  $a=b^*$ and each $b \in \R^l\setminus\{0\}$ we cannot have (\ref{fatb}) for 
\begin{equation}\label{badt} 
t= \frac{\epsilon}{|b|}, 
\end{equation} 
and thus $r_n(\omega)$ holds. 
 %In particular, Condition \ref{condsuffstrict} holds. 
% Let now consider the ECM setting and assume Condition \ref{condtX}. Then, for $\wh{\est}_n=\wh{\msq}_n$, 
%  $b\to L(b)(\omega)$ is strictly convex and if 
% %  (\ref{slestg}) holds, then we must have $Z(\omega_i)\neq 0$ we dis
% % for some $i\in\{1,\ldots,n\}$. In particular, Condition \ref{condsuffstrict} holds. 
Let us now consider the ECM setting. Then, if Condition \ref{condsuffstrict} holds for $\wh{\est}_n=g_{\var,n}$ (see (\ref{gvardef}))
then for $a=b^*$ and each $b\in \R^l\setminus\{0\}$ we cannot have (\ref{gnvarsm}) for $t$ as in (\ref{badt}). Thus, 
from Theorem \ref{thgvarInf}, in such a case system (\ref{sysIneq}) has only the zero solution.  
\end{remark}

For the GSSM and GMSM methods in the below examples we shall discuss when Condition \ref{condSmooth} holds in them and we consider 
Condition \ref{condMesDiff} holding in them as a result of Lemma \ref{lemCondSmooth}. For GMSM in all the below examples we assume that conditions 
\ref{condLiKi} and \ref{condLi} hold (where in the ECM setting we mean the counterparts of these conditions). 
% in such setting). Furthermore, we consider Condition 
% Condition \ref{condSmooth}, 
% $\wt{D}_n=\{\omega  \in \Omega_1^n:r_n(\omega)\}$, $n \in \N_+$ 
% (by which we mean that we consider the corresponding Condition \ref{condMesDiff} as in Lemma \ref{lemCondSmooth})

Let us first discuss GSSM and GMSM for $\wh{\est}_n=\wh{\msq}_n$, $n \in \N_+$, in the LETGS setting. 
From Theorem \ref{thMsqStr}, we can and shall take in 
Condition \ref{condSmooth}, $\wt{D}_n=\{\omega  \in \Omega_1^n:r_n(\omega)\}$, $n \in \N_+$.  
Let us assume conditions \ref{condzntau} and \ref{cond1} and that for some $b\in A$, $\msq(b)<\infty$, so that 
from Theorem \ref{thMsq}, we can and shall take in Condition \ref{condAf}, $f=\msq$. 
In GSSM, from Condition \ref{condT} and Theorem \ref{thPos}, Condition \ref{condKappaDk} holds. 
From the SLLN and Lemma \ref{lemSuff} for 
\begin{equation}\label{hmsq}
r(b,x)=(Z^2L'L(b))(x),\quad b \in A,\ x \in \Omega_1, 
\end{equation} 
and $Y_i=\kappa_i$, $i \in \N_+$, for $\PR$ a.e. $\omega \in \Omega$, Condition \ref{condConvBel} 
holds for $B=A$ and $f_n(b)=\wh{\msq}_{n}(b',b)(\wt{\kappa}_n(\omega))$. 
Thus, from Theorem \ref{convMinGenFun}, (\ref{nablawhfdt}), and (\ref{epstozero}), conditions \ref{condESSM1} and \ref{condESSM1f} hold. 
For GMSM let us assume that Condition \ref{condSpr} holds 
for $S=Z^2$. Then, from the third point of Theorem \ref{thunifDiffMult} and 
remarks \ref{remsuffBelUnif}, \ref{remIneqs}, and \ref{remsuffstrict}, a.s. for a sufficiently large $k$, $r_{n_k}(\wt{\chi}_k)$ 
holds, i.e. Condition \ref{condKappaDk} holds. Thus, from (\ref{epstozero}), (\ref{nablawhfdt}), and 
Theorem \ref{convMinGenFun},  conditions \ref{condESSM1} and \ref{condESSM1f} hold too. 

Let us now consider the ECM setting, assuming the following condition. 
\begin{condition}\label{condECMFull} 
Conditions \ref{condtX} and \ref{condPartvm} hold and $f=\msq$ satisfies Condition \ref{condAf}. 
\end{condition} 
From Lemma \ref{lemMinPoint}, $f=\msq$ satisfies Condition \ref{condAf} 
for instance when $f=\msq$ satisfies Condition \ref{condStrict},
which due to Lemma \ref{lempmsq} holds e.g. if for some $b \in A$, $\msq(b)<\infty$, and 
\begin{equation}\label{pq1msq} 
\PQ_1(p_{\msq})>0. 
\end{equation} 

Let us first consider the case of $\wh{\est}_n=\wh{\msq}_n$, $n \in \N_+$, for which let us assume (\ref{pq1msq}). 
From remarks \ref{remMsqEst} and \ref{remMsqECM}, we can and shall take in Condition \ref{condSmooth},
%\ref{condMesDiff} 
$\wt{D}_n=\{\omega  \in \Omega_1^n:p_{\msq}(\omega_i) \text{ holds for some } i \in \{1,\ldots,n\}\}$, $n \in \N_+$. 
%Note that Condition \ref{lemMesDiff} is fulfilled. 
In GSSM, from the SLLN, a.s. $\lim_{n\to \infty}\overline{(p_{\msq})}_n(\wt{\kappa}_n)= \PQ_1(p_{\msq})$, so that 
from (\ref{pq1msq}) and the SLLN we have a.s. $\wt{\kappa}_k\in \wt{D}_k$ for a sufficiently 
large $k$. Thus, from Condition \ref{condT}, Condition \ref{condKappaDk} holds. Using further Lemma \ref{lemSuff} for $r$ as in (\ref{hmsq}) and 
Theorem \ref{convMinGenFun},  conditions \ref{condESSM1} and \ref{condESSM1f} hold too. 
In GMSM, let $A=\R^l$. Since the counterpart of Condition \ref{condSpr} for ECM is fulfilled 
for $S=\I(p_{\msq})$, from the first point of Theorem \ref{thunifDiffMult}, 
a.s. $\overline{(L(b_{k-1})S)}_{n_k}(\wt{\chi}_k) \to \PQ_1(p_{\msq})$. Thus, from (\ref{pq1msq}), 
Condition \ref{condKappaDk} holds. 
Let us assume Condition \ref{condSpr} for $S=Z^2$. Then, 
from the third point of Theorem \ref{thunifDiffMult} and
%remarks \ref{remIneqs} or \ref{remsuffstrict}, and
Theorem \ref{convMinGenFun}, conditions \ref{condESSM1} and \ref{condESSM1f} hold. 

Let us now consider for each $n \in \N_2$ 
\begin{equation}\label{estngvar} 
\wh{\est}_n(b',b)=\wt{\msq2}_n(b',b):=\frac{1}{n^2}g_{\var,n}(b',b)+\frac{1}{n}\overline{((ZL')^2)}_n,\quad b'\in A,\ b \in B=\R^l 
\end{equation}
(see (\ref{gvardef})). 
Then, from (\ref{gvarfvar}), (\ref{fnvardef}), and (\ref{msq2sum}) %(\ref{msq2ndef}), and an easy calculation we have
% \wh{\var}_n(b',b)&=\frac{1}{n(n-1)}\sum_{i<j \in \{1,\ldots,n\}} \frac{L_i'L_j'}{L_i(b)L_j(b)}(Z_iL_i(b) - Z_jL_j(b))^2\\
% &= \frac{1}{n(n-1)}\left(\sum_{i=1}^n\left(Z_i^2L_i'L_i(b)\sum_{j\in \{1,\ldots,n\}, j \neq i}\frac{L_j'}{L_j(b)}\right)
% -\sum_{i<j \in \{1,\ldots,n\}}2Z_iZ_jL_i'L_j'\right)\\
%&=\frac{n}{n-1}\left(\wh{\msq}_n(b',b)\overline{\left(\frac{L'}{L(b)}\right)}_n - \overline{(ZL')}_n^2\right).
%\wh{\msq2}_{n}(b',b)= \wh{\msq}_n(b',b)\overline{\left(\frac{L'}{L(b)}\right)}_n,
%\wh{\msq2}_{n}(b',b)= 
%\wt{\msq2}_n(b',b)= \sum_{i=1}^n\left(Z_i^2L_i'\sum_{j\in \{1,\ldots,n\}, j \neq i}\frac{L_j'L_i(b)}{L_j(b)}\right),\quad b',b \in A.
%\frac{1}{n^2}\sum_{i=1}^n\left(Z_i^2L_i'L_i(b)\sum_{j\in \{1,\ldots,n\}, j \neq i}\frac{L_j'}{L_j(b)}\right)+\frac{1}{n}\overline{((ZL')^2)}_n
%=\frac{1}{n}\sum_{i=1}^n\left(Z_i^2L_i'L_i(b)(\overline{\left(\frac{L'}{L(b)}\right)}_n 
%- \frac{1}{n}\frac{L_i'}{L_i(b)}\right)+\frac{1}{n}\overline{((ZL')^2)}_n
\begin{equation}\label{estnbpb}
\wt{\msq2}_n(b',b)= \wh{\msq2}_{n}(b',b), \quad b', b \in A.
\end{equation}
Note that $\frac{1}{n}\overline{((ZL')^2)}_n$ does not depend on $b$. 
Thus, from Theorem \ref{thgvarInf}, we can and shall take in Condition \ref{condSmooth}, 
$\wt{D}_n=\{\omega  \in \Omega_1^n: \text{ system (\ref{sysIneq}) has only the zero solution}\}$.
In GSSM, from the SLLN and Lemma \ref{lemProdBel} for $r_1(b,y)=(Z^2L'L(b))(y)$, $r_2(b,y)=\frac{L'}{L(b)}(y)$, $b \in \R^l$, $y \in \Omega_1$,
and $Y_i=\kappa_i$, $i \in \N_+$, for $\PR$ a.e. $\omega \in \Omega$, Condition \ref{condConvBel} holds for $B=A$ and
$f_n(b)=\wh{\msq2}_{n}(b',b)(\wt{\kappa}_n(\omega))$, $b \in A$, and thus from (\ref{estnbpb}) it holds also for $B=\R^l$ and 
%\begin{equation}\label{fnbsum} 
$f_n(b)=\wt{\msq2}_n(b',b)(\wt{\kappa}_n(\omega))$, $b \in B$. 
%\end{equation}
Therefore, from remarks \ref{remIneqs} and \ref{remsuffstrict} and Condition \ref{condT}, Condition \ref{condKappaDk} holds.
Using further Theorem \ref{convMinGenFun}, conditions \ref{condESSM1} and \ref{condESSM1f} hold as well. 
In GMSM, let $A=\R^l$ and let us assume that the counterpart of 
Condition \ref{condSpr} for ECM 
holds for $S=Z^2$ (note that for $S=1$ it holds automatically). Then, from the fifth 
point of Theorem \ref{thunifDiffMult} and from 
remarks \ref{remIneqs} and \ref{remsuffstrict}, Condition \ref{condKappaDk} holds. 
Using further Theorem \ref{convMinGenFun}, conditions \ref{condESSM1} and \ref{condESSM1f} 
hold as well. 

\begin{remark}\label{remDetGM}
Checking if $G_t$ holds in possible practical realizations of GSSM or GMSM methods, 
as it can be done when using the damped Newton method as discussed in Remark \ref{remRealizGM},
may be inconvenient. For instance, for $\wh{\est}_n=\wh{\msq}_n$ in the LETGS setting  or 
$\wh{\est}_n$ as in (\ref{estngvar}) in the ECM setting as above, this typically cannot be done precisely due to numerical errors, 
and one has to make a rather arbitrary decision when such a condition holds approximately. 
From the below discussion, in the latter case one can avoid checking if $G_t$ holds 
and perform some minimization method of a $g$ as in Remark \ref{remRealizGM} for each $\omega \in \Omega$ such that $\epsilon(\omega)>0$. 
From the Zoutendjik theorem (see Theorem 3.2 in \cite{nocedal2006numerical}), 
for a number of line search minimization methods of a function  $g:B\to \R$ started at $x_0 \in B$ to be globally convergent 
it is sufficient if $g$ is bounded from below and continuously differentiable on some open neighbourhood 
$\mc{N} \subset B$ of the sublevel set $\{x \in B:g(x)\leq g(x_0)\}$, and if 
$\nabla g$ is Lipschitz on $\mc{N}$. 
In particular, it is sufficient if, in addition to the boundedness from below, $g$ is twice differentiable and 
$||\nabla^2g||_{\infty}$ is bounded on such an $\mc{N}$. One of the methods for which this holds is gradient descent 
with step lengths satisfying the Wolfe conditions; 
see \cite{nocedal2006numerical}. 
Note that from (\ref{gvardef}), (\ref{nabla2g}), and $||vv^T||_{\infty}=|v|^2$, $v \in \R^l$, for $\omega \in \Omega_1^n$ and 
$K=\max_{i,j \in \{1,\ldots,n\}} |v_{j,i}(\omega)|^2$, we have for each $b'\in A$ and $b\in \R^l$
\begin{equation}
\begin{split}
||\nabla^2_bg_{\var,n}(b',b)(\omega)||_{\infty}&\leq 
\sum_{i=1}^n(Z^2L')(\omega_i)\sum_{j\in \{1,\ldots,n\}, j \neq i}L'(\omega_j)||v_{j,i}(\omega)v_{j,i}(\omega)^T||_{\infty}\exp(b^Tv_{j,i}(\omega))\\
&\leq K g_{\var,n}(b',b)(\omega). \\
\end{split}
\end{equation}
Thus, for $\wh{\est}_n$ as in (\ref{estngvar}) it also holds
$||\nabla^2_b\wh{\est}_{n}(b',b)(\omega)||_{\infty}\leq K \wh{\est}_{n}(b',b)(\omega)$.
From this it follows that for $g$ as in Remark \ref{remRealizGM} corresponding to the GSSM or GMSM methods 
for $\wh{\est}_n$ as above, for each $x_0 \in \R^l$ and $\delta\in\R_+$, the assumptions of the Zoutendjik theorem as above
hold for $\mc{N}= \{x \in \R^l: g(x)< g(x_0)+\delta\}$. 
\end{remark}

\section{Helper theorems for proving the convergence properties of 
multi-phase minimization methods}
\begin{theorem}\label{thConvStrong}
Let $U\subset \R^l$ be an open ball with a center $b^*$ and 
$f:U\to \R$ be strongly convex with a constant $s\in \R_+$.  
Let $f_n:U\to \R$, $n \in \N_+$, be twice differentiable and such that 
$\nabla^2f_n \rightrightarrows \nabla^2f$. 
Then, for each $0<m<s$, for a sufficiently large $n$, 
$f_n$ is strongly convex with a constant $m$. 
Let further $b^*$ as above be the minimum point of $f$ and 
%\begin{equation}\label{nablafnf}
$\nabla f_n \rightrightarrows \nabla f$.  
%\end{equation}
Then, for a sufficiently large $n$, $f_{n}$ 
possesses a unique minimum point $a_n$, which is equal to the unique point $x\in U$ for which $\nabla f_n(x)=0$, and 
each $b \in U$ is a $\frac{1}{2m}|\nabla f_{n}(b)|^2$-minimizer of $f_{n}$. 
Furthermore,  $\lim_{n\to \infty}a_n = b^*$. 
\end{theorem} 
\begin{proof} 
Let $0<m<s$. From Lemma \ref{lemLipschMin}, for the sufficiently large $n$ for which 
$||\nabla^2f_n(x)-\nabla^2f(x)||_{\infty}<s-m$, $x \in U$, 
we have $m_l(\nabla^2f_n(x)) >m$, $x\in U$, 
so that $f_n$ is strongly convex with a constant $m$. Under the additional assumptions as above, 
let $h_n=f_n + f(b^*)-f_n(b^*)$, $n \in \N_+$. Then, $h_n(b^*)=f(b^*)$ 
and $\nabla h_n=\nabla f_n\rightrightarrows \nabla f$, so that $h_n \rightrightarrows f$. 
Furthermore, since $\nabla^2h_{n}=\nabla^2f_{n}$, $n \in \N_+$, $h_{n}$ is strongly convex for a sufficiently large $n$. 
Thus, from Remark \ref{remsuffBelUnif} and Theorem 
\ref{convMinGenFun}, for a sufficiently large $n$, $h_{n}$ and thus also $f_{n}$ possesses 
a unique minimum point $a_n$
and $\lim_{n\to \infty }a_n=b^*$. 
The rest of the thesis follows from the discussion in Section \ref{secStrong}. 
\end{proof}

\begin{condition}\label{condUniffn} 
%\begin{condition}\label{condUffn} 
A function $f:\R^l\to \R$ is continuous, functions $f_n:\R^l\rightarrow \R$, $n \in \N_+$, are such that $f_n\overset{loc}{\rightrightarrows}f$, 
and for a sequence $d_n \in \R^l$, $n \in \N_+$, we have 
\begin{equation}\label{limdnds}
\lim_{n\rightarrow\infty}d_n=d^*\in \R^l.
\end{equation}
\end{condition}

% \begin{condition}\label{condUniffn}
% Condition \ref{condUffn} holds, $f$ is continuous, and for a sequence $d_n \in \R^l$, $n \in \N_+$, we have
% \begin{equation}\label{limdnds}
% \lim_{n\rightarrow\infty}d_n=d^*\in \R^l  
% \end{equation}
% \end{condition}

We have the following easy-to-prove lemma. 
\begin{lemma}\label{lemConvUnif}
%If Condition \ref{condUffn} holds and for some  $d_n \in \R^l$, $n \in \N_+$, for some compact $K \subset \R^l$ we have
%$d_n\in K$ for a sufficiently large $n\in \N_+$, 
%then $\lim_{n\to \infty}(f_n(d_n)-f(d_n))=0$. 
If Condition \ref{condUniffn} holds, then $\lim_{n\to\infty} f_n(d_n)=f(d^*)$. 
% \begin{equation} 
% \end{equation}
% Condition \ref{condUniffn} we have 
\end{lemma}

\begin{theorem}\label{thConvda}
Assuming condition \ref{condUniffn}, let for a bounded sequence $s_n \in \R^l$, $n\in \N_+$, 
it hold $f_n(s_n) \leq f_n(d_n)$, $n \in \N_+$. Then, 
\begin{equation}\label{limsupfan}
\limsup_{n\rightarrow \infty}f(s_n) \leq f(d^*).
\end{equation}
Let further $d^*\in \R^l$ be the unique minimum point of $f$. Then, 
\begin{equation}\label{limfan}
\lim_{n\to \infty}f(s_n) = f(d^*).
\end{equation}
If further $f$ is convex, then
\begin{equation}\label{limsnds}
\lim_{n\to \infty}s_n=d^*.
\end{equation}
\end{theorem}
\begin{proof}
Let $\epsilon > 0$. From the boundedness of the set $D=\{s_n: n \in \N\}$ and $f_n\overset{loc}{\rightrightarrows}f$,
let $N_1 \in \N_+$ be such that for $n \geq N_1$, $|f_n(x)-f(x)|< \frac{\epsilon}{2}$, $x\in D$. From Lemma \ref{lemConvUnif}, 
let $N_2 \geq N_1$ be such that for $n \geq N_2$, $|f_n(d_n) - f(d^*)| < \frac{\epsilon}{2}$.
Then, for each $n \geq N_2$,
\begin{equation}\label{fsnleq}
f(s_n) < f_n(s_n) +\frac{\epsilon}{2} \leq f_n(d_n) +\frac{\epsilon}{2} <   f(d^*) +\epsilon,
\end{equation}
which proves (\ref{limsupfan}). Let $d^*$ be the unique minimum point of $f$. 
Then, (\ref{limfan}) follows from $f(s_n)\geq f(d^*)$, $n\in \N_+$, and (\ref{limsupfan}). 
Let now $f$ be convex and $\delta \in \R_+$. Then, from the continuity of $f$, 
there exists $x_0 \in S_l(d^*,\delta)$
such that $f(x_0)=\inf_{x\in S_{l}(d^*,\delta)}f(x)$. From the uniqueness of $d^*$, $m:=f(x_0) -f(d^*)>0$. 
%so $m>f(a^*)$ follows from uniqueness of $b^*$.
From the convexity of $f$, for $x \in \R^l$ such that $|x-d^*|\geq\delta$ we have
\begin{equation}
f(x)-f(d^*) \geq \frac{|x-d^*|}{\delta}(f(d^* + \delta\frac{x-d^*}{|x-d^*|}) -f(d^*))\geq m. 
\end{equation}
Thus, when (\ref{fsnleq}) holds for $\epsilon \leq m$, 
then we must have $|s_n-d^*|<\delta$, 
which proves (\ref{limsnds}). 
\end{proof}
% For $x \in \R^l$ and $B \subset \R^l$, let us denote 
% \begin{equation}
% \dist(x,B)=\inf_{y \in B}|x-y|.
% \end{equation}

\section{\label{secTwoPhase}Two-phase minimization of estimators 
with gradient-based stopping criteria and constraints or function modifications} 
In this section we describe minimization methods of estimators 
in which two-phase minimization can be used. In their first phase one can use some 
GM method as in Section \ref{secminGrad} 
and in the second phase e.g. constrained minimization 
of the estimator considered or unconstrained minimization 
of such a modified estimator, using gradient-based stopping criteria. 
The single- and multi-stage versions of these methods shall be 
abbreviated as CGSSM and CGMSM respectively. 
We also discuss applications of these methods to the minimization of 
the new mean square estimators 
in the LETGS setting and the inefficiency constant estimators 
in the ECM setting for $C=1$. 

Let us further in this section assume  that $A=\R^l$ and 
that the following condition holds. 
\begin{condition}\label{condG}
For some $\epsilon \in\R_+$, functions $g_1,g_2:\R^l \rightarrow \mc{B}(\R^l)$ are such that 
for each $x \in \R^l$, $g_1(x)$ is open,
\begin{equation}\label{klepsg}
B_l(x,\epsilon) \subset g_1(x), 
\end{equation} 
\begin{equation}
\overline{g_1(x)}\subset g_2(x), 
\end{equation}
and for each bounded set $B\subset \R^l$, the set $\bigcup_{x \in B}g_2(x)$ is bounded. 
\end{condition}
CGSSM and CGMSM will be defined as special cases of the following CGM method. 
In CGM we assume that for some unbounded $T\subset \R_+$, for each $t \in T$ 
we are given a $[0,\infty]$-valued random variable $\wt{\epsilon}_t$, 
an $A$-valued random variable $d_t$, a 
function $\wt{d}_t:\Omega \to A$, and a random function $\wt{f}_t:\mc{S}(A)\otimes (\Omega,\mc{F}) \to\mc{S}(\R)$, such that
$b \to \wt{f}_t(b,\omega)$ is differentiable, $\omega \in \Omega$, we have
\begin{equation}\label{ckingck}
\wt{d}_t \in g_2(d_t), 
\end{equation}
\begin{equation}\label{whsmSSM}
\wt{f}_t(\wt{d}_t)\leq  \wt{f}_t(d_t), 
\end{equation}
and if $\wt{d}_t \in g_1(d_t)$, then
\begin{equation}\label{nablawtestbSSM}
|\nabla_b\wt{f}_{t}(\wt{d}_t)|\leq \wt{\epsilon}_t. 
\end{equation} 
Furthermore, we assume that Condition \ref{condESSM1} holds for the above variables $d_t$, $t \in T$, and some $b^* \in A$.
%\begin{condition}\label{condESSM1} 
% It holds a.s. $\lim_{t\to \infty }d_t = b^*$. 
% \end{condition} 

\begin{remark} 
Functions $\wt{d}_t$ as above always exist, assuming that the other variables as above are given. 
Indeed, without loss of generality let 
$\wt{\epsilon}_t=0$. Then, if $\wt{d}_t$ fulfilling (\ref{ckingck}), (\ref{whsmSSM}), 
and $\wt{d}_t \notin g_1(d_t)$ 
does not exist, then $\wt{d}_t$ can be chosen to be a minimum point of 
$b \in g_1(d_t)\to \wt{f}_{t}(b)$, which exists
due to $\overline{g_1(d_t)}$ being compact (see Condition \ref{condG}) and $\wt{f}_t(b)>  \wt{f}_t(d_t)$, $b\in \partial g_1(d_t)
\subset g_2(d_t)\setminus g_1(d_t)$. 
\end{remark} 
%Let us now define 
% CGSSM and CGMSM are defined using some GSSM and GMSM procedures as in Section \ref{secminGrad}, which 
% let us further consider. 
 %let us consider some functions $\wh{\est}_k$ as in the previous section and 
Consider some functions $\wt{\est}_k$, $k \in \N_p$, as in (\ref{estnDef}) such that 
$b\to \wt{\est}_{k}(b',b)(\omega)$ is differentiable, $b' \in \R^l$, $\omega \in \Omega_1^k$, $k \in \N_p$. 
%are defined using some GSSM and GMSM procedures as in the previous section, which 
%which are $\mc{F}_k$-measurable in CGSSM and $\mc{G}_k$-measurable in CGMSM. 
%For some $p \in \N_+$ and
\begin{definition}\label{defCGM}
CGSSM is defined as CGM in which conditions \ref{condT} and \ref{condKappa} hold and
$\wt{f}_t(b)=\I(N_t= k \in \N_p)\wt{\est}_{k}(b',b)(\wt{\kappa}_k)$, $t \in T$.  
CGMSM is defined as CGM in which $T=\N_+$, Condition \ref{condKi} holds, Condition \ref{condChi} holds for $n_k \in \N_p$, $k \in \N_+$,
and we have $\wt{f}_k(b)=\wt{\est}_{n_k}(b_{k-1},b)(\wt{\chi}_k)$, $k \in \N_+$. 
\end{definition}

% let for $t \in T$, $\wt{c}_t$ be functions 
% from $\Omega$ to $\R^l$ such that 
% \begin{equation}\label{ckingck} 
% \wt{c}_t \in g_2(c_t), 
% \end{equation} 
% when  $N_t= k \in \N_p$, then 
% \begin{equation}\label{whsmSSM} 
% \wt{\est}_{k}(b',\wt{c}_t)(\wt{\kappa}_k)\leq  \wt{\est}_{k}(b',c_t)(\wt{\kappa}_k), 
% \end{equation} 
% and if $\wt{c}_t \in g_1(c_t)$, then 
% \begin{equation}\label{nablawtestbSSM} 
% |\nabla_b\wt{\est}_{k}(b',\wt{c}_t)(\wt{\kappa}_k)|\leq \wt{\epsilon}_t. 
% \end{equation} 
% In CGMSM, for GMSM as in the previous section, let for $k \in \N_+$, $\wt{a}_k$ be functions from $\Omega$ to $\R^l$ such that 
% \begin{equation}\label{akingak} 
% \wt{a}_k \in g_2(a_k), 
% \end{equation}
% \begin{equation}\label{whsmMSM}
% \wt{\est}_{n_k}(b_{k-1},\wt{a}_k)(\wt{\chi}_k)\leq  \wt{\est}_{n_k}(b_{k-1},a_k)(\wt{\chi}_k), 
% \end{equation}
% and if $\wt{a}_k \in g_1(a_k)$, then
% \begin{equation}\label{nablawtestbMSM}
% |\nabla_b\wt{\est}_{n_k}(b_{k-1},\wt{a}_k)(\wt{\chi}_k)|\leq \wt{\epsilon}_k.  
% \end{equation}
%Sometimes we will need the following conditions to hold. 
The following condition is needed e.g. if we want to investigate the asymptotic properties of $\wt{d}_t$, $t \in T$. %(see \cite{Badowski2015}). 
\begin{condition}\label{condVarck}
The functions $\wt{d}_t$, $t \in T$, are random variables (i.e. they are measurable functions from $(\Omega,\mc{F})$ to $\mc{S}(A)$).
 \end{condition}
\begin{remark}\label{remMesUnc}
Whenever dealing with some set $D\in \Omega$ for which it is not clear if $D \in \mc{F}$, when trying to prove that 
$\PR(D)=1$ and in particular $D \in \mc{F}$, 
we shall implicitly assume that we are working on a complete probability space, so that to achieve the goal 
it is sufficient to prove that $\PR(E)=1$ for some $E \in \mc{F}$ such that $E\subset D$. 
Such a $D$ will further typically appear when considering functions $g_t:\Omega \to \R^l$,  like $\wt{d}_t$ as above, without assuming that 
they are random variables. For instance, for some  $b^* \in \R^l$, we will consider 
$D=\{\omega\in \Omega: \lim_{t\to \infty}g_t(\omega)\to b^*\}$ or
$D=\{\omega\in \Omega: g_t(\omega)= b^* \text{ for a sufficiently large }t\}$. 
\end{remark}

\begin{condition}\label{condRE}
It holds $\wt{\epsilon}_t(\omega) >0$, $t \in T$, $\omega\in \Omega$, and we are given a function $R:\mc{S}(A)\to  \mc{S}((\epsilon,\infty))$
such that $\sup_{|x|\leq M}R(x)<\infty $, $M \in \R_+$.
\end{condition}
As the function $R$ in the above condition one can take e.g. $R(x)=a|x|+b$ for some $a \in (0,\infty)$ and $b\in(\epsilon, \infty)$.
\begin{remark}\label{rem1Comp}
Let us assume Condition \ref{condRE}. Then, using  e.g. boxes 
$g_1(x)=\{y \in \R^l: |x_i-y_i|< R(x),\ i= 1,\ldots,l\}$, 
or balls $g_1(x)=B_l(x,R(x))$, and $g_2(x)=\overline{g}_1(x)$, $x \in A$, for each $t \in T$,
under some additional regularity assumptions on $b\to \wt{f}_t(\omega,b)$, $\omega \in \Omega$
(which in the case of CGSSM and CGMSM reduce to appropriate such assumptions on $\wt{\est}_k$, $k \in \N_p$), 
$\wt{d}_t$ as above can be a result of some constrained minimization method of the respective 
$\wt{f}_t(b)$, started at $d_t$,  constrained to 
$g_2(d_t)$, and stopped in the first point $\wt{d}_t$ in which the respective requirements for CGM as above are fulfilled. 
See e.g. \cite{nocedal2006numerical,Coleman94,coleman1996interior} for some examples of such constrained 
minimization algorithms (also called minimization methods with bounds when box constraints are used). In such a case 
(and assuming that the same minimization algorithm is used for each $\omega \in \Omega$) Condition \ref{condVarck} 
typically holds and can be proved using the definition of the algorithm used. 
\end{remark}
Consider the following condition. 
\begin{condition}\label{condh} 
It holds $\delta\in \R_+$ and $h: A \to [0,\infty)$ is a twice continuously differentiable function 
such that $h(x)=0$ for $x \in \overline{B}_l(0,1)$ and $h(x)> 1$ for $|x| > 1+\delta$.
\end{condition}
An example of an easy to compute function fulfilling Condition \ref{condh} is 
$h(x)= 0$ for $|x|\leq 1$ and $h(x)=\frac{(|x|^2-1)^{3}}{\delta^{3}}$ for $|x|>1$. 
\begin{remark}\label{rem2Comp} 
Let us assume conditions \ref{condRE} and \ref{condh}, and let $\wt{f}_t$ be nonnegative, $t \in T$. 
%Let $\wt{\epsilon}_t >0$, $t \in T$, %$\wt{\est}_k$ be nonnegative,
Let $g_1(x) = B_l(x, R(x))$ and $g_2(x)=\overline{B}_l(x,R(x)(1+\delta))$, $x \in A$. 
Then, under some additional assumptions on $b\to \wt{f}_t(b,\omega)$,  $\omega \in \Omega$,
rather than using constrained minimization as in Remark \ref{rem1Comp}, 
to obtain $\wt{d}_t$ in CGM one can use some globally convergent unconstrained minimization method the following modification of $\wt{f}_{t}$
\begin{equation}\label{fbhdef}
b\to h_t(b)=\wt{f}_t(b) + \wt{f}_{t}(d_t)h(\frac{|b-d_t|}{R(d_t)}). 
\end{equation}
Such a method could start at $d_t$ and stop in the first point $\wt{d}_t$  %($\wt{a}_k$)
in which  
\begin{equation}\label{hwtdt}
h_t(\wt{d}_t)\leq h_t(d_t) 
\end{equation}
and if $\wt{d}_t \in g_1(d_t)$, then
\begin{equation}\label{nablabh}
|\nabla_b h_t(\wt{d}_t)|\leq \wt{\epsilon}_t. 
\end{equation}
Sufficient assumptions for the global convergence 
of a class of such minimization methods are given by the Zoutendjik theorem, as discussed in Remark \ref{remDetGM}.
Such assumptions are fulfilled in the above case 
if we have twice continuous differentiability of $b\to\wt{f}_t(b,\omega)$, $\omega \in \Omega$, and $h$, which is why we assumed the latter in Condition \ref{condh}. 

Let us check that the assumptions of CGM are satisfied for such constructed $\wt{d}_t$. 
If $\wt{f}_t(d_t)=0$, then from $\wt{f}_t$ being nonnegative, it holds $\nabla_b \wt{f}_t(d_t)=0$,
and thus $\wt{d}_t=d_t$ and we have (\ref{ckingck}). If 
$\wt{f}(d_t)>0$, then from (\ref{fbhdef}) and Condition \ref{condh}, 
$h_t(b)>h_t(d_t)$ for $|b-d_t|> (1+\delta)R(d_t)$, and thus from (\ref{hwtdt}) we also have (\ref{ckingck}). 
From (\ref{fbhdef}) we have $h_t(d_t)=\wt{f}_t(d_t)$
and  $\wt{f}_t(\wt{d}_t) \leq h_t(\wt{d}_t)$, so that from (\ref{hwtdt}) we have (\ref{whsmSSM}). 
Finally, if  $\wt{d}_t \in g_1(d_t)$, then from  (\ref{nablabh}) and 
$\nabla_b h_t(x)= \nabla_b \wt{f}_t(x)$, $x \in g_1(d_t)$, we have (\ref{nablawtestbSSM}). 
Similarly as in Remark \ref{rem1Comp}, Condition \ref{condVarck} typically holds for such constructed $\wt{d}_t$. 
\end{remark}

Consider the following condition, which will be useful for 
proving the asymptotic properties of minimization results 
of CGSSM, CGMSM, and some further methods. 
\begin{condition}\label{condgradck}
Almost surely for a sufficiently large $t$, (\ref{nablawtestbSSM}) holds. 
\end{condition}

The following theorem will be useful for proving the convergence properties of CGM methods. 
\begin{theorem}\label{thConstrOpt} 
Let us assume that Condition \ref{condUniffn} holds for $A=\R^l$ and $f$ which is convex and has a unique minimum point $d^*$. 
Let $\wt{d}_n\in g_2(d_n)$ be such that $f_n(\wt{d}_n)\leq f_n(d_n)$, $n \in \N_+$. Then, 
\begin{equation}\label{snb}
\lim_{n\to \infty}\wt{d}_n = d^* 
\end{equation}
and
\begin{equation}\label{fnsnb}
\lim_{n\to \infty}f_n(\wt{d}_n) = f(d^*).
\end{equation}
Let further $f$ be twice continuously differentiable with a positive definite Hessian on $A$ and
let $f_n$, $n \in \N_+$, be twice differentiable and whose $i$th derivatives for $i=1,2$, converge locally uniformly to such derivatives of 
$f$. Let $\epsilon_n \geq 0$, $n\in \N_+$, be such that $\lim_{n\to \infty} \epsilon_n=0$.  
Let for $n \in \N_+$ it hold that if $\wt{d}_n \in g_1(d_n)$ then
\begin{equation}\label{gradfnsn} 
|\nabla f_n(\wt{d}_n)| \leq \epsilon_n. 
\end{equation}
Then, for a sufficiently large $n$, (\ref{gradfnsn}) holds. 
Let further $D$ be a bounded neighbourhood of $d^*$. 
Then, for a sufficiently large $n$, $f_{n|D}$ has a unique minimum point equal to a unique  $\wt{d}_n \in D$ such that $\nabla f_n(\wt{d}_n)=0$. 
\end{theorem}
\begin{proof}
From (\ref{limdnds}) and Condition \ref{condG}, the set $\bigcup_{n\in \N_+} g_2(d_n)$ is bounded 
and so is the sequence $(\wt{d}_n)_{n\in \N_+}$. Thus, (\ref{snb}) and (\ref{fnsnb}) follow from Theorem \ref{thConvda}. 
From (\ref{limdnds}) and (\ref{klepsg}), for a sufficiently large $n$, $d_n \in \overline{B}_l(d^*,\frac{\epsilon}{2})\subset g_1(d_n)$, 
and thus (\ref{gradfnsn}) holds. The rest of the thesis follows from 
Theorem \ref{thConvStrong}, in which from Lemma \ref{lemStrConvA} as $U$ one can take any open ball with the center $d^*$,
such that $D\subset U$, and as $f$ and $f_n$ in that theorem use restrictions to $U$ of the above $f$ and $f_n$. 
\end{proof}

For CGMSM in the below examples we assume conditions \ref{condLiKi} and \ref{condLi}.
By saying that the counterparts of conditions \ref{condESSM1} 
or \ref{condESSM1f} hold for CGSSM or CGMSM
or that Condition \ref{condbkdef} holds for CGMSM, 
we mean that these conditions hold for 
$d_t$ and $\wh{f}_t$ replaced by $\wt{d}_t$ and $\wt{f}_t$,
in the counterpart of Condition \ref{condbkdef} additionally assuming that Condition \ref{condVarck} holds. 
\begin{remark}\label{remCountKiA}
Note that we have a counterpart of Remark \ref{remcondKiA} 
with $d_k$ replaced by $\wt{d}_k$ and conditions \ref{condESSM1}, \ref{condESSM1f}, 
and \ref{condbkdef} replaced by the counterparts of such conditions for CGMSM. 
\end{remark}

In the below CGSSM methods let us assume that Condition \ref{condlemYmore} holds for $S=Z^2$, while for the CGMSM methods
that Condition \ref{condSpr} holds for $S=Z^2$ (where when considering the ECM setting we mean the counterparts of these conditions), 
in the LETGS setting additionally assuming these conditions for $S=1$.

Let us now consider CGSSM or CGMSM in the LETGS setting, assuming conditions \ref{condzntau} and \ref{cond1}, 
that $b^*$ as above is a unique minimum point of $\msq$, and 
% let us consider CGSSM and CGMSM for $\wh{\est}_n=\wh{\msq}_n$ and 
that $\wt{\est}_n=\wh{\msq2}_{n}$, $n \in \N_+$. Note that in such a case the variables $d_t$ satisfying Condition \ref{condESSM1} as assumed for 
CGM above can be e.g. the results of GSSM or GMSM  respectively for $\wh{\est}_n=\wh{\msq}_n$ as in Section \ref{secminGrad}. 
%(the reasoning for $\wt{\est}_n=\wh{\var}_{n}$ is the same). 
% For CGSSM let us assume that Condition \ref{condlemYmore} holds for $S=Z^2$ and $S=1$, 
% while for CGMSM that Condition \ref{condSpr} holds for $S=Z^2$ and $S=1$. 
From conditions \ref{condESSM1}, \ref{condT}, and the fourth point of Theorem \ref{thunifDiff} for CGSSM
or the fifth point of Theorem \ref{thunifDiffMult} for CGMSM, as well as
from theorems \ref{thLETGSStrong}, \ref{thConstrOpt}, and Remark \ref{remMesUnc}, the
counterparts of conditions \ref{condESSM1} and \ref{condESSM1f} and Condition 
\ref{condgradck} hold for CGSSM and CGMSM. 

Let us now consider CGSSM or CGMSM in the ECM setting for $\wt{\est}_n=\wh{\ic}_{n}$. Let us assume Condition \ref{condECMFull} for $b^*$ as above
and that $C=1$ as discussed in Remark \ref{remECMcost}, so that $\ic =\var$ and $b^*$ is its unique minimum point. 
Note that in such a case the variables $d_t$ satisfying Condition \ref{condESSM1} as above 
can be e.g. the results of GSSM or GMSM as in Section \ref{secminGrad} respectively for $\wh{\est}_n=\wh{\msq}_n$ or $\wh{\est}_n=\wh{\msq2}_{n}$, 
for $\wh{\est}_n=\wh{\msq}_n$ additionally assuming (\ref{pq1msq}) as in that section. 
From Condition \ref{condESSM1}, %as discussed in Section \ref{secminGrad},
the fifth point of Theorem \ref{thunifDiff} for CGSSM or
the sixth point of Theorem \ref{thunifDiffMult} for CGMSM, as well as from
theorems \ref{thECMPos} and \ref{thConstrOpt}, the %, and Remark \ref{remMesUnc}, 
% Remark \ref{remdnprop}, theorems \ref{thLETGSStrong}, \ref{thConvda} and \ref{thConvStrong}, and inequality (\ref{strongEst}), 
counterparts of conditions \ref{condESSM1} and \ref{condESSM1f} 
and Condition \ref{condgradck} hold for such a CGSSM and CGMSM. 

\section{\label{secECG}Three-phase minimization of estimators with gradient-based stopping criteria and function modifications}
In this section we define minimization methods of estimators in which three-phase minimization can be used. 
% with gradient-based stopping criteria and 
% function modifications can be used. 
In their first phase one can perform some GM method as in Section \ref{secminGrad}, 
in the second a search of 
step lengths satisfying the Wolfe conditions can be carried out on a modification of the estimator considered, 
and in the third phase one can perform unconstrained minimization 
of the modified estimator using gradient-based stopping criteria. 
The single- and multi-stage versions of these methods 
shall be abbreviated as MGSSM and MGMSM respectively. 
We also discuss the possibility of the application of such methods 
to the minimization of the inefficiency constant estimators in the LETGS setting. 

MGSSM and MGMSM will be defined as special cases of the following MGM method. 
We assume in it that Condition \ref{condh} holds, $0<\alpha_1<\alpha_2<1$, and $A=\R^l$. 
Let $d^*\in A$ and for $T$ as in Section \ref{secESSM}, let $d_t$, $t \in T$, be $A$-valued random variables such that a.s. 
\begin{equation}\label{limdtd}
\lim_{t\to \infty}d_t =d^*. 
\end{equation}
Let random functions 
$\wt{f}_t:\mc{S}(A)\otimes (\Omega,\mc{F}) \to\mc{S}([0,\infty))$, $t \in T$, be such that 
$b \to \wt{f}_t(b,\omega)$ is continuously differentiable, $\omega \in \Omega$, $t\in T$. 
Let $\wt{\epsilon}_t$, $t \in T$, be as in the previous section and such that additionally a.s. 
\begin{equation}\label{wtepstozero}
\lim_{t\to \infty}\wt{\epsilon}_t=0. 
\end{equation} 
Let for each $t \in T$, $r_t$ be an $\R_+$-valued random variable, 
\begin{equation}\label{fbh2def}
h_t(b)=\wt{f}_{t}(b) + \wt{f}_{t}(d_t)h(\frac{|b-d_t|}{r_t}), \quad b \in A, 
\end{equation}
and a function $\wt{d}_t:\Omega\to A$ and an $A$-valued random variable $d_t'$  be such that 
\begin{equation}\label{wtcdef}
\wt{d}_t, d_t' \in \overline{B}_l(d_t, r_t(1+\delta)). 
\end{equation}
For each $t \in T$, let for some $[0,\infty)$-valued random variable $p_t$ it hold 
\begin{equation}\label{dtpdt}
d_t'-d_t=-p_t\nabla h_t(d_t),  
\end{equation}
and let the following inequalities hold (which are the Wolfe conditions on the step length $p_t$ 
when considering the steepest descent search direction, 
see e.g. (3.6) in \cite{nocedal2006numerical}) 
\begin{equation}\label{wolf1} 
h_t(d_t') \leq h_t(d_t)-p_t\alpha_1|\nabla h_t(d_t)|^2, 
\end{equation}
\begin{equation}\label{wolf2}
\nabla h_t(d_t')\nabla h_t(d_t) \leq \alpha_2|\nabla h_t(d_t)|^2.
\end{equation}
Finally, in MGM we assume that for each $t \in T$ 
\begin{equation}\label{whsmMGSSM}
h_t(\wt{d}_t)\leq h_t(d_t') 
\end{equation}
and 
\begin{equation}\label{nablawtestbEMSSM}
|\nabla h_t(\wt{d}_t)| \leq \wt{\epsilon}_t. 
\end{equation}
%Let $r_t$, $t \in T$, be $\R_+$-valued random variables.
Let $\wt{\est}_k$, $k\in \N_p$, as in (\ref{estnDef}) be such that  
$b\to \wt{\est}_{k}(b',b)(\omega) \in [0,\infty)$ is continuously differentiable, $b' \in \R^l$, $\omega \in \Omega_1^k$, $k \in \N_p$.
We define MGSSM and MGMSM as special cases of MGM in the same way as CGSSM and CGMSM are defined as special cases of CGM 
in Definition \ref{defCGM}.  
\begin{remark}\label{remconstrEMG}
For a given $t \in T$, assuming that the other variables and constants as above are given, 
a possible construction of $p_t$, $d_t'$, and $\wt{d}_t$ in MGM is as follows. 
% MGSSM is as follows. Let $\N_t \in \N_p$ (otherwise we can take $c_t$ and $c_t'$ to be e.g. constant).
On the event $\nabla h_t(d_t)=0$, let us set $d_t'=\wt{d}_t=d_t$. Let further $\omega \in \Omega$ be such that 
\begin{equation}\label{nablaht}
\nabla h_t(d_t(\omega),\omega)\neq 0.  
\end{equation}
Then, to obtain $p_t(\omega)$ and thus also $d_t'(\omega)$, one can perform a
line search of $h_t(\cdot,\omega)$ in the steepest descent direction
$-\nabla h_t(d_t(\omega),\omega)$, started in $d_t(\omega)$ and stopped when $p_t(\omega)$ and the corresponding $d_t'(\omega)$ (see (\ref{dtpdt}))
start to satisfy the Wolfe conditions 
(\ref{wolf1}) and (\ref{wolf2}) (evaluated on such an $\omega$). The line search can be performed e.g. using  Algorithm 3.5 from \cite{nocedal2006numerical}.
If this algorithm is used for each $\omega$ as above, then such constructed $d_t'$ is a random variable.  
Let further the variable $d_t'$ as above be given. 
Consider now $\omega\in \Omega$ satisfying (\ref{nablaht}) and 
%\begin{equation}\label{epsg0}
$\wt{\epsilon}_t(\omega)>0$.  
%\end{equation}
Then, under some additional assumptions on $\wt{f}_t(\cdot,\omega)$, 
to construct $\wt{d}_t(\omega)$
one can use some convergent unconstrained minimization algorithm of $h_t(\cdot,\omega)$ 
%the assumption of twice continuous differentiability of $h_t$),
started in $d_t'(\omega)$ and stopped in the first point $\wt{d}_t(\omega)$ in which 
(\ref{whsmMGSSM}) and (\ref{nablawtestbEMSSM}) hold.  See e.g. the assumptions of the Zoutendjik theorem
in Remark \ref{remDetGM}. %If such minimization algorithm finding $\wt{d}_t$ is a steepest descent 
%line search method with step-lengths satisfying a Wolfe condition as above, 
%then $d_t'$ and $\wt{d}_t$ can be the
%as above can be thought of out as the first step of this algorithm and . 
% See e.g. \cite{nocedal2006numerical} for a number of such algorithms applicable to $h_t$ as above under some additional
% smoothness assumptions on . 
Note that from (\ref{nablaht}) and $h_t$ being nonnegative  it holds $h_t(d_t(\omega),\omega)>0$, and 
thus $h_t(x,\omega)>h_t(d_t(\omega),\omega)$ for $|x -d_t(\omega)|>r_t(\omega)(1+\delta)$, so that 
from $h_t(\wt{d}_t(\omega),\omega)\leq h_t(d_t'(\omega),\omega)\leq h_t(d_t(\omega),\omega)$, (\ref{wtcdef}) holds. 
If $\wt{\epsilon}_t(\omega)>0$ for each $\omega$ such that (\ref{nablaht}) holds, 
and the same unconstrained minimization algorithm is used for each such $\omega$,
then from the definition of such an algorithm
it typically follows that such constructed $\wt{d}_t$ is a random variable. 
For $\omega \in \Omega$ such that we have (\ref{nablaht}) and $\wt{\epsilon}_t(\omega)= 0$,  $\wt{d}_t(\omega)$ can be e.g. some (global) minimum point of 
$h_t(\cdot,\omega)$.  
\end{remark}

For $x \in \R^l$ and $B \subset \R^l$, let us denote 
\begin{equation}
\dist(x,B)=\inf_{y \in B}|x-y|.
\end{equation}
The following theorems will be useful for proving the convergence properties of MGM methods. 

\begin{theorem}\label{thfnddb}
Let $K\subset \R^l$ be nonempty and compact and let $g_n:K\to \R$, $n\in \N_+$, converge uniformly to a continuous function 
$g:K\to \R$. Let $m$ be the minimum of $g$ and $B$ be its set of minimum points. 
Then, for each sequence of points $d_n\in K$, $n \in \N_+$, such that 
%\begin{equation}
$\lim_{n\to \infty}g_n(d_n)= m$,  
%\end{equation}
we have 
%\begin{equation}
$\lim_{n\to \infty}\dist(d_n,B)=0$. 
%\end{equation}
\end{theorem}
\begin{proof}
Let $\epsilon \in \R_+$. From the continuity of $x \in K \to \dist(x,B)$, 
$K_2:=\{x \in K:\dist(x,B)\geq \epsilon\}$ is a closed subset of $K$ and thus it is compact. 
From $g$ being continuous, it attains its infimum 
$w:=\inf_{x\in K_2}g(x)$ on $K_2$, and thus we must have $\delta:=w-m>0$. 
For sufficiently large $n$ for which $|g_n(x)-g(x)|< \frac{\delta}{2}$, $x \in K$, and $|g_n(d_n)- m|< \frac{\delta}{2}$, 
we have $|g(d_n)-m|\leq |g(d_n)-g_n(d_n)|+ |g_n(d_n)-m|< \delta$ and thus $d_n \notin K_2$, i.e. $\dist(d_n,B)<\epsilon$. 
\end{proof}

\begin{theorem}\label{thECGSM}
Let $f:\R^l\to $ and $f_n:\R^l\rightarrow \R$, $n \in \N_+$, be continuously differentiable and 
such that $f_n \overset{loc}{\rightrightarrows}f$ and $\nabla f_n \overset{loc}{\rightrightarrows}\nabla f$. 
Let further for some $d^* \in \R^l$, $s\in \R_+$, and $0<w<r<\infty$ it hold 
\begin{equation} \label{fdbddelta}
f(b)\geq f(d^*)+s, \quad b \in \R^l,\ |b-d^*|\geq w, 
\end{equation}
and let $r_n \in \R_+$, $n \in \N_+$,  be such that $\lim_{n\to \infty} r_n=r$. 
Let for a sequence $d_n \in \R^l$, $n \in \N_+$, it hold
\begin{equation}
\lim_{n\rightarrow\infty}d_n=d^*. 
\end{equation}
Let for each $n \in \N_+$, $h_n :\R^l \to \R$ be such that
\begin{equation}\label{hnbfnb}
h_n(b)=f_n(b) + f_n(d_n)h(\frac{|b-d_n|}{r_n}), \quad b \in \R^l.
\end{equation} 
Let $\epsilon_n \geq 0$, $n \in \N_+$, be such that 
$\lim_{n\to \infty}\epsilon_n=0$, and
let for each $n \in \N_+$, for some $p_n \in [0,\infty)$, points
$d_n', \wt{d}_n \in B_l(d_n,(1+\delta)r_n)$ be such that 
\begin{equation}\label{dnpdn}
d_n' -d_n=-p_n\nabla h_n(d_n),  
\end{equation}
\begin{equation}\label{wolfh1}
h_n(d_n')\leq f_n(d_n) - p_n\alpha_1|\nabla f_n(d_n)|^2,
\end{equation}
\begin{equation}\label{wolfh2}
\nabla h_n(d_n')\nabla f_n(d_n) \leq \alpha_2|\nabla f_n(d_n)|^2,
\end{equation}
\begin{equation}\label{dleqd}
h_n(\wt{d}_n)\leq h_n(d_n'),
\end{equation}
and 
\begin{equation}\label{nablahneps}
|\nabla h_n(\wt{d}_n)| \leq \epsilon_n.
\end{equation}
Then, for a sufficiently large $n$ we have
\begin{equation}\label{dnpdninbl}
d_n', \wt{d}_n \in B_l(d^*, w)
\end{equation}
and
\begin{equation}\label{nablafnwtdn}
|\nabla f_n(\wt{d}_n)| \leq \epsilon_n. 
\end{equation}
%  Furthermore,
%  \begin{equation}\label{limfnwtdn}
%  \lim_{n\to \infty} (f_n(d_n')- f(d_n'))=\lim_{n\to \infty} (f_n(\wt{d}_n)- f(\wt{d}_n))=0. 
%  \end{equation}
% (\ref{limfnwtdn} \ref{lemConvUnif}
Let further
\begin{equation}\label{sudef}
\phi(u)=\nabla f(d^*-u\nabla f(d^*))\nabla f(d^*) - \alpha_2|\nabla f(d^*)|^2,\quad u \in \R,  
\end{equation}
and $v= \inf\{u\geq 0:\phi(u)=0\}$. Then, if $|\nabla f (d^*)|= 0$ then  $v=0$
and if $|\nabla f (d^*)|\neq 0$ then $v \in (0,w)$. Furthermore, for 
\begin{equation}
\mu= f(d^*) -v \alpha_1|\nabla f(d^*)|^2
\end{equation}
we have
\begin{equation}\label{limsupfwtdp}
\limsup_{n\to \infty} f(d_n')\leq \mu,
\end{equation}
\begin{equation}\label{limsupfwtd}
\limsup_{n\to \infty} f(\wt{d}_n)\leq \mu,
\end{equation}
for 
%\begin{equation}
$E = \{x \in \R^l: f(x)\leq \mu\}$,
%\end{equation}
we have
\begin{equation}\label{dnsnbe}
\lim_{n\to \infty}\dist(d_n',E)= 0, 
\end{equation}
and for the set 
%\begin{equation}\label{Bdef}
$D=\{x \in \R^l: \nabla f(x)=0,\ f(x)\leq \mu\}\subset E$, 
%\end{equation}
containing the nonempty set of minimum points of $f$, we have
\begin{equation}\label{dnsnb}
\lim_{n\to \infty}\dist(\wt{d}_n,D) =0. 
\end{equation}
% and if $\nabla f(d^*) \neq 0$, then for some $p >0$, for a sufficiently large $n$, 
% $f(d_n')\leq f(d^*)-p$ and $f(\wt{d}_n)\leq f(d^*)-p$. 
% In particular, for some $p> 0$, for $m:= f(d^*) -p\I(\nabla f(d^*)\neq 0)$, 
% $\limsup_{n\to \infty} f(d_n')\leq m $ 
\end{theorem}
\begin{proof}
Let $K= \overline{B}_l(d^*, w)$. Let $N_1\in \N_+$ be such that for $n\geq N_1$,
$|d_n-d^*|\leq \frac{r-w}{2}$ and $r-r_n\leq \frac{r-w}{2}$, in which case 
for $x \in K$ we have %$|x-d^*|\leq w+(\frac{r-w}{2} - (r-r_n)$
\begin{equation}
|x-d_n|\leq |d^*-d_n|+ |x-d^*|\leq \frac{r-w}{2} +w \leq \frac{r-w}{2} +w +(\frac{r-w}{2} - (r-r_n))= r_n,
\end{equation}
so that 
\begin{equation}\label{kldtkldr}
K\subset \overline{B}_l(d_n,r_n). 
\end{equation}
From the set $F:=\bigcup_{n=1}^{\infty} B_l(d_n,(1+\delta)r_n)$ being bounded, let $N_2 \geq N_1$ be such that for $n\geq N_2$ 
\begin{equation}\label{fnxfxd3}
|f_n(x)-f(x)| \leq \frac{s}{2},\quad x \in F.
\end{equation}
From Lemma \ref{lemConvUnif},
\begin{equation}\label{limfndnfds}
\lim_{n\to\infty}f_n(d_n)=f(d^*), 
\end{equation}
and thus let $N_3 \geq N_2$ be such that for $n\geq N_3$ 
\begin{equation}\label{fndnfd}
|f_n(d_n)-f(d^*)| < \frac{s}{2}. 
\end{equation}
Then, for $n\geq N_3$ and $x \in F$ such that $|x-d^*| \geq w$, we have 
\begin{equation}
h_n(x)\geq f_n(x)\geq f(x)-\frac{s}{2}\geq f(d^*) +\frac{s}{2}> f_n(d_n), 
\end{equation} 
where in the first inequality we used (\ref{hnbfnb}) and Condition \ref{condh}, in the second (\ref{fnxfxd3}), in the third 
(\ref{fdbddelta}), and in the last (\ref{fndnfd}). Thus, since from (\ref{wolfh1}) and (\ref{dleqd}), 
\begin{equation}\label{leqhns} 
h_n(\wt{d}_n)\leq h_n(d_n')\leq f_n(d_n), 
\end{equation} 
for $n\geq N_3$ we have (\ref{dnpdninbl}). 
%(\ref{limfnwtdn} \ref{lemConvUnif} 
% d_n', \wt{d}_n \in B_l(d^*, w). 
% \end{equation} 
For $n \geq N_3$, from (\ref{kldtkldr}) and (\ref{dnpdninbl}), we have 
$h_n(d_n')=f_n(d_n')$, $\nabla h_n(d_n')=\nabla f_n(d_n')$, and similarly for 
$d_n'$ replaced by $\wt{d}_n$, so that from (\ref{nablahneps}),  (\ref{nablafnwtdn}) holds, and from (\ref{leqhns}), 
\begin{equation}\label{leqfns} 
f_n(\wt{d}_n)\leq f_n(d_n')\leq f_n(d_n). 
\end{equation} 
%(\ref{limfnwtdn}) follows from (\ref{dnpdninbl}) and Lemma \ref{lemConvUnif}. 
If $\nabla f(d^*)=0$, then $v=0$, in which case (\ref{limsupfwtdp}) and (\ref{limsupfwtd}) follow 
from (\ref{leqfns}), the sequences $(d_n')_{n\in\N_+}$ and $(\wt{d}_n)_{n\in\N_+}$ being bounded, and Theorem \ref{thConvda}. 
Let now $\nabla f(d^*)\neq 0$. Then, $\phi(0)=(1-\alpha_2)|\nabla f(d^*)|^2>0$ 
and thus from the continuity of $\phi$, $v>0$. The fact that 
$v <w$ follows from (\ref{fdbddelta}) and Lemma 3.1 in \cite{nocedal2006numerical} 
about the existence of steps $u>0$ satisfying the Wolfe conditions: $\phi(u) \leq 0$ and
$f(d^* - \nabla f(d^*)u)\leq f(d^*) - u\alpha_1|\nabla f(d^*)|^2$. 
Let $0<v'<v$. Then, from the continuity of $\phi$, 
\begin{equation}\label{infsu} 
\inf_{0\leq u\leq v'} \phi(u)>0. 
\end{equation}
For $n \in \N_+$, and $u \in \R$, let
$\phi_n(u)=\nabla f_n(d_n-\nabla f_n(d_n)u)\nabla f(d^*)- \alpha_2|\nabla f(d^*)|^2$. 
Since from Lemma \ref{lemConvUnif}
\begin{equation}\label{limnablafndn}
\lim_{n\to \infty }\nabla f_n(d_n)=\nabla f(d^*),  
\end{equation}
the function $u\to d_n-u\nabla f_n(d_n)$ converges uniformly to $u\to d^* - u\nabla f(d^*)$ on $[0,v']$, and thus 
from Theorem \ref{thCompact}, 
$\phi_n$ converges to $\phi$ uniformly on $[0,v']$. Thus, from (\ref{infsu}), let $N_4 \geq N_3$, be such that for $n\geq N_4$, 
$\inf_{u \in [0,v']} \phi_n(u)>0$. For such an $n$, from (\ref{dnpdn}) and (\ref{wolfh2}) it must hold $p_n > v'$ and from (\ref{wolfh1}) we have
\begin{equation}
f_n(d_n')\leq f_n(d_n) - v'\alpha_1|\nabla f_n(d_n)|^2,
\end{equation}
and thus 
\begin{equation}\label{fwtdnleq}
f(d_n')\leq f(d_n') - f_n(d_n') + f_n(d_n) - v'\alpha_1|\nabla f_n(d_n)|^2. 
\end{equation}
From (\ref{dnpdninbl}) and the fact that $f_n$ converges to $f$ uniformly on $K$, we have 
$\lim_{n\to\infty}(f(d_n') - f_n(d_n'))=0$. Thus, from (\ref{fwtdnleq}), (\ref{limfndnfds}), and 
(\ref{limnablafndn}), 
\begin{equation}
\limsup_{n\to \infty} f(d_n')\leq f(d^*) - v'\alpha_1|\nabla f(d^*)|^2.
\end{equation}
Since this holds for each $v'<v$, we have (\ref{limsupfwtdp}), and from (\ref{leqfns}), we also have (\ref{limsupfwtd}). 
Due to (\ref{fdbddelta}), $f$ attains a minimum. For each minimum point $x_0$ of $f$ we have $\nabla f(x_0)=0$
and from (\ref{limsupfwtdp}) and $f(x_0)\leq f(d_n')$, $n \in \N_+$, we have $f(x_0) \leq \mu$. Thus, $x_0 \in D$.
% from $\phi(u)\geq 0$, $u \in [0,v]$
% \begin{equation}
% f(x_0) \leq f(d^*-v\nabla f(d^*))= f(d^*)-\int_{0}^v\! \nabla f(d^*)\nabla f(d^*-t\nabla f(d^*))\, \mathrm{d}t
% \leq \alpha_2  
% \end{equation}
The minimum of $g:= (f\vee\mu)_{|K}$ is equal to $\mu$ and $E\subset K$ is its set of minimum points. 
From (\ref{limsupfwtdp}), we have $\lim_{n\to \infty}f(d_n')\vee \mu = \mu$.  
Thus, from (\ref{dnpdninbl}) and Theorem \ref{thfnddb} for such a $g$ and $g_n= g$, $n \in \N_+$,  we receive (\ref{dnsnbe}). 
The minimum of $g:=(|\nabla f| + f\vee\mu)_{|K}$ is $\mu$ and its set of minimum points is $D\subset K$. 
From $\epsilon_n\to 0$, (\ref{nablafnwtdn}), and (\ref{limsupfwtd}), 
$\lim_{n\to\infty}(|\nabla f_n(\wt{d}_n)| + f(\wt{d}_n)\vee \mu) = \mu$. Thus, from (\ref{dnpdninbl}) and
Theorem \ref{thfnddb} for such a $g$ and $g_n=(|\nabla f_n| + f\vee\mu)_{|K}$, $n \in \N_+$, we receive (\ref{dnsnb}).
\end{proof}

Let us now discuss how MGSSM and MGMSM can be applied in the LETGS setting for $\wt{\est}_n=\wh{\ic}_n$, $n \in \N_p$, for $p=2$. 
We assume conditions \ref{cond1},  \ref{condcgegc0}, and 
Condition \ref{condlemYmore} for $S=Z^2$. Then, from Theorem \ref{thLETGSStrong}, 
$\msq$ has a unique minimum point $d^*$. The variables $d_t$, $t\in T$, such that (\ref{limdtd}) holds a.s. for such a $d^*$, can be obtained e.g. 
using GSSM or GMSM methods respectively for $\wh{\est}_n=\wh{\msq}_n$ as in Section \ref{secminGrad}. Furthermore,
for a positive definite matrix $M$ and its lowest eigenvalue $m>0$ 
as in Theorem \ref{thLETGSStrong}, %for which $\msq$ is strongly convex with constant $m$,
we have from (\ref{strongfb}) that 
\begin{equation}
\var(d^*+b) \geq \var(d^*)+\frac{m}{2}|b|^2,\quad b \in \R^l, 
\end{equation}
and thus 
\begin{equation}\label{icneq}
\ic(d^*+b) \geq c_{min}(\var(d^*)+\frac{m}{2}|b|^2),\quad b \in \R^l.
\end{equation}
For some $\sigma_1,\sigma_2\in \R_+$, $\sigma_1<\sigma_2$, let us define 
\begin{equation}
r=\sqrt{\frac{2}{m}\left(\frac{\ic(d^*)}{c_{min}}-\var(d^*)\right)+\sigma_2}
\end{equation}
and 
\begin{equation}
w=\sqrt{\frac{2}{m}\left(\frac{\ic(d^*)}{c_{min}}-\var(d^*)\right)+\sigma_1}.
\end{equation}
It holds $r>w>0$ and from (\ref{icneq}), for $b \in \R^l$, $|b|\geq w$,
\begin{equation}
\ic(d^*+b) \geq c_{min}(\var(d^*)+\frac{mw^2}{2})= \ic(d^*)+\frac{m}{2}c_{min}\sigma_1,
\end{equation}
so that we have (\ref{fdbddelta}) for $f=\ic$ and $s = \frac{m}{2}c_{min}\sigma_1$. 

Let us assume that Condition \ref{condlemYmore} holds for $S=C$ (in addition to this condition holding for $S=Z^2$ as assumed above), 
so that from the fourth point of Theorem \ref{thDiff}, $\ic$ is smooth. 
Let $\mu$, $E$, and $D$ be as in Theorem \ref{thECGSM} for $f=\ic$ and $d^*$ as above. 
Note that we have $\mu < \ic(d^*)$ only if $\nabla \ic(d^*)\neq 0$, which from 
Remark \ref{remicvar} holds only if $\var(d^*)\neq 0$ and $\nabla c(d^*)\neq 0$.
Let for $n \in \N_+$ and $b \in \R^l$
\begin{equation}
\wh{M}_n(b)=\overline{\left(2L(b)GZ^2\exp(-\frac{1}{2}\sum_{i=1}^{\tau}|\eta_i|^2)\right)}_n,
\end{equation}
and let $\wh{m}_n(b)=m_l(\wh{M}_n(b))$ for $m_l$ as in Section \ref{secSomeCondLETGS}, i.e. $\wh{m}_n(b)$ is the lowest eigenvalue of $\wh{M}_n(b)$. 
%For some measurable $d_k:\Omega_1^n\to \R^l$,  
For $n \in \N_p$, and $b,d \in \R^l$, 
let us define $\wh{r}_n(b,d):\Omega_1^n\to \R$ to be such that for $\omega \in \Omega_1^n$ for which 
$\wh{m}_n(b)(\omega)>0$ and $\wh{\ic}(b,d)(\omega)-c_{min}\wh{\var}_n(b,d)(\omega)>0$
\begin{equation}
\wh{r}_n(b,d)(\omega)=\sqrt{\frac{2}{\wh{m}_n(b)(\omega)}\left(\frac{\wh{\ic}_n(b,d)(\omega)}{c_{min}}-\wh{\var}_n(b,d)(\omega)\right)+\sigma_2}, 
%r=\sqrt{\frac{2}{m}\left(\frac{\ic(d^*)}{c_{min}}-\var(d^*)\right)+\sigma_2}
\end{equation}
and otherwise $\wh{r}_n(b,d)(\omega)=a$ for some $a \in \R_+$.

Let us now focus on MGSSM, for which let us assume Condition \ref{condlemYmore} for $S=1$ 
and that $r_t=\I(N_t=k\in \N_p)\wh{r}_{k}(b',d_t)(\wt{\kappa}_k)$, $t\in T$. 
Then, from Theorem \ref{thunifDiff}, a.s. $b\to\wh{\var}_n(b',b)(\wt{\kappa}_n)\overset{loc}{\rightrightarrows}\var$ 
and $b\to\wh{\ic}_n(b',b)(\wt{\kappa}_n)\overset{loc}{\rightrightarrows}\ic$. 
%Thus, from the fact that for MGSSM, from the discussion in Section \ref{secminGrad} for GSSM, a.s. $d_t\to d^*$, 
Thus, from Lemma \ref{lemConvUnif} and Condition \ref{condT}, we have a.s. $\I(N_t=k\in \N_p)\wh{\var}_{k}(b',d_t)(\wt{\kappa}_k)\to \var(d^*)$ 
and $\I(N_t=k\in \N_p)\wh{\ic}_{k}(b',d_t)(\wt{\kappa}_k)\to \ic(d^*)$. 
Furthermore, from the SLLN, a.s. $\wh{M}_n(b')(\wt{\kappa}_n)\rightarrow M$ and thus from Lemma \ref{lemLipschMin}, 
$\wh{m}_n(b')(\wt{\kappa}_n)\to m$. Therefore, a.s. 
$\lim_{t\to \infty}r_t = r$. Thus, from Theorem \ref{thECGSM} 
and Remark \ref{remMesUnc} we receive that the following condition holds for MGSSM. 
\begin{condition}\label{condMGM}
Condition \ref{condgradck} holds, a.s. 
%\begin{equation}\label{wtdmu}
$\limsup_{t\to \infty}\ic(d_t')\leq \mu$,  
%\end{equation}
$\limsup_{t\to \infty}\ic(\wt{d}_t)\leq \mu$, $\lim_{t\to \infty} \dist(d_t',E)=0$, and
\begin{equation}\label{limdistcnD}
\lim_{t\to \infty} \dist(\wt{d}_t,D)=0. 
\end{equation}
%\end{equation}
Furthermore, a.s. for a sufficiently large $t$, $d_t',\wt{d}_t \in B_l(d^*, w)$.  
\end{condition}

For MGMSM let us assume that conditions \ref{condLiKi} and \ref{condLi} hold, that 
Condition \ref{condSpr} holds for $S=Z^2$, $S=C$, and $S=1$, and that
$r_k=\wh{r}_{n_k}(b_{k-1},d_k)(\wt{\chi}_k)$, $k \in \N_+$. 
%\lim_{k\to \infty}\wh{r}_k(b_{k-1},d_k)(\wt{\chi}_k) = r$.  
From Theorem \ref{thunifDiffMult}, a.s. $b\to\wh{\var}_{n_k}(b_{k-1},b)(\wt{\chi}_k)\overset{loc}{\rightrightarrows}\var$ and 
$b\to\wh{\ic}_{n_k}(b_{k-1},b)(\wt{\chi}_k)\overset{loc}{\rightrightarrows}\ic$. 
From H\"{o}lder's inequality and Theorem \ref{thzb} it easily follows that Condition
\ref{condSpr} holds for $S$ equal to the different entries of $\wh{M}_1(0)$.  
Thus, from the first point of Theorem \ref{thunifDiffMult} a.s. 
$\wh{M}_{n_k}(b_{k-1})(\wt{\chi}_k)\rightarrow M$, and thus $\wh{m}_{n_k}(b_{k-1})(\wt{\chi}_k)\to m$. Therefore, we have  
a.s. $\lim_{k\to \infty}r_k=r$. 
%\begin{equation}
Thus, from Theorem \ref{thECGSM} and Remark \ref{remMesUnc} it follows that Condition 
\ref{condMGM} holds for MGMSM.
% that a.s. $\limsup_{n\to \infty}\ic(\wt{a}_n')\leq \mu$,
% $\limsup_{n\to \infty}\ic(\wt{a}_n)\leq \mu$, 
% \begin{equation}\label{limdistanE}
% \lim_{n\to \infty} \dist(\wt{a}_n',E)=0, 
% \end{equation}
% % \begin{equation}
% % \end{equation}
% \begin{equation}\label{limdistanD}
% \lim_{n\to \infty} \dist(\wt{a}_n,D)=0, 
% \end{equation}
% a.s. for a sufficiently large $n$, $a_n'$ and $a_n$ are in $B_l(d^*, w)$, and Condition \ref{condgradck} holds.

\begin{theorem}\label{thicvar}
Let functions $\var$, $c$, and $\ic$ be as in Section \ref{secCoeffDiv} for $A$ open,
let $\var$ be lower semicontinuous and convex 
and have a unique minimum point $b^*\in A$, and let $\ic$ be continuous in $b^*$.
Let for some $d_n \in A$, $n \in \N_+$, 
\begin{equation}\label{limsupic}
\limsup_{n\to \infty}\ic(d_n) < \ic(b^*).
\end{equation}
Then,
\begin{equation}\label{liminfvar}
\liminf_{n\to \infty}\var(d_n)> \var(b^*)
\end{equation}
and  
\begin{equation}\label{limsupc}
\limsup_{n\to \infty}c(d_n)< c(b^*). 
\end{equation}
\end{theorem}
\begin{proof}
For some $\ic(b^*)>s> \limsup_{n\to \infty}\ic(d_n)$, let $\epsilon\in\R_+$ be such that %$B_l(b^*,\epsilon)$ be a ball such that 
$\ic(b)> s$ for $b\in B_l(b^*,\epsilon)\subset A$. Then, $d_n \in A\setminus B_l(b^*,\epsilon)$
for a sufficiently large $n$. 
From the semicontinuity of $\var$, for some $b_0 \in S_l(b^*,\epsilon)$, $\var(b_0)=\min_{b\in S_l(b^*,\epsilon)}\var(b)>\var(b^*)$, and thus 
from the convexity of $\var$ it holds $\var(b)\geq\var(b_0)$ for $b\in A\setminus B_l(b^*,\epsilon)$, 
and we have (\ref{liminfvar}). Note that from (\ref{limsupic}), $\ic(b^*)>0$ and thus $\var(b^*)>0$. Therefore,
\begin{equation}
\limsup_{n\to \infty}c(d_n)\leq \frac{\limsup_{n\to \infty}\ic(d_n)}{\liminf_{n\to \infty} \var(d_n)}< \frac{\ic(b^*)}{\var(b^*)}=c(b^*). 
\end{equation}
\end{proof}

If $\nabla \ic(d^*)\neq 0$, so that $\mu < \ic(d^*)$, then 
for $c_t=\wt{d}_t$ or $c_t=d_t'$ as above for which we have a.s. $\limsup_{t\to \infty}\ic(c_t)\leq \mu$, 
from Theorem \ref{thicvar} it also holds a.s.  $\liminf_{t\to \infty}\var(c_t)> \var(d^*)$ and $\limsup_{t\to \infty}c(c_t)< c(d^*)$. 

\begin{condition}\label{condbbsrl}
$D=\{b^*\}$ for some $b^* \in \R^l$. 
\end{condition}
\begin{remark}\label{remCasesconbbsrl}
Note that Condition \ref{condbbsrl} holds under the assumptions as above e.g. if $C$ is a positive constant or if $\var(d^*)=0$, and in both these cases $b^*=d^*$.  
\end{remark}
\begin{remark}\label{remLimIC}
%Note that Condition \ref{condbbsrl} holds under assumptions as above e.g. if $C$ is a positive constant or if $\var(d^*)=0$, and in both these cases $b^*=d^*$.  
Let us assume Condition \ref{condbbsrl}. 
Then, $b^*$ is the unique minimum point of $\ic$ as above. 
Furthermore, under the above assumptions for MGSSM and MGMSM, 
from (\ref{limdistcnD}) 
and Lemma \ref{lemConvUnif}, counterparts of conditions 
\ref{condESSM1} and \ref{condESSM1f} hold in these methods
(by which we mean the same as above Remark \ref{remCountKiA}).
\end{remark}
Note that Remark \ref{remCountKiA} applies also to MGMSM.  

\section{\label{secCompFirst}Comparing the first-order asymptotic efficiency of minimization methods} 
Let $A \in \mc{B}(\R^l)$ be nonempty and $T \subset \R_+$ be unbounded. Consider 
a function $\phi:\mc{S}(A)\to \mc{S}(\overline{\R})$ and an $A$-valued stochastic process $d=(d_t)_{t \in T}$. For $t \in T$, 
$d_t$ can be an adaptive random parameter trying to minimize $\phi$ for $t$ being e.g. the simulation budget, 
the total number of steps in SSM methods, or the number of stages or simulations in MSM methods, used to compute $d_t$. 
We describe some such possibilities in more detail in the below remark. 

\begin{remark}\label{remdk} 
 In the various SSM methods as in the previous sections, for some $T$ as in Condition \ref{condT}, 
 we can consider $d_t$ equal to $d_t$, $\wt{d}_t$, or $d_t'$ as in these methods, $t \in T$ (see Remark \ref{remCondT}). 
 Let us further consider the case of the various MSM methods as in the previous sections. % one can e.g. take $T=\N_+$ and as $d_k$ use 
 Then, for variables $p_k$ equal to 
 $d_k$, $\wt{d}_k$, or $d_k'$ as in these methods, $k \in \N_+$, %, or $b_k$, %. Alternatively, for some such $p_k$, $k \in \N_+$ 
 for some $\N\cup\{\infty\}$-valued random variables $N_t$, $t \in T$, and some $A$-valued random variables $p_0$ and $p_\infty$, one can set 
 \begin{equation}\label{ptau}
 d_t=p_{N_t}\I(N_t \neq \infty) +p_\infty\I(N_t=\infty), \quad t\in T. 
 \end{equation}
 The simplest choice would be to take $T=\N_+$ and $N_k=k$, so that $d_k=p_k$, $k\in \N_+$, i.e. $k$ is the number of stages of MSM in which 
 $d_k$ is computed. If we want $t\in T$ to correspond 
 to the number of samples generated to compute $d_t$, 
 then for $s_k:=\sum_{i=1}^kn_i$, $k \in \N_+$, and $T=\{s_k,k\in\N_+\}$, we can take $N_{s_k}=k$, $k \in \N_+$. 
 Alternatively, we can take $T=\R_+$ and for each $t\in T$, 
 $N_t$ can be the smallest number of stages using the simulation budget $t$, 
 or the highest such number before we exceed that budget.  
 Let us discuss how one can model this. 
For some $[0,\infty)$-valued random variables $M_i$ modelling the costs of the minimization algorithms 
in the $i$th stage of MSM (we can set $M_i=0$ if we do not want to consider them), $i \in \N_+$, 
and $U$ being a theoretical cost variable analogous as of a step of SSM in Remark \ref{remCondT},
under Condition \ref{condChi} we can take e.g. 
\begin{equation}\label{taucpinf} 
N_t=\inf\{k \in \N:\sum_{i=1}^{k}(M_i+\sum_{j=1}^{n_i}U(\chi_{i,j})) \geq t\} 
\end{equation} 
or 
\begin{equation}\label{taucpsup} 
N_t=\sup\{k \in \N:\sum_{i=1}^{k}(M_i+\sum_{j=1}^{n_i}U(\chi_{i,j})) \leq t\}. 
\end{equation} 
%Let us denote $d$ as in (\ref{ptau}) for $\tau$ replaced by some such  $\tau_t$ as $a_t$. 
Note that if we have (\ref{ptau}) and a.s. (\ref{ntoinfty}) and (\ref{nt0})
%  \begin{equation}\label{tauslinf} 
%  N_t<\infty,\quad t \in \R_+, 
%  \end{equation} 
%  and 
%  \begin{equation}\label{taustoinf} 
%  \lim_{t \to \infty}N_t = \infty, 
%  \end{equation} 
and one of the following holds: a.s. $p_k \to b^*$, for some $f:A\to\R$ a.s. $f(p_k)\to f(b^*)$, or 
for some $m \in \R$, a.s. $\limsup_{k\to \infty} f(p_k)\leq m$, 
then we have respectively that a.s. $d_t \to b^*$ (compare with Remark \ref{remtaua}), 
$f(d_t) \to f(b^*)$, or $\limsup_{t\to \infty}f(d_t) \leq m$. 
% For different $N_t$ corresponding to SSM methods as above, %in Remark \ref{remdk}, 
% sufficient assumptions for (\ref{nt0}) and (\ref{ntoinfty}) to hold a.s. were 
% discussed in Chapter \ref{secIneff}. 
For $N_t$ as in (\ref{taucpinf}) or (\ref{taucpsup}), (\ref{ntoinfty}) holds a.s. if $U < \infty$, $\PQ(b)$ a.s., $b \in A$. Furthermore,  
(\ref{nt0}) holds a.s. if Condition \ref{condcgegCmin} holds for $C=U$,
or from Theorem \ref{thSLLNh}, if for some $K \in \mc{B}(A)$ such that $\inf_{b \in K} \E_{\PQ(b)}(U)>0$, 
Condition \ref{condKi} holds for
$K_i=K$, $i \in \N_+$, and the assumptions of Theorem \ref{thSLLNS} hold for $g(v,x)=U(x)$, $v\in A$, $x\in \Omega_1$.  
\end{remark}

For each $\overline{\R}$-valued stochastic process
$b=(b_t)_{t\in T}$, let us denote $\sigma_{-}(b)= \sup\{x \in \R: \lim_{t\to \infty} \PR(b_t>x)=1\}$
% $A_-(b)=\{x \in \R: \lim_{t\to \infty} \PR(b_t>x)=1\}$, 
% $A_+(b)=\{x \in \R: \lim_{t\to \infty} \PR(b_t<x)=1\}$,  
% $\sigma_{-}(b)= \sup A_-(b)$, %\{x \in \R: \lim_{t\to \infty} \PR(c_t<x)=0\}$ 
and $\sigma_+(b)= \inf \{x \in \R: \lim_{t\to \infty} \PR(b_t<x)=1\} = -\sigma_-(-b)$. 
Note that $\sigma_-(b)\leq \sigma_+(b)$ 
and $\sigma_-(b)=\sigma_+(b)=x\in\overline{\R}$ only if $b_t \overset{p}{\to}x$.
For $b'$ analogous as $b$ we have $\sigma_-(b'-b)\geq \sigma_-(b')-\sigma_+(b)$.
In particular, for each $\delta \in \R$ such that
$\sigma_-(b')-\sigma_+(b)>\delta$, we have $\lim_{t\to \infty}\PR(b'-b >\delta)= 1$,
and such a $\delta$ can be chosen positive if $\sigma_-(b')>\sigma_+(b)$. 

For $d=(d_t)_{t\in T}$ as above, let us denote $\phi(d)=(\phi(d_t))_{t\in T}$. 
Let $d'$ be analogous as $d$. 
We shall call $d$ asymptotically not less efficient 
than $d'$ for the minimization of $\phi$ %(i.e. in the sense of considering $a_t$ and $a_t'$ for $t\to \infty$))
if $\sigma_+(\phi(d))\leq \sigma_-(\phi(d'))$, and (asymptotically) more efficient for this purpose
if this inequality is strict. If $\phi(d_t)$ and $\phi(d_t')$ both converge in probability to the same real number, 
then $d$ and $d'$ shall be called equally efficient.
% As we shall discuss in Section \ref{secTwo},
% such efficiency relations are well-suited for comparing asymptotic efficiency of processes 
% for the minimization of variance or inefficiency constant in the first stage of two-stage estimation methods. 
We call such defined relations the first-order asymptotic efficiency relations, to distinguish them 
from such second-order relations which will be defined in Section \ref{secSecond}.

For instance, 
for some $d=(d_t)_{t\in T}$ as above, which can be some parameters corresponding to the single- or multi-stage minimization of some mean square estimators 
as in the above remark, and $d^*$ 
being the unique minimum point of mean square, let
it hold  a.s. $d_t \to d^*$, and thus assuming further that $\ic$ is continuous in $d^*$, also a.s. 
$\ic(d_t)\to\ic(d^*)$. Let further for
$(d_t')_{t\in T}$, which can be some parameters corresponding to the minimization of the inefficiency constant estimators as in the above remark,
it hold for $\mu$ as in Section \ref{secECG} (for which $\mu \leq \ic(d^*)$
and if $\nabla\ic(d^*)\neq 0$, then $\mu<\ic(d^*)$), that a.s. $\limsup_{t\to \infty}\ic(d_t')\leq \mu$.
Then, $d'$ is asymptotically not less efficient for the minimization  of $\ic$
than $d$ and more efficient if $\nabla \ic(d^*)\neq 0$. 
% For the MSM, under appropriate conditions as in Remark \ref{remdk}  
% we also have such strong limits and asymptotic
% efficiency relations for the corresponding 
% budget-dependent parameters. % $d_t$ and $d_t'$, $t \in \R$.

Let now some 
$d=(d_t)_{t \in T}$ as above, which can be some parameters corresponding to the single- or multi-stage minimization of the cross-entropy 
estimators as in the above remark, fulfill a.s. 
$d_t \to p^*$ for $p^*$ being the unique minimum point of 
the cross-entropy. Let further $d'=(d_t')_{t \in T}$, which can correspond to the minimization of mean square
or inefficiency constant estimators for $C=1$, 
fulfill a.s. $d_t\to b^*$ for 
$b^*$ being the unique minimum point of $\msq$, which is continuous in $b^*$ and convex on the set on which it is finite.
%(for which to hold for minimization inefficiency constant estimator one should assume that $C$ is constant).
Then, $d'$ is asymptotically not less efficient than $d$
for the minimization of $\msq$, and more efficient if $b^*\neq p^*$. 

\section{\label{secFindOpt}Finding exactly a zero- or optimal-variance IS parameter} 
In this section we describe situations in which a.s. for a sufficiently large $t$, the minimization results $d_t$ of our new estimators 
are equal to a zero- or optimal-variance IS parameter $b^*$ as in Definition \ref{defParamZero}. 
When proving that this holds in the below examples  we shall impose an assumption that we can find the minimum or critical points 
of these estimators exactly. Even though such an assumption is unrealistic when minimization is performed 
using some iterative numerical methods, it is a frequent idealisation used to simplify the convergence analysis 
of stochastic counterpart methods, see e.g. \cite{Jourdain2009,Shapiro2003,KimH07}. 
Note also that such an assumption is fulfilled for linearly parametrized control variates 
(at least when the numerical errors occurring when solving a system of linear equations are ignored) 
when a zero-variance control variates parameter exists (see e.g. Chapter 5, Section 3 in \cite{asmussen2007stochastic}), 
in which case the below theory can be easily applied to prove that a.s. such a parameter is found by the method after sufficiently many steps. 

For a nonempty set $A \in \mc{B}(\R^l)$, consider a function
%\begin{equation}\label{hdef} 
$h:\mc{S}(A) \otimes \mc{S}_1 \to \mc{S}(\overline{\R})$. 
%\end{equation}
We will be most interested in the cases corresponding to the below two conditions. 
\begin{condition}\label{condhzlb}
Condition \ref{condLmes} holds and $h(b,\omega)= (ZL(b))(\omega)$, $\omega \in \Omega_1$, $b \in A$. 
\end{condition}
\begin{condition}\label{condhzlb2}
Condition \ref{condLmes} holds and 
$h(b,\omega)= (|Z|L(b))(\omega)$, $\omega \in \Omega_1$, $b \in A$. 
\end{condition}

Let $b^* \in A$ and consider a family of distributions as in Section \ref{secFamily}, satisfying Condition \ref{condpqpq1}. 
%The following condition states that $h(b^*,\cdot)$ is a zero-variance estimator of $\alpha$ under $\PQ(b')$. 
\begin{condition}\label{condhperf}
For some $\beta \in \R$, $\PQ_1$ a.s. (and thus also $\PQ(b)$ a.s. for each $b \in A$)
\begin{equation}
h(b^*,\cdot)=\beta. 
\end{equation}
\end{condition}

%\begin{remark}\label{remcondshperf}
\begin{condition}\label{condSuffhperf}
Condition \ref{condhzlb} holds and $b^*$ is a zero-variance IS parameter. 
\end{condition}

\begin{condition}\label{condSuffhperf2}
Condition \ref{condhzlb2} holds and $b^*$ is an optimal-variance IS parameter. 
\end{condition}

\begin{remark}\label{remCondhperf}
From the discussion in Section \ref{secIS}, under Condition \ref{condSuffhperf}, Condition \ref{condhperf} holds for $\beta= \alpha=\E_{\PQ_1}(Z)$, and under 
Condition \ref{condSuffhperf2} --- for $\beta=\E_{\PQ_1}(|Z|)$.
\end{remark}

% \begin{remark}
% Note that assuming conditions \ref{condB1}, \ref{condg0}, and \ref{condhzlb}, from Theorem \ref{thexistqs} 
% and Remark \ref{rempqbznpq1}, Condition \ref{condSuffhperf} holds if
% $\PQ(b^*)(B_1)=\PQ(b')(B_1)=1$ and for some $\beta \in \R \setminus 0$, $\PQ_1$ a.s. $Z = \I_{B_1}\frac{\beta}{L(b^*)}$.
% \end{remark}
% Let us assume Condition \ref{condB1}.
% Then, conditions \ref{condpqe} and \ref{condZB} hold and there exists $b_1\in A$ for which $\PQ^*=\PQ(b_1)$ only 
% if for some $b_2 \in A$, $\PQ(b_2)(B_1)=1$ and for some $\beta \in \R \setminus 0$, 
% \begin{equation}\label{zperffrom}
% Z = \I_{B_1}\I(L(b_2)\neq 0)\frac{\beta}{L(b_2)}\quad \PQ_1\text{ a.s., } 
% \end{equation}
% in which case $\beta=\alpha$ and $\PQ(b_2)=\PQ^*$.
% Note that for $Z$ such that Condition \ref{condZB} holds, we have $\PQ(b) \sim_{Z \neq 0\cup B_1} \PQ_1$ and 
% \begin{equation} 
% L(b)=\left(\frac{d\PQ_1}{d\PQ(b)}\right)_{B_1\cup \{Z \neq 0\}},\quad b\in A, 
% \end{equation} 
% so that Condition \ref{condpqbllpq1} holds. 
% Let us assume Condition \ref{condB1}.
% Then, conditions \ref{condpqe} and \ref{condZB} hold and there exists $b_1\in A$ for which $\PQ^*=\PQ(b_1)$ only 
% if for some $b_2 \in A$, $\PQ(b_2)(B_1)=1$ and for some $\beta \in \R \setminus 0$, 
% \begin{equation}\label{zperffrom}
% Z = \I_{B_1}\I(L(b_2)\neq 0)\frac{\beta}{L(b_2)}\quad \PQ_1\text{ a.s., } 
% \end{equation}
% in which case $\beta=\alpha$ and $\PQ(b_2)=\PQ^*$.

\begin{remark}\label{remhalphakappa}
If  conditions \ref{condKappa} and \ref{condhperf} hold, then a.s. 
\begin{equation}
h(b^*,\kappa_i)=\beta,\quad i \in \N_+. 
\end{equation} 
\end{remark}

\begin{lemma}\label{lemhalphachi}
If conditions \ref{condChi} and \ref{condhperf} hold, then a.s. 
\begin{equation}\label{bchik}
h(b^*,\chi_{k,i})=\beta,\quad i \in \{1,\ldots,n_k\},\ k \in \N_+.
\end{equation} 
\end{lemma}
\begin{proof}
For each $k \in \N_+$ and $i\in \{1,\ldots,n_k\}$
\begin{equation}
\begin{split}
\PR(h(b^*,\chi_{k,i})=\beta)=\E(\PQ(b_{k-1})(h(b^*,\cdot)=\beta))=1, %&=\E(\E(h(b^*,\chi_{k,i})=\alpha|b_{k-1}))\\
\end{split}
\end{equation} 
so that (\ref{bchik}) holds a.s. as a conjunction of a countable number of conditions holding with probability one. 
\end{proof}

Let $D \in\mc{B}(A)$ be such that $b^* \in D$. For each  $n \in \N_+$ and $\omega \in \Omega_1^n$, let
us define 
\begin{equation}\label{defbnomega} 
B_{n}(\omega)=\{b \in D:h(b,\omega_i)=h(b,\omega_j)\in \R,\ i, j\in\{1,\ldots,n\}\}. %\\
%&=\{b\in A:h(b,\omega_i)=r, i \in \{1,\ldots,n\}\}. \\
\end{equation} 

For some $p \in \N_+$ and each $n \in \N_p$, let $\wh{\est}_n$ be as in (\ref{estnDef}). 
Consider the following condition. 
\begin{condition}\label{condEst} 
For each $n \in \N_p$, if the set 
$B_n(\omega)$ is nonempty, then for each $d\in A$, $B_n(\omega)$ is a subset of the set of the minimum points of 
$b\in D\to\wh{\est}_n(d,b)(\omega)$. 
\end{condition} 

Note that if Condition \ref{condEst} holds, then it holds also for $D$ replaced by its arbitrary subset (where
$D$ is replaced also in (\ref{defbnomega})).

\begin{remark}\label{remExact}
Let us assume Condition \ref{condg0} and let $D=A$. It holds for $b',b\in A$
\begin{equation}
\wh{\msq2}_n(b',b)=\frac{1}{n^2}\sum_{i<j \in \{1,\ldots,n\}} \frac{L_i'L_j'}{L_i(b)L_j(b)}(|Z_i|L_i(b) - |Z_j|L_j(b))^2+ (\overline{(|Z|L')}_n)^2. 
\end{equation}
Thus, under Condition \ref{condhzlb2}, Condition \ref{condEst} is satisfied for $p=1$ and $\wh{\est}_n$ equal to 
$\wh{\msq2}_n$ or $\wt{\msq2}_n$ (for the latter see (\ref{estngvar}) and (\ref{estnbpb})), or 
for $p=2$ and $\wh{\est}_n$ equal to $\wh{\var}_n$ (which is positively linearly equivalent to $\wh{\msq2}_n$ in the function of $b$
as discussed in Section \ref{secEstMin}).  
Furthermore, under Condition \ref{condhzlb}, Condition \ref{condEst} is satisfied for 
 $p=2$ and $\wh{\est}_n=\wh{\ic}_n$ (see formulas (\ref{varest})
and (\ref{icdef})) or  $p=3$ and $\wh{\est}_n=\wh{\ic2}_n$ (see (\ref{ic2Est})). 
%Condition \ref{condEst} is satisfied for $p=2$ and 
% we have $\wh{\est}_n(d,b)(\omega)=0$, $b \in B_n(\omega)$, $d\in A$, $\omega \in \Omega_1^n$, $n \in \N_p$. 
% From formulas (\ref{varest}) and (\ref{msq2ndef}), 
% Condition \ref{condEst} holds also for $\wh{\est}_n=\wh{\msq2}_{n}$, 
% and we have $\wh{\msq2}_{n}(d,b)(\omega)=\overline{(ZL(d))}_n^2(\omega)$, $b \in B_n(\omega)$, 
% $d\in A$, $\omega \in \Omega_1^n$, $n \in \N_p$. 
\end{remark} 

\begin{lemma}\label{lemMinKappa}
If conditions \ref{condKappa} and \ref{condhperf} hold, then a.s. for each $k \in \N_p$,
$b^* \in B_k(\wt{\kappa}_k)$. If further Condition \ref{condEst} holds then 
a.s. for each $k \in \N_p$ and $d \in A$, 
$b^*$ is a minimum point of $b\in D\to\wh{\est}_k(d,b)(\wt{\kappa}_k)$. 
\end{lemma}
\begin{proof}
It follows from Remark \ref{remhalphakappa} and (\ref{defbnomega}).    
\end{proof}

\begin{condition}\label{condUniqSSM}
Conditions \ref{condKappa} and \ref{condT} hold and functions $d_t:\Omega\to B$, $t \in T$, are such that
a.s. for a sufficiently large $t$, $b\in D\to \wh{\est}_{N_t}(b',b)(\wt{\kappa}_{N_t})$ has a unique minimum point equal to $d_t$.
\end{condition}

\begin{theorem}\label{thckbs} 
If conditions \ref{condhperf}, \ref{condEst}, and \ref{condUniqSSM} hold, 
then a.s. for a sufficiently large $t$, $d_t=b^*$. % and $B_{N_t}(\wt{\kappa}_{N_t})=\{b^*\}$. 
\end{theorem}
\begin{proof}
It follows from Lemma \ref{lemMinKappa}. % and Condition \ref{condUniqSSM}. 
\end{proof}

\begin{lemma}\label{lemMinChi} 
If Condition \ref{condChi} holds for $n_k \in \N_p$, $k \in \N_+$, 
and Condition \ref{condhperf} holds, then a.s. for each $k \in \N_+$, $b^* \in B_{n_k}(\wt{\chi}_k)$. 
If further Condition \ref{condEst} holds, then a.s. for each $k \in \N_+$ and
$d \in A$, $b^*$ is a minimum point of $b\in D\to\wh{\est}_{n_k}(d,b)(\wt{\chi}_k)$. 
\end{lemma}
\begin{proof}
It follows from Lemma \ref{lemhalphachi} and (\ref{defbnomega}).    
\end{proof}

\begin{condition}\label{condUniqMSM}
Condition \ref{condChi} holds for $n_k \in \N_p$, $k \in \N_+$, and functions $d_k:\Omega\to B$, $k \in \N_+$, are such that
a.s. for a sufficiently large $k$,  
$b\in D\to \wh{\est}_{n_k}(b_{k-1},b)(\wt{\chi}_k)$ has a unique minimum point equal to $d_k$.
\end{condition}

\begin{theorem}\label{thakbs}
If conditions \ref{condhperf}, \ref{condEst}, and \ref{condUniqMSM}
hold, then a.s. for a sufficiently large $k$, $d_k=b^*$. %and $B_{n_k}(\wt{\chi}_k)=\{b^*\}$. 
\end{theorem}
\begin{proof}
It follows from Lemma \ref{lemMinChi}. % and Condition \ref{condUniqMSM}.
\end{proof}

% In all the below examples we assume that Condition \ref{condpqe} holds and $\PQ(b^*)=\PQ^*$ 
%  for $\PQ^*$ denoting the zero-variance IS distribution. 
%Under additional assumptions given in the below examples, from Theorem \ref{thexistqs} this is equivalent to demanding that $\PQ_1$ a.s. 
% $Z = \frac{\alpha}{L(b^*)}$.  
%Furthermore, 
% In the below examples, for $b^*$ being the zero-variance (optimal-variance) IS parameter,  
% under Condition \ref{condhzlb} (Condition \ref{condhzlb}), 
% Condition \ref{condSuffhperf} (Condition \ref{condSuffhperf2}) can be easily shown to
% hold for each $b'\in A$. 

Let us consider the ECM setting and assume Condition \ref{condECMFull} and that $b^*$ 
is an optimal-variance IS parameter. Let us take $\wh{\est}_n=\wt{\msq2}_n$ as in (\ref{estngvar}), and $D=A$. 
As discussed in Section \ref{secminGrad} for ESSM (that is for GSSM for $\epsilon_t=0$, $t \in T$), 
conditions \ref{condKappaDk} and \ref{condESSM1} hold. Thus, a.s. for a sufficiently 
large $t$ for which $\wt{\kappa}_{N_t} \in \wt{D}_{N_t}$ and $d_t \in A$, 
$b\in D\to \wt{\msq2}_{N_t}(b',b)(\wt{\kappa}_{N_t})$ 
has a unique minimum point equal to $d_t$, i.e. Condition \ref{condUniqSSM} holds. 
Thus, from remarks \ref{remCondhperf}, \ref{remExact}, 
and Theorem \ref{thckbs}, a.s. for a sufficiently large $t \in T$, $d_t=b^*$. 
For EMSM, under the assumptions as in Section \ref{secminGrad} which ensure that Condition 
\ref{condKappaDk} holds, Condition \ref{condUniqMSM} holds and thus 
from remarks \ref{remCondhperf},
\ref{remExact}, and Theorem \ref{thakbs}, a.s. for a sufficiently large $k\in \N_+$, $d_k=b^*$. 

Let us now consider CGSSM and CGMSM
in the LETGS setting for $\wt{\est}_n=\wh{\msq2}_{n}$ and $b^*$ being an optimal-variance IS parameter, or
in the ECM setting for $\wt{\est}_n=\wh{\ic}_n$, $C=1$, and $b^*$ being a zero-variance IS parameter. 
Let $D$ be a bounded neighbourhood of $b^*$. Let us consider the corresponding assumptions 
as in Section \ref{secTwoPhase} for $\wt{\epsilon}_t=0$, $t \in T$, for the appropriate $T$ as in that section. 
Then, from Theorem \ref{thConstrOpt} we receive that for $d_t=\wt{d}_t$ and $\wh{\est}_n=\wt{\est}_n$,
Condition \ref{condUniqSSM} holds for CGSSM and  Condition \ref{condUniqMSM} holds for CGMSM. Thus, from 
remarks \ref{remCondhperf}, \ref{remExact}, and theorems \ref{thckbs} and \ref{thakbs}, 
a.s. for a sufficiently large $t \in T$, $\wt{d}_t=b^*$ in CGSSM and CGMSM. 

Let now $b^*$ be a zero-variance IS parameter and consider the MGSSM and MGMSM methods
in the LETGS setting for $\wt{\est}_n=\wh{\ic}_n$, $n\in \N_2$, and $\wt{\epsilon}_t=0$, $t \in T$,
under the assumptions as in Section \ref{secECG}. 
Then, from Remark \ref{remicpos}, $\nabla^2\ic(b^*)$ is positive definite. Thus, from the continuity of $\nabla^2\ic$ and from Lemma \ref{lemLipschMin},
$\nabla^2\ic$ is strongly convex on some open ball $U$ with center $b^*$. Therefore, from theorems \ref{thECGSM} and \ref{thConvStrong},
as well as remarks \ref{remCasesconbbsrl} and \ref{remLimIC}, conditions \ref{condUniqSSM} and  \ref{condUniqMSM} hold 
for such a MGSSM and MGMSM respectively for $d_t=\wt{d}_t$, $\wh{\est}_n=\wh{\ic}_n$, and $D\subset U$ being some neighbourhood of $b^*$. 
% Furthermore, from the fact that in such MGSSM we have for each $d \in A$, a.s. 
% $\lim_{k\to \infty} \I(N_t=k)\wh{c}_k(b',d)(\wt{\kappa}_k)\in \R_+$, 
% while in the MGMSM a.s. $\lim_{k\to \infty} \wh{c}_{n_k}(b',d)(\wt{\chi}_k)\in \R_+$, 
% for $A_k$ as in (\ref{Andef}) it holds in MGSSM a.s. $\wt{\kappa}_k \in A_k$ for a sufficiently large $k$
% and in MGMSM a.s. $\wt{\chi}_k \in A_{n_k}$ for a sufficiently large $k$. 
Therefore, by similar arguments as above, a.s. for a sufficiently large $t$, $\wt{d}_t=b^*$ in MGSSM and MGMSM. 
% Analogous reasoning as above can be used also for other adaptive variance reduction methods, like  
% multiple controls method (see e.g. Chapter 5, Section 3 in \cite{asmussen2007stochastic}) to show that a.s. the parameters 
% of the adaptive estimators are equal to the ones leading to zero variance after sufficiently many steps. 

\chapter{\label{secAsymp}Asymptotic properties of minimization methods}
% In this section we first discuss techniques for determining the asymptotic properties of 
% minimization results of estimators as well as of functions of such results. Then, we apply these techniques to 
% minimization methods from the previous sections. 
% use theorems \ref{thAsympMin} and \ref{thAsymp0}
% to prove asymptotic properties of minimization results from methods discussed in the 
% previous sections. 

\section{\label{secHelpAsymp}Helper theory for proving the asymptotic properties of minimization results}
% Let us present some theory which will be useful for proving asymptotic properties of results 
% of single- and multi-stage minimization methods and their functions in Section \ref{secAsymp}. 
%The following conditions will be useful for proving asymptotic properties of IS parameter
%in Section \ref{secAsymp}. 
%Let us consider some open neighbourhood $A$ of $b^*\in \R^l$ and
%let $f:A\to \R$ be such that $f(x_n)\to \infty$ as  $x_n$ tries to leave $A$.

For some $l\in \N_+$, let $A \subset \R^l$ be open and nonempty, $f:A\to \R$ be twice continuously differentiable, 
$b^* \in A$, and $H=\nabla^2f(b^*)$. 
%Let us assume the following condition. 
\begin{condition}\label{condftwice}
$\nabla f(b^*)=0$ and $H$ is positive definite. 
\end{condition}

Condition \ref{condftwice} is implied e.g. by the following one. 
\begin{condition}\label{condfNice}
$H$ is positive definite and $b^*$ is the unique minimum point of $f$.
\end{condition}

\begin{remark}\label{remneighb}
Let us assume Condition \ref{condftwice}. 
Then, from Lemma \ref{lemLipschMin} and the continuity of $\nabla^2f$, 
for an open or closed ball $B\subset A$ with center $b^*$ and sufficiently small positive radius,
$f$ is strongly convex on $B$ and from the discussion in Section \ref{secStrong},
$b^*$ is the unique minimum point of $f_{|B}$.  
\end{remark}

%Let us consider the following Condition which implies Condition \ref{condftwice}. 
%Note that from Condition \ref{condftwice}, 

Let $T\subset \R_+$ be unbounded. %(typically we will take $T=\R_+$ or $T=\N_+$). 
Consider functions $f_t:\mc{S}(A)\otimes (\Omega,\mc{F}) \to \mc{S}(\R)$, $t \in T$, such that
$b\to f_t(b,\omega)$ is twice continuously differentiable, $t \in T$, $\omega\in \Omega$. 
We shall denote $f_t(b)=f_t(b,\cdot)$ and $\nabla^i f_t(b)=\nabla^i_bf_t(b,\cdot)$, $i=1,2$. 

\begin{condition}\label{condfntwice}
For some neighbourhood $D\in \mc{B}(A)$ of $b^*$, $\nabla^2f_t$ converges to  $\nabla^2f$ on $D$ uniformly in probability
(as $t \to \infty$), i.e. $\sup_{b\in D} ||\nabla^2f_t(b) - \nabla^2f(b)||_{\infty}\overset{p}{\to}0$. 
\end{condition} 

Let $d_t$, $t \in T$, be $A$-valued random variables.  
\begin{condition}\label{condbntobs}
It holds 
%\begin{equation}\label{bntobs}
$d_t \overset{p}{\to} b^*$. 
%\end{equation}
\end{condition}
% Condition \ref{condbntobs} is implied by the following one. 
% \begin{condition}\label{condbntobsSLLN}
% It holds a.s. 
% %\begin{equation}\label{bntobs}
% $\lim_{t\to \infty} d_t = b^*$. 
% %\end{equation}
% \end{condition}
For $g_t \in \R_+$, $t \in T$, we shall write $X_t = o_{p}(g_t)$ if 
$\frac{X_t}{g_t}\overset{p}{\to}0$ (as $t \to \infty$). %converges to zero in probability (as $t \to \infty$). 
Let $r_t \in \R_+$, $t\in T$, be such that $\lim_{t\to \infty} r_t=\infty$.

\begin{condition}\label{condbdelta}
For some nonnegative random variables $\delta_t$, $t \in T$, such that $\delta_t=o_p(r_t^{-1})$ and 
for some neighbourhood $U\in \mc{B}(A)$ of $b^*$, for 
the event $A_t$ that $d_t$ is a $\delta_t$-minimizer of $f_{t|U}$ (in particular, $d_t \in U$, see Section \ref{secStrong}), we have 
\begin{equation}\label{pranlim}
\lim_{t\to \infty }\PR(A_t)=1. 
\end{equation}
\end{condition}

\begin{condition}\label{condftY}
For some $\R^l$-valued random variable $Y$, $\sqrt{r_t}\nabla f_t(b^*) \Rightarrow Y$.
\end{condition}

The below theorem 
is a consequence of Theorem 2.1 and the discussion of its assumptions in sections 2 
and 4 in \cite{Shapiro1993}, of the implicit function theorem, 
and of the below Remark \ref{remAs} (see also formula (4.14) and its discussion in \cite{Shapiro1996}). 

\begin{theorem}\label{thAsympMin}
Under conditions \ref{condftwice}, \ref{condfntwice}, \ref{condbntobs}, \ref{condbdelta}, %, \ref{condbnDelta},
and \ref{condftY}, we have 
\begin{equation}\label{dnbsn}
\sqrt{r_t}(d_t-b^*) \Rightarrow -H^{-1}Y. 
\end{equation}
In particular, if $Y \sim \ND(0,\Sigma)$ for some covariance matrix $\Sigma \in \R^{l\times l}$, then 
\begin{equation}\label{lndnb}
\sqrt{r_t}(d_t-b^*)\Rightarrow \ND(0,H^{-1}\Sigma H^{-1}). 
\end{equation}
\end{theorem}
We will need the following trivial remark. 
\begin{remark}\label{remgnan}
Note that if for random variables $\wt{a}_t$ and $a_t$, $t \in T$, with probability tending to one (as $t \to \infty$)
we have $a_t=\wt{a}_t$, then for each $g:T \to \R$,  $g(t)(\wt{a}_t-a_t) \overset{p}\to 0$. 
\end{remark}

\begin{remark}\label{remAs} 
Theorem 2.1 in \cite{Shapiro1993} uses $T=\N_+$ and $r_n=n$, $n\in T$, but its proof for the general $T$ 
and $r_t$, $t\in T$, as above is analogous. 
Let us assume the conditions mentioned in Theorem \ref{thAsympMin}. Let from Remark \ref{remneighb}, 
$B$ be a closed ball contained in the set $D$ as in Condition \ref{condfntwice} and $U$ 
as in Condition \ref{condbdelta}, and such that 
$f$ is strongly convex on $B$ and $b^*$ is the unique minimum point of $f_{|B}$. 
%Let $K\subset B$ be a compact convex neighbourhood of $b^*$. 
From the generalization of Theorem 2.1 in \cite{Shapiro1993} to the general $T$
and $r_t$ as above and the discussion of assumptions of this theorem in \cite{Shapiro1993}, 
one easily receives the thesis (\ref{dnbsn}) of Theorem \ref{thAsympMin} under the additional assumptions that 
we have $d_t \in B$, $t \in T$, and Condition \ref{condbdelta} holds with $A_t=\Omega$, $t \in T$. 
From Remark \ref{remgnan} and Condition \ref{condbntobs}, 
to prove Theorem \ref{thAsympMin} it is sufficient to prove (\ref{dnbsn}) with $d_t$ replaced by
$\wt{d}_t=\I(d_t \in B)d_t + \I(d_t \notin B)b^*$. 
For $C_t=A_t \cap \{d_t \in B\}$, let $\wt{\delta}_t(\omega)=\delta_t(\omega)$, $\omega \in C_t$, and 
$\wt{\delta}_t(\omega)=\infty$, $\omega \in  \Omega\setminus C_t$.
%+ \I_{C_n'}\cdot\infty$, $n \in \N_+$. 
Then, from  Remark \ref{remgnan}, Condition \ref{condbntobs}, and (\ref{pranlim}), 
for $d_t$ replaced by $\wt{d}_t$ and $\delta_t$ by  $\wt{\delta}_t$,
the conditions of Theorem \ref{thAsympMin} are satisfied and the above additional 
assumptions hold. Thus, (\ref{dnbsn}) with $d_t$ replaced by $\wt{d}_t$ follows from Theorem 2.1 in \cite{Shapiro1993} as discussed above. 
\end{remark}

\begin{condition}\label{condfndiff}
On some neighbourhood $K\in \mc{B}(A)$ of $b^*$, for $i=1,2$, 
the $i$th derivatives of $b\to f_t(b)$ (i.e. $\nabla f_t$ and $\nabla^2f_t$) converge to such derivatives of $f$
uniformly in probability. 
\end{condition} 
Condition \ref{condfndiff} is implied e.g. by the following one. 

\begin{condition}\label{condfndiffAS} 
On some neighbourhood $K\in \mc{B}(A)$ of $b^*$, for $i=1,2$,
a.s. the $i$th derivatives of $b\to f_t(b)$, converge uniformly to such derivatives of $f$. 
\end{condition} 

\begin{condition}\label{condnablafn}
It holds $|\nabla f_t(d_t)| =o_p(r_t^{-\frac{1}{2}})$. 
\end{condition}

\begin{lemma}\label{lemfeps}
If conditions  \ref{condftwice}, \ref{condbntobs}, \ref{condfndiff},
and \ref{condnablafn} hold, then Condition \ref{condbdelta} holds.
\end{lemma}
\begin{proof}
From Remark \ref{remneighb}, let $U$ be an open ball with center $b^*$, contained in 
$K$ as in Condition \ref{condfndiff}, and such that 
$f$ is strongly convex on $U$ with a constant $s\in \R_+$. Let $m\in (0,s)$.
Then, from Theorem \ref{thConvStrong}, there exists $\epsilon\in\R_+$ such that for each twice differentiable function $g:U \to \R$ for which 
\begin{equation}
\sup_{x\in U}(||\nabla^2f(x) -\nabla^2g(x)||_{\infty} + |\nabla f(x)-\nabla g(x)|)<\epsilon, 
\end{equation}
each $b \in U$ is a $\frac{1}{2m}|\nabla g(b)|^2$-minimizer of $g$. Thus, 
Condition \ref{condbdelta} 
for the above $U$ and $\delta_t = \frac{1}{2m}|\nabla f_t(d_t)|^2$ follows from conditions
\ref{condbntobs}, \ref{condfndiff}, and \ref{condnablafn}.
\end{proof}

From the above lemma we receive the following remark.
\begin{remark}\label{remEquivAsymp}
The assumption that conditions \ref{condfntwice} and \ref{condbdelta}
hold in Theorem \ref{thAsympMin} can be replaced by the assumption that conditions 
\ref{condfndiff} and \ref{condnablafn} hold. 
\end{remark}

Consider the following composite condition. 
\begin{condition}\label{condAllAs}
Conditions \ref{condftwice}, \ref{condbntobs}, \ref{condfndiff}, and \ref{condnablafn} hold. 
\end{condition}

\begin{theorem}\label{thAsymp0}
Let us assume Condition \ref{condAllAs} and that 
\begin{equation}\label{limtprnft0}
\lim_{t\to\infty}\PR(\nabla f_t(b^*)=0)=1. 
\end{equation}
Then, 
\begin{equation}
\sqrt{r_t}(d_t-b^*)\overset{p}{\to} 0. 
\end{equation}
\end{theorem}
\begin{proof}
From (\ref{limtprnft0}), 
Condition \ref{condftY} holds for $Y =0$, so that the thesis 
follows from Remark \ref{remEquivAsymp} and Theorem \ref{thAsympMin}. 
\end{proof}

%Consider the following condition. 
\begin{condition}\label{condbnEps}
For some nonnegative random variables $\delta_t$, $t \in T$, such that  $\delta_t=o_p(r_t^{-\frac{1}{2}})$,
with probability tending to one (as $t\to \infty$) we have 
\begin{equation}\label{nablafnbn}
 |\nabla f_t(d_t)| \leq \delta_t. 
\end{equation}
\end{condition}

\begin{lemma}\label{lemEquivEps}
Conditions \ref{condnablafn} and \ref{condbnEps} are equivalent.
\end{lemma}
\begin{proof}
If Condition \ref{condnablafn} holds, then Condition \ref{condbnEps} holds for 
$\delta_t=|\nabla f_t(d_t)|$. Let us assume Condition \ref{condbnEps}. Then, for 
$\wt{\delta}_t$ equal to $\delta_t$ if (\ref{nablafnbn}) holds and $\infty$ otherwise, 
we have  $|\nabla f_t(d_t)| \leq \wt{\delta}_t$  and from 
Remark \ref{remgnan}, $\wt{\delta}_t=o_p(r_t^{-\frac{1}{2}})$, from which Condition \ref{condnablafn} follows.
\end{proof}

\section{\label{secAsympFun}Asymptotic properties of functions of minimization results}
Let us consider $T$ and $r_t$, $t\in T$, as in the previous section. 
We shall further need the following theorem on the delta method (see e.g. Theorem 3.1 and Section 3.3 in \cite{Vaart}). 
\begin{theorem}\label{thdelta} 
Let $m,n \in \N_+$, let $D\in\mc{B}(\R^m)$ be a neighbourhood of $\theta \in \R^m$, and consider a function 
$\phi:\mc{S}(D)  \to  \mc{S}(\R^n)$. Let  $Y_t$, $t \in T$, and $Y$ be $D$-valued random variables 
such that we have  $\sqrt{r_t}(Y_t-\theta) \Rightarrow  Y$ (as $t \to \infty$). 
If $\phi$ is differentiable in $\theta$ with a differential $\phi'(\theta)$ (which we identify with its matrix), then 
\begin{equation}\label{lntr} 
\sqrt{r_t}(\phi(Y_t)-\phi(\theta)) \Rightarrow \phi'(\theta)Y. 
\end{equation} 
If $\phi$ is twice differentiable in $\theta$ with the second differential $\phi''(\theta)$ and we have $\phi'(\theta)=0$, then 
\begin{equation}\label{lntnpp} 
r_t(\phi(Y_t)-\phi(\theta)) \Rightarrow \frac{1}{2}\phi''(\theta)(Y,Y). 
\end{equation} 
\end{theorem} 

\begin{remark}\label{remChi2gen} 
For $m \in \N_+$, let $\chi^2(m)$ denote the $\chi$-squared distribution with $m$ degrees of freedom. 
Let $m \in \N_+$, $S \sim \ND(0,I_m)$, $B\in\Sym_m(\R)$, and  $X=S^TBS$. Then, $X$ has a special case of the generalized $\chi$-squared 
distribution, which we shall denote as $\wt{\chi^2}(B)$. For $B$ being a diagonal matrix $B=\diag(v)$ for some $v \in \R^m$, 
$\wt{\chi^2}(B)$ will be also denoted as $\wt{\chi^2}(v)$. It holds $\E(X)=\Tr(B)$. 
If $B=wI_m$ for some  $w\in \R$, then we have $X \sim w\chi^2(m)$ (by which we mean that $X \sim wY$ for $Y \sim \chi^2(m)$).
Let $\eig(B)\in \R^m$ denote a vector of eigenvalues of $B$
and let $\lambda=\eig(B)$. 
Consider an orthogonal matrix $U \in \R^{m\times m}$ such that $B=U\diag(\lambda)U^T$.  
%(for which we have $\lambda \in [0,\infty)^m$ only if $B$ is positive semidefinite). 
Then, for $W =U^TS \sim \ND(0,I_m)$ we have 
\begin{equation}
X= W^T\diag(\lambda)W=\sum_{i=1}^m\lambda_i W_i^2, 
\end{equation}
and thus $\wt{\chi^2}(B)=\wt{\chi^2}(\lambda)$. 
Let $\Lambda=\{0\}\cup\{v \in (\R\setminus \{0\})^k: k \in \N_+,\ v_1\leq v_2\leq\ldots \leq v_k\}$, 
i.e. $\Lambda$ is the set of all real-valued vectors in different dimensions with ordered nonzero coordinates or having only one zero coordinate.  
Let for $v \in \R^m$,  $\ord(v)\in \Lambda$ be equal to $0\in \R$ if $v=0$ and otherwise 
result from ordering the coordinates of $v$ in nondecreasing order and removing the zero coordinates. 
Then, we have $\wt{\chi^2}(v)=\wt{\chi^2}(\ord(v))$.  
For $Y \sim \chi^2(1)$ we have a moment-generating function $M_{Y}(t):=\E(\exp(tY))=(1-2t)^{-\frac{1}{2}}$, $t <\frac{1}{2}$. 
Thus, for each $k \in \N_+$, $v \in \Lambda\cap \R^k$,  $Y \sim \wt{\chi^2}(v)$, and $t \in \R$ such that $1-2v_it>0$, $i =1,\ldots,k$, 
we have $M_Y(t)=\prod_{i=1}^k(1-2v_it)^{-\frac{1}{2}}$. Such an $M_Y$ is defined on some neighbourhood of $0$
and it is a different function for different $v \in \Lambda$. Thus, 
for $v_1, v_2 \in \Lambda$ such that $v_1\neq v_2$, we have $\wt{\chi^2}(v_1)\neq\wt{\chi^2}(v_2)$. 
It follows that for two real symmetric matrices $B_1$, $B_2$, we have
$\wt{\chi^2}(B_1)=\wt{\chi^2}(B_2)$ only if $\ord(\eig(B_1))=\ord(\eig(B_2))$.
\end{remark}
\begin{remark}\label{remdelta}
Using notations as in Theorem \ref{thdelta},
let $Y \sim \ND(0,M)$ for some covariance matrix $M\in \R^{m\times m}$. Then, (\ref{lntr}) implies that
\begin{equation}
\sqrt{r_t}(\phi(Y_t)-\phi(\theta)) \Rightarrow \ND(0,\phi'(\theta)M\phi'(\theta)^T).
\end{equation}
Let further $n=1$.  Then, (\ref{lntnpp}) is equivalent to
\begin{equation}
r_t(\phi(Y_t)-\phi(\theta)) \Rightarrow R:=\frac{1}{2}Y^T\nabla^2\phi(\theta)Y.
\end{equation}
For $S \sim \ND(0,I_l)$ we have $Y \sim M^{\frac{1}{2}}S$. Thus, from Remark \ref{remChi2gen}, for 
\begin{equation}\label{bdef}
B=\frac{1}{2}M^{\frac{1}{2}}\nabla^2\phi(\theta)M^{\frac{1}{2}},
\end{equation}
we have $R \sim \wt{\chi^2}(B)$ and 
\begin{equation}\label{ertr}
\E(R)= \Tr(B)= \frac{1}{2}\Tr(\nabla^2\phi(\theta)M). 
\end{equation}
Note that if $\theta$ is a local minimum point of $\phi$, then $\nabla^2\phi(\theta)$ is positive semidefinite and so is $B$.   
\end{remark}

\begin{remark}\label{remHY}
If we have assumptions as in Theorem \ref{thAsympMin} leading to (\ref{lndnb}), then from Remark \ref{remdelta}, for $M=H^{-1}\Sigma H^{-1}$ and 
\begin{equation}
B=\frac{1}{2}M^{\frac{1}{2}}HM^{\frac{1}{2}} = \frac{1}{2}H^{-\frac{1}{2}}\Sigma H^{-\frac{1}{2}}, 
\end{equation}
we have 
\begin{equation}\label{lnfnfb}
r_t(f(d_t)-f(b^*))\Rightarrow \wt{\chi^2}(B). 
\end{equation}
Note that for $R \sim \wt{\chi^2}(B)$ it holds
\begin{equation}\label{ersigmah}
\E(R)=\Tr(B)= \frac{1}{2}\Tr(\Sigma H^{-1}).  
\end{equation}
\end{remark}

\section{\label{secSecond}Comparing the second-order asymptotic efficiency of minimization methods}
For some $T$ and $r_t$, $t\in T$, as in Section \ref{secHelpAsymp}, let $\phi$ and $d=(d_t)_{t\in T}$ be as in Section \ref{secCompFirst} and such that 
for some probability $\mu$ on $\R$ and $y \in \R$ we have 
\begin{equation} \label{rtdty}
r_t(\phi(d_t)-y) \Rightarrow  \mu. % \wt{\chi^2}(B). 
\end{equation} 
%For instance, $a_t$ can arise from some situations as discussed in the above remarks. 
%For instance we can have $T=\N_+$ and $a_n$ can arise from situations like in the above remarks. 
Let further for some analogous $d'=(d_t')_{t\in T}$ and $\mu'$ it hold $r_t(\phi(d_t')-y) \Rightarrow  \mu'$. 
Then, $\phi(d_t)\overset{p}{\to}y$ and similarly in the primed case, so that the
processes $d$ and $d'$ are equivalent from the point of view of the first-order asymptotic efficiency for the minimization 
of $\phi$ as discussed in Section \ref{secCompFirst}. Their 
second-order asymptotic efficiency for this purpose can be compared by comparing the asymptotic distributions $\mu$ and $\mu'$. 
For instance, if $\mu=\mu'$, then they can be considered equally efficient. If 
 for each $x\in \R$, $\mu((-\infty,x])\geq\mu'((-\infty,x])$, then it is natural to consider 
 the unprimed process to be not less efficient and more efficient if further for some $x$ this inequality is strict. 
The second-order asymptotic efficiency as above can be also compared using some moments like means or some quantiles like medians 
of the asymptotic distributions, 
where the process corresponding to lower such parameter can  be considered more efficient. 
For $\mu=\wt{\chi^2}(B)$ and $\mu'=\wt{\chi^2}(B')$ for some symmetric 
matrices $B$ and $B'$, which can arise e.g. from situations like in remarks \ref{remdelta} or \ref{remHY}, 
it may be convenient to compare the second-order efficiency of 
the corresponding processes using the 
means $\Tr(B)$ and $\Tr(B')$ of these distributions. In situations like in remarks \ref{remdelta} or \ref{remHY}, 
such means can be alternatively expressed by formula
(\ref{ertr}) or (\ref{ersigmah}) respectively, using which in some cases they can be estimated or even 
computed analytically (see e.g. Section \ref{secAsympProp}). 
For a number of stochastic optimization methods from the literature we do not have formulas like (\ref{rtdty}) and some
other ways of comparing the second-order asymptotic efficiency of such methods are needed; see \cite{StochOptComp} for some ideas. 

\begin{remark}\label{remmoreeff} 
  Under the assumptions as above, let $\mu([0,\infty))=1$ and let for some $X \sim \mu$ and $s \in [1,\infty)$ it hold $\mu'\sim sX$. 
  Note that from Remark \ref{remChi2gen} this holds e.g. if for some symmetric positive semidefinite matrices $B$ and $B'$ we have
  $\mu=\wt{\chi^2}(B)$, $\mu'=\wt{\chi^2}(B')$, and $\ord(\eig(B'))=s\ord(\eig(B))$. 
  If further $r_t=t$, $t\in T=\R_+$, 
  then $t(\phi(d_{t})-y)$ and $t(\phi(d_{st}')-y)$ both converge in distribution to 
 $\mu$, but computing $d_{st}'$ requires $s$ times higher budget than $d_t$, $t\in T$ (assuming that such interpretation holds). 
   Thus, the unprimed process can be called $s$ times (asymptotically) more efficient for the minimization of $\phi$. 
% For $s=1$ we shall also call them equally efficient for this purpose.
  Similarly, if $M'=sM$ for some nonzero covariance matrix $M\in \R^{l\times l}$, and 
  for some $\theta\in A$, $\sqrt{t}(d_t -\theta)\Rightarrow \ND(0,M)$ and  $\sqrt{t}(d_t' -\theta)\Rightarrow \ND(0,M')$, then 
  $d$ can be said to converge $s$ times faster to $\theta$ than $d'$. In such a case, 
  if additionally $\phi$ is twice differentiable in $\theta$ with a zero gradient and a positive definite Hessian in this point, then from Remark \ref{remdelta}, 
  for  $B$ as in (\ref{bdef}) we have $t(\phi(d_{t})-y)\Rightarrow \wt{\chi^2}(B)$
  and similarly for the primed process for $B'=sB$.  Thus, the unprimed process is $s$ times more efficient
  for the minimization of $\phi$. %If s=1$
  \end{remark}
  
\begin{remark}\label{remOpt}
 Analogously as we have discussed certain properties 
 of minimization results in Section \ref{secHelpAsymp} and their functions in Section \ref{secAsympFun}, 
 or proposed how to compare the asymptotic efficiency of stochastic 
 minimization methods in Section \ref{secCompFirst} and
 this section, one can formulate such a theory for maximization methods. It is sufficient to notice that 
maximization of a function is equivalent to the minimization of its negative, so that it is sufficient to 
apply the above reasonings to the negatives of appropriate functions. 
\end{remark}

\section{\label{secSomeUs}Discussion of some conditions useful for proving the asymptotic properties in our methods}
Let us discuss when, under appropriate identifications given below,
Condition \ref{condAllAs} holds in the LETGS and ECM settings for the different
SSM methods from the previous sections, i.e. for ESSM, GSSM, CGSSM, and MGSSM, 
and for such MSM methods, i.e. for EMSM, GMSM, CGMSM, and MGMSM.   
We consider in this condition $f$ equal to  $\ce$, $\msq$, or $\ic$, each defined on $A=\R^l$.  
Furthermore, we take $T$ as for the minimization methods in the previous sections, in particular for the MSM methods we take $T=\N_+$. 
For the EM and GM methods we take $f_t=\wh{f}_t$ and $d_t$ as in these methods,
while for the CGM and MGM methods $f_t=\wt{f}_t$ and $d_t=\wt{d}_t$,  assuming Condition \ref{condVarck}. 

Sufficient conditions for the smoothness of such functions $f$ %(and in particular for their twice continuous differentiability)
follow from the discussion in sections \ref{secCEECM}, \ref{secCeLETGS}, and \ref{secDiff}. 
The smoothness of such $b\to f_t(b,\omega)$, $\omega \in \Omega$, $t \in T$, in the LETGS setting is obvious and in the ECM setting 
it holds under Condition \ref{condPartvm}, which follows from $A=\R^l$ as discussed in Remark \ref{remPartv}. 
Sufficient assumptions for Condition \ref{condfNice}, implying Condition \ref{condftwice}, to hold for $f$ equal to $\ce$,
in the ECM setting were provided in Section \ref{secCEECM}, and in the 
LETGS setting --- in Section \ref{secCeLETGS}. From the discussion in Section \ref{secStrong}, for $A=\R^l$ as above, 
Condition \ref{condfNice} follows from the
strong convexity of $f$, sufficient assumptions for which for $f$ equal to $\msq$ or $\var$ in the LETGS setting were provided 
in Theorem \ref{thLETGSStrong}. For $f$ equal to $\msq$ in the ECM setting, sufficient conditions for 
it to have a unique minimum point were discussed in sections \ref{secMsqGen} 
(see e.g. Lemma \ref{lempmsq}) and \ref{secECM}, and for $\nabla^2f$ to be positive definite --- in Theorem \ref{thECMPos}. 
For $f=\ic$, if $C$ is a positive constant, then Condition \ref{condfNice} follows from such a condition for $f=\var$, 
and some other sufficient assumptions for Condition \ref{condfNice} to hold were discussed in Remark \ref{remicpos}. 
From the discussion in Section \ref{secUnifEst} we receive assumptions for which Condition \ref{condfndiffAS}, implying Condition \ref{condfndiff}, 
holds in the MSM methods as well as in the SSM methods for $T=\N_+$ and $N_k=k$, and thus from Condition \ref{condT} also in the general case. 
For $d_t$ as above, Condition \ref{condbntobs} follows from Condition \ref{condESSM1} and its counterparts. 
% and its counterparts for CGSSM and MGSSM. For $d_k$  in EMSM and GMSM or 
% $\wt{a}_k$ in CGMSM or MGMSM, Condition \ref{condbntobsSLLN} is equivalent to Condition \ref{condESSM1} and its counterparts. 

Recall that from Lemma \ref{lemEquivEps}, conditions \ref{condnablafn} and \ref{condbnEps} are equivalent. 
For the EM methods, if $b\to \wh{\est}_n(b',b)(\omega)$ is differentiable, $b'\in A$, $\omega \in \Omega_1^n$, $n \in \N_p$, 
and if Condition \ref{condKappaDk} holds, then Condition \ref{condbnEps} holds in these methods even for $\delta_t=0$, $t \in T$. 
For the GM methods, if Condition \ref{condKappaDk} holds and  
$\epsilon_t=o_p(r_t^{-\frac{1}{2}})$ 
(which holds e.g. if $\epsilon_t = r_t^{-q}$ for some $\frac{1}{2}<q<\infty$),
then Condition \ref{condbnEps} holds for $\delta_t=\epsilon_t$. 
%t_n \to \infty 
%\epsilon_n
% Let us assume that 
% \begin{equation}
% \lim_{t\to \infty}\frac{N_t}{t} =\frac{1}{w(b')}\in \R_+. 
% \end{equation}
% $\epsilon_n = r_n$ 
% then $\epsilon_n=o_p(r_n^{-\frac{1}{2}})$ implies $\delta$ 
% Under Condition \ref{condKappaDk} for GSSM, 
% and assuming that a.s. 
% \epsilon_t=
% %\epsilon_n 
% %N()
% $\epsilon_N(t)=r_t$ 
%n^{-\frac{1}{2}})
%$\frac{N_t}{t}\to b$ 
In the CGM and MGM methods, if Condition \ref{condgradck} holds and $\wt{\epsilon}_t=o_p(r_t^{-\frac{1}{2}})$, 
then Condition \ref{condbnEps} holds for $\delta_t=\wt{\epsilon}_t$. 

\section{\label{secAsympSSM}Asymptotic properties of single-stage minimization methods} 
Let $T \subset \R_+$ be unbounded, conditions \ref{condpqbllpq1}, \ref{condLmes}, \ref{condg0}, 
\ref{condKappa}, and \ref{condT} hold, $A$ be open, $b^*\in A$,  
%hold for some open parameter set $A\subset \R^l$ and that 
and $b\in A\to L(b)(\omega)$ be differentiable, $\omega \in \Omega_1$. 
%\begin{condition}
For some function $u \in A\to [0,\infty]$ such that $u(b')\in \R_+$, let us assume that
\begin{equation}\label{nttpto}
\frac{N_t}{t} \overset{p}{\to}\frac{1}{u(b')}, \quad (t\to \infty).  
\end{equation}
\begin{remark}\label{remub} 
Let $N_t$ be given by some $U$ as in Remark \ref{remCondT}, and 
let $u(b)=\E_{\PQ(b)}(U)$, $b \in A$. For such an $U$ being 
the theoretical cost variable of a step of SSM as in Remark \ref{remCondT}, 
$u(b')$ is such a mean cost. 
If $u(b')\in\R_+$, then, as discussed in Chapter \ref{secIneff} 
(see (\ref{nttc})), 
we have a stronger fact than (\ref{nttpto}), namely that
a.s. $\frac{N_t}{t} \to\frac{1}{u(b')}$. For the special case of $U=1$, 
we have $u(b)=1$, $b \in A$. 
\end{remark}
Below we shall prove that for $g$ substituted by $\ce$, $\msq$, $\msq2$, 
or $\ic$, under appropriate assumptions, 
for some covariance matrix $\Sigma_g(b') \in \R^{l\times l}$ and 
\begin{equation}\label{ftdef}
f_t(b)=\I(N_t=k\in\N_+)\wh{g}_{k}(b',b)(\wt{\kappa}_{k}),\quad t \in T,  
\end{equation}
we have
\begin{equation}\label{nablaftng}
\sqrt{t}\nabla f_t(b^*)\Rightarrow \ND(0,u(b')\Sigma_{g}(b')). 
\end{equation}
\begin{remark}\label{remSSMAsymp}
Let (\ref{nablaftng}) hold for some $g$ as above and let Condition \ref{condAllAs} hold for the corresponding $f_t$ as above and 
$r_t=t$, $t\in T$, as well as for the minimized function $f$ 
equal to $\msq$ if $g= \msq2$, and to $g$ otherwise.  
Then, from Theorem \ref{thAsympMin}, denoting $H_f=\nabla^2f(b^*)$ and 
\begin{equation}\label{vgdef}
V_g(b')=H_f^{-1}\Sigma_g(b')H_f^{-1}, 
\end{equation}
we have
\begin{equation}\label{dtbsSSM}
\sqrt{t}(d_t-b^*) \Rightarrow \ND(0,u(b')V_g(b')).
\end{equation}
Furthermore, from Remark \ref{remHY}, for 
\begin{equation}\label{bgdef}
B_g(b')=\frac{1}{2}H_f^{-\frac{1}{2}}\Sigma_g(b')H_f^{-\frac{1}{2}},
\end{equation}
it holds
\begin{equation}\label{fdtSSM}
t(f(d_t)-f(b^*)) \Rightarrow \wt{\chi^2}(u(b')B_g(b')). 
\end{equation}

Let $u(b')$ be interpreted as the 
mean theoretical cost of a step of SSM as in Remark \ref{remub}. Let
us consider different processes  $d=(d_t)_{t\in T}$ from SSM methods for which 
(\ref{fdtSSM}) holds for possibly different $b'$ and $g$, and 
whose SSM methods have the same proportionality constant $p_{\dot{U}}$ of the theoretical to the practical cost variables of SSM steps 
(see Remark \ref{remCondT}). Then,
from the discussion in Section \ref{secSecond}, %for $T=\R_+$ and $r_t=t$, $t\in T$, 
for $R \sim \wt{\chi^2}(u(b')B_g(b'))$, 
the second-order asymptotic efficiency  of such processes for the minimization of $f$ can be compared using the quantities
\begin{equation}\label{quantSSM}
\E(R) = \frac{u(b')}{2}\Tr(\Sigma_g(b')H_f^{-1}).  
\end{equation}
For the SSM methods having different constants $p_{\dot{U}}$, one can compare 
such quantities multiplied by such a $p_{\dot{U}}$.
\end{remark}
Let $p(b)=\frac{1}{L(b)}$, 
let us define the likelihood function $l(b)=\ln (p(b))=-\ln L(b)$, and the score function
\begin{equation}
S(b)=\nabla l(b) = \frac{\nabla p(b)}{p(b)}=-\frac{\nabla L(b)}{L(b)},  
\end{equation}
$b\in A$, where such a terminology is used in maximum likelihood estimation; see \cite{Vaart}.
 %Suuch $S(b)$ is the negative of score function in maximum likelihood estimation (MLE), 
%(see e.g. page 77 in \ref{Vaart}).  
%\end{equation}
Then, $\wh{\ce}_n(b',b)=  -\overline{(ZL'l(b))}_n$ and 
%\begin{equation}
$\nabla_{b}\wh{\ce}_n(b',b)=  -\overline{(ZL'S(b))}_n$. 
%\end{equation}
Thus, if  %(and $\PQ'=\PQ(b')$)
\begin{equation}
\E_{\PQ'}((ZL'S_i(b^*))^2) = \E_{\PQ_1}(L'(ZS_i(b^*))^2)<\infty,\quad i=1,\ldots,l,  
\end{equation}
(for which to hold in the LETS setting, from Theorem \ref{thYmore} 
and Remark \ref{remCondUN}, 
it is sufficient if Condition \ref{condlemYmore} holds for $S=Z^2$),
and $\nabla \ce(b^*)=-\E_{\PQ_1}(ZS(b^*))=0$, 
then, from Theorem \ref{theConvIneff}, for 
\begin{equation}\label{sigmacedef}
\Sigma_{\ce}(b') = \E_{\PQ'}((ZL')^2S(b^*)S(b^*)^T)=\E_{\PQ_1}(L'Z^2S(b^*)S(b^*)^T),
\end{equation}
we have (\ref{nablaftng}) for $g=\ce$  (and $f_t$ as in 
(\ref{ftdef}) for such a $g$).

It holds 
\begin{equation}
\nabla_{b}\wh{\msq}_n(b',b)=  \overline{(Z^2L'\nabla L(b))}_n. 
\end{equation}
Thus, if 
\begin{equation}\label{msqCLTInt}
\E_{\PQ'}(Z^4(L'\partial_iL(b^*))^2) = \E_{\PQ_1}(Z^4L'(\partial_iL(b^*))^2)<\infty,\quad  i=1,\ldots,l,
\end{equation}
and 
\begin{equation}\label{msqbs0}
\nabla\msq(b^*)=\E_{\PQ_1}(Z^2\nabla L(b^*))=0, 
\end{equation}
%(for which to hold in the LETS setting, from Theorem \ref{thYmore} and Remark \ref{remCondUN}, it is sufficient if Condition \ref{condlemYmore} holds for $S=Z^4$),
then, from Theorem \ref{theConvIneff}, for
%multivariate CLT, for
\begin{equation}\label{sigmamsq}
\Sigma_{\msq}(b') =\E_{\PQ_1}(L'Z^4\nabla L(b^*)(\nabla L(b^*))^T),\\
\end{equation}
we have (\ref{nablaftng}) for $g=\msq$. 

Let us further in this section assume Condition \ref{condpqpq1} and let
\begin{equation}\label{rnbdef}
\wh{1}_n(b',b)=\overline{\left(\frac{L'}{L(b)}\right)}_n,\quad n \in \N_+
\end{equation}
(see (\ref{cons1})). 
Consider now the case of $g=\msq2$. We have $\wh{\msq2}_{n}=\wh{\msq}_n\wh{1}_n$ and thus
%\overline{\left(L^{-2}(b)L'\nabla L(b)\right)}_n 
%\end{equation}
\begin{equation}
\nabla_b\wh{\msq2}_{n}=(\nabla_b\wh{\msq}_n)\wh{1}_n+\wh{\msq}_n\nabla_b \wh{1}_n.
\end{equation}
 Let 
 \begin{equation}\label{tndef}
 \begin{split}
 T_t(b')&= \sqrt{t}\I(N_t =k\in \N_+)(\nabla_b\wh{\msq}_{k}(b',b^*) + \msq(b^*)\nabla_b\wh{1}_{k}(b',b^*))(\wt{\kappa}_{k})\\
 &= \sqrt{t}\I(N_t =k\in \N_+)\overline{(L'\nabla L(b^*)(Z^2 -\msq(b^*)L(b^*)^{-2}))}_{k}(\wt{\kappa}_{k}) \\
 \end{split}
 \end{equation}
 and
 \begin{equation}\label{znDef}
 \begin{split}
 Z_{t}(b')&=\sqrt{t}\I(N_t= k\in \N_+)\nabla_b\wh{\msq2}_{k}(b',b^*)(\wt{\kappa}_{k}) - T_t(b')\\
 &= \I(N_t =k\in \N_+)(((\wh{1}_{k}-1)\sqrt{t}\nabla_b\wh{\msq}_{k})(b',b^*)\\
 &+(\wh{\msq}_{k}(b',b^*)-\msq(b^*))\sqrt{t}\nabla_b\wh{1}_{k}(b',b^*))(\wt{\kappa}_{k}). \\
 \end{split}
 \end{equation}
 Let
 \begin{equation}\label{zeroderivL}
 0 =\E_{\PQ_1}(\nabla_b(L^{-1}(b^*)))
 \end{equation}
 (see the first point of Theorem \ref{thDiff} for sufficient conditions for this in the LETS setting) 
 and $\E_{\PQ_1}(L'L(b)^{-4}(\partial_i L(b))^2) < \infty$, $i=1,\ldots,l$.
Then, from Theorem \ref{theConvIneff},
% %multivariate CLT,
 \begin{equation}\label{Rnas}
 \sqrt{t}\I(N_t=k\in\N_+)\nabla_b\wh{1}_{k}(b',b)(\wt{\kappa}_{k})\Rightarrow \ND(0,u(b')\E_{\PQ_1}(L'L(b)^{-4}\nabla L(b)(\nabla L(b))^T)).
 \end{equation}
Assuming further (\ref{msqCLTInt}) and (\ref{msqbs0}), 
from (\ref{nablaftng}) for $g=\msq$,  (\ref{cons1}), (\ref{Rnas}), the fact that 
 from the SLLN and Condition \ref{condT}, a.s. $\I(N_t=k\in\N_+)\wh{\msq}_{k}(b',b^*)(\wt{\kappa}_{k})\to \msq(b^*)$,
 as well as from (\ref{znDef}) and Slutsky's lemma,
 %(\ref{CLTnablamsq})
 \begin{equation}\label{zn0}
 Z_{t}(b') \overset{p}{\rightarrow} 0.
 \end{equation}
 Let
 \begin{equation}\label{sigmamsq2}
 %\begin{split}
 \Sigma_{\msq2}(b')%\E_{\PQ'}((Z^2 -\msq(b^*)L(b^*)^{-2})^2(L')^2\nabla L(b^*)(\nabla L(b^*))^T)\\
 =\E_{\PQ_1}(L'(Z^2 -\msq(b^*)L(b^*)^{-2})^2\nabla L(b^*)(\nabla L(b^*))^T).\\
 %\end{split}
 \end{equation}
 From Theorem \ref{theConvIneff}, $T_t(b')\Rightarrow \ND(0,u(b')\Sigma_{\msq2}(b'))$, so that from (\ref{zn0}), 
the first line of (\ref{znDef}), and Slutsky's lemma, we receive (\ref{nablaftng}).

Let us finally consider the case of $g=\ic$. We have for $n \in \N_2$
\begin{equation}
\nabla_b\wh{\ic}_{n}=(\nabla_b\wh{c}_{n})\wh{\var}_{n}+\wh{c}_n\nabla_b\wh{\var}_{n},
\end{equation}
where 
\begin{equation} 
\nabla_b\wh{\var}_{n}=\frac{n}{n-1}((\nabla_b\wh{\msq}_n(b',b))\wh{1}_n
+\wh{\msq}_n(b',b)\nabla_b \wh{1}_n)
\end{equation} 
Let for $D=(\R\times\R^l)^3\times\R$ and $n \in \N_+$, $U_n(b'):\Omega_1^n \to D$ be equal to
\begin{equation}\label{sndef}
(\wh{c}_{n}(b',b^*),\nabla_b\wh{c}_{n}(b',b^*), 
\wh{\msq}_n(b',b^*),\nabla_b \wh{\msq}_n(b',b^*), \wh{1}_n(b',b^*), 
\nabla_b\wh{1}_n(b',b^*), \overline{(ZL')}_n)
\end{equation}
and let
\begin{equation}
\theta:= (c(b^*),\nabla c(b^*), \msq(b^*),\nabla \msq(b^*), 1,0,\alpha)=\E_{\PQ'}(U_1(b')) \in D.  
\end{equation}
Let the coordinates of $U_1(b')$ be square-integrable under $\PQ'$ and   
%move derivatives under integrals in formulas for $\nabla \msq(b)$, $\nabla c(b)$, and $\nabla\E_{\PQ_1}(L^{-1}(b))=0$ (see theorems \ref{thYmore} and \ref{thDiff}
%for sufficient assumptions for this in the LETS setting). Let us denote 
\begin{equation}\label{psidef}
\Psi:=\E_{\PQ'}((U_1(b')-\theta)(U_1(b')-\theta)^T). 
\end{equation}
Then, from Theorem \ref{theConvIneff}, %(\wt{\kappa}_{N_t})
\begin{equation}\label{uconv}
\sqrt{t}\I(N_t=k \in\N_+)(U_{k}(\wt{\kappa}_{k})-\theta)\Rightarrow \ND(0,u(b')\Psi). 
\end{equation}
For $\phi: D\to \R^l$ such that 
\begin{equation}
\phi((x_i)_{i=1}^7)=x_2(x_3x_5 -x_7^2) +  x_1(x_4x_5 +x_3x_6), 
\end{equation}
we have for $n \in \N_2$ 
\begin{equation}\label{icphi}
\nabla_b\wh{\ic}_{n}(b',b^*)=\frac{n}{n-1}\phi(U_n(b')) 
\end{equation}
and $\nabla\ic(b^*)= \frac{n}{n-1}\phi(\theta)$. Let us assume that 
\begin{equation}\label{nablaicb}
\nabla\ic(b^*)= \phi(\theta)= 0.
\end{equation}
Using the delta method from Theorem \ref{thdelta}, as well as Remark \ref{remdelta} and (\ref{uconv}), 
%(see e.g. Theorem 3.1 in \cite{Vaart}), 
we receive that for %for $T\sim \ND(0,\Psi)$
\begin{equation}\label{sigmaicdef}
\Sigma_{\ic}(b') = \phi'(\theta)\Psi (\phi'(\theta))^T
\end{equation}
we have 
\begin{equation}\label{uasymp}
\sqrt{t}\I(N_t=k \in \N_+)\phi(U_{k}(b')(\wt{\kappa}_{k})) \Rightarrow \ND(0, u(b')\Sigma_{\ic}(b')). %\phi'(\theta)T. 
\end{equation}
% and from Remark \ref{remdelta}, for
% we have 
% \begin{}
% $\phi'_{\theta}T\sim \ND(0, \Sigma_{\ic}(b'))$. 
From $\lim_{n\to \infty}\frac{n}{n-1}=1$, (\ref{icphi}), (\ref{uasymp}), and Slutsky's lemma, we thus have
(\ref{nablaftng}).
% \begin{equation}\(\ref{sigmaicdef}),
% \sqrt{t}\I(N_t \in \N_+)(\nabla_b\wh{\ic}_{N_t}(b',b^*)(\wt{\kappa}_{N_t})\Rightarrow  \ND(0, u(b')\Sigma_{\ic}(b')).
% \end{equation}
% Thus, assuming that Condition \ref{condAllAs} 
% holds for $f=\ic$ and $f_t(b)=\I(N_t \in \N_+)\wh{\ic}_{N_t}(b',b)(\wt{\kappa}_{N_t})$, from Theorem \ref{thAsympMin}, for 
% \begin{equation}\label{vicdef}
% V_{\ic}(b')=H_{\ic}^{-1}\Sigma_{\ic}(b')H_{\ic}^{-1}
% \end{equation}
% we have
% \begin{equation}\label{bnbsic}
% \sqrt{t}(d_t-b^*) \Rightarrow \ND(0, u(b')V_{\ic}(b')). 
% \end{equation}
From (\ref{psidef}) and (\ref{sigmaicdef}), for 
%\begin{equation}
$W(b'):=\phi'(\theta)(U_1(b')-\theta)$
%\end{equation}
\begin{equation}\label{sigmaicdef2}
\Sigma_{\ic}(b') = \E_{\PQ'}(W(b')W(b')^T). 
\end{equation}
We have
\begin{equation}
 \begin{split}
 \phi'(\theta)((x_i)_{i=1}^7)&=\nabla \msq(b^*)x_1+\var(b^*)x_2+\nabla c(b^*)x_3\\
 &+c(b^*)x_4+(\nabla c(b^*)\msq(b^*)+c(b^*)\nabla\msq(b^*))x_5\\
 &+c(b^*)\msq(b^*)x_6-2\nabla c(b^*)\alpha x_7,  \\
 \end{split}
 \end{equation}
% %\phi((x_i)_{i=1}^7)=x_2(x_3x_5 -x_7^2) +  x_1(x_4x_5 +x_3x_6) 
% %\theta= (c(b^*),\nabla c(b^*), \msq(b^*),\nabla \msq(b^*), 1,0,\alpha) \in D.  
% and 
% \begin{equation}
% \begin{split}
% U_1(b')-\theta&= (CL'L(b^*)^{-1}-c(b^*), -CL'L^{-2}(b^*)\nabla L(b^*)-\nabla c(b^*),
% Z^2L'L(b^*)-\msq(b^*),\\
% &Z^2L'\nabla L(b^*) -\nabla \msq(b^*), L' L^{-1}(b^*)-1, -L'L(b^*)^{-2}\nabla L(b^*), ZL'-\alpha), 
% \end{split}
% \end{equation}
so that
\begin{equation}\label{wbdef}
\begin{split}
W(b')&= \nabla \msq(b^*)(CL'L(b^*)^{-1}-c(b^*))+ \var(b^*)(-CL'L^{-2}(b^*)\nabla L(b^*)-\nabla c(b^*))\\
&+\nabla c(b^*)(Z^2L'L(b^*)-\msq(b^*))
+c(b^*)(Z^2L'\nabla L(b^*) -\nabla \msq(b^*))\\
&+(\nabla c(b^*)\msq(b^*)+c(b^*)\nabla\msq(b^*))(L' L^{-1}(b^*)-1)\\
&-c(b^*)\msq(b^*)L'L(b^*)^{-2}\nabla L(b^*)-2\nabla c(b^*)\alpha(ZL'-\alpha). 
\end{split}
\end{equation}
\begin{remark}
%\phi'(\theta)((x_i)_{i=1}^7)&=\nabla \msq(b^*)x_1+\var(b^*)x_2+\nabla c(b^*)x_3+c(b^*)x_4\\
%&+(\nabla c(b^*)\msq(b^*)+c(b^*)\nabla\msq(b^*))x_5+c(b^*)\msq(b^*)x_6-2\nabla c(b^*)\alpha x_7  
Let us make assumptions as above and that $C=1$. Then, we have $c(b)=1$, $\nabla c(b)=0$, and $\ic(b)=\var(b)$, $b \in A$, and
from (\ref{nablaicb}), $\nabla \msq(b^*) =\nabla \var(b^*)=0$, so that from (\ref{wbdef}),
\begin{equation}
W(b')=L'\nabla L(b^*)(Z^2-(\var(b^*)+\msq(b^*))L(b^*)^{-2}),  
\end{equation}
and thus 
\begin{equation}\label{sigmaicc1}
\Sigma_{\ic}(b')= \E_{\PQ_1}(L'(Z^2-(\var(b^*)+\msq(b^*))L(b^*)^{-2})^2\nabla L(b^*)(\nabla L(b^*))^T).   
\end{equation}
\end{remark}

\section{A helper CLT}\label{secHelpCLT}
For some $l\in \N_+$, consider a nonempty set $A \in \mc{B}(\R^l)$ and a 
corresponding family of probability distributions as in Section \ref{secFamily}. Let $m \in \N_+$, 
$u:\mc{S}(A)\otimes \mc{S}_1 \to \mc{S}(\R^m)$, and $B\in \mc{B}(A)$ be  nonempty.
\begin{condition}\label{condCLT}
For
\begin{equation}
f(b,M):= \E_{\PQ(b)}(|u(b,\cdot)|\I(|u(b,\cdot)|>M)),\quad b \in B,\ M \in \R,
\end{equation}
and $R(M):= \sup_{b \in B} f(b,M)$, $M \in \R$, we have 
\begin{equation}\label{minfR}
\lim_{M\rightarrow \infty} R(M)= 0.
\end{equation}
\end{condition}
Note that the above condition is equivalent to saying that for random variables $\psi_b \sim \PQ(b)$, $b \in B$, %(assuming that they exist) 
the family $\{|u(b,\psi_b)|:b \in B\}$ is uniformly integrable. In particular, similarly as for uniform integrability,
using H\"{o}lder's inequality one can prove the following criterion for the above condition to hold. 
\begin{lemma}
% Note that for $s:\mc{F}_1\otimes A\to \R$, and $B\subset A$, under (\ref{rtdef}),  for (\ref{mtoinftys}) to hold 
If for some $p>1$, 
\begin{equation}
\sup_{b\in B}\E_{\PQ(b)}(|u(b,\cdot)|^p)<\infty,  
\end{equation}
then Condition \ref{condCLT} holds.
\end{lemma}
%\begin{proof}
%from H\"{o}lder inequality for $b \in B$ and $q$ such that $\frac{1}{p}+ \frac{1}{q}=1$
% \begin{equation}
% \begin{split}
% \E_{\PQ(b)}(|s(b,\cdot)|\I(|s(b,\cdot)|>t)&\leq  (\PQ(b)(s(b,\cdot)>t))^{\frac{1}{q}}(\E_{\PQ(b)}(s(b,\cdot)^p))^{\frac{1}{p}}\\
% &\leq (\frac{\E_{\PQ(b)}(s(b,\cdot)^p)}{t^p})^{\frac{1}{q}}(\E_{\PQ(b)}(s(b,\cdot)^p))^{\frac{1}{p}}\\
% &\leq Mt^{-\frac{p}{q}}.
% \end{split}
% \end{equation}
% and thus $r_t\leq Mt^{-\frac{p}{q}}\to 0$ as $t \to \infty$. 
% \end{remark}
For some $T \in \N_+\cup \infty$, $\overline{\R}$-valued random variables $(\psi_i)_{i=1}^T$  are said to be martingale differences for 
a filtration $(\mc{F}_i)_{i=0}^T$, if  $M_n= \sum_{i=1}^n\psi_i$, $n \in \N$,
is a martingale for $(\mc{F}_i)_{i=0}^T$, that is if 
$\E(|\psi_i|)<\infty$, $\psi_i$ is $\mc{F}_i$-measurable, and $\E(\psi_i|\mc{F}_{i-1})=0$, $i=1,\ldots,T$. 
The following martingale CLT is a special case of Theorem 8.2 with conditions II, page 442 in \cite{lipster89}. 
%(see also Corollary 3.6 with conditions 3.22 in \cite{rootzen1983}). 
\begin{theorem}\label{thMartCLT} 
For each $n \in \N_+$, let $m_n \in \N_+$, $(\mc{F}_{n,k})_{k=0}^{m_n}$ be a filtration, 
and $(\psi_{n,k})_{k=1}^{m_n}$ be martingale differences for it such that $\E(\psi_{n,k}^2)<\infty$, $k=1,\ldots,m_n$. Let further
\begin{enumerate} 
\item for each $\delta>0$, $\sum_{k=1}^{m_n}\E(\psi_{n,k}^2\I(|\psi_{n,k}|>\delta)|\mc{F}_{n,k-1}) \overset{p}{\to} 0$ (as $n \to \infty$), 
\item for some $\sigma \in [0,\infty)$, $\sum_{k=1}^{m_n}\E(\psi_{n,k}^2|\mc{F}_{n,k-1}) \overset{p}{\to} \sigma^2$. 
\end{enumerate} 
Then, 
\begin{equation} 
\sum_{k=1}^{m_n}\psi_{n,k} \Rightarrow \ND(0,\sigma^2). 
\end{equation} 
\end{theorem} 
% for some neighbourhood $D$ of $b^*$ 
% (in particular a.s. for a sufficiently large $i$, $b_i \in D$). 

Let $r:\mc{S}(A)\otimes \mc{S}_1 \to \mc{S}(\R^m)$ be such that for each $b \in A$,
\begin{equation}
\E_{\PQ(b)}(r(b,\cdot))=0.
\end{equation} 
Consider a matrix 
\begin{equation}
\Sigma(b)=\E_{\PQ(b)}(r(b,\cdot)r(b,\cdot)^T)
\end{equation}
for $b \in A$ for which it is well-defined. 
Note that under Condition \ref{condCLT} for $u=|r|^2$, from (\ref{minfR}) we have $R(M)<\infty$ for some $M>0$, and thus
% \begin{equation}\label{rofin}
$R(0) \leq M  + R(M)<\infty$ and $\Sigma(b)\in \R^{m\times m}$,  $b\in B$. 

\begin{theorem}\label{thCLTM} 
Let us assume that Condition \ref{condChi} holds and we have
\begin{equation}\label{limnk}
\lim_{k\rightarrow \infty}n_k=\infty.
\end{equation}
Let further $B$ as above be a neighbourhood of $b^*\in A$, 
\begin{equation}\label{bnpto0} 
b_n\overset{p}{\rightarrow} b^*,
\end{equation} 
Condition \ref{condCLT} hold for $u=|r|^2$, and
$b\in B\to \Sigma(b)$ be continuous in $b^*$.  
Then, %for $W_k:=\frac{1}{\sqrt{n_k}}\sum_{i=1}^{n_k}r(b_{k-1},\chi_{k,i})$, $k \in \N_+$,
%\begin{equation}
% $\wh{r}_k = \frac{1}{n_k}\sum_{i=1}^{n_k}r(b_{k-1},\chi_{k,i})$ 
% %\end{equation}
% and $W_k=\sqrt{n_k}\wh{r}_k$
% \begin{equation}\label{nakam} 
% W_k=\sqrt{n_k}\wh{r}_k=\frac{1}{\sqrt{n_k}}\sum_{i=1}^{n_k}r(b_{k-1},\chi_{k,i})
% \end{equation}
%we have
\begin{equation}
\frac{1}{\sqrt{n_k}}\sum_{i=1}^{n_k}r(b_{k-1},\chi_{k,i})\Rightarrow \ND(0,\Sigma(b^*)).
\end{equation}
\end{theorem}
\begin{proof}
Let $W_k=\frac{1}{\sqrt{n_k}}\sum_{i=1}^{n_k}r(b_{k-1},\chi_{k,i})$
and $D_k=\I(b_{k-1} \in B)W_k$, $k \in \N_+$. From (\ref{bnpto0}) we have $W_k-D_k = \I(b_{k-1} \notin B)W_k\overset{p}{\to} 0$,
so that from Slutsky's lemma it is sufficient to prove 
that $D_k \Rightarrow \ND(0,\Sigma(b^*))$. Furthermore, using Cram\'{e}r-Wold device %(see p. 16 in \cite{Vaart}) 
it is sufficient to prove that for each $t \in \R^l$, for 
$v(b):=t^T\Sigma(b)t$, $b \in B$, 
% \begin{equation}
% v(b):=t^T\Sigma(b)t,\quad b \in B, 
% \end{equation}
and $S_k: =t^TD_k$, we have $S_k \Rightarrow \ND(0,v(b^*))$. For $t=0$ this is obvious, so let us consider $t\neq 0$.  
It is sufficient to check that the assumptions of Theorem \ref{thMartCLT} hold for 
$m_k=n_k$, $\mc{F}_{k,i}=\sigma(b_{k-1};\ \chi_{k,j}:j\leq i)$, 
$\psi_{k,i}=\frac{1}{\sqrt{n_k}}\I(b_{k-1}\in B)t^Tr(b_{k-1},\chi_{k,i})$, and $\sigma^2=v(b^*)$. 
From Condition \ref{condCLT} (for $u=|r|^2$),  
\begin{equation}\label{psiki2}
\begin{split}
\E(\psi_{k,i}^2)&= \frac{1}{n_k}\E((\E_{\PQ(b)}(\I(b\in B)(t^Tr(b,\cdot))^2))_{b=b_{k-1}}) \\
&\leq \frac{|t|^2}{n_k}\E((\E_{\PQ(b)}(\I(b\in B)|r(b,\cdot)|^2))_{b=b_{k-1}}) \\
&\leq \frac{|t|^2}{n_k}\E(\I(b_{k-1}\in B)f(b_{k-1},0))\\
&\leq \frac{|t|^2}{n_k} R(0) <\infty.  
\end{split}
\end{equation}
For $\delta>0$, from Condition \ref{condCLT} and (\ref{limnk}), 
\begin{equation}\label{sumCLT1}
\begin{split}
\sum_{i=1}^{m_k}\E(\psi_{k,i}^2\I(|\psi_{k,i}|>\delta)|\mc{F}_{k,i-1})&=
(\I(b\in B)\E_{\PQ(b)}(|t^Tr(b,\cdot)|^2\I(t^Tr(b,\cdot)>\sqrt{n_k}\delta)))_{b=b_{k-1}}\\
&\leq |t|^2(\I(b\in B)\E_{\PQ(b)}(|r(b,\cdot)|^2\I(|t||r(b,\cdot)|>\sqrt{n_k}\delta)))_{b=b_{k-1}}\\
%&=\frac{1}{n_k}\sum_{i=1}^{n_k}\E(\E(|(ZL(v))(\xi(\beta_{k,i},v))-\alpha|^2\I(|ZL(v)(\xi(\beta_{k,i},v))-\alpha|>\sqrt{n_k}\delta))_{v=c_{k-1}})\\
&= |t|^2\I(b_{k-1}\in B)f(b_{k-1},n_k(\frac{\delta}{|t|})^2)\\
&\leq |t|^2R(n_k(\frac{\delta}{|t|})^2) \rightarrow 0, \quad k \to \infty.\\
\end{split}
\end{equation}
 To prove the second point of Theorem \ref{thMartCLT} let us notice that 
\begin{equation}\label{limpvbs}
\begin{split}
\sum_{i=1}^{m_k}\E(\psi_{k,i}^2|\mc{F}_{k,i-1})&= (\I(b\in B)\E_{\PQ(b)}(|t^Tr(b,\cdot)|^2))_{b=b_{k-1}}\\
&= (\I(b\in B)v(b))_{b=b_{k-1}} \overset{p}{\rightarrow} v(b^*), 
\end{split}
\end{equation}
where in the last line we used (\ref{bnpto0}) and the continuity of $b\to \I(b \in B)v(b)$ in $b^*$.
\end{proof}

We will be most interested in the IS case in which we shall assume the following condition.
\begin{condition}\label{condIScaseM}
For some $\R^m$-valued random variable $Y$ on $\mc{S}_1$, Condition 
\ref{condpqbllpq1} holds for $Z$ replaced by $Y$,
and we have 
\begin{equation}\label{IScaseM}
r(b,\omega)=(YL(b))(\omega), \quad  b\in A,\ \omega \in \Omega_1.
\end{equation}
\end{condition}

\begin{lemma}\label{lemSuffCLT}
Let us assume Condition \ref{condIScaseM}
and that for $F=\sup_{b\in B}\I(Y\neq 0)L(b)$ we have 
\begin{equation}\label{pqzfm}
\E_{\PQ_1}(|Y|^2F)<\infty. 
\end{equation}
Then, Condition \ref{condCLT} holds for $u=|r|^2$. If further $\PQ_1$ a.s. $b\to L(b)$ is continuous, then 
$b\in B\to \Sigma(b)$ is continuous.
\end{lemma}
\begin{proof}
For each $M \in \R$ and $b \in B$
\begin{equation}\label{fbmm}
\begin{split}
f(b,M)&= \E_{\PQ(b)}(|YL(b)|^2\I(|YL(b)|>\sqrt{M}))\\
&=\E_{\PQ_1}(|Y|^2L(b)\I(|YL(b)|>\sqrt{M}))\\
&\leq\E_{\PQ_1}(|Y|^2F\I(|YF|>\sqrt{M})),\\
\end{split}
\end{equation}
so that
\begin{equation}
R(M)\leq\E_{\PQ_1}(|Y|^2F\I(|YF|>\sqrt{M})) \leq \E_{\PQ_1}(|Y|^2F)<\infty. 
\end{equation}
From (\ref{pqzfm}) we have $0=\PR(|Y|^2F=\infty)=\PR(|Y|F=\infty)$. Therefore, 
as $M\rightarrow\infty$, we have $\I(|YF|>\sqrt{M})\to 0$ and thus from Lebesgue's dominated convergence theorem also
$\E_{\PQ_1}(|Y|^2F\I(|YF|>\sqrt{M})) \to 0$ and $R(M) \to 0$, i.e. Condition \ref{condCLT} holds. We have 
\begin{equation}
\Sigma_{i,j}(b)= \E_{\PQ(b)}(Y_iY_jL(b)^2)= \E_{\PQ_1}(Y_iY_jL(b)) 
\end{equation}
and $|Y_iY_jL(b)| \leq |Y|^2F$, $b \in B$. Thus, the continuity of $b\in B\to \Sigma_{i,j}(b)$ (and thus 
also of $b\in B\to \Sigma(b)$) follows from Lebesgue's dominated convergence theorem. 
\end{proof}
\begin{remark}\label{remSuffpqzfm}
From Theorem \ref{thYmore} and Remark \ref{remCondUN}, in the LETS setting, 
for $B$ bounded,
under Condition \ref{condIScaseM}, and for $F$ as in Lemma \ref{lemSuffCLT},  
(\ref{pqzfm}) holds if Condition \ref{condlemYmore} holds for $S=|Y|^2$. 
%(\ref{bnpto0}) in Theorem \ref{thCLTM} in different MSM methods from previous sections follows from Condition \ref{condbntobs} for $d_n=b_n$. 
%sufficient assumptions for which were discussed in Section \ref{secSomeUs}. 
\end{remark}

\section{\label{secAsympMSM}Asymptotic properties of multi-stage minimization methods} 
Consider the following conditions. 
%For various MSM methods we shall need the following condition to hold. 
\begin{condition}\label{condbkconv} 
We have $d^* \in A \in \mc{B}(\R^l)$ and $A$-valued random 
variables $b_k$, $k\in \N$, are such that 
\begin{equation}\label{bktods} 
b_k \overset{p}{\to} d^*. 
\end{equation} 
\end{condition} 
% CGMSM and MGMSM instead of conditions \ref{condbkdef} and \ref{condESSM1} we mean their counterparts as discussed 
% above Remark \ref{remCountKiA}. 
% \begin{condition}\label{condAll} 
% Conditions \ref{condbkdef}, \ref{condESSM1}, and \ref{condKiA} hold. 
% \end{condition} 
\begin{remark}\label{remdsbs}
From Remark \ref{remcondKiA} and its counterparts for CGMSM and MGMSM as discussed in Remark \ref{remCountKiA}, 
under conditions \ref{condbkdef} and \ref{condESSM1} for EMSM and GMSM or their 
counterparts for CGMSM and MGMSM as discussed above Remark \ref{remCountKiA},
as well under Condition \ref{condKiA}, we have a.s. for a sufficiently large $k$, $b_k=d_k$ for EMSM and GMSM, and 
$b_k=\wt{d}_k$ for CGMSM and MGMSM. In particular, a.s. $\lim_{k\to \infty} b_k=b^*$ and thus 
Condition \ref{condbkconv} holds for $d^*=b^*$.   
\end{remark}
Let us further in this section assume conditions \ref{condChi} and \ref{condbkconv}. 
Using analogous reasonings and assumptions as in Section \ref{secAsympSSM}, but 
using CLT from Theorem \ref{thCLTM} for $b^*$ replaced by $d^*$ instead of Theorem \ref{theConvIneff}, 
we receive that under the appropriate assumptions as in that theorem we have 
for $g$ in the below formulas substituted by $\ce$, $\msq$, or $\ic$, that 
for $f_k(b)= \wh{g}_{n_k}(b_{k-1},b)(\wt{\chi}_k)$ 
\begin{equation}\label{asympnkdvc} 
\sqrt{n_k} \nabla f_k(b^*) \Rightarrow \ND(0,\Sigma_{g}(d^*)). 
\end{equation} 
To prove (\ref{asympnkdvc}) for $g=\msq2$ using a reasoning analogous as in Section \ref{secAsympSSM} 
we additionally need the facts that 
$\wh{\msq}_{n_k}(b_{k-1},b^*)(\wt{\chi}_k)\overset{p}{\to} \msq(b^*)$ and $\wh{1}_{n_k}(b_{k-1},b^*)(\wt{\chi}_k)\overset{p}{\to} 1$. 
Under appropriate assumptions, such convergence results 
follow from the convergence in distribution of $\sqrt{n_k}(\wh{\msq}_{n_k}(b_{k-1},b^*)(\wt{\chi}_k)-\msq(b^*))$ 
and $\sqrt{n_k}(\wh{1}_k(b_{k-1},b^*)(\wt{\chi}_k)- 1)$, which can be proved using Theorem \ref{thCLTM} as above.
For different  $g$ as above, assuming (\ref{asympnkdvc}) and that 
Condition \ref{condAllAs} holds for $T=\N_+$, $r_k=n_k$, $f_k$ as above, and $f$ corresponding to $g$ as in Remark \ref{remSSMAsymp} 
(see Section \ref{secSomeUs} for some sufficient assumptions), 
for $V_g$ and $B_g$ as in that remark 
%$f=g$ for $g$ other than $\msq2$ for which $f=\msq$), 
we have from Theorem \ref{thAsympMin} 
\begin{equation}\label{nkdkvf} 
\sqrt{n_k}(d_k - b^*)\Rightarrow \ND(0,V_{g}(d^*)), 
\end{equation} 
and from Remark \ref{remHY} 
\begin{equation}\label{nkfdkfb} 
n_k(f(d_k) - f(b^*))\Rightarrow \wt{\chi^2}(B_{g}(d^*)). %\frac{1}{2}H_f^{-\frac{1}{2}}\Sigma_{g}(d^*)H_f^{-\frac{1}{2}}). 
\end{equation} 
If for $s_k$ denoting the number of samples generated till the $k$th stage of MSM, i.e. $s_k:=\sum_{i=1}^kn_i$, $k \in \N_+$, we have 
\begin{equation}\label{limsknk} 
\lim_{k\to \infty}\frac{s_k}{n_k}=\gamma \in [1,\infty), 
\end{equation} 
then from (\ref{nkdkvf}) it follows that 
\begin{equation}\label{skdk}
\sqrt{s_k}(d_{k} - b^*)\Rightarrow \ND(0,\gamma V_{g}(d^*)), 
\end{equation}
while from (\ref{nkfdkfb}) --- that 
\begin{equation}\label{skfdk} 
s_k(f(d_k) - f(b^*))\Rightarrow \wt{\chi^2}(\gamma B_{g}(d^*)). 
\end{equation}
%In such case, from the discussion in Section \ref{secAsympFun} for $r_k=s_k$, $k \in \N_+$, the asym
%$\gamma=1$ or 
% in Section \ref{secAsympFun} if $\gamma>1$ 
% and $B_{g}(d^*)\neq 0$, and they are equally efficient otherwise. 
For instance, for $n_k=A_1 + A_2m^k$ as in Remark \ref{remniLi}, we have $s_k=A_1k +A_2\frac{m^{k+1} - 1}{m-1}$, 
so that $\gamma=\frac{m}{m-1}$. For $n_k=A_1 + A_2k!$ as in that remark, we have 
\begin{equation}
\frac{s_k}{n_k}\leq\frac{s_k}{A_2k!}=\frac{A_1}{A_2(k-1)!} +1+  \frac{1}{k}+\frac{1}{k(k-1)}+\ldots +\frac{1}{k!}\leq \frac{A_1}{A_2(k-1)!} + 1+\frac{2}{k}, 
\end{equation}
so that $\gamma=1$.

\begin{remark}\label{remPs}
Let us assume that (\ref{skfdk}) holds and let $p_{s_k}=d_{k}$, $k \in \N_+$, 
i.e. $p$ is the process of MSM results but indexed 
by the total number of the generated samples rather than the number of stages.  
 Consider now $d_k$ as in SSM in Section \ref{secAsympSSM} 
 for $T=\N_+$, $N_k=k$, $k \in T$,  and $b'=d^*$, so 
 that we have (\ref{nttpto}) for $u(b')=u(d^*)=1$. 
 Let us assume that (\ref{fdtSSM}) holds for such $d_k$, 
 and let $p_{s_k}'= d_{s_k}$, $k \in \N_+$. Let further $T=\{s_k:k\in\N_+\}$.  
 Then, $t(f(p_t) - f(b^*))\Rightarrow \wt{\chi^2}(\gamma B_{g}(d^*))$ 
 and $t(f(p_t') - f(b^*))\Rightarrow \wt{\chi^2}(B_{g}(d^*))$,
 $t \to \infty$, $t \in T$. 
 Thus, for $\gamma=1$ or $B_{g}(d^*)=0$, the processes $p$ and $p'$ 
 can be considered asymptotically equally efficient for the minimization of 
 $f$ in the second-order sense as discussed in Section \ref{secSecond}, 
 while for $\gamma>1$ and $B_{g}(d^*)\neq 0$
 the process from SSM can be considered more efficient 
 than the one from MSM  in such sense.
\end{remark}  

\begin{remark}\label{remCompMSM}
From the discussion in Section \ref{secSecond} for $T=\N_+$ and $r_k=n_k$, $k\in T$, for $R\sim \wt{\chi^2}(B_{g}(d^*))$, 
the second-order asymptotic inefficiency for the minimization of $f$ of processes
$d=(d_k)_{k\in \N_+}$  from MSM satisfying (\ref{nkfdkfb}) like above, 
e.g. for different $g$ or $d^*$ but for the same $b^*$ and $n_k$, can be quantified using
\begin{equation}\label{quantMSM}
\E(R)= \frac{1}{2}\Tr(\Sigma_{g}(d^*)H_f^{-1}). 
\end{equation}
Let (\ref{skfdk}) hold and consider a process $p_{s_k}=d_{k}$, $k \in \N_+$, as in Remark \ref{remPs}. 
Then, the asymptotic inefficiency of $p$ for the minimization of $f$ can be quantified using
\begin{equation}\label{quantMSM2}
\gamma\frac{1}{2}\Tr(\Sigma_{g}(d^*)H_f^{-1}). 
\end{equation}
Using (\ref{quantMSM2}) one can compare the asymptotic efficiency of such a process $p$ from MSM with that of a process $p'$ from SSM as in Remark \ref{remCompMSM}, 
but this time without assuming that $b'=d^*$, so that 
the inefficiency of $p'$ is quantified by $\frac{1}{2}\Tr(\Sigma_{g}(b')H_f^{-1})$. 
In particular, if $\gamma\Tr(\Sigma_{g}(d^*)H_f^{-1})
<\Tr(\Sigma_{g}(b')H_f^{-1})$ then $p$ can be considered asymptotically more efficient for the minimization of $f$
than $p'$.

Consider further the mean theoretical cost $u$ of MSM steps, analogous as in Remark \ref{remub} for SSM. 
For two MSM processes $d$ as above 
for which $u$ is continuous in the corresponding points $d^*$, 
(\ref{quantMSM}) is positive and not higher for the first process than for the second one, 
and $u(d^*)$ is lower for the first process than for the second one 
if the constants $p_{\dot{U}}$ as in Remark \ref{remCondT} for these processes are the same, or
the mean practical cost $p_{\dot{U}}u(d^*)$ of this process is lower if these constants are different,
it seems reasonable to consider the first process asymptotically more efficient for the minimization of $f$. 
More generally, by analogy to formula (\ref{quantSSM}) for SSM, rather than using (\ref{quantMSM}),
one can quantify the asymptotic inefficiency of MSM processes $d$ as above by
\begin{equation}\label{quantMSMfix}
u(d^*)\frac{1}{2}\Tr(\Sigma_{g}(d^*)H_f^{-1}), 
\end{equation}
or such a quantity multiplied by $p_{\dot{U}}$ respectively. 
\end{remark}

A more desirable possibility than having Condition \ref{condbkconv} satisfied for $d^*=b^*$ as discussed in Remark \ref{remdsbs} 
(where for the minimization methods from the previous sections such a $b^*$ is equal to the unique minimum point of the minimized function $f$), 
may be to have it fulfilled for 
$d^*$ minimizing some measure of the asymptotic inefficiency of $d_k$ for the minimization of $f$, like 
(\ref{quantMSM}) or (\ref{quantMSMfix}) (assuming that such a $d^*$ exists). 
See Chapter \ref{secConcl} for further discussion of this idea. 
From Remark \ref{remRabmin} in Section \ref{secAsympProp} it will follow that for $g=\msq$, the minimum point of  $\msq$
does not need to be the minimum of (\ref{quantMSM}) in the function of $d^*$. 

\section{\label{secAsympZero}Asymptotic properties of the minimization results of the new estimators when a zero- or optimal-variance IS parameter exists} 
\begin{condition}\label{conddiffest}
We have the assumptions as in Section \ref{secFindOpt} above Condition \ref{condEst}, $D$ as in that section is a neighbourhood of $b^*$, conditions 
\ref{condhperf} and \ref{condEst} hold, and for each $b' \in A$, $k \in \N_p$, and $\omega\in \Omega_1^k$,
$b\to \wh{\est}_k(b',b)(\omega)$ is differentiable in $b^*$. 
\end{condition}

\begin{theorem}\label{thAsd0Kappa} 
Let conditions \ref{condKappa} and \ref{conddiffest} hold, $T \subset \R_+$ be unbounded, 
$N_t$, $t \in T$, be $\N\cup\{\infty\}$-valued random variables, and
$b\in B\to f_t(b):=\I(N_t=k \in \N_p)\wh{\est}_{k}(b',b)(\wt{\kappa}_{k})$, $t\in T$. Then, it holds a.s.
\begin{equation}\label{nfnbs0p}
\nabla f_t(b^*)=0, \quad t \in T. 
\end{equation}
If further Condition \ref{condAllAs} holds for such $f_t$, then %(for $t\in T$ going to $\infty$)
\begin{equation}\label{superCan}
\sqrt{r_t}(d_t-b^*)\overset{p}{\to} 0, \quad t \to \infty. 
\end{equation} 
\end{theorem} 
\begin{proof}
From  Lemma \ref{lemMinKappa} we receive (\ref{nfnbs0p}). 
Thus, (\ref{superCan}) follows from Theorem \ref{thAsymp0}. 
\end{proof}

\begin{theorem}\label{thAsd0Chi} 
Let Condition \ref{condChi} hold for $n_k \in \N_p$, $k \in \N_+$, let Condition \ref{conddiffest} hold, 
and let $b\in B\to f_k(b):=\wh{\est}_{n_k}(b_{k-1},b)(\wt{\chi}_k)$, $k \in T:=\N_+$. Then, we have a.s. (\ref{nfnbs0p}). 
If further Condition \ref{condAllAs} holds for such $f_k$, then (\ref{superCan}) holds. 
\end{theorem} 
\begin{proof} 
From  Lemma \ref{lemMinChi} we have (\ref{nfnbs0p}), 
so that (\ref{superCan}) follows from Theorem \ref{thAsymp0}. 
\end{proof} 
\begin{remark}\label{remSymp0Est} 
Let conditions \ref{condLmes}, \ref{condpqpq1}, and \ref{condg0} hold, $A$ be a neighbourhood of $b^*$, and $b\to L(b)(\omega)$ be differentiable, 
$\omega \in\Omega_1$. Let $\wh{\est}_{n}$ be equal to $\wh{\msq2}_{n}$ and  $b^*$ be an optimal-variance IS parameter, or 
$\wh{\est}_{n}$ be equal to $\wh{\ic}_{n}$ and $b^*$ be a zero-variance IS parameter. 
Let further for SSM Condition \ref{condKappa} hold and $T$ and  $N_t$, $t \in T$, be as in Theorem \ref{thAsd0Kappa}, while for MSM 
let Condition \ref{condChi} hold for $n_k\in \N_p$, $k \in T:=\N_+$, for $p=1$ for $\wh{\est}_{n}=\wh{\msq2}_{n}$ or $p=2$ for 
$\wh{\est}_{n}=\wh{\ic}_{n}$. 
% Let for $A_k$ corresponding to the respective $\wh{\est}_{k}$ as in Remark \ref{remExact}, for SSM the variables
% $N_t$, $t \in T$, be as in Theorem \ref{thAsd0Kappa} and for MSM let (\ref{chiAnk}) hold. 
% $T \subset \R_+$ be unbounded, 
% $b\in B\to f_t(b):=\I(N_t=k \in \N_p)\wh{\est}_{k}(b',b)(\wt{\kappa}_{k})$, $t\in T$, 
% and $\N\cup\{\infty\}$-valued random variables $N_t$, $t \in T$, be such that $N_t\overset{p}{\to}\infty$ and 
% $\lim_{t\to \infty}\PR(N_t= k\in \N_p\wedge \wt{\kappa}_k\in A_k)=1$. 
Then, from remarks \ref{remCondhperf}, \ref{remExact}, and the above theorems, 
we have a.s. (\ref{nfnbs0p}) for $f_t$ as in Theorem \ref{thAsd0Kappa} for SSM, and as in Theorem \ref{thAsd0Chi} 
for MSM. If further Condition \ref{condAllAs} holds (see Section \ref{secSomeUs} for some sufficient assumptions), 
then we also have (\ref{superCan}) in these methods. If we have (\ref{superCan}) 
for $r_t$ growing to infinity faster than $t$ for SSM 
or than $n_t$ for MSM, i.e. such that $\lim \frac{t}{r_t}\to 0$ or $\lim \frac{n_t}{r_t}\to 0$ respectively, 
then we have in a sense faster rate of convergence of $d_t$ to $b^*$ than in Section \ref{secAsympSSM} for SSM 
or in Section \ref{secAsympMSM} for MSM respectively. 
\end{remark}

\section{\label{secAsympProp}Some properties of the matrices characterizing the asymptotic distributions when a zero- or optimal-variance IS parameter exists}
Let us further in this section assume conditions \ref{condpqpq1} and \ref{condg0}. 
Consider matrix-valued functions $\Sigma_g$, $V_g$, and $B_g$, 
given by the formulas from Section \ref{secAsympSSM} and considered on the subsets of $A$ on which these formulas make sense. 

From the reasonings in  Section \ref{secAsympSSM}, for each $b' \in A$, under Condition \ref{condKappa}, for $g$ replaced by $\msq2$ or $\ic$, 
for $f_n(b)=\wh{g}_n(b',b)(\wt{\kappa}_n)$, $n \in \N_+$, under appropriate assumptions $\Sigma_{g}(b')$ is the asymptotic covariance matrix of 
$\sqrt{n}\nabla f_n(b^*)$. Under appropriate assumptions as in Remark \ref{remSymp0Est},  including
$b^*$ being an optimal-variance IS parameter in the case of $g=\msq2$ or a zero-variance one for $g=\ic$, 
we have a.s. $\nabla f_k(b^*)=0$, $k \in \N_+$. Thus, in such cases $\Sigma_g(b')=0$, $b' \in A$. 
This can be also verified by the following more generally valid direct calculations. 
For $b^*$ being an optimal-variance IS parameter, from  (\ref{zlsmsq}), $\PQ_1$ a.s. 
(and thus from Condition \ref{condpqpq1} also $\PQ(b)$ a.s., $b \in A$) 
\begin{equation}\label{zmsql}
(ZL(b^*))^{2} =\msq(b^*),
\end{equation}
and thus from (\ref{sigmamsq2}), $\Sigma_{\msq2} =0$. Let now $b^*$ be a zero-variance IS parameter. 
Then, we have $\var(b^*)=0$ and under appropriate differentiability assumptions $\nabla \msq(b^*)=0$. 
%If $\PQ(b^*)=\PQ^*$, we have $\var(b^*)=0$ and $\nabla \msq(b^*)=0$, so that 
%\begin{equation}
%\nabla^2 \ic(b^*)= c(b^*)\nabla^2\msq(b^*).  
%\end{equation}
Using further (\ref{zmsql}) and the fact that $\PQ_1$ a.s. $ZL(b^*)= \alpha$, from (\ref{wbdef}) we receive that for each $b'\in A$, $\PQ'$ a.s.
\begin{equation} 
\begin{split}\label{wb0ic} 
W(b')&=\nabla c(b^*)(Z^2L'L(b^*)-\msq(b^*) +\msq(b^*)(L' L^{-1}(b^*)-1)-2\alpha(ZL'-\alpha))\\ 
&+c(b^*)\nabla L(b^*)L'(Z^2-\msq(b^*)L(b^*)^{-2}) =0.\\ 
\end{split} 
\end{equation} 
Thus, from (\ref{sigmaicdef2}), $\Sigma_{\ic}=0$. Using the notations as in Remark \ref{remSSMAsymp}, 
if $H_{\msq}$ is positive definite for $g=\msq2$ or $H_{\ic}$ is positive definite for $g=\ic$, %(\ref{Vmsq2}), 
and we have $\Sigma_{g}=0$ as above, then it also holds $V_{g}=B_g=0$. 
%not using the formulas for \Sigma_g is given below. 
%An alternative reasoning under slightly stronger assumptions 
%is given in below remark. 
% The below remark says how, under appropriate assumptions, 
% one can reach similar conclusions as above without using the formulas for $\Sigma_{\msq}$ or $\Sigma_{\ic}$. 
%\begin{remark} %and from Remark \ref{remicpos}, $H_{\ic}=c(b^*)\nabla^2\msq(b^*)$. 
% Note that $\Sigma_{g}(b')$ for $g$ replaced by $\msq2$ or $\ic$ arise as in Section \ref{secAsympSSM} for $T=\N_+$ and $N_k=k$ 
% as asymptotic covariance matrices of 
% $\sqrt{n}\nabla f_n(b)$ for $f_n(b)=\wh{g}_n(b',b^*)(\wt{\kappa}_n)$. However, from Remark \ref{remSymp0Est}, 
% for $b^*$ being an optimal-variance IS parameter for $g=\msq2$ or zero-variance IS parameter for $g=\ic$ 
%  we have $\nabla f_n(b^*)=0$ a.s., $n \in\N_+$. Thus, we receive that 
%  $\Sigma_g(b')=0$ as above. Since 
% \end{remark} 
%Let us assume that conditions (\ref{condB1}). 
%\ref{condpqbllpq1} holds. 

Let us further use the notations $p(b)$, $l(b)$, and $S(b)$ as in Section \ref{secAsympSSM}. We define the Fisher information matrix as 
\begin{equation}\label{Ibdef1} 
I(b)=\E_{\PQ(b)}(S(b)S(b)^T), \quad b\in A 
\end{equation} 
(assuming that it is well defined). 
It is well known that under appropriate assumptions, allowing one to move the derivatives inside the expectations in the below derivation, 
we have
\begin{equation}\label{ibdef} 
I(b)= -\E_{\PQ(b)}(\nabla^2l(b)). 
\end{equation}
The following derivation is as on page 63 in \cite{Vaart}. 
From $\E_{\PQ_1}(p(b))=1$, $b \in A$, we have $\E_{\PQ_1}(\nabla p(b))=0$, $b\in A$, and 
$0=\E_{\PQ_1}(\nabla^2 p(b))=\E_{\PQ(b)}(\frac{\nabla^2 p(b)}{p(b)})$, $b \in A$, so that taking the expectation with respect to $\PQ(b)$ of
\begin{equation}
\nabla^2 l(b)=\frac{\nabla^2 p(b)}{p(b)}-\frac{\nabla p(b)(\nabla p(b))^T}{p^2(b)},
\end{equation}
we receive (\ref{ibdef}). 

Let us define
\begin{equation}
R(a,b)= \E_{\PQ(a)}(\frac{L(b)}{L(a)}S(a)S(a)^T), \quad a,b \in A, 
\end{equation}
Note that
\begin{equation}\label{rbb}
R(b,b)=I(b).
\end{equation}
\begin{remark}\label{remRabpos}
For some $a,b \in A$, let $R(a,b)$ and $I(a)$ have real-valued entries. Then, $R(a,b)$ 
is positive definite only if for each $v \in \R^l$, $v \neq 0$, 
$\E_{\PQ(a)}(\frac{L(b)}{L(a)}|S^T(a)v|^2)>0$, which holds only if $\PQ(a)(|S^T(a)v|\neq 0)>0$, $v \in \R^l$, $v \neq 0$, 
and thus only if $I(a)$ is positive definite. 
\end{remark}
Let $b^*$ be an optimal-variance IS parameter. Then, from (\ref{zmsql})
\begin{equation}\label{sigmaceopt}
\begin{split}
\Sigma_{\ce}(b)&= \E_{\PQ_1}(Z^2L(b)S(b^*)S(b^*)^T)\\
&=\msq(b^*)\E_{\PQ_1}(\frac{L(b)}{L(b^*)^2}S(b^*)S(b^*)^T)\\
&= \msq(b^*)R(b^*,b)
\end{split}
\end{equation}
and
\begin{equation}\label{sigmamsqopt}
\begin{split}
\Sigma_{\msq}(b) &= \E_{\PQ_1}(Z^4L(b)L(b^*)^2S(b^*)S(b^*)^T)\\
&=\msq(b^*)^2\E_{\PQ_1}(\frac{L(b)}{L(b^*)^2}S(b^*)S(b^*)^T)\\
&=\msq(b^*)^2R(b^*,b). 
\end{split}
\end{equation}
For the cross-entropy, let us assume that $b^*$ is a zero-variance IS parameter, so that  $\PQ_1$ a.s. we have $ZL(b^*)= \alpha$. Then
\begin{equation}
\ce(b)= -\E_{\PQ_1}(Zl(b))=-\alpha\E_{\PQ_1}(L(b^*)^{-1}l(b)).
\end{equation}
Thus, assuming that one can move the derivatives inside the expectation %(for which in the LETS setting one can easily check that it  
%it is sufficient if Condition \ref{condlemYmore} holds for $S=Z$)
\begin{equation}\label{nabla2ce}
\nabla^2\ce(b)= -\alpha\E_{\PQ_1}(L(b^*)^{-1}\nabla^2l(b))=-\alpha\E_{\PQ(b^*)}(\nabla^2l(b)),
\end{equation}
in which case from (\ref{ibdef})
\begin{equation}\label{hceopt}
\begin{split}
H_{\ce}&=\nabla^2\ce(b^*)= \alpha I(b^*).\\
\end{split}
\end{equation}
Assuming that $I(b^*)$ is positive definite and $\alpha\neq 0$,  from $\msq(b^*)=\alpha^2$, 
(\ref{vgdef}), (\ref{sigmaceopt}), and (\ref{hceopt}) we have
\begin{equation}\label{vceeq}
V_{\ce}(b)= H_{\ce}^{-1}\Sigma_{\ce}(b)H_{\ce}^{-1}=I(b^*)^{-1}R(b^*,b)I(b^*)^{-1}, 
\end{equation}
which is positive definite from Remark \ref{remRabpos}. 
\begin{remark} 
Under the assumptions as above, from (\ref{rbb}) and (\ref{vceeq}) we have $V_{\ce}(b^*)=I(b^*)^{-1}$. Note that this is the asymptotic covariance 
matrix of maximum likelihood estimators for $b^*$ being the true parameter, see e.g. page 63 in \cite{Vaart}. This should be the case, 
since under Condition \ref{condKappa} for $b'=b^*$, a.s. 
\begin{equation}\label{cenbsb}
\begin{split}
\wh{\ce}_n(b^*,b)(\wt{\kappa}_n)&= \overline{(ZL(b^*) \ln(L(b)))}_n(\wt{\kappa}_n)\\
&=-\alpha\overline{\ln(p(b)))}_n(\wt{\kappa}_n),\\ 
\end{split}
\end{equation} 
and $b \to \overline{\ln(p(b)))}_n(\wt{\kappa}_n)=\frac{1}{n}\sum_{i=1}^n\ln(p(b)(\kappa_i))$ 
is maximized in such maximum likelihood estimation 
(note that for $\alpha>0$ we should minimize (\ref{cenbsb}) while for $\alpha<0$ --- maximize it). 
\end{remark} 
Under the assumptions as above and for $\alpha>0$,  from (\ref{bgdef}), (\ref{sigmaceopt}), and (\ref{hceopt})
\begin{equation} 
% \begin{split} 
%B_g(b')=\frac{1}{2}H_f^{-\frac{1}{2}}\Sigma_g(b')H_f^{-\frac{1}{2}},
B_{\ce}(b)= \frac{1}{2}H_{\ce}^{-\frac{1}{2}}\Sigma_{\ce}(b)H_{\ce}^{-\frac{1}{2}}=\frac{1}{2}\alpha I(b^*)^{-\frac{1}{2}}R(b^*,b)I(b^*)^{-\frac{1}{2}},\\ 
% \end{split} 
\end{equation} 
which is positive definite. Note that $B_{\ce}(b^*)=\frac{1}{2}\alpha I_l$ and $\Tr(B_{\ce}(b^*))=\frac{l\alpha}{2}$. 

For the mean square, let us assume that $b^*$ is an optimal-variance IS parameter.
Then, from (\ref{zmsql})
\begin{equation}\label{msqbopt}
\msq(b)= \E_{\PQ_1}(Z^2L(b))=\msq(b^*)\E_{\PQ_1}(\frac{L(b)}{L(b^*)^2})\\
\end{equation}
Thus, assuming that one can move the derivatives inside the expectation we have
\begin{equation}\label{nabla2msq2}
\begin{split}
\nabla^2 \msq(b)&= -\msq(b^*)\nabla \E_{\PQ_1}(\frac{L(b)}{L(b^*)^2}S(b)^T)\\
&=\msq(b^*)\E_{\PQ_1}(\frac{L(b)}{L(b^*)^2}(S(b)S(b)^T - \nabla^2l(b))).\\
\end{split}
\end{equation}
In such a case, from (\ref{ibdef}),
\begin{equation}\label{hmsqform}
\begin{split}
H_{\msq}&=\msq(b^*)\E_{\PQ_1}(\frac{L(b)}{L(b^*)^2}(S(b^*)S(b^*)^T - \nabla^2l(b^*)))\\
&=\msq(b^*)2I(b^*).\\
\end{split}
\end{equation}
\begin{remark}
Under the appropriate assumptions as in Remark \ref{remfdif}, 
when $b^*$ is a zero-variance IS parameter, then, for $\dist$ denoting the cross-entropy distance,
$b\to \dist(\PQ(b^*),\PQ(b))$ and $\ce$ are 
linearly equivalent with a linear proportionality constant $\alpha$ (see (\ref{dKul})). 
Thus, in such a case (\ref{hceopt}) follows from the well-known fact that under appropriate assumptions
\begin{equation}
(\nabla^2_b\dist(\PQ(b^*),\PQ(b))_{b=b^*}=I(b^*). 
\end{equation}
Furthermore, from Remark (\ref{remPers}), when $b^*$ is an optimal-variance IS parameter, then, for $\dist$ denoting the Pearson divergence, 
$b\to \dist(\PQ(b^*),\PQ(b))$ and $\msq$ are 
 linearly equivalent with a linear proportionality constant 
 $\msq(b^*)$ (see (\ref{distvar}) and (\ref{zlsmsq})). Thus, in such a case (\ref{hmsqform}) is equivalent to the fact that 
 \begin{equation}
(\nabla^2_b\dist(\PQ(b^*),\PQ(b))_{b=b^*}=2I(b^*). 
\end{equation}
\end{remark} 

Assuming that $I(b^*)$ is positive definite and $\msq(b^*)\neq 0$, from (\ref{vgdef}),
(\ref{sigmamsqopt}), and (\ref{hmsqform}), 
\begin{equation}\label{vmsqdef}
\begin{split}
V_{\msq}(b)&= H_{\msq}^{-1}\Sigma_{\msq}(b) H_{\msq}^{-1}\\
&=\frac{1}{4}I(b^*)^{-1}R(b^*,b)I(b^*)^{-1},\\
%&= \frac{1}{4}V_{\ce}(b).\\
\end{split}
\end{equation}
which is positive definite from Remark \ref{remRabpos}.
Thus, assuming further that $b^*$ is 
a zero-variance IS parameter and we have (\ref{vceeq}), it holds
\begin{equation}\label{Vmsqce}
V_{\msq}(b)=\frac{1}{4}V_{\ce}(b). 
\end{equation}
% In particular, in SSM, assuming that (\ref{dtbsSSM}) holds for $g$ equal to $\ce$, as well as for $g$ equal to $\msq$ and $d_t$ replaced by $d_t'$,
% $d'=(d_t')_{t\in T}$ converges to $b^*$ four times faster than $d=(d_t)_{t\in T}$ in
% the sense of Remark \ref{remmoreeff}. 
%$d_t$ converges to $b^*$ four times faster when minimizing 
%which is positive definite from Remark \ref{remRabpos}, 
From (\ref{vmsqdef}), we have $V_{\msq}(b^*)=\frac{1}{4}I(b^*)^{-1}$. Furthermore, from (\ref{bgdef})
\begin{equation}\label{bmsqdef}
\begin{split}
B_{\msq}(b)&=\frac{1}{2}H_{\msq}^{-\frac{1}{2}}\Sigma_{\msq}(b)H_{\msq}^{-\frac{1}{2}}\\
%(V_{\msq}(b))^{\frac{1}{2}}H_{\msq} (V_{\msq}(b))^{\frac{1}{2}} \\
&= \frac{\msq(b^*)}{4}I(b^*)^{-\frac{1}{2}}R(b^*,b)I(b^*)^{-\frac{1}{2}},\\ 
\end{split}
\end{equation}
which is positive definite. Note that $B_{\msq}(b^*)=\frac{\msq(b^*)}{4}I_l$ and $\Tr(B_{\msq}(b^*))=\frac{l\msq(b^*)}{4}$. 
%Note that from Remark \ref{remChi2gen}, $\wt{\chi^2}(B_{\msq}(b^*))=\frac{\msq(b^*)}{4}\chi^2(l)$. 
% Under assumptions as in the previous sections, for $b=b'$ in SSM or $b=d^*$ in MSM, for appropriate $u(b)\in\R_+$ 
% for SSM or $u(b)=1$ for MSM for appropriate $d_t$ from minimization of $\wh{\msq}_n$ we have
%  \begin{equation}\label{rtmsqbmsq}
%  r_t(\msq(d_t)-\msq(d^*))\Rightarrow \wt{\chi^2}(u(b)B_{\msq}(b))
%  \end{equation}
\begin{remark}\label{rembcemsq}
Let $A$, $T$, $d_t$, and $r_t$ be as in Section \ref{secHelpAsymp}, let $b\in A$, and 
$u:A\to\overline{\R}$ be such that $u(b)\in \R_+$ and $r_t(d_t-b^*) \Rightarrow \ND(0,u(b)V_{\ce}(b))$ 
(see Section \ref{secAsympSSM} for sufficient conditions for this to hold 
for $b=b'$ and $r_t=t$ for the SSM of the cross-entropy estimators and Section \ref{secAsympMSM} for $b=d^*$, $u(d^*)=1$, and $r_k =n_k$ for 
the MSM of such estimators).
Let further $\msq$ be twice differentiable in $b^*$ with $\nabla \msq(b^*)=0$ and $H_{\msq}=\nabla^2\msq(b^*)$. Then, 
from Remark \ref{remdelta}, for
\begin{equation}\label{bcemsq}
B_{\ce,\msq}(b):= \frac{1}{2}V_{\ce}(b)^{\frac{1}{2}}H_{\msq} V_{\ce}(b)^{\frac{1}{2}},
\end{equation}
we have
\begin{equation}\label{rtmsqbcemsq}
r_t(\msq(d_t)-\msq(b^*))\Rightarrow \wt{\chi^2}(u(b)B_{\ce,\msq}(b)).
\end{equation}
%\end{remark}
For $b^*$ being a zero-variance IS parameter, $I(b^*)$ being positive definite, and $\alpha\neq 0$,
from (\ref{bcemsq}), (\ref{hmsqform}), (\ref{vceeq}), and (\ref{bmsqdef})
\begin{equation}\label{bcemsq2}
\begin{split}
B_{\ce,\msq}(b)&=\msq(b^*)I(b^*)^{-\frac{1}{2}}R(b^*,b)I(b^*)^{-\frac{1}{2}}\\
&=4B_{\msq}(b).  
\end{split}
\end{equation}
\end{remark}

\begin{remark}\label{remWellDef}
Consider the LETS setting and let Condition \ref{condlemYmore} hold for $S=1$. 
Then, from $\E_{\PQ(a)}(|\frac{L(b)}{L(a)}S_i(a)S_j(a)|)=\E_{\PQ_1}(|\frac{L(b)}{L(a)^2}S_i(a)S_j(a)|)$, 
Theorem \ref{thYmore}, and Remark \ref{remCondUN}, we have $R(a,b)\in \R^{l\times l}$,$a,b \in A$, and thus also $I(b)\in \R^{l\times l}$, $b \in A$. 
Furthermore, from Theorem \ref{thDiffLp}, the above derivation leading to (\ref{ibdef}) can be carried out.
From these theorems and remark we can also move the derivatives inside the expectation in
(\ref{nabla2msq2}), and using analogous reasoning as in the proof of Theorem \ref{thDiffLp} - also in (\ref{nabla2ce}). 
\end{remark}

Consider now the ECM setting as in Section \ref{secECM}, assuming Condition \ref{condPartvm}. Then, from (\ref{Ibdef1})
\begin{equation}
I(b)=\E_{\PQ(b)}((X-\mu(b))(X-\mu(b))^T)=\Sigma(b). 
\end{equation}
%which, as discussed in Section \ref{secECM} is positive definite only if Condition \ref{condtX} holds. 
%the above $R(a,b)$ can be expressed in terms of values of $\Psi$, $\mu$ and $\Sigma$
Furthermore,
\begin{equation}\label{rab}
\begin{split}
R(a,b)&= \E_{\PQ_1}(\frac{L(b)}{L(a)^2}S(a)S(a)^T)\\
&= \exp(\Psi(b)-2\Psi(a))\E_{\PQ_1}(\exp((2a-b)^TX)(X-\mu(a))(X-\mu(a))^T)\\
&= \exp(\Psi(b)+\Psi(2a-b)-2\Psi(a))\E_{\PQ(2a-b)}((X-\mu(a))(X-\mu(a))^T)\\
&= \exp(\Psi(b)+\Psi(2a-b)-2\Psi(a))(\Sigma(2a-b)\\
&+(\mu(2a-b)-\mu(a))(\mu(2a-b)-\mu(a))^T).\\
\end{split}
\end{equation}
For a positive definite matrix $B\in\Sym_l(\R)$ and $b \in \R^l$, let 
$|b|_B=\sqrt{b^TBb}$. Then, $|\cdot|_B$  is a norm on $\R^l$. 
For $X \sim \ND(\mu_0, M)$ under $\PQ_1$ for some positive definite covariance matrix $M$, we have from (\ref{rab}) and discussion in Section \ref{secECM} 
\begin{equation}
R(a,b)=\exp(|a-b|_M^2)(M+M(a-b)(a-b)^TM). 
\end{equation}
In particular, for $b^*$ being 
a zero-variance IS parameter,  from (\ref{vceeq}) and (\ref{Vmsqce}) we have
\begin{equation}
V_{\ce}(b)= \exp(|b^*-b|_M^2)(M^{-1}+(b^*-b)(b^*-b)^T) = 4V_{\msq}(b), 
\end{equation}
and
\begin{equation}
B_{\msq}(b)=\frac{\msq(b^*)}{4}\exp(|b^*-b|_M^2)(I_l+M^{\frac{1}{2}}(b^*-b)(b^*-b)^TM^{\frac{1}{2}}). 
\end{equation}
%&= \frac{\msq(b^*)}{4}I(b^*)^{-\frac{1}{2}}R(b^*,b)I(b^*)^{-\frac{1}{2}},\\ 
Thus, the mean of $\wt{\chi^2}(B_{\msq}(b))$ is
\begin{equation}
\Tr(B_{\msq}(b)) =\frac{\msq(b^*)}{4}\exp(|b^*-b|_M^2)(l+ |b^*-b|_{M}^2).  
\end{equation}
%(where we used the fact that $\Tr M^{\frac{1}{2}}(b^*-b)(b^*-b)^TM^{\frac{1}{2}}= M^{\frac{1}{2}}(b^*-b) $)
Note that $b^*$ is the unique minimum point of $b\to\Tr(B_{\msq}(b))$. From the below remark it follows that 
for $X$ having a different distribution under $\PQ_1$ this may be not the case. 
%TODO trace
% Then, $\Psi(b)=b^T\mu_0 +\frac{|b|_M}{2}$ and $L(b)=\exp(-b^TX +\Psi(b))$, $b \in A$.  
% For some $a \in \R^l$, let %us consider a random variable on $\mc{S}_1$
%  \begin{equation}
%  Z=\exp(a^T(X-\mu_0))= \exp(\frac{\frac{1}{2}|a|_M^2})L(a)^{-1},
%  \end{equation}
% so that $\alpha = \exp(\frac{|a|_M^2}{2})$ and $b^*=a$ is the unique minimum point of each of 
% $\ce$, $\msq$, and $\ic$ (for which we take $C=1$, so that $\ic=\var$). 
% From an easy calculation using (\ref{densC}) we receive
% $\ce(b)=\alpha(\frac{|b|_M^2}{2}-a^Tb)$, so that $\nabla^2\ce(b)=\alpha M= H_{\ce}$. Furthermore,
% \begin{equation}
% \Sigma_{\ce}(b)=\exp((|b-a|_M^2+|a|_M^2)(M+M(b-a)(b-a)^TM),
% \begin{equation}
%  R(a,b)=\exp(a-bM)=
% \end{equation}
\begin{remark}\label{remRabmin}
Consider the ECM setting for $A=\R$. Then, 
$R(a,b)= \E_{\PQ(a)}(\frac{L(b)}{L(a)}(S(a))^2)$ and 
% assuming that one can move the derivative inside the expectation 
% in the below formula, which from theorems 
% \ref{thYmore} and \ref{thDiffLp} is the case when $A=\R$, we have
%under appropriate regularity assumptions %TODO assumptions
\begin{equation}
\begin{split}
(\nabla_{b} R(a,b))_{b=a}&= -\E_{\PQ(a)}(S(a)^3)\\
 &=-\E_{\PQ(a)}((X-\mu(a))^3).\\
\end{split}
\end{equation}
From the convexity of $b \to R(a,b)$ 
(which follows from the convexity of $b \to L(b)$), %as discussed above, 
a necessary and sufficient condition for $b=a$ to be its minimum point 
(and thus for a zero-variance IS parameter $b^*=a$ to be 
the minimum point of $b \to B_{\msq}(b)$ as in (\ref{bmsqdef})),
is that $X$ has a zero third central moment under $\PQ(a)$. 
This does not hold e.g. for $X \sim \Pois(\mu(a))$ under $\PQ(a)$ as in Section 
\ref{secECM}, for which $\E_{\PQ(a)}((X-\mu(a))^3)=\mu(a)>0$. 
\end{remark}
In the LETGS setting, assuming Condition \ref{equivCond}, that for some $b \in A$,
$G$ has $\PQ(b)$-integrable entries, and that we have (\ref{ibdef}), from (\ref{lnlG}) 
\begin{equation}
I(b)= 2\E_{\PQ(b)}(G),
\end{equation}
which from Lemma \ref{lemPosK} is positive definite only if Condition \ref{cond1} holds for $Z=1$.

\section{\label{sec3Point}An analytical example of a symmetric three-point distribution}
Let us consider ECM as in Section \ref{secECM} for $l=1$, 
assuming that $\PQ_1(X=-1)=\PQ_1(X=0)=\PQ_1(X=1)=\frac{1}{3}$. Then, we have 
$A=\R$, $\Phi(b)= \frac{e^b+e^{-b} +1}{3}$, 
$L(b)=\Phi(b)\exp(-bX)$, $\Psi(b)=\ln(\Phi(b))$, $\mu(b)=\nabla\Psi(b)=\frac{e^b-e^{-b}}{e^b+e^{-b} +1}$, and 
$\nabla^2\Psi(b)=\frac{e^b+e^{-b}+4}{(e^b+e^{-b}+1)^2}$. 
%TODO assumptions condtX etc. 
For some $d \in \R$, let $Z=\I(X=0)+d\I(X\neq 0)$. 
We have $\alpha=\frac{1+2d}{3}$ and 
\begin{equation}
\msq(b)=\E_{\PQ_1}(Z^2L(b))=\frac{1}{9}(1+e^b+e^{-b})(1+d^2(e^b+e^{-b})).
\end{equation}
The unique minimizer of $\msq$ is $0$, which corresponds to crude MC. It holds 
%\begin{equation}
$\msq(0)= \frac{1+2d^2}{3}$
%\end{equation}
and
%\begin{equation}
$\var(0)=\frac{1 +2d^2}{3}-(\frac{1+2d}{3})^2= \frac{2}{9}(1-d)^2$. 
%\end{equation}
Thus, there exists a zero-variance IS parameter only if $d=1$, in which case such a parameter is $b^*:=0$. 
%The optimal variance IS parameter exists also for $d=-1$
We have $\E_{\PQ_1}(ZX)=0$, so that from (\ref{cebform}),
$\ce(b)=\alpha\Psi(b)$. Thus, if $d\neq -\frac{1}{2}$, then $\ce$ has a unique optimum point $0=b^*$,
which is a minimum point if $\alpha>0$ and a maximum point if $\alpha<0$. We have $\nabla^2 \Psi(0)=\frac{2}{3}$ and 
\begin{equation}
H_{\ce}=\nabla^2\ce(0)= \alpha\nabla^2\Psi(0)= \frac{2(1+2d)}{9}.
\end{equation}
It holds %$L(0)=1$ and
$\nabla L(b)=-X \exp(-bX)\Phi(b) + \exp(-bX) \frac{e^b-e^{-b}}{3}$, so that 
%\begin{equation}
$\nabla L(0)=-X$.  
%\end{equation}
Thus, from (\ref{sigmacedef}), for 
\begin{equation}
g(b):= (1+e^b+e^{-b})(e^b+e^{-b}),
\end{equation}
we have
\begin{equation}
\Sigma_{\ce}(b)= \E_{\PQ_1}(Z^2L(b)X^2)=\frac{d^2}{9}g(b)
\end{equation}
and for $d \neq {-\frac{1}{2}}$
\begin{equation}
V_{\ce}(b) = H_{\ce}^{-2}\Sigma_{\ce}(b)= \left(\frac{3d}{2(1+2d)}\right)^2g(b). %\frac{9d^2}{4(1+2d)^2}g(b).
\end{equation}
We have
\begin{equation}
H_{\msq}=\nabla^2\msq(0)=\frac{2}{9}(1+5d^2),
\end{equation}
from (\ref{sigmamsq})
%\msq(b)=\E_{\PQ_1}(Z^2L(b))=\frac{1}{9}(1+e^b+e^{-b})(1+d^2(e^b+e^{-b})).
%\nabla^2\msq(0)=\nabla^2\frac{1}{9}(1+e^b+e^{-b})(1+d^2(e^b+e^{-b}))=\frac{2}{9}(1+2d^2)+ \frac{6}{9}d^2S
\begin{equation}
\Sigma_{\msq}(b)=\E_{\PQ_1}(L(b)Z^4X^2)=\frac{d^4}{9}g(b),
\end{equation}
and thus 
\begin{equation}
V_{\msq}(b)=H_{\msq}^{-2} \Sigma_{\msq}(b)=
\left(\frac{3d^2}{2(1+5d^2)}\right)^2g(b).
\end{equation}
From (\ref{sigmamsq2})
\begin{equation}
%\begin{split}
\Sigma_{\msq2}(b)=\E_{\PQ_1}((Z^2 -\msq(0))^2L(b)X^2)
%&= \frac{1}{9}(d^2 -\frac{1 +2d^2}{3})^2g(b)\\
=\frac{1}{3^4}(d^2 -1)^2g(b),
%\end{split}
\end{equation}
so that
\begin{equation}
V_{\msq2}(b)= H_{\msq}^{-2} \Sigma_{\msq2}(b)=  \left(\frac{d^2 -1}{2(1+5d^2)}\right)^2g(b).
\end{equation}
From (\ref{sigmaicc1})
\begin{equation}
%\begin{split}
\Sigma_{\ic}(b)=\E_{\PQ_1}((Z^2 -\msq(0)-\var(0))^2L(b)X^2)
%&= \frac{1}{9}(d^2 -\frac{1 +2d^2}{3} - \frac{2}{9}(1-d)^2)^2g(b)\\
=\frac{1}{3^6}((d-1)(d+5))^2g(b),
%\end{split}
\end{equation}
and thus
\begin{equation} 
%\begin{split}
V_{\ic}(b)=H_{\msq}^{-2} \Sigma_{\ic}(b)=\left(\frac{(d-1)(d+5)}{6(1+5d^2)}\right)^2g(b).
%\end{split}
\end{equation}
For $s$ substituted by $\ce$, $\msq$, $\msq2$, and $\ic$, 
let us further write $V_s(b,d)$ rather than $V_s(b)$ to mark its dependence on $d$.
%Since $g(b)$ has unique 
%(\frac{3d}{2(1+2d)})^2g(b).%\frac{9d^2}{4(1+2d)^2}g(b).
Let for such different $s$, $f_s(d)=2(1+5d^2)\sqrt{\frac{V_s(b,d)}{g(b)}}$, so that 
$f_{\ce}(d)=|\frac{3d(1+5d^2)}{1+2d}|$, 
$f_{\msq}(d)=3d^2$,  $f_{\msq2}(d)=|d^2-1|$, 
and $f_{\ic}(d)=\frac{1}{3}|(d-1)(d+5)|$. 
For different substitutions of $s_1$ and $s_2$ as for $s$ above, 
it holds $\frac{V_{s_1}}{V_{s_2}}(b,d)=(\frac{f_{s_1}}{f_{s_2}}(d))^2$ whenever $f_{s_2}(d)\neq 0$. 
The graphs of functions $f_s$ for different $s$ are shown in Figure \ref{figCompfs}. 
Note that $b\to g(b)$ is positive and has a unique minimum point in $0=b^*$. Thus, if $f_s(d)\neq 0$, then 
$b\to V_s(b,d)=g(b)(\frac{f_s(d)}{2(1+5d^2)})^2$ also has a unique minimum point in $b^*$. 
%TODO figure 
%We have f_s(d)
%It holds \lim_{d\to \infty}f_s(d) $\lim_{d\to -\frac{1}{2}}f_{\ce}(d)= \infty$, 
It holds $f_{\ce}(0)=f_{\msq}(0)=0$. For $d \in \R \setminus \{-\frac{1}{2},0\}$, let
$r(d)=\frac{f_{\ce}}{f_{\msq}}(d)= |\frac{1+5d^2}{d(1+2d)}|$. One can easily show that this function assumes a unique minimum in 
$d=1$, in which $r(d)=2$. 
In particular, for $d \in \R\setminus\{-\frac{1}{2}\}$ we have
$f_{\ce}(d)\geq 2f_{\msq}(d)$ and thus $V_{\ce}(b,d)\geq 4V_{\msq}(b,d)$, $b \in \R$, with equalities holding only for $d$ equal 
to $0$ and $1$, the latter being in 
agreement with the theory in Section \ref{secAsympProp} since for $d=1$, $b^*$ is a zero-variance IS parameter. 
From an easy calculation we receive that $f_{s_1}<f_{s_2}$ on $D$, $f_{s_1}=f_{s_2}$ on $\partial D$, and 
$f_{s_1}>f_{s_2}$ otherwise for $s_1=\msq$, $s_2=\msq2$, and $D=(-\frac{1}{2},\frac{1}{2})$, 
$s_1=\msq$, $s_2=\ic$, and $D=(\frac{-2-3\sqrt{2}}{10},\frac{-2+3\sqrt{2}}{10})$, as well as 
$s_1=\msq2$, $s_2=\ic$, and $D=(-2,1)$. 
Since for $s$ equal to $\msq2$ or $\ic$,  $f_s$ is continuous and $f_s(0)>0$, 
such $f_s$ are higher than $f_{\ce}$ (and thus also than $f_{\msq}$) on some neighbourhood of $0$. 
Note that in agreement with the theory in Section \ref{secAsympProp}, 
%as opposed to the well-known ones. %, and they also lead to lower values of these functions on some neighbourhood of one. 
for $d=1$, for which $b^*$ is a zero-variance IS parameter, for $s=\ic$, $V_s(b,d)=0$, $b \in \R$, %TODO comment 
and for $s=\msq2$,  this holds also for $d=-1$, for which $b^*$ is an optimal-variance IS parameter. 
For $B_s(b,d)=\frac{1}{2}V_s(b,d)H_{\msq}$ for $s$ equal to $\msq$, $\msq2$, or $\ic$, and for 
$B_{\ce,\msq}(b,d)=\frac{1}{2}V_{\ce}(b,d)H_{\msq}$ for $s=\ce$, 
we have analogous relations as for the different $V_s(b,d)$ above. 
\begin{figure}[h]
\includegraphics[width=0.6\textwidth]{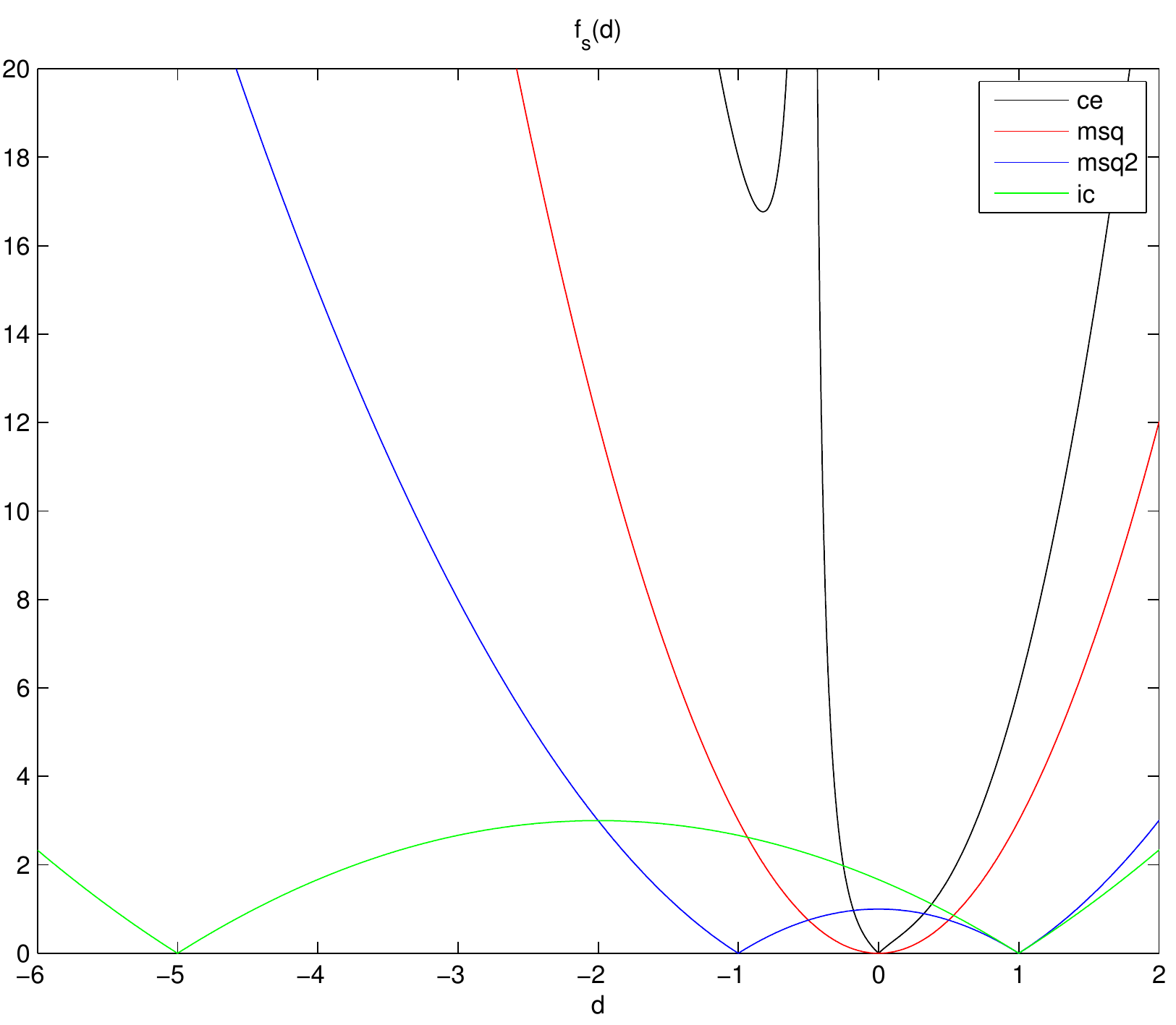}
\caption{Functions $f_s$ as in the main text for $s$ equal to $\ce$, $\msq$, $\msq2$, and $\ic$.}
\label{figCompfs}
\end{figure}

\chapter{\label{secTwo}Two-stage estimation} 
%In this section we shortly describe a two-stage approach for adaptive estimation of the quantity of interest 
In this chapter we shortly discuss a two-stage adaptive method for estimation, in the first stage of which 
some parameter vector $d$ is computed, with the help of which in the second stage an estimator of the quantity of interest $\alpha \in \R$ 
is calculated. 
%Detailed analysis of such approach, including its various limiting properties, is relegated to a future paper. 
%Some alternative estimation approaches are discussed in the below Remark \ref{remAlt} and in Section \ref{secConcl}. 
For a nonempty parameter set $A \in \mc{B}(\R^l)$ and a family of distributions as in Section \ref{secFamily}, 
let a measurable function $h:\mc{S}(A) \otimes \mc{S}_1 \to \mc{S}(\overline{\R})$ be such that for each $b \in A$, 
\begin{equation} 
\E_{\PQ(b)}(h(b,\cdot))=\alpha.
\end{equation} 
Let $\var(b)= \Var_{\PQ(b)}(h(b,\cdot))$, $b \in A$.
Consider some cost variable $C$ and functions $c$ and $\ic$ as in Section \ref{secCoeffDiv}, 
but in the definition of $\ic$ using the above $\var$ (in particular, for $\ic$ we assume Condition \ref{condicwelldef}). 
% %We shall assume that  $\ic(b)$ is well-defined for each $b \in A$.  
In the IS case conditions \ref{condpqbllpq1} and \ref{condhzlb} hold, 
but the above setting is more general and 
can also describe e.g. control variates or control variates in conjunction with IS. 
The parameter vector $d$ computed in the first stage is an $A$-valued random variable, 
which can be obtained e.g. using some adaptive algorithm like 
single- or multi-stage minimization method of some estimators described in the previous sections. This can be done in many different 
ways as discussed in the below remark. 
\begin{remark}\label{remdk2} 
One possibility is to use as $d$ some parameter $d_t$ corresponding to some first-stage budget $t$ as in Remark \ref{remdk}. 
Alternatively, in methods in which to compute some parameter 
$p_{k}$ one first needs to compute $p_l$, $l=1,\ldots,k-1$, 
like in the MSM methods from the previous sections for $p_k$ as in Remark \ref{remdk}, one can 
set $d=p_{\tau}\I(\tau \neq \infty) +p_\infty\I(\tau=\infty)$ for some $p_\infty \in A$ 
and some a.s. finite  $\N_+\cup\{\infty\}$-valued stopping time 
$\tau$  for the filtration $\mc{F}_n=\sigma\{p_l:l\leq n\}$, $n \in \N_+$. For instance, if a.s. $p_k \to b^*\in A$, 
then $\tau$ can be the moment when the change of $p_k$ from $p_{k-1}$, or such a relative change if $b^*\neq 0$, becomes smaller than some $\epsilon\in \R_+$ 
(see e.g. Remark 10 in \cite{pupa2011}). 
For the various $p_k$ from MSM methods as in Remark \ref{remdk}, 
the fact that a.s. $p_k \to b^*$ %(and thus for some neighbourhood $B$ of $b^*$ a.s. $p_k \in B$ for a sufficiently large $k$) 
follows from Condition \ref{condESSM1} or its counterparts.
When for some functions $f_k:\mc{S}(A) \otimes (\Omega, \mc{F}) \to \mc{S}(\R)$, $k \in \N_+$, (which can be e.g. the minimized estimators) and 
$f:A \to \R$ we have a.s. $f_k(p_k)\to f(b^*)$, then such a stopping time can be based on the behaviour of 
$f_k(p_k)$, similarly as for $p_k$ above. % or if $f_k$ are differentiable it can be based on values of $|\nabla f_k(p_k)|$ becoming small. 
Assuming that a.s. $f_k \overset{loc}{\rightrightarrows} f$ (see Section \ref{secUnifEst} for sufficient assumptions), if a.s. $p_k \to b^*$, 
then from Lemma \ref{lemConvUnif} a.s. $f_k(p_k) \to f(b^*)$. 
\end{remark}

One way to model the second stage of a two-stage estimation method is to consider it on a different probability space than the first one,
for a fixed computed value $v$ of the variable $d$ from the first stage. 
In such a case, for some $\kappa_1,\kappa_2,\ldots,$ i.i.d. $\sim  \PQ(v)$,
in the second stage we perform an MC procedure as in Chapter \ref{secIneff} but using $Z_i=h(v,\kappa_i)$, $i \in \N_+$,
e.g. using a fixed number of samples or a fixed approximate budget. 
We can also construct asymptotic confidence intervals for $\alpha$ as in that chapter.
From the discussion in Chapter \ref{secIneff}, the inefficiency of such a procedure 
can be quantified using the inefficiency constant $\ic(v)$. This justifies comparing the asymptotic efficiency of methods for finding the
adaptive parameters in the first stage of a two-stage method as above
by comparing their first- and, if applicable, second-order asymptotic efficiency for the minimization 
of $\ic$ as discussed in sections \ref{secCompFirst} and \ref{secSecond}. 
 
An alternative way to model the second stage of a two-stage method is to consider its second stage on the 
same probability space as the first one. Let us assume the following condition. 
 \begin{condition}\label{condKappaa}
 Random variables  $\phi_i$, $i \in \N_+$, are conditionally independent given $d$ and have the same conditional distribution $\PQ(b)$ given $d=b$. 
 \end{condition} 
 Condition \ref{condKappaa} is implied by the following one. 
 \begin{condition}\label{condabeta} 
 Condition \ref{condxi} holds, and for $\beta_{1}, \beta_2, \ldots$, i.i.d. $\sim\PR_1$ and independent of $d$ 
 we have $(\phi_i)_{i\in\N_+}=(\xi(d,\beta_i))_{i\in \N_+}$. 
 \end{condition}
In the second stage of the considered method one computes an estimator 
 \begin{equation}\label{alphaN} 
  \wh{\alpha}_n=\frac{1}{n}\sum_{i=1}^{n}h(d,\phi_i). 
 \end{equation}
Similarly as above, the number $n$ of samples can be deterministic or random. In the first case
the resulting estimator is unbiased, while in the second this needs not to be true. 
Random $n$ can correspond e.g. to a fixed approximate computational budget and be given by definition (\ref{nt1}) or (\ref{nt2}) but for  $C_i$ replaced by 
$C(\phi_i)$, $i \in \N_+$. 

 \begin{remark}\label{remAlt} %TODO generalize to your 
 A possible alternative to the above discussed two-stage estimation method is the same as its second model 
 above except that for the computation of $\wh{\alpha}_n$ in the second stage one uses the variables $\phi_i=\xi(d,\beta_i)$ as in Condition 
 \ref{condabeta} but without assuming that $\beta_i$, $i \in \N_+$, are independent of $d$.  
 In such a case, Condition \ref{condKappaa} may not hold. For example, one could reuse the
i.i.d. random  variables with distribution $\PR_1$ generated for the estimation of $d$ in the first stage 
as some (potentially all) the variables $\beta_i$ used for the computation of $\wh{\alpha}_n$, which could save the computation time. 
Under appropriate identifications, such an approach using exactly the same $\beta_i$, $i=1,\ldots,n$ in ESSM to compute $d$ and then (\ref{alphaN}) 
is used in the multiple control variates method (see \cite{asmussen2007stochastic,Szechtman2001}),
while for IS it was considered in \cite{Jourdain2009}. %The CLT of the final estimator optional 
In such a reusing approach, $\wh{\alpha}_n$ as in (\ref{alphaN}) needs not to be unbiased even for $n$ deterministic. 
Furthermore, one needs to store a potentially large random number of the generated values of random variables, 
which may be more difficult to implement and requires additional computer memory. 
Finally, in a number of situations, like in the case of our numerical experiments, 
generating the required parts of the variables $\beta_i$  forms only a small fraction of the computation time 
needed for computing the variables $h(d,\phi_i)=(ZL(d))(\xi(d,\beta_i))$ in the second stage, 
so that reusing some $\beta_i$ from the first stage would not lead to considerable time savings. 
\end{remark}

\chapter{\label{secNumExp}Numerical experiments} 
%In this chapter we describe our numerical experiments. %, already briefly discussed in the Introduction. 
Our numerical experiments were carried out using programs written in matlab2012a
and run on a laptop. 
Unless stated otherwise, we used the simulation parameters, variables, and IS basis functions as 
for the problems of estimation of the expectations $\mgf(x_0)$, $q_{1,a}(x_0)$, $q_{2,a}(x_0)$, and $p_{T,a}(x_0)$ 
in Section \ref{secSpecNumExp}. In some of our experiments we performed 
the single- or multi-stage minimization of 
estimators 
$\wh{\est}$ (where for short we write $\wh{\est}$ rather 
than $\wh{\est}_n$, $n \in \N_p$, for appropriate $p$) equal to 
$\wh{\ce}$, $\wh{\msq}$, $\wh{\msq2}$, and $\wh{\ic}$, 
as discussed in Section \ref{secSimpleMin}. 
In the MSM we used in each case $b_0=0$. 
For the minimization of $\wh{\msq2}$ and $\wh{\ic}$ in these methods we used the 
matlab fminunc unconstrained minimization function 
with the default settings and exact gradients, for $\wh{\msq}$ 
additionally using their exact Hessians,  
as discussed Section \ref{secSimpleMin}. 
The minimum points of $\wh{\ce}$ were found by solving the linear systems of equations 
as in that section. 
Both for the crude MC (CMC) and when using IS, the computation times 
of the MC replicates in our experiments were typically approximately proportional 
to the replicates of the exit times $\tau$ for the MGF and translated committors, and 
to the replicates of $\tau'$ as in (\ref{taup}) for $p_{T,a}(x_0)$. 
Thus, we consider the theoretical 
cost variables $C$ equal to  $h\tau$ for the MGF and translated committors and to 
$h\tau'$ for $p_{T,a}(x_0)$. 
The proportionality constants $p_{\dot{C}}$ of the replicates of such $C$ to the simulation times 
as in  Chapter \ref{secIneff}) were different for CMC and when using different 
basis functions in IS. 

The remainder of this chapter is organized as follows. 
In Section \ref{secStat} we discuss some methods for 
testing statistical hypotheses, which are later 
used for interpreting the
results of our numerical experiments.
In Section \ref{secExpEst} we describe two-stage estimation experiments 
as in Section \ref{secSimpleMin}, performing MSM 
in the first stages and in the second stages estimating the expectations of the
functionals of the Euler scheme as above. In the second stages we also 
estimated some other quantities, like 
inefficiency constants, variances, mean costs, and the proportionality constants $p_{\dot{C}}$ as above. 
We use these quantities to compare the 
efficiency of applying in a IS MC method the IS parameters obtained from the MSM of different estimators, as well as 
of using different added constants $a$ and IS basis functions
in such adaptive IS procedures.  
In Section \ref{secSpread} we compare the spread of the IS drifts coming from the SSM 
of different estimators and using different parameters $b'$. 
In Section \ref{secIntu} we provide some intuitions behind 
the results of our numerical experiments. 

\section{Testing statistical hypotheses}\label{secStat}
%\begin{remark} %TODO later 
Let $\mu_X, \mu_Y \in \R$, $\sigma_X, \sigma_Y \in \R_+$, and 
$\R$-valued random variables $X_n$ and $Y_n$, $n \in \N_+$, be such that 
$\wt{X}_n:=\sqrt{n}(X_n-\mu_X) \Rightarrow \ND(0,\sigma_X^2)$, $\wt{Y}_n:=\sqrt{n}(Y_n-\mu_Y) \Rightarrow \ND(0,\sigma_Y^2)$, 
and for each $j,k \in \N_+$, $X_j$ is independent of $Y_k$. 
Let $\wh{\sigma}_{X,n}$ and $\wh{\sigma}_{Y,n}$, $n \in \N_+$, be $[0,\infty)$-valued random variables such that 
$\wh{\sigma}_{X,n}\overset{p}{\to} \sigma_X$ and $\wh{\sigma}_{Y,n}$ $\overset{p}{\to}\sigma_Y$. Let
%Consider $\R$-valued random variables $X_{n}$, $Y_n$ $n \in \N_+$, such that for some $\sigma \in \R_+^2$, and $\mu \in\R^2$ we have
% \begin{equation}
% X_{n} \Rightarrow \ND(\mu,\diag((\sigma_i^2))_{i=1}^2).
% \end{equation}
% Let $m_n \in \N_+^2$, $n\in \N_+$ be such that 
$a_n,b_n \in \N_+$, $n\in \N_+$, be such that $\lim_{n\to\infty}a_n=\lim_{n\to\infty}b_n=\infty$ and
\begin{equation}
\lim_{n\to \infty} \frac{a_{n}}{b_{n}}=\rho \in \R_+.
\end{equation}
Let 
\begin{equation}
t_n=\frac{X_{a_n}-Y_{b_n}}{\sqrt{\frac{\wh{\sigma}_{X,a_n}^2}{a_n}+\frac{\wh{\sigma}_{Y,b_n}^2}{b_n}}}
=\frac{\sqrt{a_n}(X_{a_n}-Y_{b_n})}{\sqrt{\wh{\sigma}_{X,a_n}^2+\frac{a_n}{b_n}\wh{\sigma}_{Y,b_n}^2}}\\ 
\end{equation}
and $H_n=\frac{\sqrt{a_n}(\mu_Y-\mu_X)}{\sqrt{\wh{\sigma}_{X,a_n}^2+\frac{a_n}{b_n}\wh{\sigma}_{Y,b_n}^2}}$.

\begin{lemma}
Under the assumptions as above, we have%for $\mu_X=\mu_Y$
\begin{equation}\label{equhypho}
%\begin{split}
%{\sqrt{\frac{\sigma_X^2+\sigma_Y^2\frac{a_n}{b_n}}}}\\ 
t_n+ H_n= \frac{\wt{X}_{a_n}-\sqrt{\frac{a_n}{b_n}}\wt{Y}_{b_n}}{\sqrt{\wh{\sigma}_{X,a_n}^2+\frac{a_n}{b_n}\wh{\sigma}_{Y,b_n}^2}}\Rightarrow \ND(0,1). 
\end{equation}
\end{lemma}
\begin{proof}
From the asymptotic properties of  $\wt{X}_i$ and $\wt{Y}_j$ as above and their independence, 
we receive,  e.g. using Fubini's theorem and 
the fact that convergence in distribution is equivalent to the pointwise convergence of characteristic functions, that
\begin{equation}
\wt{X}_{a_n} - \sqrt{\rho} \wt{Y}_{b_n}\Rightarrow \ND(0,\sigma_X^2+\rho\sigma_Y^2).
\end{equation}
Thus, from  $G_n:=\wt{Y}_{b_n}(\sqrt{\rho} - \sqrt{\frac{a_n}{b_n}}) \overset{p}{\to} 0$,
\begin{equation}\label{wtXYto}
%\begin{split}
\wt{X}_{a_n}-\sqrt{\frac{a_n}{b_n}}\wt{Y}_{b_n} = \wt{X}_{a_n}-\sqrt{\rho}\wt{Y}_{b_n}
+G_n\Rightarrow \ND(0,\sigma_X^2+\rho\sigma_Y^2).  
%\end{split}
\end{equation}
Furthermore, from the continuous mapping theorem, 
\begin{equation}\label{overpsigma}
\sqrt{\wh{\sigma}_{X,a_n}^2+\wh{\sigma}_{Y,b_n}^2\frac{a_n}{b_n}}\overset{p}{\to}\sqrt{\sigma_X^2+\rho\sigma_Y^2}.
\end{equation}
Now, (\ref{equhypho}) follows from (\ref{wtXYto}), (\ref{overpsigma}), and Slutsky's lemma. 
\end{proof}
%\end{split}
If $\mu_X\leq \mu_Y$, then from (\ref{equhypho}) and $H_n \geq 0$, $n \in \N_+$,
for each $\alpha \in (0,1)$ and $z_{1-\alpha}$ as in Remark \ref{remConfIC},
\begin{equation}
\limsup_{n\to \infty}\PR(t_n > z_{1-\alpha}) 
\leq \lim_{n\to \infty}\PR(t_n + H_n > z_{1-\alpha}) = \alpha, 
\end{equation}%\frac{\mu_Y-\mu_X}{\sqrt{\frac{\wh{\sigma}_{X,a_n}^2}{a_n}+\frac{\wh{\sigma}_{Y,b_n}^2}{b_n}}}
i.e. the tests of the null hypothesis $\mu_X\leq \mu_Y$ with the regions of rejection $t_n > z_{1-\alpha}$, $n \in \N_+$, 
are pointwise asymptotically level $\alpha$ (see Definition 11.1.1 in \cite{lehmann2005testing}). 
We shall further use such tests for $z_{1-\alpha}=3$, so that $\alpha \approx 0.00270$. 
If for some selected $n$ we have $t_n \geq 3$, i.e. the null hypothesis as above can be rejected, then we shall informally say that the estimate 
$X_{a_n}\pm \frac{\wh{\sigma}_{X,a_n}}{\sqrt{a_n}}$ of $\mu_X$ is (statistically significantly) higher than such an estimate
$Y_{b_n}\pm \frac{\wh{\sigma}_{Y,b_n}}{\sqrt{b_n}}$ of $\mu_Y$. 

Most frequently, for some i.i.d. square integrable random variables $X_i'$, $i \in \N_+$, and 
such variables $Y_i'$, $i \in \N_+$, independent of $X_j'$, in such tests 
we shall use $X_n=\frac{1}{n}\sum_{i=1}^n X_i'$ and $\wh{\sigma}_{X,n}=\sqrt{\frac{1}{n-1}\sum_{i=1}^n(X_i'-X_n)^2}$,
and analogously for $Y_n$ and $\wh{\sigma}_{Y,n}$. In such a case $\frac{\wh{\sigma}_{X,n}}{\sqrt{n}}$ shall be called an 
estimate of the standard deviation of the mean $X_n$. 

\section{\label{secExpEst}Estimation experiments} 
We first performed $k$-stage minimization methods 
of the different estimators and for the different estimation problems, 
using $n_i=50\cdot 2^{i-1}$ samples in the $i$th stage 
for $i=1,\ldots,k$ for various $k\in \N_+$ (see Section \ref{secSimpleMin}). 
We chose $k=3$ for the problem of estimating $q_{1,a}(x_0)$, 
$k=5$ for $\mgf(x_0)$, and $k=6$ 
for $p_{T,a}(x_0)$ and $q_{2,a}(x_0)$. 
%Such $n_i$ were chosen only for illustrative purposes. %Note that from the discussion in Remark \ref{remPs}
%it could be in a sense asymptotically more efficient to choose $n_i$ growing more rapidly. 
%more efficient to use a faster growing sequence $n_i$). %TODO 
We first used $a=0.05$ and $M=10$ time-independent IS basis functions as in (\ref{riexp}). 
For $i=1,2,\ldots, 6$, the IS drifts $r(b_i)$ corresponding to the minimization results $b_i$ from the $i$th 
stage of the MSM of different estimators 
for estimating the translated committor $q_{2,a}(x_0)$ are shown in Figure \ref{fig3WAll}. 
The IS 
drifts corresponding to the final results of 
MSM for the estimation of all the expectations are shown in Figure \ref{fig3WOpt}. 
In figures \ref{fig3WAll} and \ref{fig3WOpt} we also show for comparison approximations 
of the zero-variance IS drifts $r^*$ for the diffusion problems 
for the translated committors and MGF, computed 
from formula (\ref{rstar2}) using finite differences instead of derivatives and finite difference approximations of $u$ 
in that formula computed as in Section \ref{secSpecNumExp}. 
In Figure \ref{fig3WAll}, the IS drifts from the consecutive stages of the MSM of 
$\wh{\msq2}$ and $\wh{\ic}$ seem to converge the fastest to some limiting drift 
close to (the approximation of) $r^*$, from the MSM of $\wh{\msq}$ --- slower, and of $\wh{\ce}$ --- the slowest. 
See Section \ref{secIntu} for some intuitions behind these results. 

\begin{figure}[h]%
\centering
\subfloat[]{\includegraphics[width=0.45\textwidth]{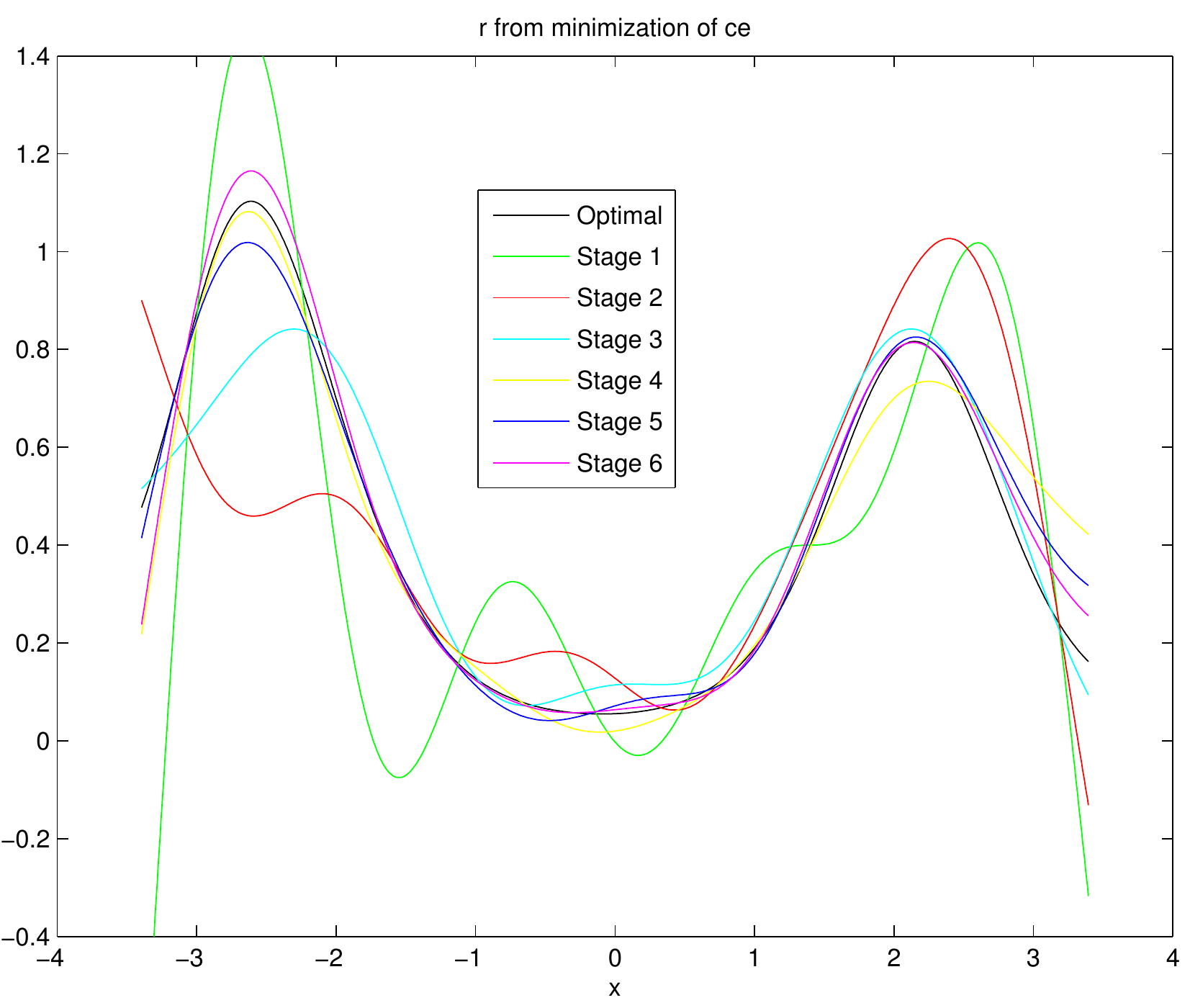}}
\qquad  
\subfloat[]{\includegraphics[width=0.45\textwidth]{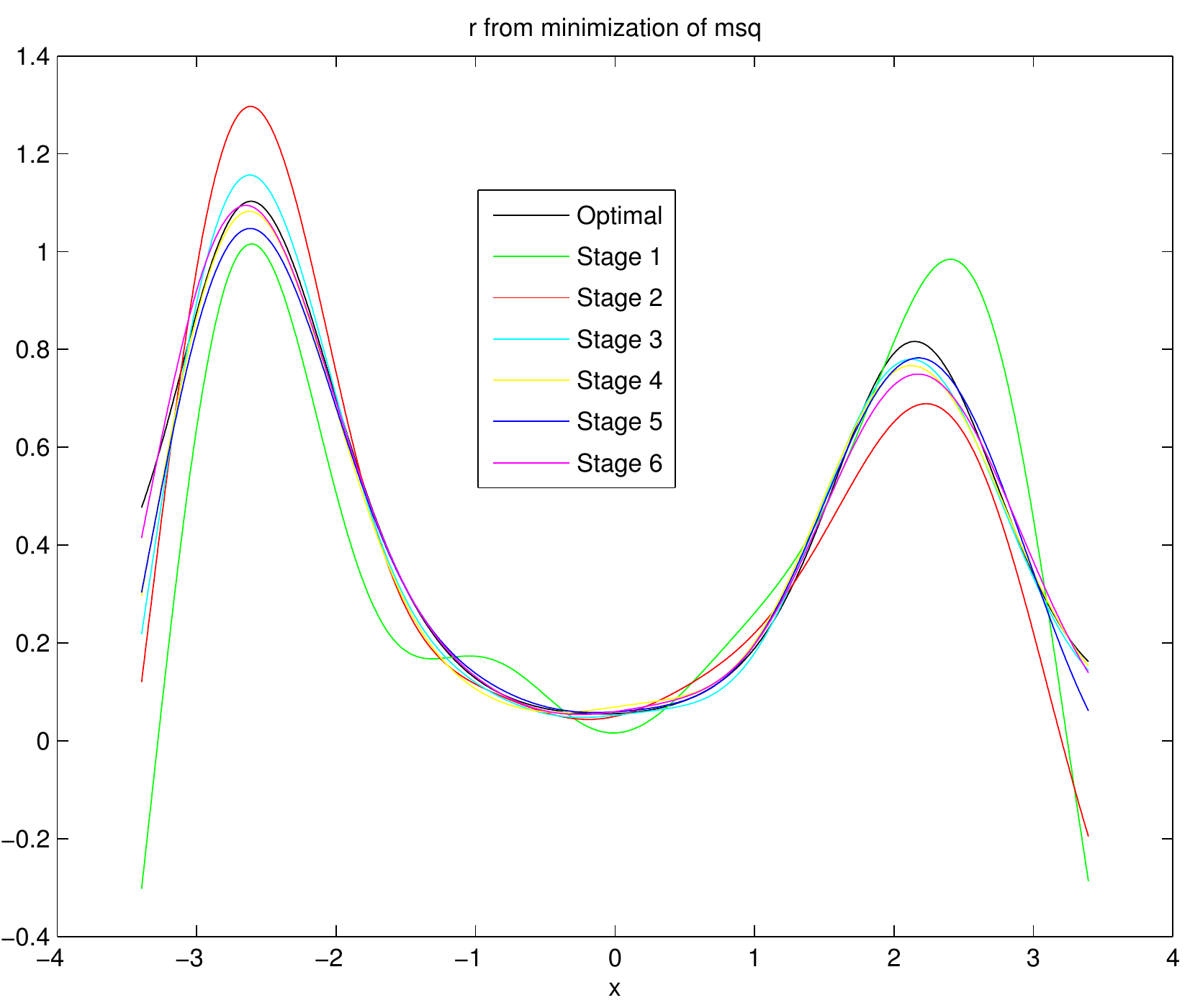}} 
\\
\subfloat[]{\includegraphics[width=0.45\textwidth]{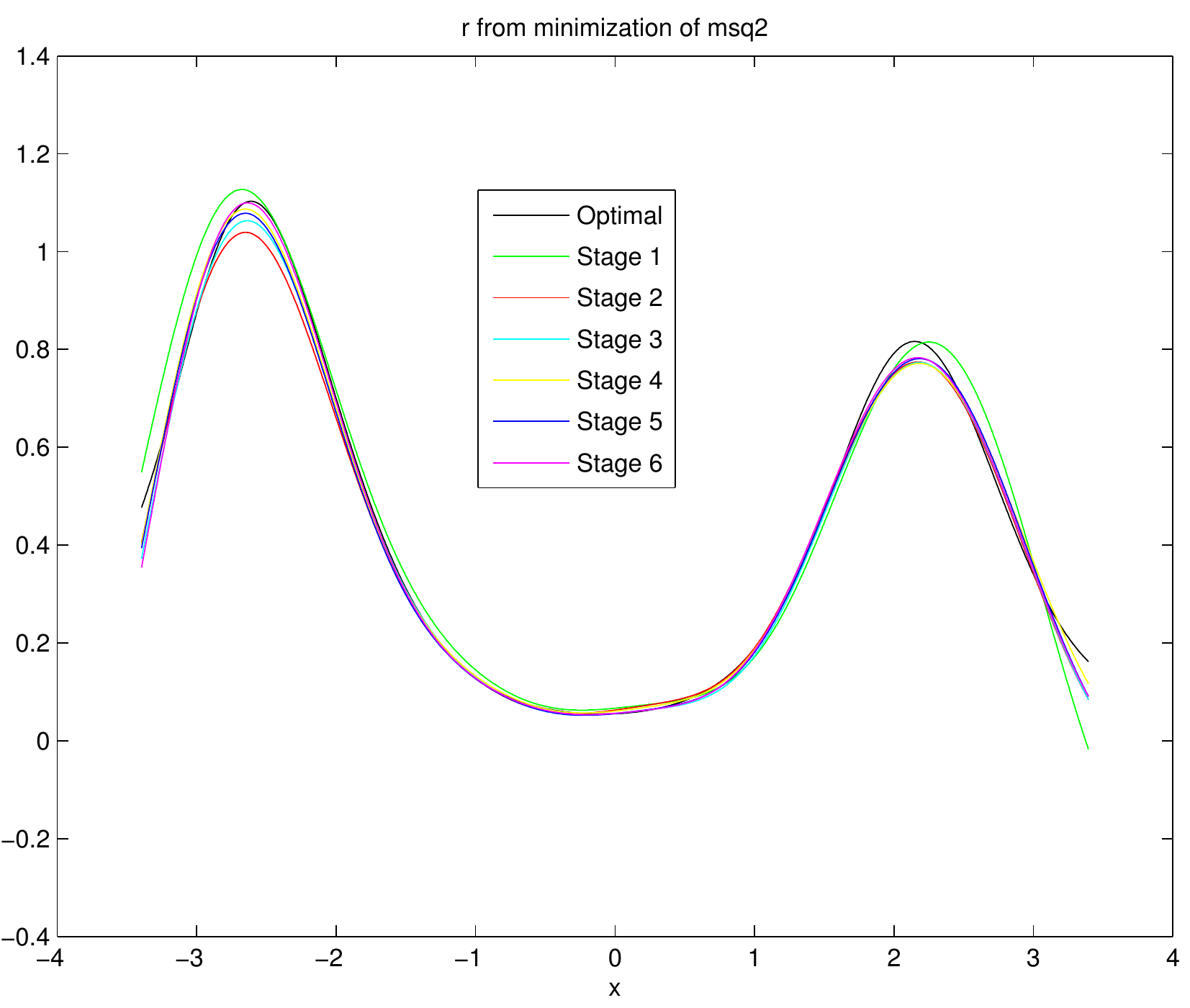}}
\qquad  
\subfloat[]{\includegraphics[width=0.45\textwidth]{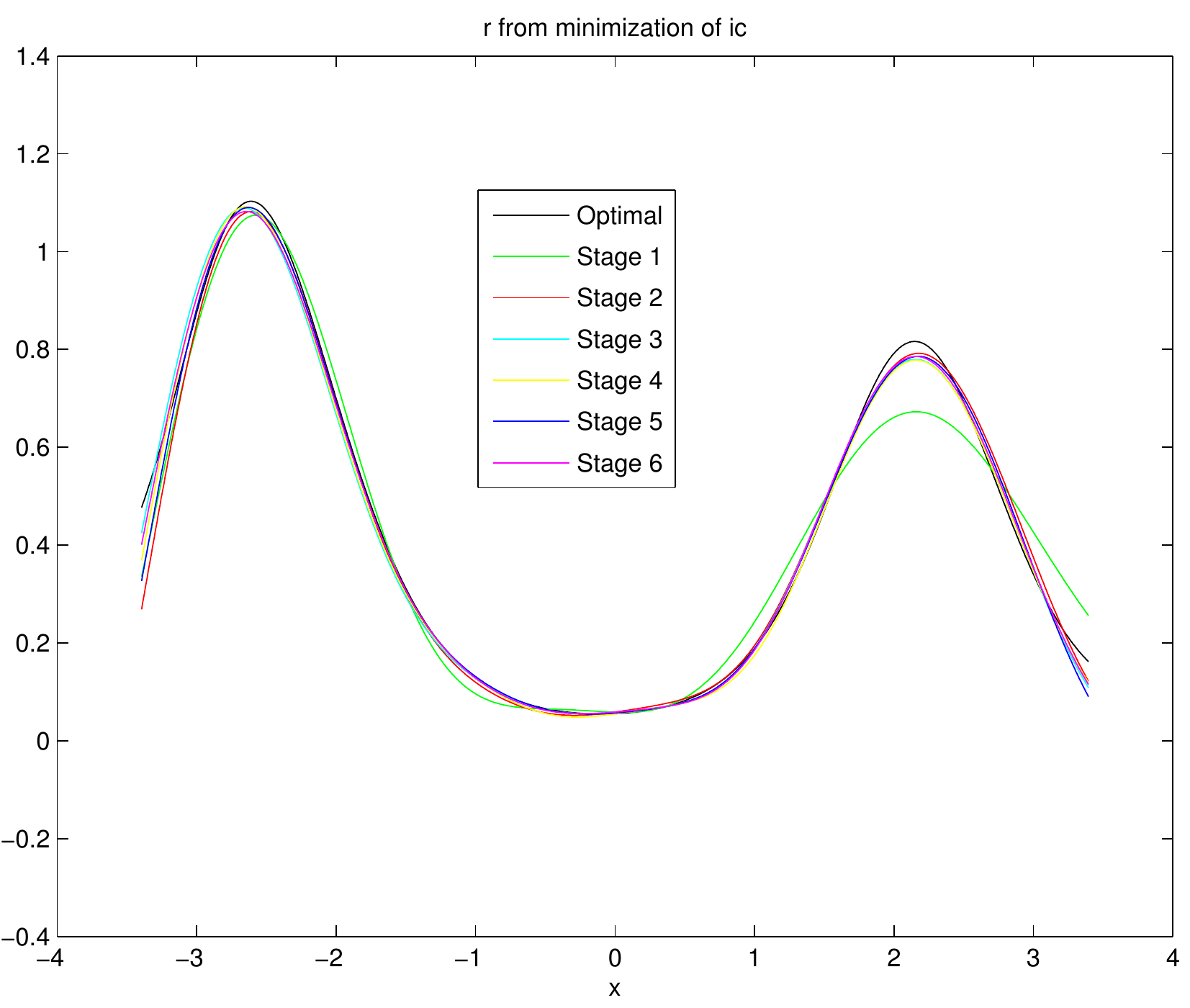}}

\caption{\label{fig3WAll} IS drifts from different stages of MSM for estimating $q_{2,a}$, minimizing $\wh{\ce}$ in (a),
$\wh{\msq}$ in (b), $\wh{\msq2}$ in (c), and $\wh{\ic}$ in (d). 'Optimal' denotes an approximation of the zero-variance IS drift $r^*$.}
\end{figure}

\begin{figure}[h]%
\centering
\subfloat[]{\includegraphics[width=0.45\textwidth]{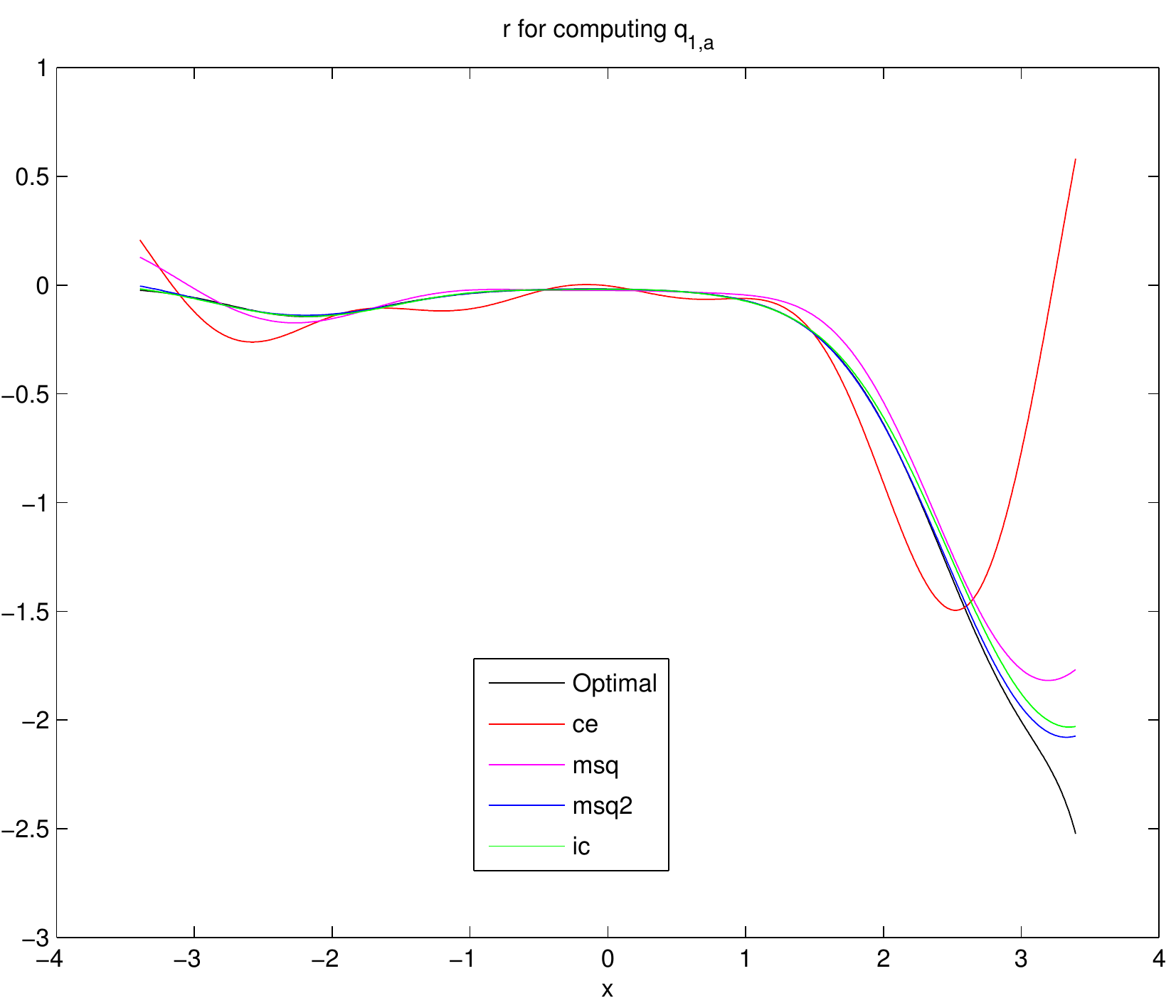}}
\qquad  
\subfloat[]{\includegraphics[width=0.45\textwidth]{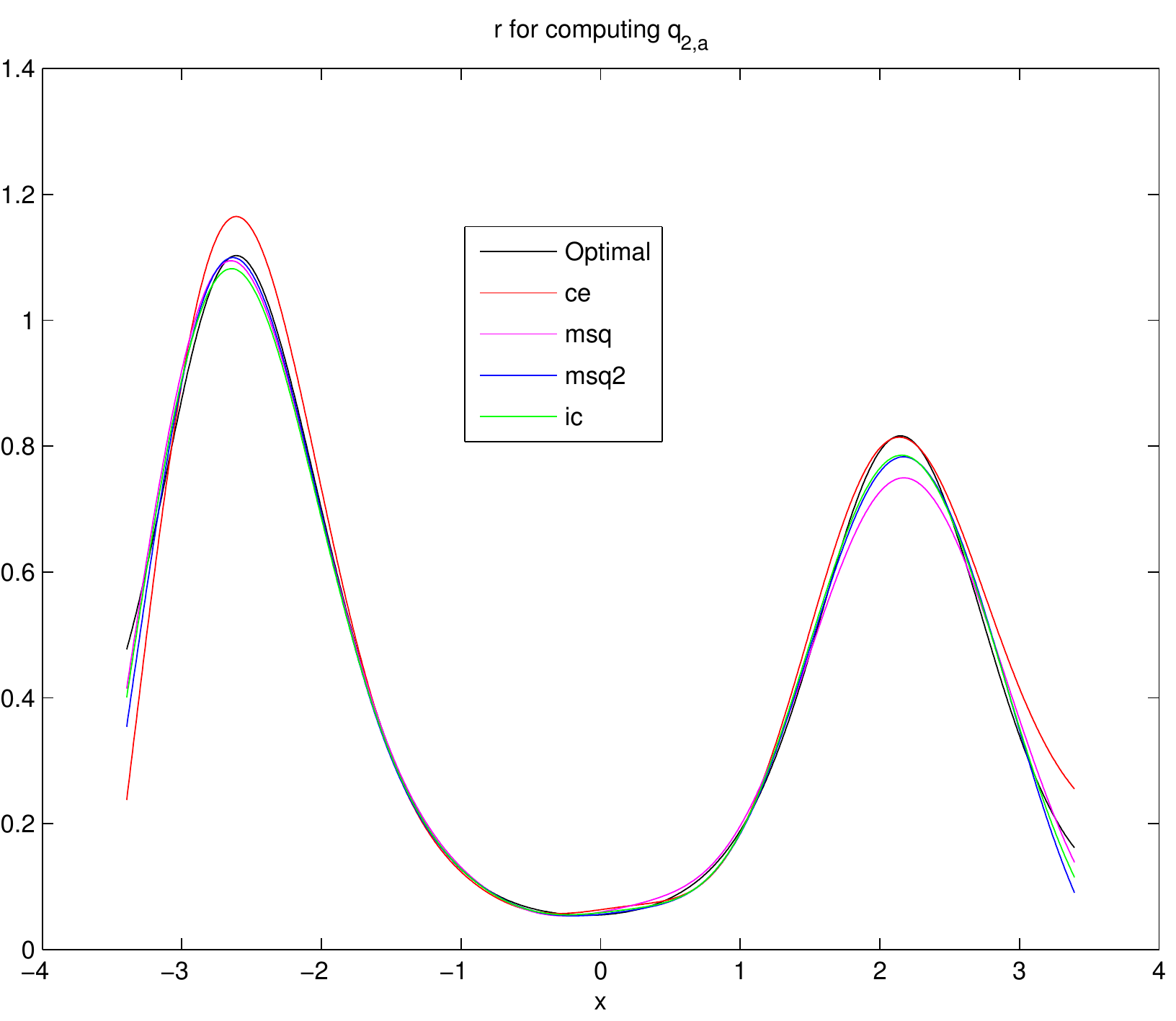}} 
\\
\subfloat[]{\includegraphics[width=0.45\textwidth]{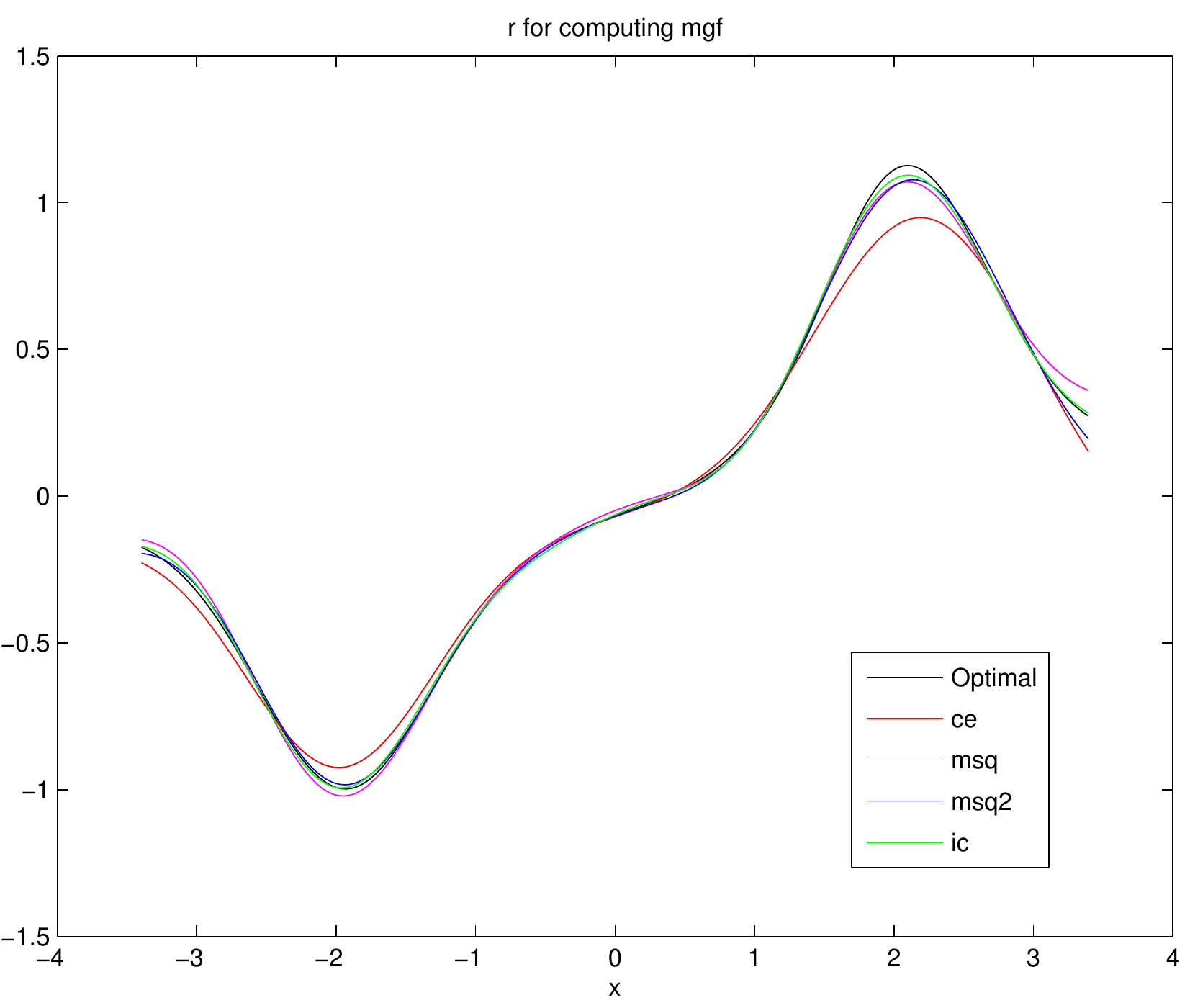}}
\qquad  
\subfloat[]{\includegraphics[width=0.45\textwidth]{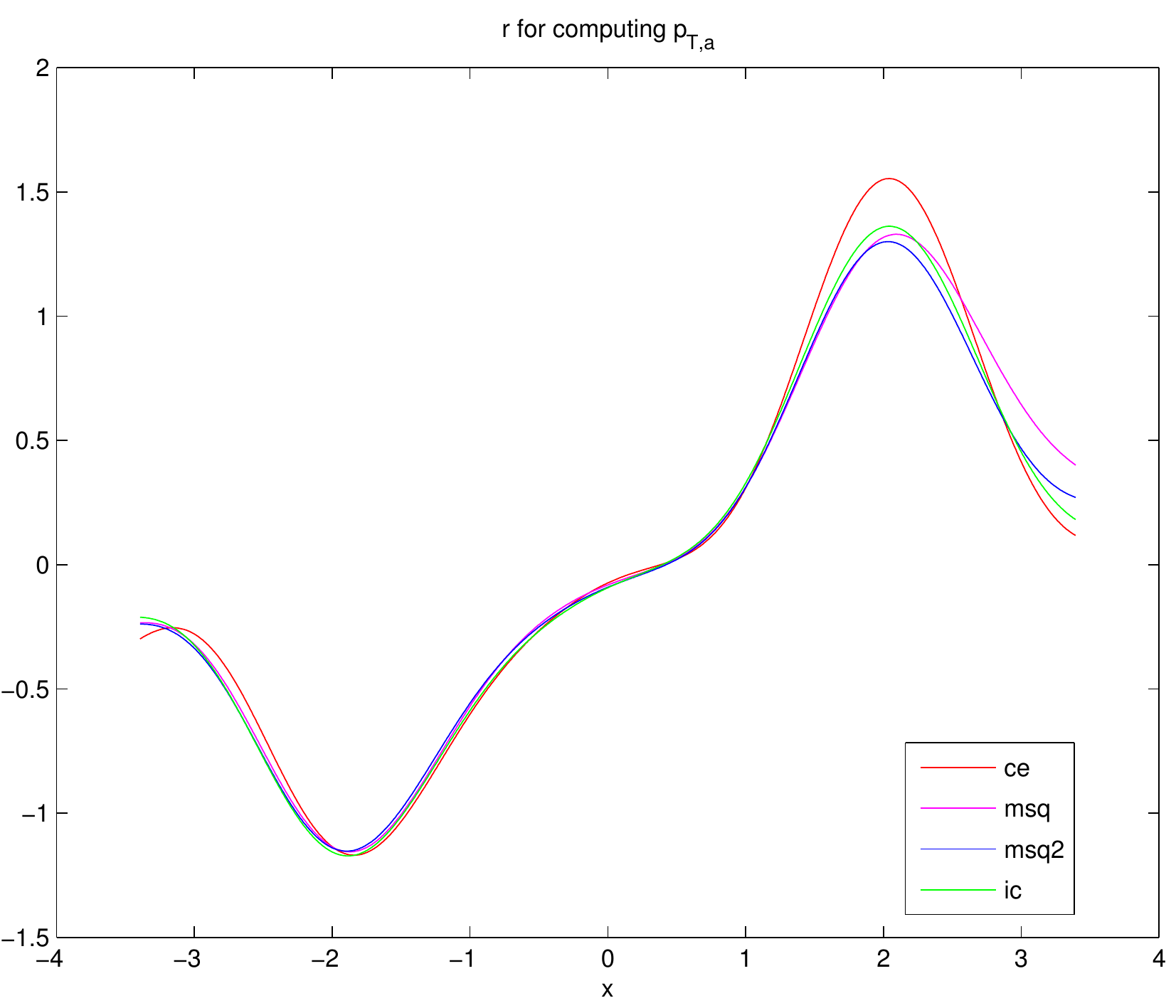}} 
\caption{\label{fig3WOpt} Final IS drifts from the MSM experiments minimizing different estimators,  for $q_{1,a}(x_0)$
in (a), $q_{2,a}(x_0)$ in (b), $\mgf(x_0)$ in (c), and $p_{T,a}(x_0)$ in (d). ``Optimal`` denotes an approximation of the zero-variance IS drift.}
\end{figure}

Consider a numerical experiment in which, for a given IS parameter $v \in A$, we 
compute unbiased estimates of the mean cost $c(v)$ as well as of the
variance $\var(v)$ and the
(theoretical) inefficiency constant $\ic(v)$ 
of the IS estimator of the expectation of the functional of the Euler scheme  
of interest, using  estimators (\ref{cest}),
(\ref{ic2Est}), and (\ref{varest}) respectively for $b=b'=v$ and $n=10$. 
For some $\beta_1, \ldots,\beta_{n}$ i.i.d. $\sim\PU$,
these estimators are evaluated on $(\xi(\beta_{i},v))_{i=1}^{n}$ (for $\xi$ as in (\ref{xigss})),  
so that the computations involve simulating $n$ 
independent Euler schemes with an additional drift $r(v)$ as in (\ref{Xkb}). 
Such experiments for the 
different estimation problems and for $v$ equal to the final results 
of the MSM of $\wh{\ce}$, $\wh{\msq}$, 
$\wh{\msq2}$, and $\wh{\ic}$ as above, and for $v=0$ for CMC,  
were repeated independently $K$ times in an outer MC loop for different $K$. 
For the problem of the estimation of $q_{2,a}(x_0)$ we additionally used as $v$ 
the minimization result from the third step of MSM. 
For the MGF we made in all the cases $K=75000$ repetitions. 
For the translated committors, when using CMC or IS with a parameter $v$ 
from the $3$-stage MSM of estimators other than $\wh{\ce}$,  
we took $K=2000$, while in the other cases we chose $K=5000$. 
For $p_{T,a}$ we made $K=20 000$ repetitions both for the CMC 
and when using $v$ from the MSM of $\wh{\ce}$, and $K=2\cdot10^5$ repetitions for $v$ from the MSM of 
$\wh{\msq}$, $\wh{\msq2}$, and $\wh{\ic}$. 
The MC means of the inefficiency constant and variance estimators from the outer loops 
for the translated committors and MGF are given in Table \ref{tabCompIC}, along with the
estimates of the standard deviations of such means. 
For $p_{T,a}(x_0)$, such outer MC loop estimates of the
inefficiency constants, variances, and mean costs %and standard deviations of the results 
are given in Table \ref{tabCompT}. 

\begin{table}[h]
\resizebox{14cm}{!} {
\begin{tabular}{|l|c|c|c|c|c|}
\hline
& CMC & $\wh{\ce}$&$\wh{\msq}$&$\wh{\msq2}$&$\wh{\ic}$ \\
\hline
&\multicolumn{5}{c|}{\text{Estimates of inefficiency constants ($\cdot 10^{-3}$)}}\\
\hline
$q_{1,a}$ & $7204\pm 92 $ & $4855\pm 649$ & $507 \pm 10 $ & $129.1\pm 3.8$ & $132.9\pm 3.5$\\
$q_{2,a}$,\ $k=3$& $7300 \pm 88 $& $553 \pm 26$ & $60.8 \pm 1.2$ & $53.4 \pm 1.0$& $50.5 \pm 1.1$\\
$q_{2,a}$ ,\ $k=6$&  & $73.1\pm 1.0$ & $56.10 \pm 0.69$ & $48.78 \pm 0.59$& $47.99 \pm 0.61$\\
$\mgf$ & $2041 \pm 15$& $ 6.276 \pm0.055$ & $3.691 \pm 0.035$ & $3.177 \pm 0.029 $& $3.163\pm 0.028$\\
\hline 
&\multicolumn{5}{c|}{\text{Estimates of variances $(\cdot 10^{-3})$}}\\
\hline
$q_{1,a}$ & $174.0 \pm 1.7$ & $126\pm 14$ & $12.36  \pm 0.20$ & $3.155\pm0.083$& $3.208\pm 0.073$ \\
$q_{2,a}$, $k=3$ & $177.43 \pm 1.7$ & $16.35\pm 0.79$ & $1.437 \pm 0.021$ & $1.268 \pm 0.019$ & $1.223 \pm 0.020$\\
$q_{2,a}$, $k=6$&  & $1.794 \pm 0.021$ & $1.347 \pm 0.013$ &$1.164 \pm 0.011$ & $1.169\pm 0.012$\\
$\mgf$& $49.41 \pm 0.32$ & $0.9040\pm 0.0072$ & $0.3732\pm 0.0029$ & $0.3195\pm 0.0024$  & $0.3209 \pm 0.0024$\\
\hline
\end{tabular}
}
\caption{\label{tabCompIC} Estimates of the inefficiency constants and variances of the estimators of the translated committors for $a=0.05$ and the MGF 
when using CMC or IS with IS parameters from the MSM of different estimators. For $q_{2,a}$ we consider using the IS parameters from the 
$k$th stages of MSM for $k\in\{3,6\}$.} 
\end{table} 

\begin{table}[h]
\resizebox{14cm}{!} {
\begin{tabular}{|l|c|c|c|c|c|}
\hline
& CMC & $\wh{\ce}$&$\wh{\msq}$&$\wh{\msq2}$&$\wh{\ic}$ \\
\hline
$\ic\ (\cdot 10^{-3})$ & $1391 \pm 5$& $65.22 \pm 0.30 $ & $63.332 \pm 0.09 $ & $62.677 \pm 0.090 $& $61.863 \pm 0.091 $\\
$\var\ (\cdot 10^{-3})$ & $150.78 \pm 0.58$ & $10.258 \pm 0.042$ & $10.036 \pm 0.013$ & $9.8937\pm 0.0126 $& $ 9.9491 \pm 0.0130 $\\
$c$ & $9.227 \pm 0.004 $& $6.3608 \pm 0.0066 $ & $6.3117\pm 0.0021 $ & $6.3355 \pm 0.0021 $& $ 6.2173 \pm 0.0021$\\
\hline
\end{tabular}
% mean t: =6.2173 \pm 0.0021
}
\caption{\label{tabCompT} Estimates of the inefficiency constants 
and variances of the estimators of $p_{T,a}$ for $a=0.05$ as well as of the
mean costs, when using CMC or IS with the IS parameters 
from the MSM of different estimators.} 
\end{table} 

%In tables \ref{tabCompIC} and \ref{tabCompT} we also provide 
\begin{remark}
Note that due to Remark \ref{remCond}, the variables $C$ as 
above have all moments (and thus also variance) finite under $\PQ(v)$, $v \in A$. 
For $v=0$ (i.e. for CMC), from the boundedness of the considered $Z$, we have the finiteness of the mean costs and of 
the variances and inefficiency constants of the estimators 
of the Euler scheme expectations of interest as well as of the variances of the utilized estimators of such quantities. 
For the general $v$ and bounded stopping times (as is the case for such times equal to $\tau'$ as in (\ref{taup})
when estimating $p_{T,a}$), the finiteness of the quantities as in the previous sentence 
follows from the corresponding $Z$ being bounded, as well as from Theorem \ref{thYmore} and Remark \ref{remCondUN}. 
In cases when the stopping time is not bounded (like for such a time equal to $\tau$ for the MGF and translated committors as above), 
one can ensure such boundedness by terminating the simulations at some fixed time as discussed in Section \ref{secCondLETS}. 
We did not terminate our simulations, but still our results can be interpreted as coming from simulations 
terminated at some time larger than any of the exit times encountered in our experiments. 
\end{remark}

% In the below remark we introduce some methods for testing statistical hypotheses 
% and some conventions, needed further on. 

From tables \ref{tabCompIC} and \ref{tabCompT} we can see that using the IS parameters from the MSM of 
$\wh{\ic}$ and $\wh{\msq2}$ led in each case to the lowest estimates of variances and (theoretical) inefficiency constants, followed by 
the ones from 
using the parameters from the MSM of $\wh{\msq}$, and finally $\wh{\ce}$. % and CMC. %(see Section \ref{secIntu} for discussion of some intuitions behind these results). 
%Some intuitions behind these results 
%TODO Some possible intuitions behind these observations 
% are given by remarks \ref{remMsq2Close}, \ref{remICClose} (however, note that we do not know if  latter Remark) 
% and possible intuition for these observations %was given in 
Using CMC led in each case to the highest such estimates. 
%of variances and inefficiency constants. % in all cases except for $q_{1,a}$, for which 
For $q_{2,a}(x_0)$ and each of $\wh{\ce}$, $\wh{\msq}$, and $\wh{\msq2}$, using the IS parameter from the sixth stage of MSM led to a
lower estimate of variance and inefficiency constant than using such a parameter from the third stage. 
Note also that the estimates of the inefficiency constants and variances for $q_{2,a}(x_0)$ when using 
the IS parameters from the third stage of the MSM of $\wh{\msq2}$ and $\wh{\ic}$ are lower than when using the parameters 
from the sixth stage of the MSM of $\wh{\ce}$ and $\wh{\msq}$ (though for the estimates of the inefficiency constants for $\wh{\msq2}$ and $\wh{\msq}$
we cannot confirm this at the desired significance level as in Section \ref{secStat}). %such estimates are comparable to these when minimizing $\wh{\ce}$. 
For $p_{T,a}(x_0)$, the estimate of the inefficiency constant, variance, and mean cost is respectively 
lower, higher, and lower when minimizing $\wh{\ic}$ than $\wh{\msq2}$. 
%TODO? Note, however, that the relative 
% difference
% of these estimates are quite small, 
% their differ by no more than $2\%$. 
%(note that we cannot exclude that  is the case also for estimation other quantities, 
%however we would need to estimate their second-stage variances, mean exit times and inefficiency constants more precisely 
%to verify this). %This result 
Some intuitions behind these results are given by Theorem \ref{thicvar} 
and Remark \ref{remicvar}, see also the discussion in Section \ref{secIntu}. 

Using the IS parameters $v$ from the MSM of $\wh{\ic}$ as above and averaging the estimates 
from the $nK$ simulations available in each case 
we computed the IS MC estimates of the quantities of interest: 
$\mgf(x_0)$, and using the translated estimators as in Section \ref{secImpSpec} 
also of $p_{T}(x_0)$ and $q_{i}(x_0)$, $i=1,2$. %(note that this computation corresponds to the second stage of algorithm from Section \ref{secSimpleMin}). 
The results are presented in Table \ref{tabResults}.  %we present estimates of $\mgf(x_0)$, $p_{T}(x_0)$, and $q_{i}(x_0)$, for $i=1,2$,
Note that we have $q_1(x_0)=1-q_2(x_0)\approx 0.78$ and the
estimates of the inefficiency constants in Table \ref{tabCompIC} for estimating the lower value committor $q_2(x_0)$ 
are lower. Thus, it seems reasonable to use the translated IS estimator for $q_2(x_0)$ also for computing $q_1(x_0)$ 
as discussed in Section \ref{secImpSpec}.

\begin{table}[h]
\resizebox{12cm}{!} {
\begin{tabular}{|l|c|c|c|}
\hline
 $q_1(x_0)$&$q_2(x_0)$&$\mgf(x_0)$&$p_T(x_0) $  \\
\hline
 $0.7751\pm0.0004$& $0.22597\pm0.00025$& $0.16682\pm (6 \cdot 10^{-5})$ & $0.18396\pm (7\cdot 10^{-5})$\\
\hline
\end{tabular}
}
\caption{\label{tabResults} 
Estimates of different expectations obtained from IS MC using IS parameters from the MSM of $\wh{\ic}$. 
} 
\end{table} 

In the above experiments utilising $nK$ simulations we also computed the MC estimates of the mean costs $c(v)$. 
For comparison we also computed an estimate of the mean cost in CMC (equal to $c(0)=\E_{\PU}(\tau)$), using an MC average of such costs from
$7.5 \cdot10^5$ simulations. The results are provided in Table \ref{tabMeans}.
Note that the estimates of the mean costs in tables \ref{tabCompT} and \ref{tabResults} are 
lower for IS using the IS parameters $v$ from the MSM methods for computing $\mgf(x_0)$ and $p_{T,a}(x_0)$ than for the respective CMC methods. 
As discussed in Section \ref{secISIneff}, an intuition behind these results is provided by Theorem \ref{thDecrCost}. 

\begin{table}[h]
\resizebox{8cm}{!} {
\begin{tabular}{|l|c|c|c|c|}
\hline
&$q_{1,a}(x_0)$&$q_{2,a}(x_0)$&$\mgf(x_0)$& CMC  \\
\hline
$c$& $41.26\pm0.28$ &$41.18\pm0.28$ & $9.89\pm0.03$& $41.44\pm 0.15$ \\
\hline
\end{tabular}
}
\caption{\label{tabMeans} 
Estimates of the mean costs when using the IS parameters from the MSM of $\wh{\ic}$ and for CMC, 
for the problems of computing the translated committors for $a=0.05$ and MGF. 
} 
\end{table} 

We also performed two-stage experiments similar as above for $q_{2,a}(x_0)$ and $p_{T,a}(x_0)$ 
for several different added constants $a\in \R_+$ other than $a=0.05$ considered above. 
For $q_{2,a}(x_0)$ we used the IS basis functions as above, while for $p_{T,a}(x_0)$ also
the time-dependent basis functions as in (\ref{rtimedep}) for $M=5$ and $M=10$ and various $p\in \N_+$. 
This time in the first stages we performed the MSM only of $\wh{\ic}$ for 
$k=3$ and $n_i=400\cdot 2^{i-1}$, $i=1,\ldots,k$, so that the number of samples $n_k =1600$ used in the final stages of MSM 
was the same as for $a=0.05$ above. 
In the second stages we estimated the inefficiency constants, mean costs, and variances in an external loop like above. 
For $q_{2,a}(x_0)$ we made $K=3000$ repetitions in such a loop, while for $p_{T,a}(x_0)$ --- $K=10000$ 
for the basis functions as in (\ref{riexp}), as well as $K=50000$ for the basis functions as in (\ref{rtimedep}) for $M=5$ and $K=30000$ for $M=10$. 
The results are presented in tables \ref{tabCompCom2a}, \ref{tabCompTa}, and \ref{tabCompTaTime}, along with the results for the case of 
$a=0.05$ considered before. The smallest estimates of the inefficiency constants and variances for $q_{2,a}(x_0)$ were obtained for $a=0.05$. 
For $p_{T,a}(x_0)$ and the basis functions as in (\ref{riexp}), we obtained the smallest variance for $a=0.2$ and the lowest 
inefficiency constants for $a=0.1$ and $a=0.2$. Among all the cases for $p_{T,a}$, the smallest variances and theoretical inefficiency 
constants were received for $a=0$ and when using the time-dependent IS basis functions (\ref{rtimedep}) for $M=10$ and $p=3$. 

\begin{table}[h]
%\resizebox{14cm}{!} {
\begin{tabular}{|l|c|c|c|c|c|c|}
\hline
$a=$& $0$ & $0.05$ & $0.1$ & $0.2$  \\
\hline
$\ic\ (\cdot 10^{-3})$ & $82.6 \pm 2.3$ & $47.99 \pm 0.61$ & $51.89 \pm 0.80$ & $55.52 \pm 0.86$\\
$\var\ (\cdot 10^{-3})$ & $2.008 \pm 0.058$ & $1.169\pm 0.012$ & $1.252 \pm 0.016$ & $1.359 \pm 0.016$\\
\hline
\end{tabular}
\caption{\label{tabCompCom2a} Estimates of the inefficiency constants and variances of the IS estimators of $q_{2,a}$ for different $a$, 
corresponding to IS with the parameters from the MSM of $\wh{\ic}$.} 
\end{table} 

\begin{table}[h]
\resizebox{14cm}{!} {
\begin{tabular}{|l|c|c|c|c|c|}
\hline
$a=$& $0$ & $0.05$ & $0.1$ & $0.2$& $0.3$ \\
\hline
$\ic\ (\cdot 10^{-3})$ &$100.8\pm 0.7$& $61.86 \pm 0.09 $& $55.95\pm 0.35$&  $56.44\pm 0.39$&$61.43 \pm 0.48$\\
$\var\ (\cdot 10^{-3})$ &$17.88 \pm 0.10$ & $ 9.949 \pm 0.013 $& $8.224\pm 0.047$&  $7.544\pm 0.050$&$7.794\pm 0.058$\\
 $c$ & $5.641\pm0.009$ &$ 6.2173 \pm 0.0021$& $6.803\pm0.009$& $7.489\pm 0.009$& $7.874\pm 0.009$\\ 
\hline
\end{tabular}
% mean t: =6.2173 \pm 0.0021
}
\caption{\label{tabCompTa} Estimates of the inefficiency constants and variances of the estimators of $p_{T,a}$ for different $a$, 
and estimates of the mean costs, corresponding to IS with the parameters from the MSM of $\wh{\ic}$
and using $M=10$ time-independent IS basis functions as in (\ref{riexp}).} 
\end{table} 

\begin{table}[h]
\resizebox{14cm}{!} {
\begin{tabular}{|l|c|c|c|c|c|c|c|}
\hline
& \multicolumn{4}{c|}{$M=5$ }& \multicolumn{3}{c|}{$M=10,\ a=0$ } \\
\hline
& $p=1,\ a=0$  &$p=1,\ a=0.05$ & $p=2,\ a=0$ & $p=3,\ a=0$& $p=1$&$p=2$ &$p=3$ \\
\hline
$\ic\ (\cdot 10^{-3})$ & $63.01 \pm 0.39$ & $93.44 \pm 0.45$&$48.18\pm 0.28$&$46.52\pm 0.23$& $34.11\pm 0.16$ &$22.95\pm 0.17$ &$17.34 \pm 0.12$\\
$\var\ (\cdot 10^{-3})$ & $10.544 \pm 0.064$ &$13.60 \pm0.06$&$7.975\pm 0.045$&$7.629\pm 0.036$&$5.714 \pm 0.025$&$3.899\pm0.027$ &$2.951\pm 0.019$ \\
$c$ & $5.978 \pm 0.003$& $6.871 \pm0.004$&$6.037\pm 0.003$&$6.086 \pm 0.003$& $5.966\pm 0.004$&$5.878 \pm 0.004$ &$5.887 \pm 0.004$\\
\hline
\end{tabular}
% mean t: =6.2173 \pm 0.0021
}
\caption{\label{tabCompTaTime} Estimates of the inefficiency constants and variances of IS estimators of $p_{T,a}$ for different $a$, as well as of
the mean costs, corresponding to IS with the parameters from the MSM of $\wh{\ic}$, 
 using the time-dependent IS basis functions as in (\ref{rtimedep}) for different $M$ and $p$.} 
\end{table} 

%   \begin{figure}[h]
%    \includegraphics[width=0.5\textwidth]{./IterICDSampleT3WnB20TimeDepT10alpha0t3GraphBest.pdf}
%    %IterICDSampleT3WnB20TimeDepT10alpha0t3GraphsBest.pdf}
%    \caption{Tilted potentials for $p_{T,a}$ from minimization of $\wh{\ic}$ at different times.}
%    \label{pot3W}
%    \end{figure}

In our experiments, when using CMC and IS MC with different sets of IS basis functions, 
the proportionality constants $p_{\dot{C}}$ as in Chapter \ref{secIneff} were considerably different. 
Thus, to compare the efficiency of the MC methods using estimators corresponding to these different bases, one should 
compare their practical rather than theoretical inefficiency constants. 
%rather than comparing the theoretical inefficiency constants of estimators corresponding to these different situations, to compare  
%we should compare their (implementation-dependent) practical inefficiency 
%constants, equal to the theoretical constants multiplied by the respective $p_{\dot{C}}$. 
We performed separate experiments approximating some $p_{\dot{C}}$ as above and computing the corresponding practical inefficiency constants 
(equal to the products of such $p_{\dot{C}}$ and the respective theoretical inefficiency constants).
For $n=10^5$, we ran $n$-step CMC and IS MC procedures for estimating $q_{2,0.05}(x_0)$ using the IS basis functions as in (\ref{riexp}), 
and for estimating $p_{T,a}$: for $a=0.05$ for IS basis functions as in (\ref{riexp}) for $M=10$, and
for $a=0$ for IS basis functions as in 
(\ref{rtimedep}): for $M=5$ and $p=1$, and for $M=10$ and $p\in \{1,3\}$. 
When performing the IS MC we used the IS parameters from the final stages of the corresponding MSM procedures as above. 
For $C_i$ being the theoretical cost of the $i$th step of a given MC procedure and $\dot{C}_i$ being its practical cost equal to
its computation time calculated using the matlab tic and toc functions, as an approximation of $p_{\dot{C}}$ we used 
the ratio $\wh{p}_{\dot{C},n}=\frac{\sum_{i=1}^n\dot{C}_i}{\sum_{i=1}^nC_i}$. Treating $(\dot{C}_i, C_i)$, $i=1,2,\ldots,$ as i.i.d. 
random vectors with square-integrable coordinates, for $p_{\dot{C}}:=\frac{\E(\dot{C}_1)}{\E(C_1)}$, from the delta method it easily follows that for 
\begin{equation}
\sigma: =p_{\dot{C}} \sqrt{\frac{\Var(C_1)}{(\E(C_1))^2} +\frac{\Var(\dot{C_1})}{(\E(\dot{C}_1))^2}-2\frac{\Cov(C_1,\dot{C_1})}{\E(\dot{C}_1)\E(C_1)}}
\end{equation}
we have
%\begin{equation}
$\sqrt{n}(\wh{p}_{\dot{C},n}-p_{\dot{C}})\Rightarrow \ND(0,\sigma^2)$. For $\wh{\sigma}_n$ being an estimate of $\sigma$ in which instead of means, variances, and 
covariances one uses their standard unbiased estimators computed using
$(C_i,\dot{C}_i)_{i=1}^n$, we have a.s. $\frac{\wh{\sigma}_n}{\sigma}\to 1$. Thus, from Slutsky's lemma, 
$\frac{\sqrt{n}}{\wh{\sigma}_n}(\wh{p}_{\dot{C},n}-p_{\dot{C}})\Rightarrow \ND(0,1)$, which can be used for constructing asymptotic confidence intervals for $p_{\dot{C}}$.
In Table \ref{tabComppC} we provide the computed estimates in form $\wh{p}_{\dot{C},n} \pm \frac{\wh{\sigma}_n}{\sqrt{n}}$. 
It can be seen that these approximations of $p_{\dot{C}}$ are 
close for $q_{2,a}(x_0)$ and $p_{T,a}(x_0)$ when using in both cases CMC or IS with the basis functions as in (\ref{riexp}), and 
for $p_{T,a}(x_0)$ when using the basis functions as in (\ref{rtimedep}) for $M=10$ and different $p$. However, such $p_{\dot{C}}$ differ significantly 
for the other pairs of MC methods. 
In Table \ref{tabComppC} we also provide the estimates of practical inefficiency constants $\dot{\ic}$ obtained by multiplying the corresponding
estimates of the theoretical inefficiency constants computed earlier by the received approximations of $p_{\dot{C}}$. From this table we can see that using IS
in the considered cases led to considerable practical inefficiency constant reductions over using CMC.

\begin{table}[h]
\resizebox{14cm}{!} {
\begin{tabular}{|l|c|c|c|c|c|c|c|}
\hline
& \multicolumn{2}{c|}{$q_{2,a}$ }& \multicolumn{5}{c|}{$p_{T,a}$ } \\
\hline
&CMC &$M=10$ &CMC & $M=10$ &$M=5$,\ $p=1$& $M=10$,\ $p=1$& $M=10$,\ $p=3$ \\
\hline
$p_{\dot{C}}$ ($\cdot 10^{-6}$s)& $33.612\pm 0.003$&$137.06\pm0.01$ &$33.998 \pm 0.004$&$138.75\pm 0.01$&$96.36\pm0.01$& $148.92\pm0.01 $ &$149.74\pm 0.01$\\
$\dot{\ic}$ ($\cdot 10^{-3}s$)& $245.4 \pm 3.0$&$6.578\pm 0.084$ &$47.3 \pm 0.2$&$8.584 \pm 0.013$&$6.072 \pm 0.038$& $5.080 \pm 0.024$ &$2.597 \pm 0.018$\\
\hline
\end{tabular}
}
\caption{\label{tabComppC} Estimates of $p_{\dot{C}}$ and $\dot{\ic}$ for computing $q_{2,a}$ and $p_{T,a}$ for various $a$ 
and IS basis functions as discussed in the main text. }
\end{table} 
% ics: =$5.080 \pm 0.024$
% ics: =2.597 \pm 0.018

\section{Experiments comparing the spread of IS drifts.} \label{secSpread}
In the experiments described in this section we consider the assumptions 
as for the estimation of $q_{2,a}(x_0)$ in Section \ref{secSpecNumExp}. For the
estimators $\wh{\est}$ equal to each of 
$\wh{\ce}$, $\wh{\msq}$, $\wh{\msq2}$, and $\wh{\ic}$, 
we performed $20$ independent SSM experiments 
for $n_1=100$ and $b'=0$ as in Section \ref{secSimpleMin}, i.e. 
minimizing $b\to\wh{\est}_{n_1}(b',b)(\wt{\chi}_1)$ 
for some $\wt{\chi}_1$ as in that section. 
For each such experiment for $\wh{\est}=\wh{\ic}$, for the same $\wt{\chi}_1$ as in that experiment, we additionally 
carried out a two-phase minimization, 
in its first phase minimizing $b\to\wh{\msq}_n(b',b)(\wt{\chi}_1)$ and in the second 
$b\to\wh{\ic}_n(b',b)(\wt{\chi}_1)$, using the first-phase minimization result as a starting point. 
The IS drifts corresponding to the IS parameters computed in the 
above experiments are shown in Figure \ref{fig3WFirst}, % and for the second in Figure \ref{fig3WFirstIS}. 
in which we also show an approximation of the zero-variance IS drift $r^*$ for the corresponding diffusion problem as in the previous section. 

\begin{figure}[h] %
\centering
\subfloat[]{\includegraphics[width=0.45\textwidth]{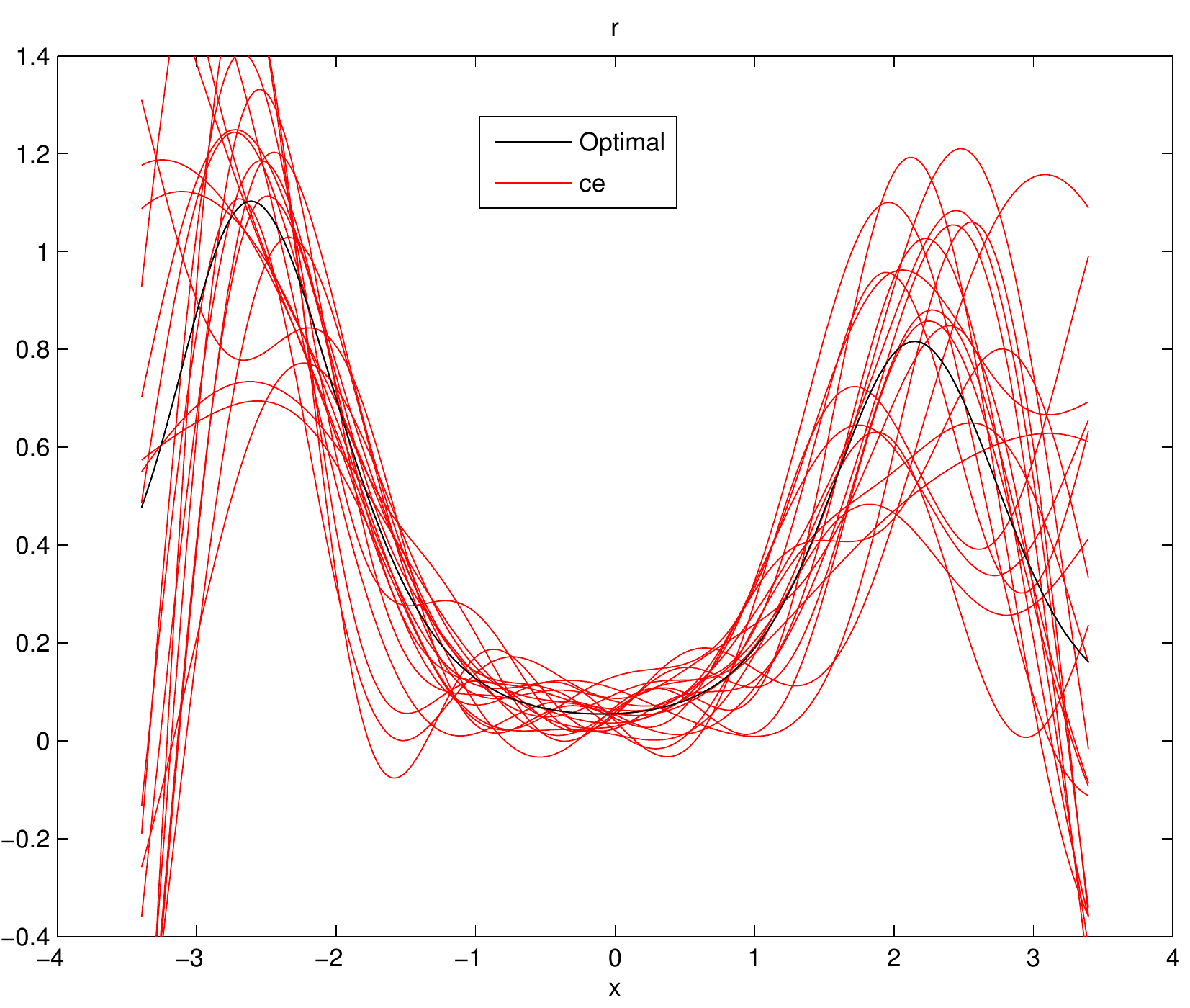}}
\qquad  
\subfloat[]{\includegraphics[width=0.45\textwidth]{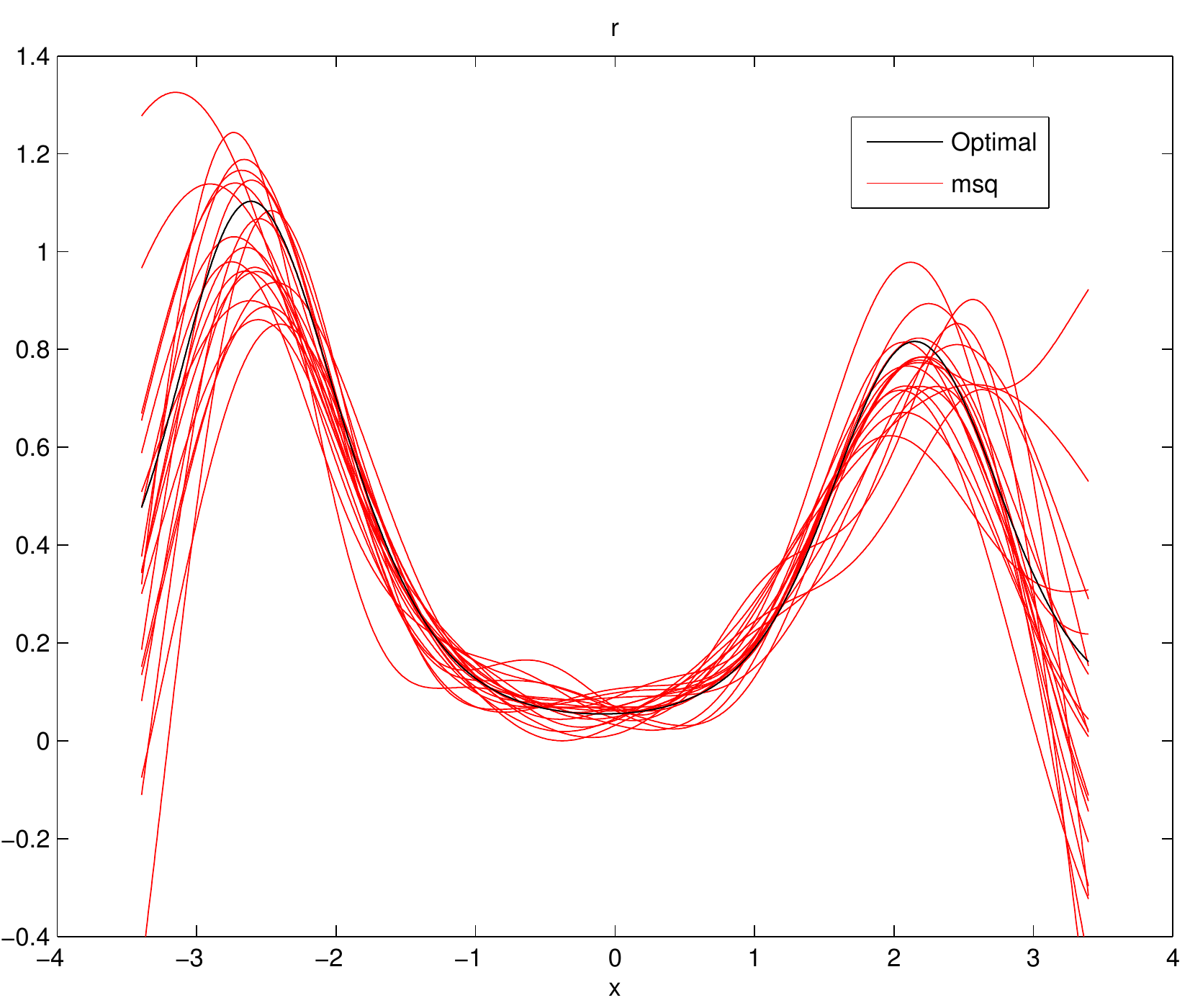}} 
\\
\subfloat[]{\includegraphics[width=0.45\textwidth]{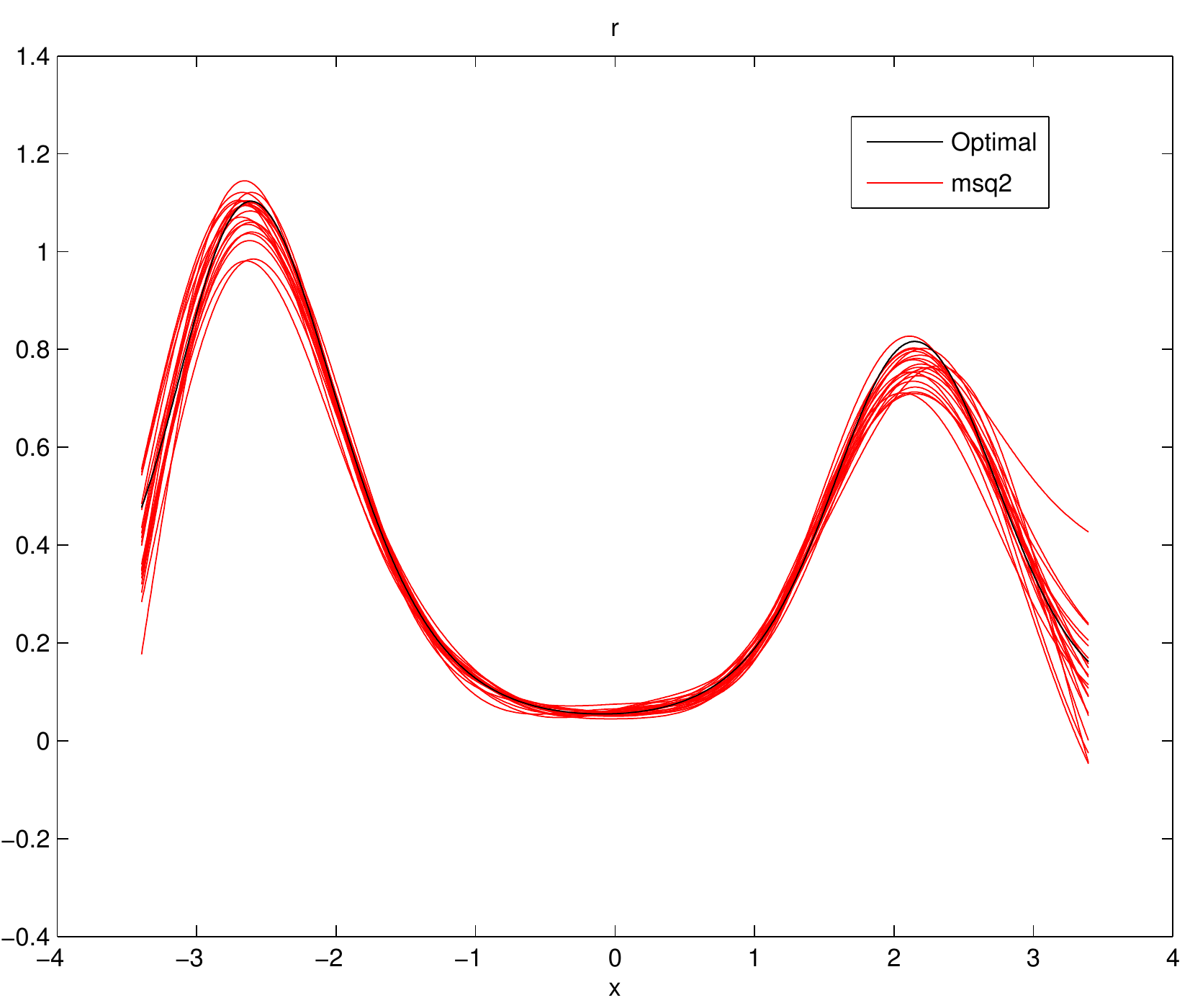}}
\qquad  
\subfloat[]{\includegraphics[width=0.45\textwidth]{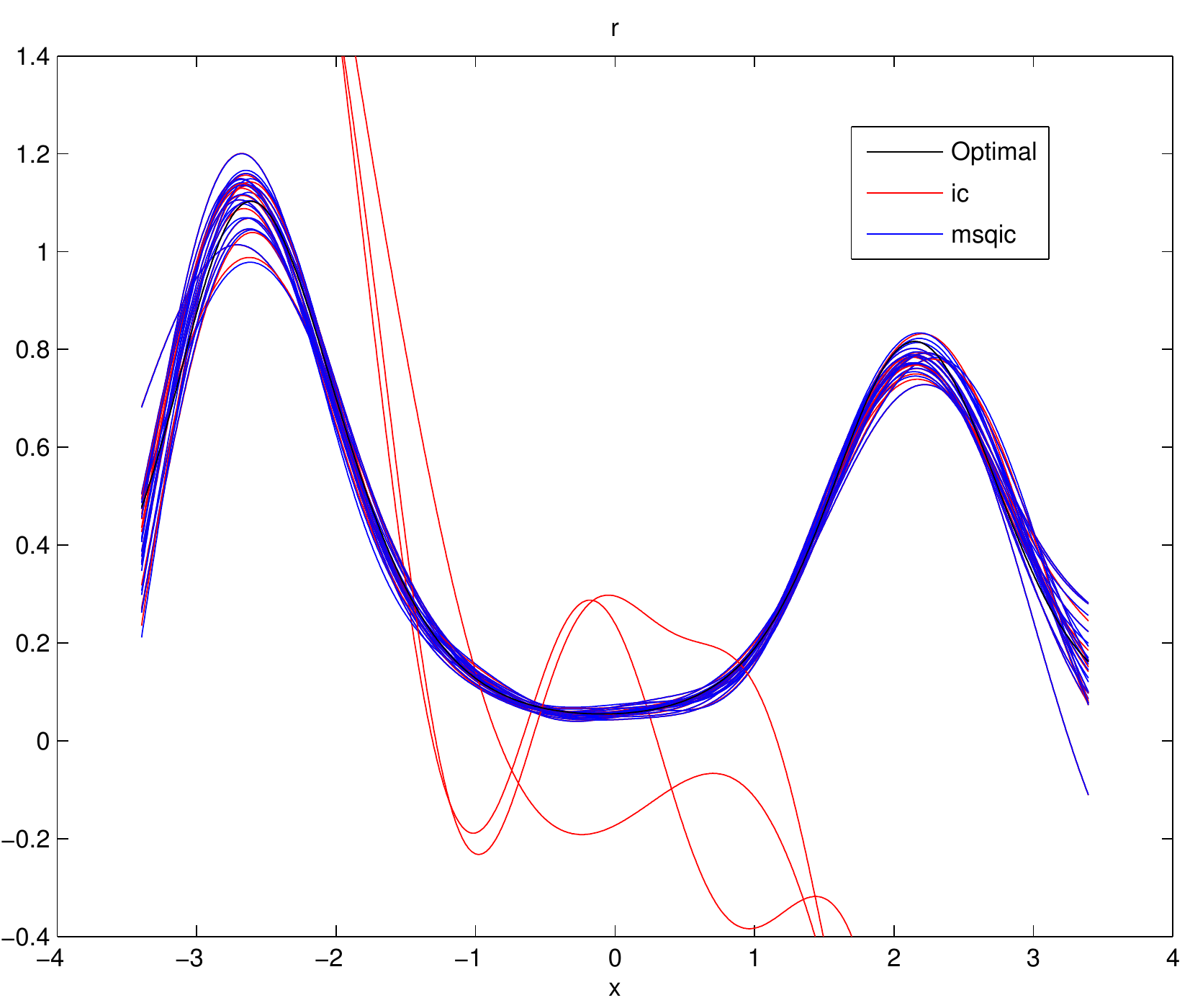}}
%\subfloat[]{\includegraphics[width=0.45\textwidth]{./IterICDSampleDrift3WCom2First100_20.pdf}}
%\\
% \subfloat[]{\includegraphics[width=0.45\textwidth]{./IterCESampleDrift3WCom2First100_20IS2.pdf}}
% %\subfloat[]{\includegraphics[width=0.45\textwidth]{./IterICDSampleDrift3WCom2First100_20IS.pdf}}
% \qquad 
% \subfloat[]{\includegraphics[width=0.45\textwidth]{./IterICDSampleDrift3WCom2First100_20IS.pdf}}
 % \subfloat[]{\includegraphics[width=0.45\textwidth]{./IterICDSampleDrift3WCom2First100_20DoubleMsq.pdf}}

\caption{\label{fig3WFirst} The IS drifts from SSM for estimating $q_{2,a}$, minimizing $\wh{\ce}$ in (a),
$\wh{\msq}$ in (b), $\wh{\msq2}$ in (c), and $\wh{\ic}$ in (d). The ''optimal`` IS drift is a finite-difference approximation of the zero-variance IS drift $r^*$.
The label
'msqic' in (d) corresponds to the two-phase minimization and 'ic' to the single-phase minimization as discussed in the main text.}
\end{figure}

%and for the second series the drifts coming from MSM as above. 
From Figure \ref{fig3WFirst} (d) it can be seen 
that ordinary (i.e. single-phase) SSM using $\wh{\ic}$ yielded in three experiments IS drifts far from 
$r^*$, while in the other $17$ experiments and in all $20$ experiments when using two-phase minimization we received drifts 
close to $r^*$. In the experiments in which ordinary minimization led to drifts far from $r^*$, the value of 
$\wh{\ic}_n(b',b)(\wt{\chi}_1)$  in the minimization result $b$ was several 
times smaller than when using two-phase minimization, while in the other 
cases these values were very close 
(e.g. the absolute value of difference of such values divided by the 
smaller of them was in each case below $1\%$). % and in most cases below $0.1\%$). 
In Figure \ref{fig3WFirst}, the IS drifts from the SSM of $\wh{\msq2}$ and the two-phase minimization 
minimizing $\wh{\ic}$ in the second phase as above seem to be the least spread, followed by the ones 
from the minimization of $\wh{\msq}$, and finally $\wh{\ce}$. 
% seem to be most spread, followed 
% by the ones from minimization of $\wh{\msq}$, and the ones from minimization of $\wh{\msq2}$ and 
%  $\wh{\ic}$ seem to be least spread (where for $\wh{\ic}$ we mean drifts from two-phase minimization).  

From the below remark it can be expected that for a sufficiently large $n$, the IS drifts from SSM experiments like above 
should have approximately normal distribution in each point. 
\begin{remark}\label{rempropr} 
Let $A$, $T$, $d_t$, and $r_t$ be as in Section \ref{secHelpAsymp} and 
let for some covariance matrix $D\in \R^{l\times l}$ it hold
\begin{equation}\label{rtdt}
\sqrt{r_t}(d_t-b^*) \Rightarrow \ND(0,D).  
\end{equation}
%let $b\in A$, and 
%$u:A\to\overline{\R}$ be such that $u(b)\in \R_+$ and for some $V_{g}(b)\in \R^{l\times l}$ 
%we have $r_t(d_t-b^*) \Rightarrow \ND(0,u(b)V_{g}(b))$. 
This holds e.g. for $d_t$ being the SSM or MSM results of the estimators $\wh{g}$ 
for $g$ replaced by $\ce$, $\msq$, $\msq2$, or $\ic$, for the SSM 
under the assumptions as in Section \ref{secAsympSSM} for $D=u(b')V_{g}(b')$ and $r_t=t$, $t\in T$, while 
for the MSM under the assumptions as in Section 
\ref{secAsympMSM} for $D=V_g(d^*)$, $T=\N_+$, and $r_k=n_k$, $k \in T$. 
Let us assume a linear parametrization of the IS drifts as in (\ref{rbxdef}), let $x \in \R^m$, and let $B\in \R^{l\times d}$ 
be such that $B_{i,j}=(\wt{r}_{i})_j(x)$, $i=1,\ldots,l$, $j=1\ldots,d$. 
Then, $r(d_t)(x)=B^Td_t$, $t \in T$, and from (\ref{rtdt})
\begin{equation}\label{rdn0} 
\sqrt{r_t}(r(d_t)-r(b^*))(x)\Rightarrow \ND(0, B^TDB). 
\end{equation} 
\end{remark}

%From the above remark we can expect the spread of 
%From the above remark it is natural to expect that for a sufficiently large $n_1$ the distributions 
%of IS drifts from SSM of different  
In the further experiments, to 
be able to carry out more simulations in a reasonable time, 
we changed the model considered by increasing the temperature $10$ times. %TODO why?
For such a new temperature we received an estimate $1.468 \pm 0.002$ of the mean cost $h\E_{\PU}(\tau)$ in CMC, as compared to $41.44\pm 0.15$ 
under the original temperature as in Table \ref{tabResults}. 
%We computed interquartile ranges (IQRs) of values at zero of 
We carried out an MSM procedure of $\wh{\msq2}$ for $k=6$ and $n_i=50\cdot 2^{i-1}$, $i=1,\ldots,k$, receiving the final minimization result $b_{MSM}$. 
For $\wh{\est}$ equal to each of the
different estimators as above and for $b'=0$ and $b'=b_{MSM}$, we carried out independently $N=5000$ times the
SSM of  $\wh{\est}$ for $n_1=200$, in the $i$th SSM receiving a result $a_i$ and then computing
$r(a_i)(0)$, i.e. the corresponding IS drift at zero, $i=1,\ldots,N$. 
The histograms of such IS drifts at zero for $b'=0$, with fitted Gaussian functions, are shown in Figure \ref{fig3WHists}. 
\begin{figure}[h] %
\centering
\subfloat[]{\includegraphics[width=0.45\textwidth]{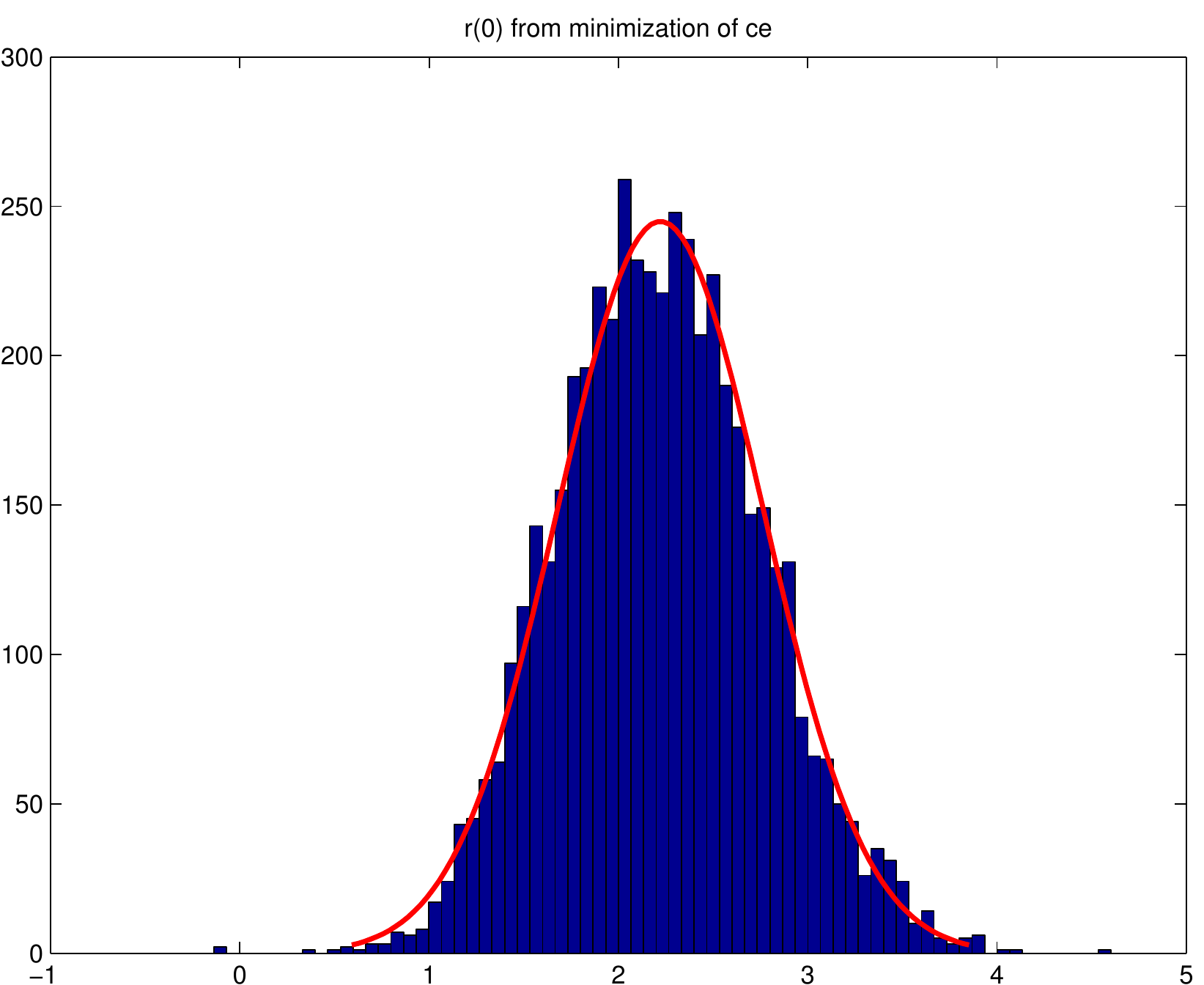}}
\qquad  
\subfloat[]{\includegraphics[width=0.45\textwidth]{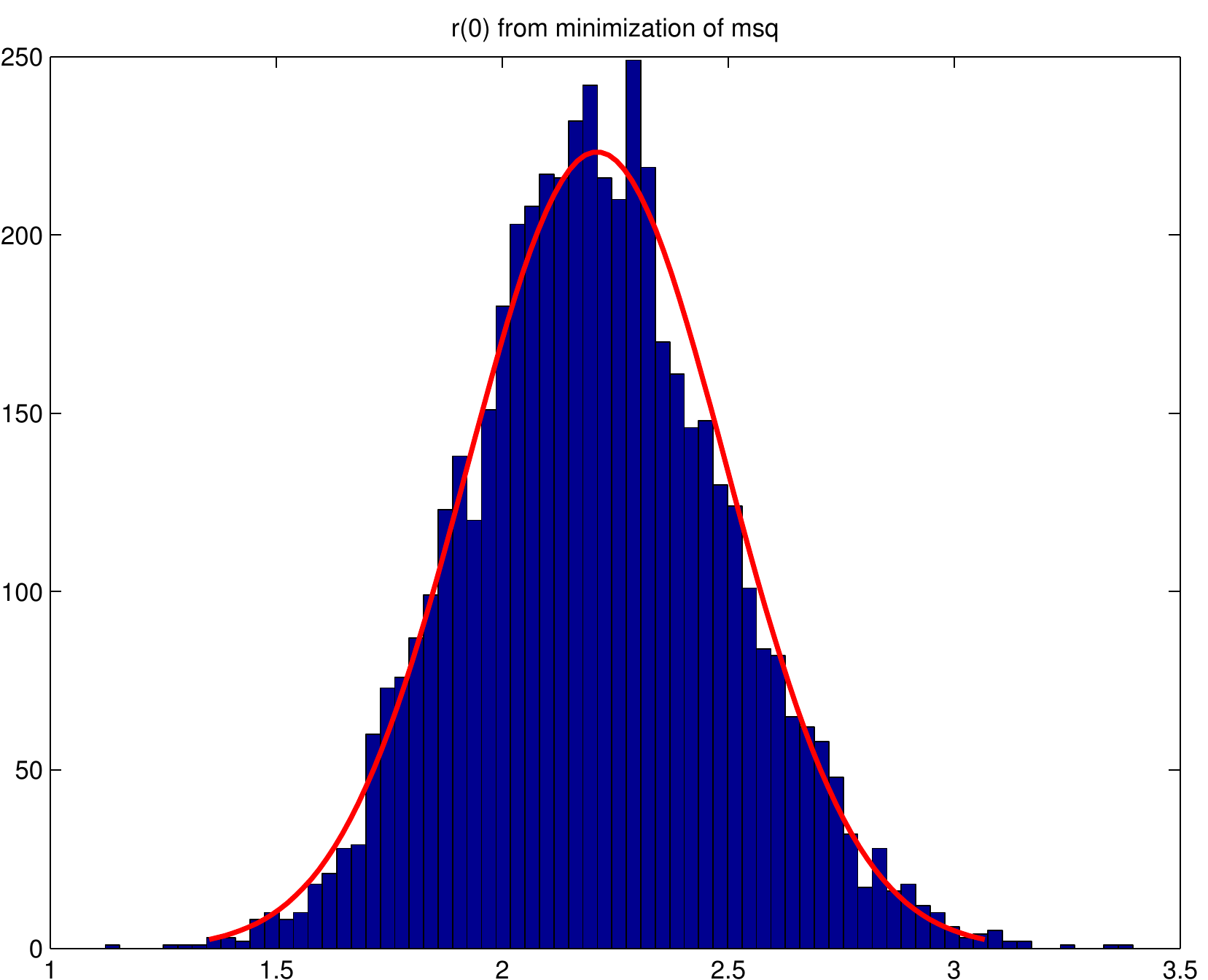}} 
\\
\subfloat[]{\includegraphics[width=0.45\textwidth]{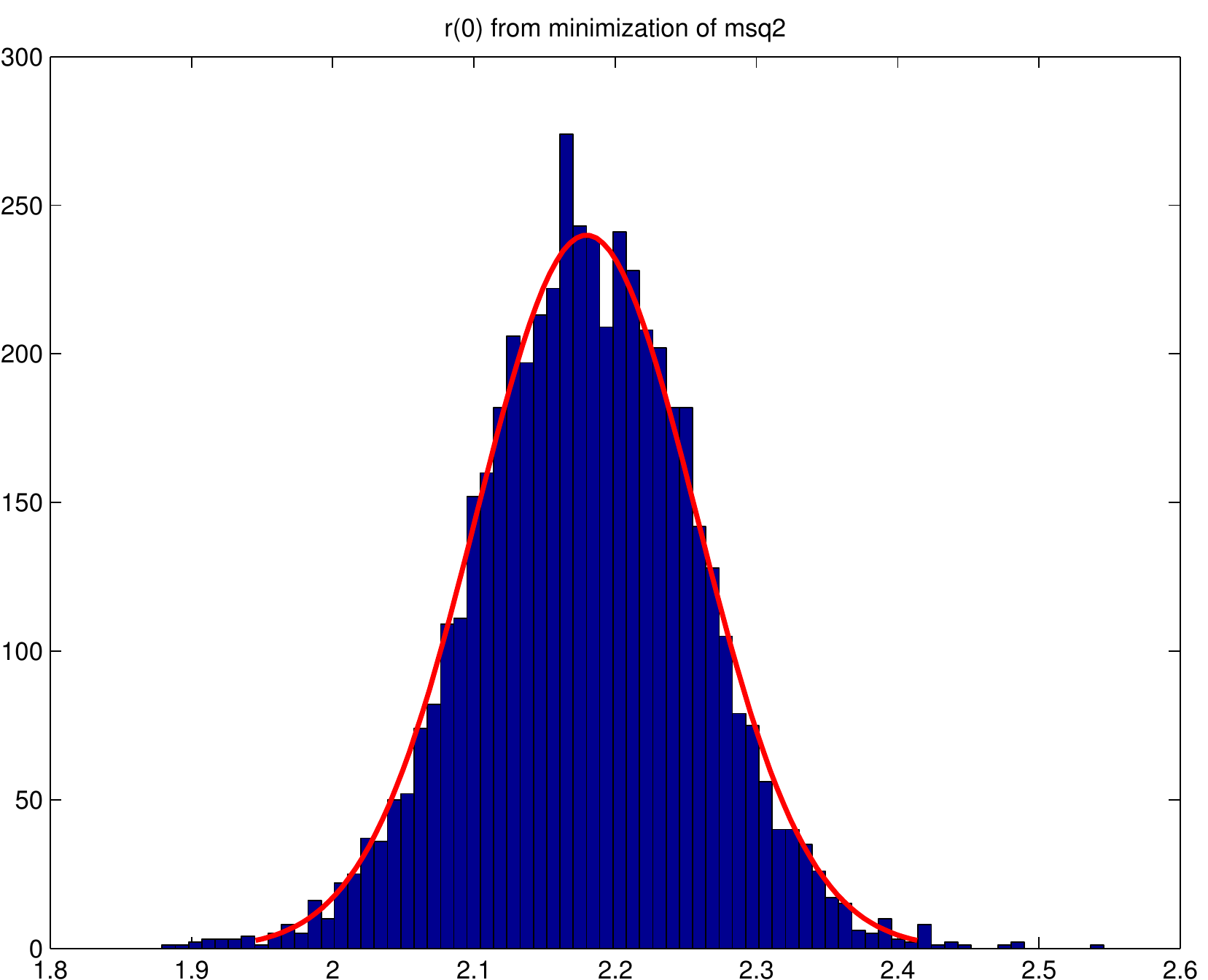}}
\qquad  
\subfloat[]{\includegraphics[width=0.45\textwidth]{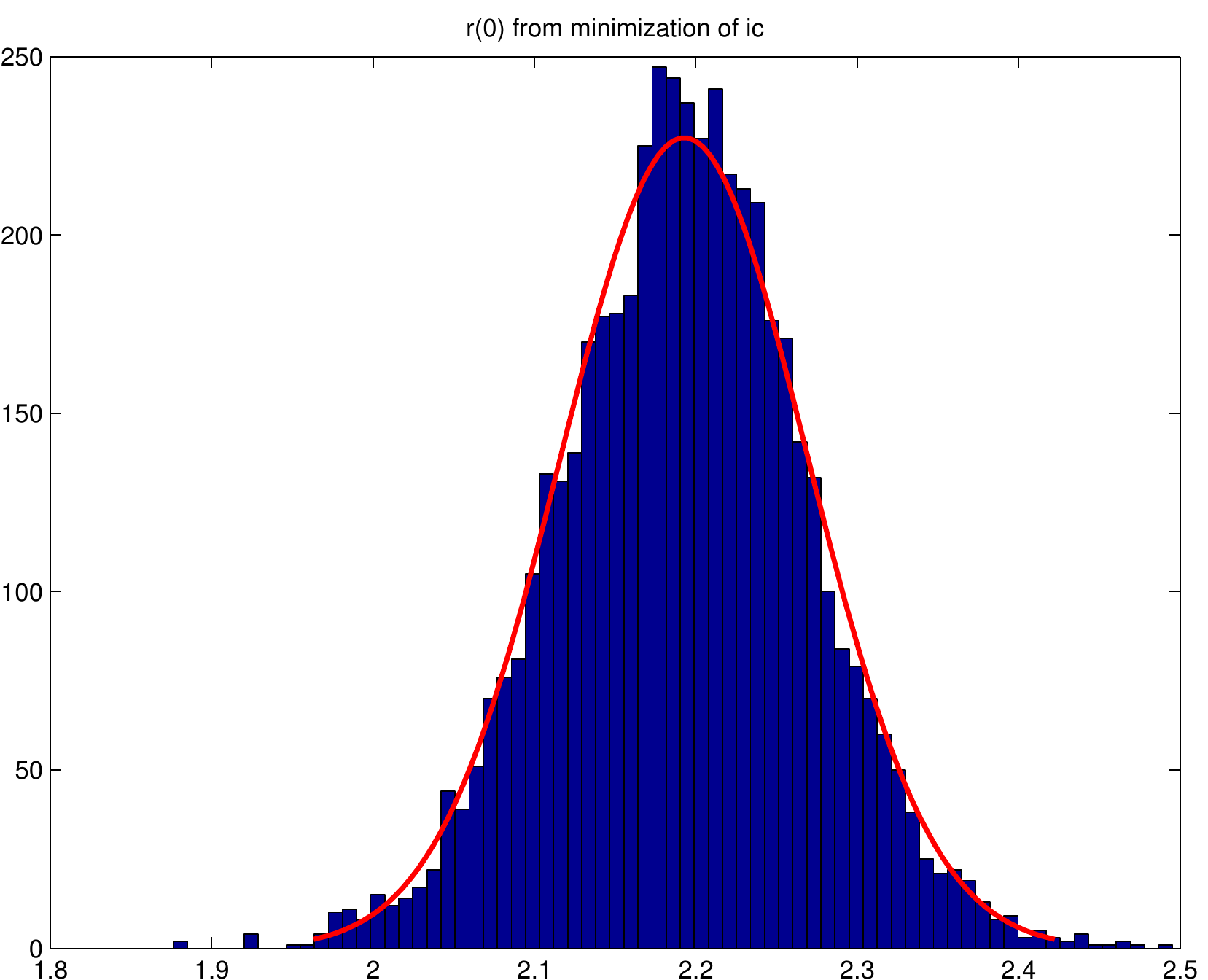}}
%\subfloat[]{\includegraphics[width=0.45\textwidth]{./IterICDSampleDrift3WCom2First100_20.pdf}}
%\\
% \subfloat[]{\includegraphics[width=0.45\textwidth]{./IterCESampleDrift3WCom2First100_20IS2.pdf}}
% %\subfloat[]{\includegraphics[width=0.45\textwidth]{./IterICDSampleDrift3WCom2First100_20IS.pdf}}
% \qquad 
% \subfloat[]{\includegraphics[width=0.45\textwidth]{./IterICDSampleDrift3WCom2First100_20IS.pdf}}
 % \subfloat[]{\includegraphics[width=0.45\textwidth]{./IterICDSampleDrift3WCom2First100_20DoubleMsq.pdf}}
\caption{\label{fig3WHists} Histograms of the IS drifts at zero from the SSM for computing $q_{2,a}$, minimizing $\wh{\ce}$ in (a),
$\wh{\msq}$ in (b), $\wh{\msq2}$ in (c), and $\wh{\ic}$ in (d).}
\end{figure}
This figure suggests that the distributions 
of the IS drifts at zero are approximately normal, as could be expected from Remark \ref{rempropr}. 
Furthermore, the (empirical) distribution of the IS drifts at zero for $b'=0$ seems to be in a sense the least 
spread when minimizing $\wh{\msq2}$ and $\wh{\ic}$, followed by $\wh{\msq}$, and finally $\wh{\ce}$.
The same observations can be made from the inspection of histograms for the case of $b'=b_{MSM}$, which are not shown. 
We shall now compare quantitatively the spread of empirical distributions of the IS drifts at zero 
in the above experiments for the different estimators and $b'$ used, 
using interquartile ranges (IQRs), 
the definition and some required properties of which are provided in the below remark. 

\begin{remark}\label{remIQR}
For i.i.d. random variables $X_1, X_2, \ldots$, and $k,n \in \N_+$, 
let $X_{k:n}$ be the $k$th coordinate of $\wt{X}_n:=(X_i)_{i=1}^n$ in the nondecreasing order. 
For $n \geq 4$, let us define the interquartile range (IQR) of the coordinates of $\wt{X}_n$ as 
$\wh{\IQR}_n=X_{\lfloor\frac{3n}{4}\rfloor:n} -X_{\lfloor\frac{n}{4}\rfloor:n}$.  
Let further $X_1\sim \ND(\mu,\sigma^2)$ for some $\mu \in \R$ and $\sigma\in \R_+$, 
and let $q$ denote the IQR of $\ND(\mu,\sigma^2)$ (i.e. the difference of its third and first quartile).  
Then, for a certain $d \approx 1.36$ we have  $\sqrt{n}(\wh{\IQR}_n-q) \Rightarrow\ND(0,dq^2)$
(see page 327 in \cite{dasgupta2011probability}). Thus, for $\wh{\sigma}_n=\sqrt{d}\wh{\IQR}_n$ we have 
$\frac{\sqrt{n}}{\wh{\sigma}_n}(\wh{\IQR}_n-q)\Rightarrow \ND(0,1)$, 
which can be used for constructing asymptotic confidence intervals for $q$. 
For some $\mu' \in \R$ and $\sigma'\in \R_+$,
consider further $X_1', X_2'\ldots,$ i.i.d. $\sim \ND(\mu', (\sigma')^2)$, such that
$(X_i')_{i\in \N_+}$ is independent of $(X_i)_{i\in \N_+}$, and let $q'\in\R_+$ be the IQR of $\ND(\mu', (\sigma')^2)$. 
Then, for $\wh{\IQR}_n'$ analogous as above but for the primed variables
we have $(\wh{\IQR}_n-q,\wh{\IQR}_n'-q')\Rightarrow \ND(0,d\diag(q^2,(q')^2))$.
Thus, for $R=\frac{q}{q'}$, $\wh{R}_n=\frac{\wh{\IQR}_n}{\wh{\IQR}_n'}$, and $\sigma_{R}=R\sqrt{2d}$, from the
delta method we have 
$\sqrt{n}(\wh{R}_n-R) \Rightarrow \ND(0,\sigma_R^2)$. Therefore, for $\wh{\sigma}_{R,n}=\wh{R}_n\sqrt{2d}$ we have
$\frac{\sqrt{n}}{\wh{\sigma}_{R,n}}(\wh{R}_n-R) \Rightarrow \ND(0,1)$.
\end{remark}
For $X_i=r(a_i)(0)$, $i=1,\ldots,N$, received from the SSM of different estimators as above, we computed the estimates $\wh{\IQR}_N$ 
of the IQRs of drifts at zero and the values $\frac{\wh{\sigma}_N}{\sqrt{N}}$ 
as in Remark \ref{remIQR}. The results are provided in Table \ref{tabMeanspreads}. 

\begin{table}[h]
%\resizebox{14cm}{!} {
\begin{tabular}{|l|c|c|c|c|c|}
\hline
& $\wh{\ce}$ &$\wh{\msq}$& $\wh{\msq2}$& $\wh{\ic}$ \\
\hline
$b'=0$& $0.736 \pm 0.012$ & $0.3770 \pm 0.0062$ & $0.1039 \pm 0.0017$ & $0.1006 \pm 0.0017$\\
\hline 
$b'= b_{MSM}$& $0.546 \pm 0.009$ & $0.2910 \pm 0.0048$ & $0.0932 \pm 0.0015$ & $0.0929 \pm 0.0015$ \\
\hline
% \hline
% $b'=0$& $0.741 \pm 0.019$ & $0.391 \pm 0.010$ & $0.1027 \pm 0.0027$ & $0.1031 \pm0.0027$\\
% \hline 
% $b'\neq 0$& $0.549 \pm 0.014$ & $0.2830 \pm 0.0074$ & $0.0937 \pm 0.0024$ & $0.0902 \pm 0.0024$\\
% \hline
\end{tabular}
% mean t: =6.2173 \pm 0.0021
%}
\caption{\label{tabMeanspreads} Estimates of the IQRs of the IS drifts at zero and the values $\frac{\wh{\sigma}_N}{\sqrt{N}}$ from the 
SSM of various estimators.} 
\end{table} 
From this table we can see that for the both values of $b'$
the computed estimates of IQRs from the minimization of $\wh{\ic}$ and $\wh{\msq2}$ are the lowest, followed by such 
estimates from the minimization of $\wh{\msq}$, and finally $\wh{\ce}$. 
The ratio of the estimates of IQRs from the minimization of $\wh{\ce}$ to $\wh{\msq}$ is
$1.951\pm  0.045$ for $b'=0$ and $1.8771 \pm   0.044$ for $b'=b_{MSM}$ 
(where the results are provided in the form $\wh{R}_n\pm \frac{\wh{\sigma}_{R,N}}{\sqrt{N}}$ under appropriate identifications with 
the variables from Remark \ref{remIQR}).  
Note that these ratios are close to $2$. Intuitions supporting the above results are given in Section \ref{secIntu}. 
Note also that the estimates of IQRs are lower when using $b'=b_{MSM}$ than $b'=0$.

\section{\label{secIntu}Some intuitions behind certain results of our numerical experiments} 
Recall that in the numerical experiments in Section \ref{secExpEst} we observed the fastest convergence of the IS drifts in the MSM 
results and in Section \ref{secSpread} the lowest spreads of such drifts in the SSM results 
when minimizing $\wh{\msq2}$ and $\wh{\ic}$, followed by $\wh{\msq}$, and finally $\wh{\ce}$. 
Furthermore, in Section \ref{secSpread} the 
IQRs of the values at zero of the IS drifts corresponding to the SSM results were approximately two times higher when minimizing $\wh{\ce}$ than $\wh{\msq}$. 
In this section we provide some intuitions behind these and some other of our experimental results. We will need the following remark. 
\begin{remark}\label{remZeroR}
Let us assume that, similarly as in Section \ref{secAsympProp}, for $b^*$ being a zero-variance IS parameter, 
for each $b \in \R^l$ we have 
$V_{\msq2}(b)=V_{\ic}(b)=0$, $V_{\ce}(b)=4V_{\msq}(b)$, and $V_{\ce}(b)$ is positive definite. 
Let further, similarly as in Remark \ref{rempropr}, for $d=1$, for $g$ replaced by each of $\ce$, $\msq$, $\msq2$, and $\ic$, 
for $x \in \R^m$ and $B=((\wt{r}_{i})(x))_{i=1}^l$, for 
$u(b')\in \R_+$, $r_t=t$, and $v_g=u(b')B^TV_g(b')B$ for SSM or 
$r_k=n_k$ and $v_g=B^TV_g(d^*)B$ for MSM, the IS drifts corresponding to the SSM or MSM results $d_t$ of the
estimators $\wh{g}$ respectively fulfill 
\begin{equation}
\sqrt{r_t}(r(d_t)-r(b^*))(x)\Rightarrow \ND(0,v_g).
\end{equation}
Then, for $g$ replaced by $\msq2$ or $\ic$ we have $v_g=0$ and the distribution $\ND(0,v_g)$ has zero IQR. If
further $\wt{r}_i(x)\neq 0$ for some $i\in\{1,\ldots,l\}$, then $0<v_{\ce}=4v_{\msq}$, so that $\ND(0,v_{\ce})$ has a positive IQR, which is exactly 
two times higher than the IQR of $\ND(0,v_{\msq})$. 
\end{remark} 

A possible reason why we received the above mentioned experimental results is that we can have approximately the 
same relations as in the above remark between the matrices $V_g(b)$  for the appropriate $b$
in our experiments, that is the entries of $V_g(b)$ can be much smaller 
for $g$ equal to $\msq2$ and $\ic$ than $\msq$ and $\ce$, and we can have $V_{\ce}(b) \approx 4V_{\msq}(b)$. 
This would lead to approximately the same relations between the asymptotic variances of the IS drifts in different points and the IQRs 
of their asymptotic distributions as as in the above remark. 

Such approximate relations between the matrices $V_g(b)$ can be a consequence of 
the IS distributions and densities corresponding to the minimum points of the minimized functions being close to the zero-variance ones, 
in the sense that the derivations as in Section \ref{secAsympProp} can be carried out approximately. 
For the estimation problems for whose diffusion counterparts there exist zero-variance IS drifts, like for the 
case of the translated committors and MGF, we also have the following possible intuition behind the hypothesized approximate 
relations between the matrices $V_g(b)$ as above. 
For the diffusion counterparts of these estimation problems, the zero-variance IS drifts 
minimize the mean square, inefficiency constant, and cross-entropy among all the appropriate drifts. 
Furthermore, as evidenced in Figure \ref{fig3WOpt}, the diffusion zero-variance IS drifts can be 
approximated very well using linear combinations of the IS basis functions considered. Thus, the diffusion IS drifts corresponding to the 
minimizers of the functions considered are likely to be close to the zero-variance ones. Therefore, using 
such drifts in the place of the zero-variance ones, the derivations as in Section \ref{secAsympProp} can be carried out approximately 
and we should have approximately the same relations between the matrices $V_g(b)$ for the diffusion case as in Remark \ref{remZeroR}. 
For small stepsizes $h$, like the ones used in our numerical experiments, the matrices $V_g(b)$ 
for the Euler scheme case can be expected to be close to their diffusion counterparts and thus 
we should also have approximately the same relations between them as above. 

For small stepsizes we can also expect the IS drifts corresponding to the minimizers of the functions 
considered for the Euler scheme case to be close to their diffusion counterparts, and thus, from the above discussion, also close to the 
diffusion zero-variance IS drifts. This would provide an intuition why in Figure \ref{fig3WOpt} the IS drifts from the minimization of 
various estimators of the functions considered are close to the approximations of the zero-variance IS drifts 
for the diffusion case. 

In the experiments from Section \ref{secExpEst} for computing $p_{T,a}(x_0)$, 
the MSM results of $\wh{\ic}$ led to a lower estimate of the inefficiency constant 
than these of $\wh{\msq2}$, at the same time yielding a higher estimate of the variance and a lower of the mean cost. 
A possible intuition behind these results is provided by Theorem \ref{thicvar}, from which it follows that 
under appropriate assumptions %and discussion in Section \ref{secECG} 
a.s. we eventually should have such relations for the corresponding functions evaluated on
some parameters converging a.s. to the minimum point of the mean square 
and the ones minimizing the inefficiency constant (see Section \ref{secECG} for some sufficient assumptions). 
Note, however, that this intuition fails when comparing the estimates of the variances in the minimization results of $\wh{\msq}$ 
and $\wh{\ic}$, as the latter were smaller in all of our estimation experiments. 
A possible factor that could have contributed to the fact that in Section \ref{secExpEst} we 
obtained the lowest estimates of the inefficiency constants and variances when minimizing the new estimators $\wh{\ic}$ and $\wh{\msq2}$, followed by 
$\wh{\msq}$, and $\wh{\ce}$, 
is that, from the above hypothesis on the approximate relations of the matrices $V_g(b)$,
we may have the lowest spread of the distributions of the minimization results of the new estimators around the minimum points of 
variances and inefficiency constants, followed by such results for $\wh{\msq}$, and $\wh{\ce}$. 
We suspect that if sufficiently long minimization 
methods are performed (i.e. for a sufficiently large $n_1$ for SSM or $k$ for MSM), so that the distributions of the minimization results of 
the estimators considered
become much less spread around the minimum points of their corresponding functions, then, as suggested by Theorem \ref{thicvar}, 
the minimization results of $\wh{\msq}$ should typically lead to lower variance 
than these of $\wh{\ic}$. However, if the above hypothesis on the entries of $V_{\msq}(b)$ being much smaller 
than these of $V_{\msq2}(b)$ is correct, then, 
for a longer minimization, the minimization results of $\wh{\msq2}$ should still typically lead to lower variance than these of $\wh{\msq}$. 
This is because such results $d_t$ for $\wh{\msq2}$ would be asymptotically much more efficient for the minimization of variance 
in the different second-order senses discussed in Section \ref{secSecond}. For instance, in the sense of 
the mean of the asymptotic distribution of $r_t(\msq(d_t)-\msq(b^*))$ (for the appropriate $r_t$), equal to 
$\frac{u(b')}{2}\Tr(V_g(b')H_{\msq})$ for SSM or $\frac{1}{2}\Tr(V_g(d^*)H_{\msq})$ for MSM, 
being much smaller for $g$ equal to $\msq2$ than $\msq$. 
Apart from the highest spread of the distributions of the minimization 
results when minimizing $\wh{\ce}$, 
another factor that could have contributed to the higher estimates of the variances 
in the minimization results of $\wh{\ce}$ than for the mean square estimators in our experiments
is that the minimum points of the cross-entropy functions are likely to be different 
from the ones of the mean square functions, so that, 
as discussed in Section \ref{secCompFirst}, in such cases minimizing the mean square estimators 
can be more efficient for the minimization of variance in the first-order sense. 

\chapter{Conclusions and further ideas}\label{secConcl} 
In this work we developed methods for obtaining the parameters of the 
IS change of measure adaptively via single- and multi-stage minimization 
of well-known estimators 
of cross-entropy and mean square, as well as of new estimators of mean square 
and inefficiency constant, 
ensuring their various convergence and asymptotic properties 
in the ECM and LETGS settings. 
It would be interesting to prove such properties 
of our methods, or some their modifications, using some other 
parametrizations of IS;
see e.g. \cite{Lemaire2010, Rubinstein_2004} for some examples. 

We proposed criteria for comparing the first- and second-order asymptotic 
efficiency of certain stochastic optimization methods of functions, 
which for such functions being equal to inefficiency constants 
can be used for comparing the efficiency of 
methods for finding the adaptive parameters in the first stage of a two-stage 
estimation method as in Chapter \ref{secTwo}. 
We also derived formulas for measures of the second-order asymptotic inefficiency 
of the above minimization methods of estimators. 

Let us now discuss some problems which one can face when trying to use in practice 
the minimization methods 
for the results of which we proved strong convergence 
and asymptotic properties, 
as well as possible solutions to 
these problems.
When using gradient-based stopping criteria in some of these methods, 
one has to choose some nonnegative 
random bounds $\epsilon_i$ or $\wt{\epsilon}_i$ on the norms of the gradients 
in the minimization results, 
converging to zero a.s. (or, equivalently, ensure that these gradients 
converge to zero a.s.). 
If chosen too large, such bounds can make the minimization algorithm 
perform in practice no steps at all, and if taken too small, 
they can make the algorithm run longer than it can be afforded. 
To ensure the a.s. convergence of the gradients to zero 
in the MSM methods and that a reasonable computational 
effort is made
by the minimization algorithm in each stage, for some $q\in (1,\infty)$,
one can perform at least a fixed number of steps of the minimization 
algorithm plus an additional number of steps needed to make the norm of the gradient 
at least $q$ times smaller 
than in the most recent step in which the final gradient was nonzero (assuming that such a step exists). 

As discussed in Remark \ref{remdsbs}, under appropriate assumptions, to ensure that $b_i \overset{p}{\to}b^*$ in the MSM methods
one can choose appropriate sets $K_i$ containing the variables $b_i$ and such that $b_i$ is equal to the $i$th minimization result
$d_i$ whenever $d_i \in K_i$.
If for some $m \in \N_+$ the sets 
$K_i$ contain  $b^*$ only for $i\geq m$, then the 
convergence of $b_i$ to $b^*$ may be very slow until $i$ 
exceeds such an $m$. One  can try to deal with this problem by performing some preliminary SSM or MSM until 
the sequence of the minimization results has approximately converged to $b^*$ 
and then taking in a new MSM all the sets $K_i$ containing some neighbourhood of the computed approximation of $b^*$. 

As discussed in Section \ref{secAsympMSM}, %and in Remark \ref{remMinMean}, 
as an alternative to 
using in MSM methods variables $b_i$ converging to $b^*$ minimizing the function $f$ considered, 
it may be reasonable to choose such $b_i$ converging to 
some $d^*$ minimizing some measure of the second-order asymptotic inefficiency of $d_k$ for the minimization of $f$,  
assuming that such a $d^*$ exists. 
Such variables $b_i$ could be potentially obtained 
by minimizing some estimators of such a measure. %It would be interesting to check 
A similar idea is to use as the parameter $b'$ in SSM methods
an estimate of $d^*$ minimizing the measure of inefficiency (\ref{quantSSM}).  

For IS in which the mean theoretical cost is not constant in the function of the 
IS parameter, minimizing the inefficiency constant estimator can be asymptotically the best option 
as under appropriate assumptions it can outperform the minimization of the other estimators in terms 
of the first-order asymptotic efficiency for the minimization 
of the inefficiency constant (see Section \ref{secCompFirst}). 
However, if the mean cost does not depend on the IS parameter, 
so that the inefficiency constant is proportional to the variance, then the 
minimization results of all the mean square and inefficiency constant estimators considered 
can converge to the minimum point of variance, in which  case
minimizing them is asymptotically equally efficient for the minimization of variance in the first-order sense. 
In such a case it may be reasonable to minimize the estimators whose minimization results are the 
most efficient for the minimization (e.g. using SSM or MSM) of the variance in the second-order sense, 
as discussed in Section \ref{secSecond}. 
A possible idea is to estimate the 
measures of the second-order asymptotic inefficiency of different estimators for the minimization of variance, 
which can be combined with the estimation of the parameters $d^*$ minimizing 
such measures as discussed above. The estimators, and potentially also the estimate of $d^*$ as above, leading to the lowest estimates 
of the inefficiency measure, can be later used in a separate SSM or MSM procedure. 

In our numerical experiments, using different IS basis functions and added constants $a$ led to considerably  
different inefficiency constants of the adaptive IS estimators. It would be interesting 
to develop adaptive methods for choosing such basis functions 
and constants. For instance, the added constant $a$ can be chosen adaptively 
via minimization of the estimators of variance or inefficiency constant in which such an $a$ is treated as an additional minimization parameter. 

In MSM, an alternative approach to the minimization of the estimators constructed using only the samples from the last stage, 
as in this work, would be to minimize some weighted average of such estimators from all the previous stages. 
In our initial numerical experiments, minimizing such averages typically yielded drifts farther 
from the approximations of the zero-variance IS drifts for the corresponding diffusions than 
the approach from this work (data not shown), which is why we focused on the current approach. 
Similarly, the mean $\alpha$ of interest could be estimated using a weighted average of the estimators from all the stages, which closely resembles 
the purely adaptive approach used in stochastic approximation methods \cite{KimH07, arouna04, Lemaire2010, Lapeyre2011}. 
For instance, under the assumptions as in Section \ref{secSimpleMin} and denoting $s_k =\sum_{i=1}^kn_i$, such an estimator of $\alpha$ from the $k$th 
stage could be $\frac{1}{s_k}\sum_{l=1}^{k}\sum_{i=1}^{n_l}(ZL(b_{l-1}))(\chi_{l,i})$. 
An SLLN and CLT for such an estimator can be proved similarly as in \cite{Lapeyre2011}. 

The model which we used for the numerical experiments in this work is only a toy one. 
It would be interesting to test and compare the performance of our minimization methods of different estimators on some realistic molecular models,
as well as on models arising in some other application 
areas of IS sampling, like computational finance and queueing theory. 
When using our methods for rare event simulation in practice one should take care to choose  
the IS parameter $b$ equal to $b'$ in SSM or $b_0$ in MSM so that the considered event is not too rare under the IS distribution $\PQ(b)$.   
This is because if such an event was too rare, then it would typically
not occur at all in a reasonable simulation time. To find such a $b$ adaptively one can use e.g. some MSM method
in which the problem is modified in the initial stages to make the considered event less rare in these stages as in 
\cite{Rubinstein_optim,  Rubinstein_2004, ZhangWang2014}.

\bibliographystyle{plain}
{\footnotesize
%\bibliography{bibfile}}
\bibliography{FEC}}
\end{document}